\documentclass[11pt]{amsart}
\usepackage{amssymb}
\usepackage{amsmath}
\usepackage{mathrsfs}
\usepackage{verbatim}
\usepackage{amsthm}
\usepackage{wasysym}
\usepackage{color}
\usepackage{accents}
\usepackage[T1]{fontenc}
\usepackage{fancyhdr}
\usepackage{hyperref}

\numberwithin{equation}{subsection}

\newtheorem{proposition}{Proposition}[subsection]
\newtheorem{lemma}[proposition]{Lemma}
\newtheorem{corollary}[proposition]{Corollary}
\newtheorem{theorem}{Theorem}[section]

\theoremstyle{definition}
\newtheorem{definition}{Definition}[subsection]
\newtheorem{remark}{Remark}[subsection]

\newcommand{\eqdef}{\overset{\mbox{\tiny{def}}}{=}}

\newcommand{\gzerozeronorm}[1]{\underline{S}_{g_{00}+1;#1}}
\newcommand{\gzerostarnorm}[1]{\underline{S}_{g_{0*};#1}}
\newcommand{\hstarstarnorm}[1]{\underline{S}_{h_{**};#1}}
\newcommand{\gnorm}[1]{\underline{S}_{g;#1}}
\newcommand{\fluidnorm}[1]{\underline{S}_{\partial \Phi;#1}}
\newcommand{\totalnorm}[1]{\underline{\mathbf{S}}_{#1}}

\newcommand{\supgzerozeronorm}[1]{S_{g_{00}+1;#1}}
\newcommand{\supgzerostarnorm}[1]{S_{g_{0*};#1}}
\newcommand{\suphstarstarnorm}[1]{S_{h_{**};#1}}
\newcommand{\supgnorm}[1]{S_{g;#1}}
\newcommand{\supfluidnorm}[1]{S_{\partial \Phi;#1}}
\newcommand{\suptotalnorm}[1]{\mathbf{S}_{#1}}

\newcommand{\gzerozeroenergy}[1]{\underline{E}_{g_{00}+1;#1}}
\newcommand{\gzerostarenergy}[1]{\underline{E}_{g_{0*};#1}}
\newcommand{\hstarstarenergy}[1]{\underline{E}_{h_{**};#1}}
\newcommand{\genergy}[1]{\underline{E}_{g;#1}}
\newcommand{\fluidenergy}[1]{\underline{E}_{\partial \Phi;#1}}

\newcommand{\supgzerozeroenergy}[1]{E_{g_{00}+1;#1}}
\newcommand{\supgzerostarenergy}[1]{E_{g_{0*};#1}}
\newcommand{\suphstarstarenergy}[1]{E_{h_{**};#1}}
\newcommand{\supgenergy}[1]{E_{g;#1}}
\newcommand{\supfluidenergy}[1]{E_{\partial \Phi;#1}}
\newcommand{\suptotalenergy}[1]{\mathbf{E}_{#1}}

\newcommand{\Ent}{\eta}
\newcommand{\speed}{c_s}
\newcommand{\bard}{\bar{\partial}}

\newcommand{\decayparameter}{\varkappa}

\begin{document}

\title{The Nonlinear Stability of the Irrotational Euler-Einstein System with a Positive Cosmological Constant}

\author{Igor Rodnianski}
\thanks{I.R. was supported in part by NSF Grant DMS-0702270.}

\author{Jared Speck}
\thanks{J.S. was supported in part by the Commission of the European Communities, 
ERC Grant Agreement No 208007.}

\begin{abstract}
	In this article, we study small perturbations of the family of Friedmann-Lema\^{\i}tre-Robertson-Walker cosmological 
	background solutions to the Euler-Einstein system with a positive cosmological constant in $1 + 3$ dimensions. The background 
	solutions describe an initially uniform quiet fluid of positive energy density evolving in a spacetime undergoing accelerated 
	expansion. We show that under the equation of state $p = \speed^2 \rho,$ $0 < \speed^2 < 1/3,$ the background solutions are 
	globally future asymptotically stable under small irrotational perturbations. In particular, we prove that the perturbed 
	spacetimes, which have the topological structure $[0,\infty) \times \mathbb{T}^3,$ are future causally geodesically complete.
\end{abstract}



\date{Version of \today}
\maketitle
\setcounter{tocdepth}{2}

\pagenumbering{roman} 
\tableofcontents 
\newpage 
\pagenumbering{arabic}

\section{Introduction}

The irrotational Euler-Einstein system models the evolution of a dynamic spacetime\footnote{By spacetime, we mean a $4-$dimensional Lorentzian manifold $\mathcal{M}$ together with a metric $g_{\mu \nu}$ on $\mathcal{M}$ of signature $(-,+,+,+).$} $(\mathcal{M},g)$ containing a perfect fluid with vanishing vorticity. In this article, we endow this system with a positive cosmological constant $\Lambda$ and consider the equation of state $p=\speed^2 \rho,$ where $p$ is the fluid \emph{pressure}, $\rho$ is the \emph{proper energy density}, and the non-negative constant $\speed$ is the \emph{speed of sound}. As is fully discussed in Section \ref{S:IrrotationalEE}, under these assumptions, the irrotational Euler-Einstein system comprises the equations\footnote{Throughout the article, we work in units with 
$8 \pi G = c = 1,$ where $c$ is the speed of light and $G$ is Newton's universal gravitational constant.}

\begin{subequations}
\begin{align} 
	R_{\mu \nu} - \frac{1}{2}R g_{\mu \nu} + \Lambda g_{\mu \nu} & = T_{\mu \nu}^{(scalar)}, 
		&& (\mu, \nu = 0,1,2,3), \label{E:EinsteinIntro} \\
	D_{\alpha} (\sigma^{s} D^{\alpha} \Phi) & = 0, && \label{E:Fluidintro}
\end{align}
\end{subequations}
where $g_{\mu \nu}$ is the \emph{spacetime metric}, $D$ is the covariant derivative operator corresponding to $g_{\mu \nu},$
$R_{\mu \nu}$ is the \emph{Ricci curvature tensor}, $R = g^{\alpha \beta} R_{\alpha \beta}$ is the \emph{scalar curvature}, $\Phi$ is the \emph{fluid potential}\footnote{In \cite{dC2007}, Christodoulou refers to $\Phi$ as the \emph{wave function}.}, $T_{\mu \nu}^{(scalar)} = 2 \sigma^{s} (\partial_{\mu} \Phi) (\partial_{\nu} \Phi) + g_{\mu \nu} (s + 1)^{-1} \sigma^{s + 1}$ is the energy-momentum tensor of an irrotational fluid, 
$\sigma = - g^{\alpha \beta}(\partial_{\alpha} \Phi)(\partial_{\beta} \Phi)$ is the \emph{enthalpy per particle},
and $s = (1 - \speed^2)/(2\speed^2).$ The fundamental unknowns are $(\mathcal{M},g,\partial \Phi),$ while
the pressure and proper energy density can be expressed as $p = \frac{1}{s+1} \sigma^{s+1},$ $\rho = \frac{2s+1}{s+1} \sigma^{s+1}.$ In this article, we will mainly restrict our attention to the case $s > 1,$ which is equivalent to $0 < \speed < \sqrt{1/3}.$ Although we limit our discussion to the physically relevant case of $1 + 3$ dimensions, we expect that our work can be easily generalized to apply to the case of $1 + n$ dimensions, $n \geq 3.$ 

\begin{remark} \label{R:Phiremark}
	Due to topological restrictions arising in the application of Poincar\'{e}'s lemma (see Section \ref{SS:IrrotationalFluids}), 
	the function $\Phi$ may only be defined locally. However, only its spacetime partial derivatives $\partial \Phi$ enter into 
	equations \eqref{E:EinsteinIntro} - \eqref{E:Fluidintro} and into our estimates. Since the quantities that can be locally 
	expressed as $\partial \Phi$ do not suffer from the same topological obstructions as $\Phi,$ we abuse notation and 
	use the single symbol $``\partial \Phi"$ to denote them globally. In fact, we could in principle avoid introducing $\partial 
	\Phi,$ choosing instead to work with the 1-form $\beta_{\mu};$ see equation \eqref{E:betaPhirelation}. The advantage of 
	working with $\partial \Phi$ is that the wave equation \eqref{E:Fluidintro} allows us to quickly and compactly obtain energy 
	estimates using a standard integration by parts argument (see Section \ref{SS:fluidvariableEnergies}).
\end{remark}

The choice of an equation of state is necessary in order to close the relativistic Euler equations. Our choice of
$p=\speed^2 \rho$ is often made in the mathematics and cosmology literature. As is explained in Section 
\ref{S:backgroundsolution}, under such equations of state, there exists a family of Friedmann-Lema\^{\i}tre-Robertson-Walker (FLRW) solutions to \eqref{E:EinsteinIntro} - \eqref{E:Fluidintro} that are frequently used to model a fluid-filled universe
undergoing accelerated expansion; these are the solutions that we investigate in detail in this article. The cases $p=0$ and $p=(1/3) \rho,$ which are known as the ``dust'' and ``radiative'' equations of state, are of special significance in the astrophysical literature. The latter is often used as a simple model for a ``radiation-dominated'' universe, while the former for a ``matter-dominated'' universe. Unfortunately, as we will see, these two equations of state lie just outside of the scope of our main theorem. Our results can be summarized as follows. We state them roughly here; they are are stated more precisely as theorems in Sections \ref{S:GlobalExistence} and \ref{S:Asymptotics}.

\begin{changemargin}{.25in}{.25in} 
\textbf{Main Results.} \
If $0 < \speed < \sqrt{1/3}$ (i.e. $s > 1$), then the FLRW background solution $([0,\infty) \times \mathbb{T}^3, \widetilde{g}, \partial \widetilde{\Phi})$ to \eqref{E:EinsteinIntro} - \eqref{E:Fluidintro}, which describes an initially uniform quiet fluid of constant \emph{positive} proper energy density, is globally future-stable under small perturbations. In particular, small perturbations of the initial data corresponding to the background solution have maximal globally hyperbolic developments that are future\footnote{Throughout this article, $\partial_t$ is future-directed.} causally geodesically complete. Here, $\widetilde{g} = -dt^2 + e^{2 \Omega(t)}\sum_{i=1}^3 (dx^3)^2,$ and $\partial \widetilde{\Phi} \eqdef (\partial_t \widetilde{\Phi},
\partial_1 \widetilde{\Phi}, \partial_2 \widetilde{\Phi}, \partial_3 \widetilde{\Phi})$
$=(\bar{\Psi} e^{-\decayparameter \Omega(t)},0,0,0),$ where $\bar{\Psi}$ is a positive constant, $\decayparameter  = 3/(2s+1) = 3 \speed^2,$ and $\Omega(t) \sim e^{(\sqrt{\Lambda/3})t}$ is defined in \eqref{E:BigOmega}. Furthermore, in the wave coordinate system introduced in Section \ref{SS:harmoniccoodinates}, the components $g_{\mu \nu}$ of the perturbed metric, its inverse $g^{\mu \nu},$ the fluid potential spacetime gradient $\partial \Phi,$ and various coordinate derivatives of these quantities converge as $t \rightarrow  \infty.$
\end{changemargin}

\begin{remark} \label{R:futurepaper}
	We made the irrotationality assumption only for simplicity. In a future article, we will
	remove this restriction. 
\end{remark}

\begin{remark}
	The kind of stability result stated above, in which the solution is shown to converge as $t \to \infty,$ is known as
	\emph{asymptotic stability}.
\end{remark}

\begin{remark}
	Note that our results only address perturbations of fluids featuring a strictly positive proper energy density. We have 
	thus avoided certain technical difficulties, such as dealing with a free boundary, that arise when $\rho$
	vanishes.
\end{remark}

\begin{remark}
	We have only shown stability in the ``expanding'' direction ($t \to \infty).$ 
\end{remark}

\begin{remark}
	We define $\mathbb{T}^d \eqdef [-\pi,\pi]^d$ with the ends identified.
\end{remark}

We would now like to make a few remarks about the cosmological constant. We do not attempt to give a detailed account of the rich history of $\Lambda,$ but instead defer to the discussion in \cite{sC2001}; we offer only a brief introduction. While the cosmological constant was originally introduced by Einstein \cite{aE1917} to allow for static solutions to the Einstein equations in the presence of matter, the present day motivation for introducing $\Lambda > 0$ is entirely different. In 1929, Hubble discovered of the expansion of the universe \cite{eH1929}. In brief, Hubble's ``law,'' which was formulated based on measurements of \emph{redshift}, states that the velocities at which distant galaxies are receding from Earth are proportional to their distance from Earth. Furthermore, the present day explanation is that the cause of these velocity shifts is the expansion of spacetime itself. For example, a metric of the form $g = - dt^2 + a^2(t) \sum_{i=1}^3 (dx^i)^2,$ with 
$\dot{a} > 0,$ creates a redshift effect. In the 1990's, experimental evidence derived from sources such as type Ia supernovae and the Cosmic Microwave Background suggest that the universe is in fact undergoing \emph{accelerated} expansion. Our main motivation for introducing the positive cosmological constant is that it allows for spacetime solutions of \eqref{E:EinsteinIntro} - \eqref{E:Fluidintro} that feature this effect. A simple example of a solution to the Einstein-vacuum equations that features such accelerated expansion is the metric $g = - dt^2 + e^{2Ht} \sum_{i=1}^3 (dx^i)^2$ on the manifold $(-\infty,\infty) \times \mathbb{T}^3,$ where $H = \sqrt{\Lambda/3}.$ 

The introduction of a positive cosmological constant is not the only known mechanism for generating solutions to Einstein's equations with accelerated expansion. In particular, Ringstr\"{o}m's work \cite{hR2008}, which is the main precursor to this article, shows that the Einstein-scalar field system, with an appropriate nonlinearity\footnote{Ringstr\"{o}m refers to $V(\Phi)$ as a ``potential.'' However, we call $V(\Phi)$ a ``nonlinearity'' in order to avoid confusing it with our ``potential'' $\Phi,$ and in order to avoid confusing it with a term of the form $V(x) \Phi,$ which is often referred to as a ``potential'' term in the literature.} $V(\Phi),$ has an open family of solutions with accelerated expansion. More specifically, the system studied by Ringstr\"{o}m can be obtained by replacing \eqref{E:Fluidintro} with $g^{\alpha \beta} D_{\alpha} D_{\beta} \Phi = V'(\Phi)$ and setting $T_{\mu \nu}^{(scalar)} = (\partial_{\mu} \Phi)(\partial_{\nu} \Phi) - \big[\frac{1}{2} g^{\alpha \beta} (\partial_{\alpha} \Phi) (\partial_{\beta} \Phi) + V(\Phi)\big]g_{\mu \nu}$ in equation \eqref{E:EinsteinIntro}, where $D$ denotes the covariant derivative induced by $g.$ $V$ is required to satisfy $V(0) > 0, V'(0) = 0, V''(0) > 0,$ so that in effect, the influence of the cosmological constant is emulated by $V(\Phi).$ His main result, which is analogous to our main result, is a proof of the future global stability of a large class of spacetimes featuring accelerated expansion.

The \textbf{Main Results} stated above allude to both the existence of an \emph{initial value problem} formulation of the Einstein equations, and the existence of a ``maximal'' solution. These notions are fleshed out in Section \ref{SS:IVP}, but we offer a brief description here. One of the principal difficulties in analyzing the Einstein equations is the lack of a canonical coordinate system. Intimately connected to this difficulty is the fact that due to the diffeomorphism invariance of the equations, their hyperbolic nature is not readily apparent until one makes some kind of gauge choice. One way of resolving these difficulties is to work in a special coordinate system known as \emph{wave coordinates} (also known as \emph{harmonic gauge} or \emph{de Donder gauge}), in which the Einstein equations become a system of \emph{quasilinear wave equations}. One advantage of such a formulation is that local-in-time existence for a system of wave equations is immediate, because a standard hyperbolic theory based on energy estimates has been developed (consult e.g. \cite[Ch. VI]{lH1997},
\cite[Ch. 16]{mT1997III}, \cite{jSmS1998}, \cite{cS1995}). Although the use of wave coordinates is often attributed solely to Choquet-Bruhat, it should be emphasized that use of wave coordinates in the context of the Einstein equations goes back to at least 1921, where it is featured in the work of de Donder \cite{tD1921}. However, the completion of the initial value formulation of the Einstein equations is in fact due to Choquet-Bruhat \cite{cB1952}, her main contribution being a proof that the wave coordinate condition is preserved during the evolution if it is initially satisfied.

The initial data for the irrotational Euler-Einstein system consist of a $3-$dimensional Riemannian manifold $\Sigma,$ a Riemannian metric $\bar{g}$ on $\Sigma,$ a symmetric two-tensor $\bar{K}$ on $\Sigma,$ and initial data $\bard \mathring{\Phi}, \mathring{\Psi}$ for the tangential and normal derivatives of $\Phi$ along $\Sigma.$ A solution consists of a $4-$dimensional manifold $\mathcal{M},$ a Lorentzian metric $g$ and the spacetime derivatives $\partial \Phi$ (see Remark \ref{R:Phiremark}) of a scalar field on $\mathcal{M}$ satisfying \eqref{E:EinsteinIntro} - \eqref{E:Fluidintro}, and an embedding $\Sigma \subset \mathcal{M}$ such that $\bar{g}$ is the first fundamental form\footnote{Recall that the first fundamental form of $\Sigma$ 
is defined to be the metric $\bar{g}$ on $\Sigma$ such that $\bar{g}|_x(X,Y) = g|_x(X,Y)$ for all vectors $X,Y \in T_x \Sigma.$} of $\Sigma,$ $\bar{K}$ is the second fundamental form\footnote{Recall that $\bar{K}$ is defined at the point x by $\bar{K}(X,Y) \eqdef g(D_{X} \hat{N},Y)$ for all $X,Y \in T_x \Sigma,$ where $\hat{N}$ is the future-directed unit normal to $\Sigma$ at $x.$} of $\Sigma,$ and the restriction of $\bard \Phi$ and $D_{\hat{N}} \Phi$ to $\Sigma$ are $\bard \mathring{\Phi}$ and $\mathring{\Psi}$ respectively. Here $D_{\hat{N}} \Phi$ denotes the derivative of $\Phi$ in the direction of the future-directed unit normal $\hat{N}$ to $\Sigma.$ It is important to note that the initial value problem is overdetermined, and that the data are subject to the \emph{Gauss} and \emph{Codazzi} constraints:

\begin{subequations}
\begin{align}
	\bar{R} - \bar{K}_{ab} \bar{K}^{ab} + (\bar{g}^{ab} \bar{K}_{ab})^2 & = 4 \mathring{\sigma}^s
		\mathring{\Psi}^2 - \frac{2s}{s+1} \mathring{\sigma}^{s+1}, && \label{E:Gaussintro} \\
	\bar{D}^a \bar{K}_{aj} - \bar{g}^{ab}  \bar{D}_j \bar{K}_{ab}  & = 2 \mathring{\sigma}^s \mathring{\Psi} 
		\partial_j \mathring{\Phi}, && (j=1,2,3), \label{E:Codazziintro}
\end{align}
\end{subequations}
where $\bar{R}$ is the scalar curvature of $\bar{g},$ $\bar{D}$ is the covariant derivative induced by $\bar{g},$
and $\mathring{\sigma} \eqdef \sigma|_{\Sigma}.$

\begin{remark}
	In this article, we do not address the issue of solving the constraint equations for the system \eqref{E:EinsteinIntro} 
	- \eqref{E:Fluidintro}.
\end{remark}

17 years after the initial value problem formulation was understood, Choquet-Bruhat and Geroch showed \cite{cBgR1969} that every initial data set satisfying the constraints launches a unique \emph{maximal globally hyperbolic development.} Roughly speaking, this is the largest spacetime solution to the Einstein equations that is uniquely determined by the data. This result is still a local existence result in the sense that it allows for the possibility that the spacetime might contain singularities. In particular, future-directed, causal geodesics may terminate, which in physics terminology means that an observer (light ray in the case of null geodesics) may run into the end of spacetime in finite affine parameter. For spacetimes launched by initial data near that of the background solution, our main result rules out the possibility of these singularities for observers (light rays) traveling in the ``future direction.''

We offer a few additional remarks concerning wave coordinates. The classic wave coordinate condition is the algebraic relation $\Gamma^{\mu}=0,$ where the $\Gamma^{\mu}$ are the contracted Christoffel symbols of the spacetime metric. In this article, we use a version of the wave coordinate condition that is closer in spirit to the one used by Ringstr\"{o}m in \cite{hR2008},
which was itself inspired by the ideas in \cite{hFaR2000}. Namely, we set $\Gamma^{\mu} = \widetilde{\Gamma}^{\mu},$ where $\widetilde{\Gamma}^{\mu}$ is the contracted Christoffel symbol of the background solution metric. Simple computations imply that $\widetilde{\Gamma}^{\mu} = 3 \omega \delta_0^{\mu},$ where $\omega = \omega(t) \sim \sqrt{\Lambda/3}$ is defined in \eqref{E:omegadef}. It follows that in our wave coordinate system, the (geometric) wave equation $g^{\alpha \beta} D_{\alpha} D_{\beta} v = 0$ for the function $v$ is equivalent to the modified (also known as the ``reduced'') wave equation $g^{\alpha \beta} \partial_{\alpha} \partial_{\beta} v = 3 \omega (\partial_t v)^2,$ which features the dissipative source term $\omega (\partial_t v)^2.$ We provide a more detailed discussion of this modified scalar equation in Section \ref{SS:Commentsonanalysis}. Furthermore, in Section \ref{S:ReducedEquations}, we modify the irrotational Euler-Einstein system in an analogous fashion, arriving at an equivalent hyperbolic system featuring dissipative terms. More precisely, the modified system is equivalent to the Einstein equations if the data satisfy the Einstein constraint equations and the wave coordinate condition.

\subsection{Comparison with previous work}

First, it should be emphasized that the behavior of solutions the fluid equation \eqref{E:Fluidintro} on exponentially expanding backgrounds is quite different than it is in flat spacetime. In particular, if one fixes a background metric on $[0,\infty) \times \mathbb{T}^3$ near $\widetilde{g},$ then our proof shows that the fluid equation \eqref{E:Fluidintro} on this background with 
$0 < \speed < \sqrt{1/3}$ has global solutions arising from data that are close to that of an initial uniform quiet fluid state, which is represented by $\partial \widetilde{\Phi}.$ This is arguably the most interesting aspect of our main result. In contrast, Christodoulou's monograph \cite{dC2007} shows that on the Minkowski space background, shock singularities can form in solutions to the irrotational fluid equation arising from data that are arbitrarily close to that of a uniform quiet fluid state. Our original intuition for this article was that rapid spacetime expansion should smooth out the fluid and discourage the formation of shocks.

Second, we would like to compare our nonlinear stability result to the well-known \emph{linear instability} 
results of Sachs and Wolfe \cite{rSaW1967}. In this work, they consider the Euler-Einstein system with
$\Lambda = 0$ under the equations of state $p=0$ and $p=\frac{1}{3} \rho,$ which correspond to borderline cases of the equations of state that we consider in this article. Sachs-Wolfe then consider a family of background solutions to this system on the manifold $(-\infty, \infty) \times \mathbb{R}^3.$ We remark that these well-known background solutions are of FLRW type, and can be obtained as special cases\footnote{The Sachs-Wolfe solutions have spatial slices diffeomorphic to $\mathbb{R}^3,$ while our solutions have spatial slices diffeomorphic to $\mathbb{T}^3.$} of solutions we have presented in Section \ref{S:backgroundsolution}. The first main result of \cite{rSaW1967} is that the linearized\footnote{That is, the system formed by linearizing the full Euler-Einstein equations around one of the FLRW background solutions.} system is \emph{unstable}; i.e., it features solutions that grow like $t^{C},$ where $C < 1$ depends on the equation of state. Because our theorem does not cover the borderline cases $p=0, p = \frac{1}{3}  \rho,$ it is not yet entirely clear whether or not the Sachs-Wolfe instabilities are caused exclusively by the absence of a cosmological constant, or whether the borderline equations of state also play a role; see Section \ref{SS:Commentsonanalysis} for some additional remarks concerning the borderline cases. As a side remark, we mention the most well-known result of the Sachs-Wolfe paper: they showed that the growing density perturbations couple back into the metric. The resulting variations in the metric lead to anisotropies in the Cosmic Microwave Background (CMB). In particular, the amount by which photons are gravitationally shifted varies with direction in the sky. The theoretical predictions of this effect, which is known was the Sachs-Wolfe effect, are consistent with the variations in the CMB detected by the Mather-Smoot team's COBE satellite in 1992 \cite{nBjMrW1992}.

Next, we note that Brauer, Rendall, and Reula have shown \cite{uBaRoR1994} a Newtonian analogue of our main result. More specifically, they studied Newtonian cosmological models\footnote{Their models were based on Newton-Cartan theory, which is a slight generalization of ordinary Newtonian gravitational theory that can be endowed with a highly geometric interpretation.} with a positive cosmological constant and with perfect fluid sources under the equation of state $p = C \rho_{Newt}^{\gamma},$ where $\rho_{Newt} \geq 0$ is the Newtonian mass density, and $\gamma > 1.$ They showed that small perturbations of a uniform quiet fluid state of constant positive density lead to a global solution. It is of particular interest to note that they do not require the fluid to be irrotational. This suggests that our main result can be extended to allow for (small) non-vanishing vorticity. As discussed in Remark \ref{R:futurepaper} we will address this issue in an upcoming article.

We also note a curious anti-correlation between our results and some well-known stability arguments for the Euler-Poisson system (a non-relativistic system with vanishing cosmological constant) which may be found e.g. in Chapter XIII of Chandrasekhar's book \cite{sC1961}. Chandrasekhar considers a simple model for an isolated body in equilibrium, namely a static compactly supported solution to the Euler-Poisson equations under the equation of state\footnote{More correctly, his equation of state is $p= \widetilde{C} \rho_{Newt}^{\gamma}.$ However, in the case of the Euler-Poisson equations, $\rho_{Newt}$ is proportional to the number density $n;$ i.e., $p= C n^{\gamma}$ is equivalent to $p = \widetilde{C} \rho_{Newt}^{\gamma}.$} $p = C n^{\gamma},$ where $n$ denotes the fluid element \emph{number density}, $C > 0$ is a constant, and $\gamma > 0.$ He uses virial identity arguments to suggest that such a configuration is stable if $\gamma > 4/3.$ However, since \eqref{E:rhopnrelation} implies that the equation of state $p=\speed^2 \rho$ (here $\rho$ denotes the proper energy density, a relativistic quantity) is equivalent to $p = C n^{1 + \speed^2},$ our main results show that our background solution is stable under irrotational perturbations if $1 < 1 + \speed^2 < 4/3;$ i.e., our results seem to anti-correlate with the aforementioned non-relativistic one\footnote{We temper this statement by noting that
our problem differs in several key ways from that of Chandrasekhar; e.g., Chandrasekhar studied compactly supported data for a non-relativistic system on a flat background, while here we study relativistic fluids of everywhere positive energy density on an expanding background.}.

In addition to the previously mentioned work of Ringstr\"{o}m, we would also like to mention other contributions related to the question of nonlinear stability for the Einstein equations with a positive cosmological constant. The first author to obtain stability results in this direction was Helmut Friedrich, first in vacuum spacetime \cite{hF1986a} in $1 + 3$ dimensions, and then later for the Einstein-Maxwell and Einstein-Yang-Mills equations \cite{hF1991}. Anderson then extended the vacuum result to cover the case of $1 + n$ dimensions, $n$ odd \cite{mA2005}. Their work was based on the \emph{conformal method}, which reduces the question of global stability for the Einstein equations to the much simpler question of local-in-time stability\footnote{The conformal field equations are symmetric hyperbolic, and for such systems, local-in-time stability is a standard result.} for the \emph{conformal field equations}, which were developed by Friedrich. In contrast, Ringstr\"{o}m has stated that one of his main motivations for his wave coordinate approach in \cite{hR2008} is that the conformal method cannot necessarily be easily adapted to handle matter models other than Maxwell and Yang-Mills fields. Our work can be viewed as an example of the robustness of his methods.

Finally, we compare our work here to the body of work on the stability of Minkowski space, which is the most well-known solution to the Einstein-vacuum equations in the case $\Lambda = 0$. This groundbreaking work, which was initiated by Christodoulou and Klainerman \cite{dCsK1993}, covered the case of the Einstein-vacuum equations in $1 + 3$ dimensions. 
Their proof, which is manifestly covariant, relied upon several geometric foliations of spacetime, including maximal $t=const$ slices and also a family outgoing null cones. In particular, it was believed that wave coordinates were unstable and therefore unsuitable for studying the global stability for Minkowski space. However, Lindblad and Rodnianski have recently devised yet another proof for the Einstein-vacuum and Einstein-scalar field systems \cite{hLiR2004}, \cite{hLiR2005}, which is much shorter but less precise, and which shows stability in the wave coordinate gauge $\Gamma^{\mu} = 0.$ In particular, Lindblad and Rodnianski were the first authors to show that a wave coordinate system can be used to prove global stability results for the Einstein equations. As we will explain in the next section, our result was technically simpler to achieve than either of these results. More specifically, in $1 + 3$ dimensions with $\Lambda = 0,$ the Einstein-vacuum equations involve nonlinear terms that are on the boarder\footnote{They contain nonlinear terms that, on the basis of their order alone, might be expected to produce blow-up. However, these terms satisfy a version of the \emph{null condition}, which means that they have a special algebraic structure that allows for small-data global existence.} of what can be expected to allow for global existence. As we will see, the addition of $\Lambda > 0$ to the Einstein equations, together with our previously mentioned wave coordinate choice, will lead to the presence of energy dissipation terms. Consequently in the parameter range $0 < \speed < \sqrt{\frac{1}{3}},$ we do not have to contend with the difficult ``borderline integrals'' that appear in the proofs of the stability of Minkowski space. A more thorough comparison of the proofs of the stability of Minkowski space to the proofs of the stability of exponentially expanding solutions can be found in the introduction of \cite{hR2008}. Moreover, we remark that
readers interested in results related to those of Christodoulou-Klainerman and Lindblad-Rodnianski can consult \cite{lBnZ2009}, \cite{sKfN2003}, \cite{jL2008}, \cite{nZ2000}.

\subsection{Comments on the analysis} \label{SS:Commentsonanalysis}

The basic idea behind our global existence argument can be seen by analyzing the inhomogeneous
wave equation $g^{\alpha \beta} D_{\alpha} D_{\beta} v = F$ for the model metric $g = -dt^2 + e^{2t} \sum_{i=1}^3 (dx^i)^2$
on the manifold-with-boundary $\mathcal{M} = [0,\infty) \times \mathbb{T}^3.$ Here, we are using standard local coordinates on $\mathbb{T}^3.$ An omitted computation implies that relative to this coordinate system, the wave equation can be expressed as follows:

\begin{align} \label{E:modelequation}
	- \partial_t^2 v + e^{-2t} \delta^{ab} \partial_a \partial_b v = 3 (\partial_t v)^2 + F.
\end{align}

To estimate solutions to \eqref{E:modelequation}, one can define the ``usual'' energy
$E^2(t) = \frac{1}{2} \int_{\mathbb{T}^3} (\partial_t v)^2 + e^{-2t} \delta^{ab} (\partial_a v)(\partial_b v) \, d^3 x,$ and a standard integration by parts argument leads to the inequality

\begin{align} \label{E:modelenergyinequality}
	\frac{d}{dt} E \leq -E + \|F\|_{L^2}.
\end{align}
It is clear from \eqref{E:modelenergyinequality} that sufficient estimates of $\|F\|_{L^2}$ in terms of $E$ will lead to 
energy decay\footnote{In our work below, we work with rescaled energies that are approximately constant in time.}. 

Our estimates for the modified irrotational Euler-Einstein equations are in the spirit of the above argument. The bulk of the work lies in estimating the inhomogeneous terms and in ensuring that the perturbed solution remains close to that of the background solution. We remark that the main tools used for estimating the inhomogeneous terms are Sobolev-Moser type estimates,
which we have placed in the Appendix for convenience. Our analysis of the spacetime metric components closely parallels the work \cite{hR2008} of Ringstr\"{o}m. In particular, based on Ringstr\"{o}m's work, we provide (in Section \ref{S:BootstrapConsequences}) a \emph{Counting Principle} based on the number of spatial indices upstairs and downstairs in a product of metric components, Christoffel symbols, and first derivatives of these quantities. This heuristic mechanism can be used to quickly (and roughly, but good enough to prove small-data local existence) determine the rates of decay of products of such terms. Despite the availability of this tool, we carefully derive many of the estimates in complete detail. We remark that the Counting Principle is not precise enough to detect the improved decay estimates derived in Section \ref{S:Asymptotics}. 

Although Ringstr\"{o}m's framework is useful for analyzing the metric components, our analysis of the fluid variables 
$\partial_{\mu} \Phi,$ $(\mu = 0,1,2,3),$ involves additional complications beyond those encountered in his analysis. We would now like to briefly discuss these complications, and also to indicate why we make the assumption $0 < \speed < \sqrt{1/3}.$ We believe that the breakdown of our proof in the case $\speed = 0$ is merely an artifact of our methods; we cannot address the equation of state $p=0$ here because we use the pressure as an unknown in our equations. In a future article, we will investigate this case using alternate methods. On the other hand, for a technical example of where the proof breaks down in the case $\speed \geq 1/3,$ see Section \ref{SS:Breakdown}. We expect that stability can be shown for the
equation of state $p = (1/3) \rho$ using the Friedrich's conformal method; we plan to investigate this issue in a future article. On the other hand, the question of what to expect in the case $\speed > 1/3$ is open.

The principal difficulty we encounter in our analysis is that our background solution fluid variable $\partial \widetilde{\Phi} = (\bar{\Psi} e^{-\decayparameter \Omega},0,0,0)$ and the associated quantity $\widetilde{\sigma} = - \widetilde{g}^{\alpha \beta}(\partial_{\alpha} \widetilde{\Phi})(\partial_{\beta} \widetilde{\Phi})$ both exponentially decay to $0$ as $t \to \infty.$ Therefore, the quantity $\sigma$ corresponding to a slightly perturbed solution will also decay. Furthermore, by examining the fluid equation \eqref{E:Fluidintro}, we see that $\sigma = 0$ corresponds to a degeneracy. In order to deal with these difficulties, and in order to close our bootstrap argument for global existence, we have to ensure that $\sigma$ doesn't decay too quickly, and that it never becomes $0$ in finite time. In fact, for the global solutions that we construct,
we show that the perturbed $\sigma$ will decay at the same rate as $\widetilde{\sigma}.$ Since this fact plays a key role in our estimates for the fluid equation, we now heuristically outline our approach its proof. Assuming that the perturbed $g_{\mu \nu}$ is near that of $\widetilde{g}_{\mu \nu},$ $(\mu,\nu= 0,1,2,3),$ it follows that $\sigma \approx -(\partial_t \Phi)^2 + e^{-2\Omega(t)} \delta^{ab}(\partial_a \Phi)(\partial_b \Phi).$ Thus, to conclude that $\sigma \approx \widetilde{\sigma},$ we \textbf{i)} show that $\partial_t \Phi \approx \partial_t \widetilde{\Phi},$ and \textbf{ii)} introduce the quantities $z_j \eqdef e^{- \Omega(t)} \partial_j \Phi/\partial_t \Phi,$ $(j=1,2,3),$ which we assume are small initially, and then show that they remain small. In fact, to close our bootstrap argument, we show that they exponentially decay to $0$ in $L^{\infty}.$ Combining \textbf{i)} and \textbf{ii)}, we deduce that $\sigma \approx \widetilde{\sigma}.$

The above mathematical conditions have a physical interpretation. To elaborate, we first note that the four-velocity $u$ of the fluid is connected to the fluid potential via \eqref{E:betadef} and \eqref{E:betaPhirelation}, which imply that $u_{\mu} = - \sigma^{-1/2} \partial_{\mu} \Phi.$ From this fact and the previous remarks, it thus follows that the smallness/decay condition on $z_j$ ensures that the fluid decays towards the  ``low velocity'' regime, in which the time component $u^0$ is dominant over the spatial components $u^j.$ 

Finally, to understand the assumption $0 < \speed < \sqrt{1/3},$ we have to explain some of the details of our estimates. If one tries estimate $z_j$ in $L^{\infty}$ using only energy estimates plus Sobolev embedding, then one can only conclude that $\| z_j \|_{L^{\infty}}$ is of size $C \epsilon$ during the time of existence when the rescaled energy is of size $\epsilon.$
Unfortunately, this estimate is not strong enough to close our bootstrap argument. However, if we use the bootstrap estimates\footnote{These estimates will follow easily from the assumption that our rescaled energy is of size $\epsilon.$} $\Big\| \frac{1}{\partial_t \Phi} \Big \|_{L^{\infty}} \leq C e^{\decayparameter \Omega},$ 
$\| \partial_t \partial_j \Phi \|_{L^{\infty}} \leq C \epsilon e^{-\decayparameter \Omega},$ and the initial condition $\| \partial_j \Phi(0) \|_{L^{\infty}} \leq C \epsilon,$ then we can estimate $\| z_j \|_{L^{\infty}}$ by first estimating 
$\| \partial_j \Phi \|_{L^{\infty}}$ through time integration using the bound on $\| \partial_t \partial_j \Phi \|_{L^{\infty}}.$ This leads to the estimate $\| z_{j} \|_{L^{\infty}} \leq C \epsilon e^{(\decayparameter - 1) \Omega(t)},$ which decays in time if $\decayparameter < 1.$ The condition $0 < \decayparameter < 1$ holds exactly when $1 < s < \infty,$ which is equivalent to $0 < \speed^2 < 1/3;$ this parameter range of stability is precisely the one mentioned in the \textbf{Main Results} statement above.

\subsection{Applications to spatial topologies other than \texorpdfstring{$\mathbb{T}^3$}{}}

The model metric $g = -dt^2 + e^{2t} \sum_{i=1}^3 (dx^i)^2$ has another feature that is of crucial relevance for possible extensions of our work. To illustrate our point, let us consider $g$ to be a metric on $[0,\infty) \times \mathbb{R}^3,$ a Lorentzian manifold with boundary that has the Cauchy hypersurface $\Sigma \eqdef \lbrace t=0 \rbrace.$ Simple computations imply that the causal past of the causal future of a point intersects $\Sigma$ in a \emph{precompact}\footnote{For example, in this model spacetime, the causal past of the causal future of the origin is contained in the set $\lbrace (t,x^1,x^2,x^3) \ | \ t \geq 0, \sum_{i=1}^3 (x^i)^2 \leq 4 \rbrace.$} set. This is in stark contrast to the situation encountered in Minkowski space, where the causal future of a point is the forward null cone emanating from that point, and the causal past of this null cone includes the \emph{entire} Cauchy hypersurface $\lbrace t = 0 \rbrace.$ One consequence of this fact is that the study of solutions to wave equations on exponentially expanding spacetimes is a ``very'' local problem; i.e., if we makes assumptions about the data in large enough ball in $\Sigma,$ then we can control the solution in a non-compact region of spacetime that includes a cylinder of the form $[0,\infty) \times B,$ where $B$ is a spatial-coordinate ball. 

Using these observations, Ringstr\"{o}m was able to prove the future-stability of various solutions to the Einstein-scalar field system for many spatial topologies in addition to $\mathbb{T}^3.$ The main idea is to choose local coordinate patches on the spatial slices on which the problem is quantitatively close to the study of case of $\Sigma = \mathbb{T}^3,$ and to piece together the future development of these patches into a global spacetime. The most difficult part of his argument is the global existence theorem on $\mathbb{T}^3.$ However, his patching argument requires that one use cut-off functions, which introduces regions in which the Einstein constraint equations are not satisfied. To deal with this difficulty, he constructs his modified system of equations in such a manner that one can still conclude global existence, even if the constraint equations are not satisfied in the cut-off regions. Finally, after patching, these artificially-introduced regions are of course ``discarded'' and are not part of the spacetime. 

The modified system \eqref{E:finalg00equation} - \eqref{E:finalfluidequation} that we study is similar to Ringstr\"{o}m's modified equations in that our global existence argument depends only on a smallness condition on the data, and not on whether or not the constraint equations are satisfied. As noted above, this is the main step in Ringstr\"{o}m's work. For these reasons, it is very likely that his patching arguments can be used to extend our result to other spatial topologies. However, for sake of brevity, we do not explore this issue in this article.

\subsection{Outline of the structure of the paper} 

\begin{itemize}
	\item In Section \ref{S:Notation}, we describe our conventions for indices and introduce some notation for
		differential operators and Sobolev norms of tensor-fields.
	\item In Section \ref{S:IrrotationalEE}, we introduce the irrotational Euler-Einstein system.
	\item In Section \ref{S:backgroundsolution}, we use a standard ODE ansatz to derive a well-known family of 
		background Friedmann-Lema\^{\i}tre-Robertson-Walker (FLRW) solutions to the irrotational Euler-Einstein system.
	\item In Section \ref{S:ReducedEquations} we introduce wave coordinates and use algebraic identities 
		valid in such a coordinate system to construct a modified version of the irrotational Euler-Einstein system. We then
		discuss how to construct data for the modified system from data for the un-modified system. Finally, we discuss 
		classical local existence for the modified system, including the continuation principle that is used in
		Section \ref{S:GlobalExistence}.
	\item In Section \ref{S:NormsandEnergies}, we introduce the relevant norms and the related energies for the modified system 
		that we use in our global existence argument. We also provide a preliminary analysis of the derivatives of the energies, 
		but the inhomogeneous terms are not estimated until Section \ref{S:BootstrapConsequences}.
	\item In Section \ref{S:LinearAlgebra}, we introduce some bootstrap assumptions on the spacetime metric $g_{\mu \nu}.$
		We then use these assumptions to prove some linear-algebraic lemmas that are useful for
		analyzing $g_{\mu \nu}$ and the inverse metric $g^{\mu \nu}.$
	\item In Section \ref{S:BootstrapAssumptions}, we introduce our main bootstrap assumption, which is a smallness condition
		on $\suptotalnorm{N},$ a norm of difference between the perturbed solution and the background solution. We also define
		the positive constants $q$ and $\eta_{min},$ which play a fundamental role in the technical estimates of the following
		sections.
	\item Section \ref{S:BootstrapConsequences} contains most of the technical estimates. We assume the bootstrap assumptions 
		from the previous sections and use them to deduce estimates for $g_{\mu \nu}, g^{\mu \nu}, m^{\mu \nu},$ and for the 
		nonlinearities appearing in the modified equations. Here, $m^{\mu \nu}$ denotes the \emph{reciprocal acoustical metric}, 
		which is the effective metric for the irrotational fluid equation.
	\item Section \ref{S:EnergyNormEquivalence} is a very short section in which we show that the Sobolev norms we have defined 
		are equivalent to the energy norms. 
	\item In Section \ref{S:GlobalExistence}, we use the estimates from the previous sections to prove our main theorem, which is 
		a small-data global (where ``small'' means close to the background solution) existence result for the modified equations. 
		We then discuss the breakdown of our proof in the case $\speed \geq 1/3.$
		Finally, we use the global existence theorem to prove a related theorem, which states that initial data satisfying the 
		irrotational Euler-Einstein constraints, the wave coordinate condition, and a smallness condition lead to a 
		future-geodesically complete solution of the irrotational Euler-Einstein system. 
	\item In Section \ref{S:Asymptotics}, we prove that the global solution from the main theorem converges as $t \to \infty.$
		The main idea is that once we have a global small solution to the modified system, we can revisit the modified equations
		and upgrade some of the estimates proved in Section \ref{S:GlobalExistence}.
\end{itemize}

\section{Notation} \label{S:Notation}

In this section, we briefly introduce some notation that we use in this article.

\subsection{Index conventions}
Greek indices $\alpha, \beta, \cdots$ take on the values $0,1,2,3,$ while Latin indices $a,b,\cdots$ 
(which we sometimes call ``spatial indices'') take on the values $1,2,3.$ Pairs of repeated indices, with one raised and one lowered, are summed (from $0$ to $3$ if they are Greek, and from $1$ to $3$ if they are Latin). We raise and lower indices with the spacetime metric $g_{\mu \nu}$ and its inverse $g^{\mu \nu}.$ The only exceptions are equations \eqref{E:Gaussintro} - \eqref{E:Codazziintro} and \eqref{E:Gauss} - \eqref{E:Codazzi}, in which we use the $3-$metric $\bar{g}_{jk}$
and its inverse $\bar{g}^{jk}$ to raise and lower indices, and in Section \ref{S:Asymptotics}, in which 
all indices are raised and lowered with $g_{\mu \nu}$ and $g^{\mu \nu}$ except for the
$3-$metric $g_{jk}^{(\infty)},$ which has $g_{(\infty)}^{jk}$ as its corresponding inverse metric. 

\subsection{Coordinate systems and differential operators}
Throughout this article, we work in a standard local coordinate system $x^1,x^2,x^3$ on $\mathbb{T}^3.$ 
Although strictly speaking this coordinate system is not globally well-defined, the vectorfields $\partial_j \eqdef \frac{\partial}{\partial x^{j}}$ are globally well-defined. This coordinate system extends to a coordinate system $x^0,x^1,x^2,x^3$ on manifolds with boundary of the form $\mathcal{M} =[0,T) \times \mathbb{T}^3,$ and we often write $t$ instead of $x^0.$ We write $\partial_{\mu}$ to denote the coordinate derivative $\frac{\partial}{\partial x^{\mu}},$ and we often write $\partial_t$ instead of $\partial_0.$ Throughout the article, we will perform all of our computations with respect to the frame $\big\lbrace \partial_{\mu} \big\rbrace_{\mu = 0,1,2,3}.$

If $\vec{\alpha} = (n_1,n_2,n_3)$ is a triplet of non-negative integers, then we define the spatial multi-index coordinate differential operator $\partial_{\vec{\alpha}}$ by $\partial_{\vec{\alpha}} \eqdef \partial_1^{n_1} \partial_2^{n_2} \partial_3^{n_3}.$ We denote the order of $\vec{\alpha}$ by $|\vec{\alpha}|,$ where $|\vec{\alpha}| \eqdef n_1 + n_2 + n_3.$

We write $D_{\mu} T_{\mu_1 \cdots \mu_s}^{\nu_1 \cdots \nu_r} = 
\partial_{\mu} T_{\mu_1 \cdots \mu_s}^{\nu_1 \cdots \nu_r} + 
\sum_{a=1}^r \Gamma_{\mu \ \alpha}^{\ \nu_{a}} T_{\mu_1 \cdots \mu_s}^{\nu_1 \cdots 
\nu_{a-1} \alpha \nu_{a+1} \nu_r} 
- \sum_{a=1}^s \Gamma_{\mu \ \mu_{a}}^{\ \alpha} T_{\mu_1 \cdots \mu_{a-1} \alpha \mu_{a+1} \mu_s}^{\nu_1 \cdots \nu_r}
$ to denote the components of the covariant derivative of a 
tensorfield on $\mathcal{M}$ with components $T_{\mu_1 \cdots \mu_s}^{\nu_1 \cdots \nu_r}.$ 

We write $\partial^{(N)} T_{\mu_1 \cdots \mu_s}^{\nu_1 \cdots \nu_r}$ to denote the array
containing of all of the $N^{th}$ order \emph{spacetime} coordinate derivatives (including time derivatives) of the component $T_{\mu_1 \cdots \mu_s}^{\nu_1 \cdots \nu_r}.$ Similarly, we write $\bard^{(N)} T_{\mu_1 \cdots \mu_s}^{\nu_1 \cdots \nu_r}$ to denote the array containing of all $N^{th}$ order \emph{spatial coordinate} derivatives of the component $T_{\mu_1 \cdots \mu_s}^{\nu_1 \cdots \nu_r}.$ We omit
the superscript $^{(N)}$ when $N=1.$

\subsection{Norms}
All of the Sobolev norms we use are relative to the coordinate system $x^1,x^2,x^3$ introduced above. We remark
that they are not coordinate invariant quantities. If $T$ is a tensorfield on $\mathcal{M}$ with components $T_{\mu_1 \cdots \mu_s}^{\nu_1 \cdots \nu_r}$ relative to this coordinate system, then we define the following norm of $T_{\mu_1 \cdots \mu_s}^{\nu_1 \cdots \nu_r}$ on the $t=cont$ slices of $\mathcal{M}:$
\begin{align} \label{E:Tensornormdef}
	\big\| T_{\mu_1 \cdots \mu_s}^{\nu_1 \cdots \nu_r}  \big\|_{H^N(U)} \eqdef
		\bigg(\sum_{\mu_1, \cdots \mu_s = 0}^3 \sum_{\nu_1, \cdots \nu_r = 0}^3 \sum_{|\vec{\alpha}| \leq N}
		\int_{x(U)} \big|\partial_{\vec{\alpha}} T_{\mu_1 \cdots \mu_s}^{\nu_1 \cdots \nu_r} \circ x^{-1} \big| dx^1dx^2dx^3 \bigg)^{1/2}.
\end{align}

Note that the above norm is in general a function $t.$ If the set $U$ is omitted, then it is understood that
$U = \mathbb{T}^3.$ This is a slight abuse of notation since the coordinate system $x^1,x^2,x^3$ is not globally well-defined on $\mathbb{T}^3.$ Furthermore, when writing integral expressions similar to the one on the right-hand side of \eqref{E:Tensornormdef}, we suppress the coordinate map $(x^1,x^2,x^3): \mathcal{M} \rightarrow \mathbb{T}^3.$ This should cause no confusion since we are working relative to the frame $\big\lbrace \partial_{\mu} \big\rbrace_{\mu = 0,1,2,3}.$

Using the above notation, we can write the $N^{th}$ order homogeneous Sobolev norm of the component 
$T_{\mu_1 \cdots \mu_s}^{\nu_1 \cdots \nu_r}$ as

\begin{align}
	\big\| \bard^{(N)} T_{\mu_1 \cdots \mu_s}^{\nu_1 \cdots \nu_r}  \big\|_{L^2} \eqdef
		\sum_{|\vec{\alpha}| = N} \big\| \partial_{\vec{\alpha}} T_{\mu_1 \cdots \mu_s}^{\nu_1 \cdots \nu_r} \big\|_{L^2}.
\end{align}

If $\mathfrak{K} \subset \mathbb{R}^n$ or $\mathfrak{K} \subset \mathbb{T}^n,$ then $C^N_b(\mathfrak{K})$ denotes the set of $N-$times continuously differentiable functions (either scalar or array-valued, depending on context) on Int$(\mathfrak{K})$ with bounded derivatives up to order $N$ that extend continuously to the closure of $\mathfrak{K}.$ The norm of a function $F 
    \in C^N_b(\mathfrak{K})$ is defined by
    \begin{equation} \label{E:CbMNormDef}
        |F|_{N,\mathfrak{K}} \eqdef \sum_{|\vec{I}|\leq N} \sup_{\cdot \in \mathfrak{K}}
        |\partial_{\vec{I}}F(\cdot)|,
    \end{equation}
    where $\partial_{\vec{I}}$ is a multi-indexed operator representing repeated partial differentiation
    with respect to the arguments $\cdot$ of $F,$ which may be either spacetime coordinates or state-space variables
    depending on context. When $N=0,$ we also use the notation
    
    \begin{align} 
    	|F|_{\mathfrak{K}} \eqdef \sup_{\cdot \in \mathfrak{K}} |F(\cdot)|.
    \end{align}
    Furthermore, we use the notation
    \begin{align} 
    	|F^{(N)}|_{\mathfrak{K}} \eqdef \sum_{|\vec{I}| = N}|\partial_{\vec{I}} F|_{\mathfrak{K}}.
    \end{align}
    The quantity $|F^{(N)}|_{\mathfrak{K}}$ is a measure of the size of the $N^{th}$ order derivatives of $F$ 
    on the set $\mathfrak{K}.$

		If $I \subset \mathbb{R}$ is an interval and $X$ is a normed function space, then we use the notation
    $C^N(I,X)$ to denote the set of $N$-times continuously differentiable maps from $I$ into $X.$

\subsection{Running constants} \label{SS:runningconstants}
We use $C$ to denote a running constant that is free to vary from line to line.
In general, it can depend on $N$ (see \eqref{E:Ndef}) and $\Lambda,$ but can be chosen to be independent of all functions
$(g_{\mu \nu}, \partial_{\mu} \Phi),$ $(\mu, \nu = 0,1,2,3),$ that are sufficiently close to the background solution $(\widetilde{g}_{\mu \nu}, \partial_{\mu} \widetilde{\Phi})$ of Section \ref{S:backgroundsolution}. We sometimes use notation such as $C(N)$ to indicate the dependence of $C$ on quantities that are peripheral to the main argument. Occasionally, we use $c, C_*,$ etc., to denote a constant that plays a distinguished role in the analysis. We remark that many of the constants blow-up as $\Lambda \rightarrow 0^+.$

\subsection{A warning on the sign of \texorpdfstring{$\hat{\Square}_g$}{}}

Although we often choose notation that agrees with the notation used by Ringstr\"{o}m in \cite{hR2008}, our reduced wave operator $\hat{\Square}_g \eqdef g^{\alpha \beta} \partial_{\alpha} \partial_{\beta}$ has the opposite sign of the one in \cite{hR2008}.

\section{The Irrotational Euler-Einstein System} \label{S:IrrotationalEE}

The Einstein equations connect the \emph{Einstein tensor} $R_{\mu \nu} - \frac{1}{2}g_{\mu \nu} R,$
which depends on the geometry of the spacetime manifold $\mathcal{M},$ to the \emph{energy-momentum-stress-density tensor} (energy-momentum tensor for short) $T_{\mu \nu},$ which models the matter content of spacetime. In $1 + 3$ dimensions, 
they can be expressed as
\begin{align} \label{E:EinsteinField}
	R_{\mu \nu} - \frac{1}{2}R g_{\mu \nu} + \Lambda g_{\mu \nu} = T_{\mu \nu}, && (\mu, \nu = 0,1,2,3),
\end{align}
where $R_{\mu \nu}$ is the \emph{Ricci curvature tensor}, $R$
is the \emph{scalar curvature}, and $\Lambda$ is the \emph{cosmological constant}. We remark that \textbf{the stability results proved in this article heavily depend upon the assumption} $\Lambda > 0.$ Recall that the Ricci curvature tensor and scalar curvature are defined in terms of $R_{\mu \alpha \nu}^{\ \ \ \ \beta},$ the \emph{Riemann curvature tensor}\footnote{Under our sign convention, $D_{\mu} D_{\nu} X_{\alpha} - D_{\nu} D_{\mu} X_{\alpha} = R_{\mu \nu \alpha}^{\ \ \ \ \beta} X_{\beta}.$}, which can be expressed in terms of $\Gamma_{\mu \ \nu}^{\ \alpha},$ the \emph{Christoffel symbols} of the metric. In a local coordinate system, these quantities can be expressed as follows:

\begin{subequations}
\begin{align}
R_{\mu \alpha \nu}^{\ \ \ \ \beta} & \eqdef 
	\partial_{\alpha} \Gamma_{\mu \ \nu}^{\ \beta}
	- \partial_{\mu} \Gamma_{\alpha \ \nu}^{\ \beta}
  + \Gamma_{\alpha \ \lambda}^{\ \beta} \Gamma_{\mu \ \nu}^{\ \lambda}
	- \Gamma_{\mu \ \lambda}^{\ \beta} \Gamma_{\alpha \ \nu}^{\ \lambda},	&& (\alpha,\beta, \mu, \nu = 0,1,2,3),
	\label{E:Riemanndef} \\
R_{\mu \nu} & \eqdef R_{\mu \alpha \nu}^{\ \ \ \ \alpha} 
	= \partial_{\alpha} \Gamma_{\mu \ \nu}^{\ \alpha}
	- \partial_{\mu} \Gamma_{\alpha \ \nu}^{\ \alpha}
  + \Gamma_{\alpha \ \lambda}^{\ \alpha} \Gamma_{\mu \ \nu}^{\ \lambda}
	- \Gamma_{\mu \ \lambda}^{\ \alpha} \Gamma_{\alpha \ \nu}^{\ \lambda}, && (\mu, \nu = 0,1,2,3),
	\label{E:Riccidef} \\
R & \eqdef g^{\alpha \beta} R_{\alpha \beta}, && \label{E:Rdef} \\
\Gamma_{\mu \ \nu}^{\ \alpha} & \eqdef \frac{1}{2} g^{\alpha \lambda}(\partial_{\mu} g_{\lambda \nu} 
	+ \partial_{\nu} g_{\mu \lambda} - \partial_{\lambda} g_{\mu \nu}), && (\alpha, \mu, \nu = 0,1,2,3). 	
	\label{E:EMBIChristoffeldef}
\end{align}
\end{subequations}

The Bianchi identities (see \cite{rW1984}) imply that the left-hand side of \eqref{E:EinsteinField} is divergence free, which leads to the following equations satisfied by $T^{\mu \nu}:$ 
\begin{align} \label{E:Tdivergence}
	D_{\nu} T^{\mu \nu} = 0, && (\mu = 0,1,2,3).
\end{align}

By contracting each side of \eqref{E:EinsteinField} with $g^{\mu \nu}$, we deduce that
$R = 4 \Lambda - T,$ where $T \eqdef g^{\alpha \beta} T_{\alpha \beta}$ is the trace of $T_{\mu \nu}.$ From this fact, it easily
follows that \eqref{E:EinsteinField} is equivalent to

\begin{align} \label{E:EinsteinFieldequivalent}
	R_{\mu \nu} = \Lambda g_{\mu \nu} + T_{\mu \nu} - \frac{1}{2}T g_{\mu \nu}, && (\mu,\nu = 0,1,2,3).
\end{align}

It is not our aim to give a complete discussion of the notion of a perfect fluid. A thorough introduction to the subject, including a discussion of its history, can be found in Christodoulou's survey article \cite{dC2007b}. We confine ourselves here to a brief introduction. In general, the energy-momentum tensor for a perfect fluid is
\begin{align} \label{E:Tlowerfluid}	
	T_{(fluid)}^{\mu \nu} = (\rho + p) u^{\mu} u^{\nu} + p g^{\mu \nu}, && (\mu,\nu = 0,1,2,3),
\end{align}
where $\rho \geq 0$ is the \emph{proper energy density}, $p \geq 0$ is the \emph{pressure}, and $u$
is the \emph{four velocity}, a unit-length (i.e., $u_{\alpha} u^{\alpha} = -1)$ future-directed vectorfield
on $\mathcal{M}.$ The \emph{Euler equations}, which are the laws of motion for a perfect relativistic fluid, are the four equations \eqref{E:Tdivergence} together with a conservation law \eqref{E:Eulernumberdensity} for the number of fluid elements. In a local coordinate system, they can be expressed as follows:

\begin{subequations}
\begin{align}	
  D_{\nu} T_{(fluid)}^{\mu \nu} & = 0, && (\mu = 0,1,2,3), \label{E:Euler} \\
	D_{\nu} (n u^{\nu}) & = 0, && \label{E:Eulernumberdensity}
\end{align}
\end{subequations}
where $n$ is the \emph{proper number density} of the fluid elements. We also introduce the thermodynamic variable $\Ent,$ the \emph{proper entropy density}, which we will discuss below.

Unfortunately, even in a prescribed spacetime $(\mathcal{M},g),$ the equations \eqref{E:Euler} - \eqref{E:Eulernumberdensity} do not form a closed system. The standard means of closing the equations is to appeal the laws of thermodynamics, which imply the following relationships between the state-space variables (see e.g. \cite{yGsTZ1998}):    
    
    \begin{enumerate}
  		\item $\rho \geq 0$ is a function of $n \geq 0$ and ${\Ent} \geq 0.$
  		\item $p \geq 0$ is defined by
        \begin{align} \label{E:rhopnrelation}
            p= n \left. \frac{\partial \rho}{\partial n} \right|_{\Ent} - \rho,                                                      		\end{align}
        where the notation $\left. \right|_{\cdot}$ indicates partial differentiation with $\cdot$ held
        constant.
        \item A perfect fluid satisfies
        \begin{align}
        \left. \frac{\partial \rho}{\partial n} \right|_{\Ent} >0, \left. \frac{\partial p}{\partial
        n} \right|_{\Ent}>0, \left. \frac{\partial \rho}{\partial {\Ent}} \right|_n \geq 0 \
        \mbox{with} \ ``='' \ \mbox{iff} \ \Ent=0.                             \label{E:EOSAssumptions}
        \end{align}
        As a consequence, we have that $\zeta,$ the speed
        of sound\footnote{In general, $\zeta$ is not constant. However, for the equations of
        state we study in this article, $\zeta$ is equal to the constant $\speed.$} in the fluid, is always real for $\Ent > 0:$
            \begin{align}
                \zeta^2 \overset{def}{=} \left.\frac{\partial p}{\partial
                \rho}\right|_{\Ent} = \frac{{\partial p / \partial n}|_{\Ent}}{{\partial \rho / \partial
                n}|_{\Ent}} > 0.                                                                         \label{E:SpeedofSound}
            \end{align}
        \item We also demand that the speed of sound is positive and less than the speed of light
        whenever $n > 0$ and $\Ent > 0$:
            \begin{align} \label{E:Causality}
                n>0 \ \mbox{and} \ \Ent > 0 \implies 0 < \zeta < 1.
            \end{align}
    \end{enumerate}
    
   	Postulates (1) - (3) express the laws of thermodynamics and fundamental thermodynamic assumptions, while Postulate 
    (4) ensures that at each $x \in \mathcal{M},$ vectors that are causal with respect to the sound cone\footnote{The
    sound cone is defined to be the subset of tangent vectors $X \in T_x \mathcal{M}$ such that $m_{\alpha \beta} X^{\alpha} 
    X^{\beta} = 0,$ where $m_{\mu \nu}$ is the \emph{acoustical metric}. The matrix $m_{\mu \nu}$ 
    is the inverse matrix of the \emph{reciprocal acoustical metric} $m^{\mu \nu}$, which is introduced in Section 
    \ref{SS:Decomposition}; i.e., $m_{\mu \nu}$ is not obtained by lowering the indices of $m^{\mu \nu}$ with $g_{\mu \nu}.$} 
    in $T_x \mathcal{M}$ are necessarily causal with respect to the gravitational null cone\footnote{By gravitational null cone 
    at $x,$ we mean the subset of tangent vectors $X \in T_x \mathcal{M}$ such that $g_{\alpha \beta} X^{\alpha} X^{\beta} = 
    0.$} in $T_x \mathcal{M}.$ See \cite{js2008a} for a more detailed analysis of the geometry of the sound cone and the 
    gravitational null cone.

    We note that the assumptions $\rho \geq 0, p \geq 0$ together imply that
    the energy-momentum tensor \eqref{E:Tlowerfluid} satisfies both the \emph{weak energy 
    condition} ($T_{\alpha \beta}^{(fluid)} X^{\alpha} X^{\beta} \geq 0$ holds whenever $X$ is timelike and future-directed
    with respect to the gravitational null cone) 
    and the \emph{strong energy condition}
    ($[T_{\mu \nu}^{(fluid)} - 1/2 g^{\alpha \beta}T_{\alpha \beta}^{(fluid)} g_{\mu \nu}]X^{\mu}X^{\nu} \geq 0$ holds whenever 
    $X$ is timelike and future-directed with respect to the gravitational null cone). Furthermore, if we assume that the 
    equation of state is such that $p=0$ when
    $\rho = 0,$ then \eqref{E:SpeedofSound} and \eqref{E:Causality} guarantee that $p \leq \rho.$ It is then easy to check
    that $0 \leq p \leq \rho$ implies the \emph{dominant energy condition}
    ($-g^{\mu \alpha}T_{\alpha \nu}^{(fluid)} X^{\nu}$ is causal and future-directed whenever $X$ is causal and future-directed 
    with respect to the gravitational null cone).

Postulate $(1)$ from above is equivalent to making a choice of an \emph{equation of state}, which is a function 
that expresses $p$ in terms of $\Ent$ and $\rho.$ An equation of state is not necessarily a fundamental law of nature, but can instead be an empirical relationship between the fluid variables. In this article, we consider the case of an \emph{irrotational, barotropic} fluid under the equation of state $p = \speed^2 \rho,$ where $0 < \speed < \sqrt{\frac{1}{3}},$ and according to \eqref{E:SpeedofSound}, the constant $\speed$ is the \emph{speed of sound}. A barotropic fluid is one for which $p$ is a function of $\rho$ alone. Because $\Ent$ plays no role in the analysis of such fluids, this quantity is absent from the remainder of our article. As discussed in the following section, these assumptions imply that the tensor $T_{(fluid)}^{\mu \nu}$ defined in \eqref{E:Tlowerfluid} is equal to the tensor $T_{(scalar)}^{\mu \nu}$ defined in \eqref{E:Tslowerperfectfluid}, and that the equations \eqref{E:Tdivergence} are equivalent to \eqref{E:ELsigmas}, a single quasilinear wave equation for a scalar function $\Phi$
(see Remark \ref{R:Phiremark}). As a consequence, it follows that $T_{(scalar)}^{\mu \nu}$ also satisfies the weak, strong, and dominant energy conditions.

\subsection{Irrotational fluids} \label{SS:IrrotationalFluids}
In this section, we introduce the notion of an irrotational fluid. We assume that the fluid is barotropic, be we do not yet impose the particular equation of state $p = \speed^2 \rho.$ The scalar formulation of an irrotational fluid in a curved spacetime goes back to at least 1937 \cite{jS1937}, but in this article, we use modern terminology found e.g. in \cite{dC2007}.
We begin by defining the quantity $\sqrt{\sigma},$ which is known as the \emph{enthalpy per particle}:
\begin{align} \label{E:enthaplydef}
	\sqrt{\sigma} \eqdef \frac{\rho + p}{n},
\end{align}
where $\sigma \geq 0.$ We note for future use that we can differentiate each side of 
\eqref{E:enthaplydef} with respect to $\sqrt{\sigma},$ and use \eqref{E:rhopnrelation} and the chain rule to conclude that
\begin{align} \label{E:enthalpypnrelation}
	\frac{dp}{d \sqrt{\sigma}} = n.
\end{align}
In the above formula we are viewing $p$ as a function of $\sigma.$

We also introduce the following one form:
\begin{align} \label{E:betadef}
	\beta_{\mu} \eqdef - \sqrt{\sigma} u_{\mu}, && (\mu = 0,1,2,3).
\end{align}
The \emph{fluid vorticity} $v$ is then defined to be $d \beta,$ where $d$ denotes the exterior derivative operator. In local coordinates, we have that

\begin{align}
	v_{\mu \nu} \eqdef \partial_{\mu} \beta_{\nu} - \partial_{\nu} \beta_{\mu}, && (\mu, \nu = 0,1,2,3). 
\end{align}

An \emph{irrotational fluid} is defined to be a fluid for which $v_{\mu \nu}$ vanishes everywhere. In this case, by 
Poincar\'{e}'s lemma, there \emph{locally} exists (see Remark \ref{R:Phiremark}) a scalar function $\Phi$ such that
\begin{align} \label{E:betaPhirelation}
	\beta_{\mu} = \partial_{\mu} \Phi, && (\mu = 0,1,2,3),
\end{align}
which implies that the four velocity is connected to $\partial \Phi$ via the equation

\begin{align} \label{E:uintermsofPhi}
	u_{\mu} & = - \sigma^{-1/2} \partial_{\mu} \Phi, && (\mu = 0,1,2,3),
\end{align}
and the enthalpy per particle is connected to $\partial \Phi$ via 

\begin{align} \label{E:sigmadef}
	\sigma = - (D^{\alpha} \Phi)(D_{\alpha} \Phi) = - g^{\alpha \beta} (D_{\alpha}\Phi) (D_{\beta} \Phi).
\end{align}

We now show that under the assumption of irrotationality, the equations \eqref{E:Euler} reduce to
a single quasilinear wave equation for $\Phi.$ We begin by identifying the Lagrangian for the purported wave equation with $p:$

\begin{align} \label{E:Lequalsp}
	\mathscr{L} = \mathscr{L}(\sigma) = p. 
\end{align}
From \eqref{E:enthalpypnrelation} and \eqref{E:Lequalsp}, it follows that
\begin{align} \label{E:Lsigmanrelation}
	2 \frac{\partial \mathscr{L}}{\partial \sigma} & = \frac{n}{\sqrt{\sigma}} = \frac{n^2}{\rho + p}.
\end{align}

We also recall that the Euler-Lagrange equation for $\Phi$ corresponding to the Lagrangian \eqref{E:Lequalsp} is
\begin{align} \label{E:ELfluid}
	D_{\alpha} \bigg[\frac{\partial \mathscr{L}}{\partial \sigma} D^{\alpha} \Phi \bigg] = 0.
\end{align}
Using \eqref{E:uintermsofPhi} and \eqref{E:Lsigmanrelation}, we conclude that for an irrotational fluid, \eqref{E:ELfluid}
is equivalent to the continuity equation $D_{\nu}(n u^{\nu}) = 0$ from \eqref{E:Eulernumberdensity}. 

To show that the remaining equations \eqref{E:Euler} follow from \eqref{E:ELfluid}, we first
recall that the energy-momentum tensor for a Lagrangian scalar-field theory can be expressed as

\begin{align} \label{E:Tlower}
	T_{\mu \nu}^{(scalar)} & = -2 \frac{\partial \mathscr{L}}{ \partial g^{\mu \nu}} + g_{\mu \nu} \mathscr{L},
\end{align}	
\noindent and that if $\partial \Phi$ is a solution to \eqref{E:ELfluid}, then $T_{\mu \nu}^{(scalar)}$ is symmetric
and divergence free:
\begin{align}
	T^{\mu \nu} & = T^{\nu \mu}, && (\mu, \nu = 0,1,2,3), \\
	D_{\nu} T_{(scalar)}^{\mu \nu} & = 0, && (\mu = 0,1,2,3).
\end{align}

For future use, we remark that if $\partial \Phi$ is a solution to the inhomogeneous equation

\begin{align} \label{E:ELfluidinhomogeneous}
	D_{\alpha} \bigg[\frac{\partial \mathscr{L}}{\partial \sigma} D^{\alpha} \Phi \bigg] + I_{\partial \Phi} = 0,
\end{align}
then

\begin{align} \label{E:Inhomogeneousenergymomentumdivergence}
	D_{\nu} T_{(scalar)}^{\mu \nu} = - 2I_{\partial \Phi} D^{\mu} \Phi, && (\mu=0,1,2,3).
\end{align}
In the case of the Lagrangian \eqref{E:Lequalsp}, one can check that \eqref{E:Tlower} is equivalent to

\begin{align} \label{E:Tlowerequivalent}
	T_{\mu \nu}^{(scalar)} \eqdef  2 \frac{\partial \mathscr{L}}{ \partial \sigma} 
	(\partial_{\mu} \Phi)(\partial_{\nu} \Phi) + g_{\mu \nu} \mathscr{L}.
\end{align}	
Using \eqref{E:uintermsofPhi}, \eqref{E:Lequalsp}, \eqref{E:Tlower}, and \eqref{E:Tlowerequivalent}, we have that 
\begin{align}
	T_{\mu \nu}^{(scalar)} = \frac{n}{\sqrt{\sigma}}(\partial_{\mu} \Phi)(\partial_{\nu} \Phi) + g_{\mu \nu} p
		= (\rho + p)u_{\mu} u_{\nu} + p g_{\mu \nu}.
\end{align}

By examining \eqref{E:Tlowerfluid}, we observe that $T_{\mu \nu}^{(scalar)} = T_{\mu \nu}^{(fluid)}.$ To summarize, we have shown that if $\partial \Phi$ is a solution to \eqref{E:ELfluid}, and if $p, n, u,$ and $\rho$ are defined through $\partial \Phi$ via equations \eqref{E:Lequalsp}, \eqref{E:enthalpypnrelation}, \eqref{E:uintermsofPhi}, and \eqref{E:Lsigmanrelation} respectively, then it follows all 5 equations from \eqref{E:Euler} - \eqref{E:Eulernumberdensity} are necessarily satisfied. Thus, it follows that for an irrotational fluid, \emph{the entire content of the Euler equations is contained in the single scalar equation} \eqref{E:ELfluid}.

We conclude with a summary of the constraints that $\partial \Phi$ and $\mathscr{L}(\sigma)$ must satisfy in order to have an irrotational fluids interpretation. First, since $\sigma \geq 0,$ \eqref{E:sigmadef} implies that 

\begin{align}
	- g^{\alpha \beta}(\partial_{\alpha} \Phi)(\partial_{\beta} \Phi) \geq 0.
\end{align}
Similarly, since $p$ is non-negative, \eqref{E:enthalpypnrelation} implies that

\begin{align}
	\mathscr{L}(\sigma) \geq 0.
\end{align}
Next, \eqref{E:enthalpypnrelation} and \eqref{E:Lequalsp} together imply that $n = \frac{d \mathscr{L}}{d \sqrt{\sigma}},$ so that the requirement $n \geq 0$ is equivalent to

\begin{align}
	\frac{d \mathscr{L}}{d \sqrt{\sigma}} \geq 0.
\end{align}
Furthermore, since \eqref{E:rhopnrelation} implies that
$\frac{d}{d \sqrt{\sigma}}(\mathscr{L}/\sigma) = \sqrt{\sigma} \frac{d \mathscr{L}}{d \sqrt{\sigma}} - p = \rho,$ and
since $\rho \geq 0,$ we must have 

\begin{align}
	\frac{d}{d \sqrt{\sigma}}\Big(\frac{\mathscr{L}}{\sigma}\Big) \geq 0.
\end{align}
Finally, since \eqref{E:EOSAssumptions} implies that $\frac{dp}{dn} > 0,$ and since \eqref{E:enthalpypnrelation} and the chain rule imply that $\frac{dp}{dn} = \frac{d \mathscr{L}}{d \sqrt{\sigma}} / \frac{d n}{d \sqrt{\sigma}}
= n / \Big( \frac{d}{d \sqrt{\sigma}} (\frac{d \mathscr{L}}{d \sqrt{\sigma}}) \Big),$
we must have

\begin{align}
	\frac{d^2 \mathscr{L}}{d \sqrt{\sigma}^2} > 0.
\end{align}

\subsection{The initial value problem for the irrotational Euler-Einstein system} \label{SS:IVP}

In this section, we discuss various aspects of the initial value problem for the Einstein equations, including the initial data and the notion of the \emph{maximal globally hyperbolic development} of the data. We assume that we are given a Lagrangian $\mathscr{L} = \mathscr{L}(\sigma)$ that is subject to the constraints discussed at the end of the previous section;
this is equivalent to fixing an equation of state $p=p(\rho).$ We remark that the discussion in this section is very standard, and we provide it only for convenience.

\subsubsection{Summary of the irrotational Euler-Einstein system}

We first summarize the results of the previous sections by stating that the irrotational Euler-Einstein system
is the following system of equations:

\begin{subequations}
\begin{align} 
	R_{\mu \nu} - \frac{1}{2}R g_{\mu \nu} + \Lambda g_{\mu \nu} & = T_{\mu \nu}^{(scalar)}, && (\mu, \nu =0,1,2,3), 
		\label{E:EinsteinFieldsummary} \\
	D_{\alpha} \bigg[ \frac{\partial \mathscr{L}}{\partial \sigma} D^{\alpha} \Phi \bigg] & = 0,  && 
		\label{E:FluidEquationSummary}
\end{align}
\end{subequations}
where $\mathscr{L} = \mathscr{L}(\sigma),$ $\sigma = - (D^{\alpha} \Phi)(D_{\alpha} \Phi),$ and
$T_{\mu \nu}^{(scalar)} = 2 \frac{\partial \mathscr{L}}{ \partial \sigma} (\partial_{\mu} \Phi)(\partial_{\nu} \Phi) + g_{\mu \nu} \mathscr{L}.$

\subsubsection{Initial data for the irrotational Euler-Einstein system} \label{SSS:InitialDataOriginalSystem}

The initial value problem formulation of the Einstein equations goes back to the seminal work \cite{cB1952} by Choquet-Bruhat. Initial data for the system \eqref{E:EinsteinFieldsummary} - \eqref{E:FluidEquationSummary} consist of a $3-$dimensional manifold $\Sigma$ together with the following fields on $\Sigma:$ a Riemannian metric $\bar{g},$ a covariant two-tensor $\bar{K},$ and functions $\bard \mathring{\Phi}, \mathring{\Psi}.$ By $\bard \mathring{\Phi},$ we mean that
these three components, which are required to be globally defined on $\Sigma,$ should be the spatial gradient of 
functions $\Phi$ that are locally defined on $\Sigma$ (see Remark \ref{R:Phiremark}).

It is well-known that one cannot consider arbitrary data for the Einstein equations. The data are in fact subject to the following constraints, where $\bar{D}$ is the covariant derivative induced by $\bar{g}:$

\begin{subequations}
\begin{align}
	\bar{R} - \bar{K}_{ab} \bar{K}^{ab} + (\bar{g}^{ab} \bar{K}_{ab})^2 & = 2T_{00}^{(scalar)}|_{\Sigma}, && \label{E:Gauss} \\
	\bar{D}^a \bar{K}_{aj} - \bar{g}^{ab}  \bar{D}_j \bar{K}_{ab}  & = 
		T_{0j}^{(scalar)}|_{\Sigma}, && (j=1,2,3). \label{E:Codazzi}
\end{align}
\end{subequations}
We remark that in the case of the equation of state $p = \speed^2 \rho,$ the calculations in Section \ref{SS:Calculations} will imply that $T_{00}^{(scalar)}|_{\Sigma} = 2\mathring{\sigma}^s \mathring{\Psi}^2 - \frac{s}{s+1} \mathring{\sigma}^{s+1}$ and 
$T_{0j}^{(scalar)}|_{\Sigma}=2 \mathring{\sigma}^s \mathring{\Psi} \partial_j \mathring{\Phi},$ where
$\mathring{\sigma} = \mathring{\Psi}^2 - \bar{g}^{ab} (\partial_a \mathring{\Phi})(\partial_b \mathring{\Phi}).$

The constraints \eqref{E:Gauss} - \eqref{E:Codazzi} are manifestations of the \emph{Gauss} and \emph{Codazzi} equations
respectively. These equations relate the geometry of the ambient Lorentzian spacetime $(\mathcal{M},g)$ (which has to be constructed) to the geometry inherited by an embedded Riemannian hypersurface (which will be $(\Sigma,\bar{g})$ after construction). Without providing the rather standard details (see e.g. \cite{dC2008}), we remark that they are consequences of the following assumptions:
\begin{itemize}
	\item $\Sigma$ is a submanifold of the spacetime manifold $\mathcal{M}$
	\item $\bar{g}$ is the first fundamental form of $\Sigma$
	\item $\bar{K}$ is the second fundamental form of $\Sigma$
	\item The irrotational Euler-Einstein system is satisfied along $\Sigma$
	\item We are using a coordinate system $(x^0 =t, x^1, x^2, x^3)$ on $\mathcal{M}$ such that $\Sigma = 
		\lbrace x \in \mathcal{M} \ | \ t = 0 \rbrace,$ 
	and along $\Sigma,$ $g_{00} = -1,$ $g_{0j} = 0,$ $g_{jk} = \bar{g}_{jk},$ $\partial_t g_{jk} = 2\bar{K}_{jk},$ 
	$(j,k=1,2,3),$ $\bard \Phi = \bard \mathring{\Phi},$ and $\partial_t \Phi = \mathring{\Psi}.$
\end{itemize}

We recall that $\bar{g}$ is the Riemann metric on $\Sigma$ defined by
\begin{align}
	\bar{g}|_x(X,Y) & = g|_x(X,Y),  & & \forall X,Y \in T_x \Sigma,
\end{align}
and that $\bar{K}$ is the symmetric tensorfield on $\Sigma$ that can be expressed as

\begin{align}
	\bar{K}|_x(X,Y) & = g|_x(D_X \hat{N},Y) = g|_x(D_Y \hat{N},X) = \frac{1}{2} (D_{\hat{N}} g)|_x(X,Y), 
	& & \forall X,Y \in T_x \Sigma,
\end{align}
where $\hat{N}$ is the future-directed normal\footnote{Under the above assumptions, it follows that at every point $x \in \Sigma,$ $\hat{N}^{\mu} = (1,0,0,0).$} to $\Sigma$ at $x,$ and $D$ is the covariant derivative induced by
$g.$

\subsubsection{The definition of a solution to the irrotational Euler-Einstein system}
In this section, we define the notion of a solution to the Einstein equations. We begin with the following definition, which describes the maximal region in which a solution is determined by its values on a set $S.$

\begin{definition}
Given any set $S \in \mathcal{M},$ we define $\mathcal{D}(S),$ the Cauchy development of $S,$ to be the union
$\mathcal{D}(S) = \mathcal{D}^+(S) \cup \mathcal{D}^-(S),$ where $\mathcal{D}^+(S)$ is the set of all points $p \in \mathcal{M}$ such that every past inextendible causal curve through $p$ intersects $S,$ and $\mathcal{D}^-(S)$ 
is the set of all points $p \in \mathcal{M}$ such that every future inextendible\footnote{A curve $\gamma :[s_0, s_{max}) \rightarrow \mathcal{M}$ is said to be \emph{future-inextendible} if $\lim_{s \to s_{max}}$ does not exist. Past inextendibility is defined in an analogous manner.} causal curve through $p$ intersects $S.$
\end{definition}

We also rigorously define a \emph{Cauchy hypersurface}:

\begin{definition}
	A Cauchy hypersurface in a Lorentzian manifold $\mathcal{M}$ is a hypersurface that is intersected exactly once by every 	
	inextendible timelike curve in $\mathcal{M}.$ 
\end{definition}

It is well-known that if $\Sigma \subset \mathcal{M}$ is a Cauchy hypersurface, then $\mathcal{D}(\Sigma) = \mathcal{M}$ (see e.g. \cite{bO1983}).

Given sufficiently smooth initial data $(\Sigma, \bar{g}, \bar{K}, \bard \mathring{\Phi}, \mathring{\Psi})$ as described in 
Section \ref{SSS:InitialDataOriginalSystem}, a (classical) solution to the system \eqref{E:EinsteinFieldsummary} - \eqref{E:FluidEquationSummary} is a $4$ dimensional manifold $\mathcal{M},$ a Lorentzian metric $g,$ the spacetime derivatives
$\partial \Phi$ of a locally defined function $\Phi$ (see Remark \ref{R:Phiremark}), and an embedding\footnote{Here, we suppress the embedding $\iota: \Sigma \rightarrow \mathcal{M},$ and identify $\Sigma$ with $\iota(\Sigma).$} $\Sigma \subset \mathcal{M}$ subject to the following conditions:

	\begin{itemize}
		\item $g$ is a $C^2$ tensorfield and and $\partial \Phi$ is a $C^1$ tensorfield
		\item $\Sigma$ is a Cauchy hypersurface in $(\mathcal{M},g)$
		\item $\bar{g}$ is the first fundamental form of $\Sigma$
		\item $\bar{K}$ is the second fundamental form of $\Sigma$ 
		\item $D_{\hat{N}} \Phi = \mathring{\Psi},$ where $\hat{N}$ is the future-directed 
			normal to $\Sigma.$
	\end{itemize}
	The triple $(\mathcal{M}, g, \partial \Phi)$ is called a \emph{globally hyperbolic development} of the initial data.

\subsubsection{Maximal globally hyperbolic development}

In this section, we state a fundamental abstract existence result of Choquet-Bruhat and Geroch \cite{cBgR1969}, which states that for initial data of \emph{arbitrary} regularity, there is a unique ``largest'' spacetime determined by it. The following definition captures the notion of this ``largest'' spacetime. 

\begin{definition}
	Given initial data for the irrotational Euler-Einstein system \eqref{E:EinsteinFieldsummary} - \eqref{E:FluidEquationSummary}, a \emph{maximal globally hyperbolic development} of the data is a globally hyperbolic development $(\mathcal{M},g,\partial \Phi)$ together with an embedding 
$\iota: \Sigma \rightarrow \mathcal{M}$ with the following property: if $(\mathcal{M}',g',\partial \Phi')$ is any other globally hyperbolic development of the same data with embedding $\iota': \Sigma \rightarrow \mathcal{M}',$ then there is a map
$\psi: \mathcal{M}' \rightarrow \mathcal{M}$ that is a diffeomorphism onto its image such that 
$\psi^* g = g', \psi^* \partial \Phi = \partial \Phi',$ and $\psi \circ \iota' = \iota.$ Here, $\psi^*$ denotes the pullback by
$\psi.$
\end{definition}

Before we can state the theorem, we also need the following definition, which captures the notion of having two different representations of the same spacetime.

\begin{definition}
The developments $(\mathcal{M},g,\partial \Phi)$ and $(\mathcal{M}',g',\partial \Phi')$ are said to be \emph{isometrically isomorphic} if the map $\psi$ from the previous definition is a diffeomorphism from $\mathcal{M}$ to $\mathcal{M}'.$
\end{definition}

We now state the aforementioned theorem. 

\begin{theorem} (Choquet-Bruhat and Geroch, 1969)
Given any initial data for the irrotational Euler-Einstein system \eqref{E:EinsteinFieldsummary} - \eqref{E:FluidEquationSummary}, there is a maximal globally hyperbolic development of the data which is unique up to isometry. 
\end{theorem}

\subsection{Calculations for the equation of state \texorpdfstring{$p= \speed^2 \rho$}{}} \label{SS:Calculations}

For the remainder of this article, we restrict our attention to equations of state of the form

\begin{align} \label{E:EOS}
	p = \speed^2 \rho,
\end{align}
where by equation \eqref{E:SpeedofSound}, $\speed$ is the \emph{speed of sound}. As mentioned in the introduction to this article, our stability results are limited to the following parameter range:

\begin{align}
	0 < \speed < \sqrt{\frac{1}{3}}.
\end{align}
Equations \eqref{E:rhopnrelation}, \eqref{E:Lequalsp}, \eqref{E:Lsigmanrelation}, and \eqref{E:EOS} imply that there
exist constants $C > 0$ and $\widetilde{C} > 0$ such that
\begin{align}
	p & = C n^{1 + \speed^2} = \widetilde{C} \sigma^{s + 1},
\end{align}
where	

\begin{align} \label{E:sArelations}
	s & = \frac{1 - \speed^2}{2\speed^2}, && \speed^2 = \frac{1}{2s + 1}.
\end{align}
Choosing a convenient normalization constant, it follows that under the equation of state \eqref{E:EOS}, the Lagrangian \eqref{E:Lequalsp} is given by
\begin{align} \label{E:Lintermsofsigma}
	\mathscr{L} = (s + 1)^{-1} \sigma^{s + 1}.
\end{align}
Furthermore, in the case of the Lagrangian \eqref{E:Lintermsofsigma}, the Euler-Lagrange equation \eqref{E:ELfluid}
is

\begin{align} \label{E:ELsigmas}
	D_{\alpha} (\sigma^{s} D^{\alpha} \Phi) = 0,
\end{align}
while the energy-momentum tensor is given by

\begin{align} \label{E:Tslowerperfectfluid}
	T_{\mu \nu}^{(scalar)} = 2 \sigma^{s} (\partial_{\mu} \Phi) (\partial_{\nu} \Phi) + g_{\mu \nu} (s + 1)^{-1} \sigma^{s + 1}.
\end{align}

For future reference, we record here the following two identities, which follow easily from \eqref{E:Tslowerperfectfluid}:
\begin{align}
	T^{(scalar)} & \eqdef g^{\alpha \beta} T_{\alpha \beta}^{(scalar)} = \bigg[\frac{2(1-s)}{s + 1} \bigg] \sigma^{s + 1}, 
		&& \\
	T_{\mu \nu}^{(scalar)} - \frac{1}{2} T^{(scalar)} g_{\mu \nu} & = 2 \sigma^s (\partial_{\mu} \Phi) (\partial_{\nu} \Phi) 
		+ \Big(\frac{s}{s + 1}\Big) \sigma^{s + 1}  g_{\mu \nu}, && (\mu,\nu = 0,1,2,3).
		\label{E:TmunuminushalfTgmunu}
\end{align}

\subsubsection{Summary of the irrotational Euler-Einstein system under the equation of state $p=\speed^2 \rho$}

To summarize, we note that under the equation of state $p = \speed^2 \rho,$ the irrotational Euler-Einstein system comprises the equations

\begin{subequations}
\begin{align} \label{E:EulerEinsteinIrrotational}
	R_{\mu \nu} - \Lambda g_{\mu \nu} - T_{\mu \nu}^{(scalar)} + \frac{1}{2}T^{(scalar)} g_{\mu \nu} = 0, 
		&& (\mu, \nu = 0,1,2,3),
\end{align}
where $T_{\mu \nu}^{(scalar)} = 2 \sigma^{s} (\partial_{\mu} \Phi) (\partial_{\nu} \Phi) + g_{\mu \nu} (s + 1)^{-1} \sigma^{s + 1}$ is as in \eqref{E:Tslowerperfectfluid}, together with \eqref{E:ELsigmas}, the equation of motion for an irrotational fluid: 

\begin{align} \label{E:irrotationalfluid}
	D_{\alpha} (\sigma^{s} D^{\alpha} \Phi) = 0.
\end{align}
\end{subequations}

\section{Background Solutions} \label{S:backgroundsolution}

Our main results concern the future-stability (with respect to irrotational perturbations) of a class of background solutions $([0,\infty) \times \mathbb{T}^3,\widetilde{g}, \partial \widetilde{\phi})$ to the system \eqref{E:EulerEinsteinIrrotational} - \eqref{E:irrotationalfluid}. These background solutions, which are of 
FLRW type\footnote{Technically, the term ``FLRW'' is usually reserved for a class of solutions that have spatial slices diffeomorphic to $\mathbb{S}^3,$ $\mathbb{R}^3,$ or hyperbolic space (see e.g. \cite{rW1984}).}, physically represent the evolution of an initially \emph{uniform} \emph{quiet} fluid in a spacetime that is undergoing rapid expansion. To find these solutions, we follow a procedure outlined in chapter 5 of \cite{rW1984} which, under appropriate ansatzes, reduces the Euler-Einstein equations to ODEs. Although our goal is to find ODE solutions to the irrotational equations \eqref{E:EulerEinsteinIrrotational} - \eqref{E:irrotationalfluid}, the procedure we follow will produce ODE solutions
to the full Euler-Einstein system \eqref{E:EinsteinField} $+$ \eqref{E:Tlowerfluid}. However, these ODE solutions will turn out to be irrotational. Thus, as discussed in Section \ref{SS:IrrotationalFluids}, these solutions can also be interpreted as solutions to the irrotational system. We remark that the derivation of these solutions is very well-known, and that we have provided it only for convenience.

\subsection{Derivation of the background solution}

To proceed, we first make the ansatz that the background metric $\widetilde{g}$ is of the form

\begin{align} \label{E:backgroundmetricform}
	\widetilde{g} = -dt^2 + a^2(t) \sum_{i=1}^3 (dx^i)^2,
\end{align}
from which it follows that the only corresponding non-zero Christoffel symbols are
\begin{align} \label{E:BackgroundChristoffel}
	\widetilde{\Gamma}_{j \ k}^{\ 0} = \widetilde{\Gamma}_{k \ j}^{\ 0}  = a \dot{a} \delta_{jk}, 
		&& \widetilde{\Gamma}_{j \ 0}^{\ k} = \widetilde{\Gamma}_{0 \ j}^{\ k}  = \omega \delta_j^k, && (j,k=1,2,3), 
\end{align}
where 

\begin{align}
	\omega \eqdef \frac{\dot{a}}{a} \label{E:omegadef},
\end{align}
and $\dot{a} \eqdef \frac{d}{dt} a.$ Using definitions \eqref{E:Riccidef} and \eqref{E:Rdef}, together with \eqref{E:BackgroundChristoffel}, we compute that

\begin{subequations}
\begin{align} \label{E:BackgroundEinsteintensor}
	\widetilde{R}_{00} - \frac{1}{2} \widetilde{R} \widetilde{g}_{00} & = 3 \Big(\frac{\dot{a}}{a}\Big)^2, && \\
	\widetilde{R}_{0j} - \frac{1}{2} \widetilde{R} \widetilde{g}_{0j} & = 0,  && (j=1,2,3),\\
	\widetilde{R}_{jk} - \frac{1}{2} \widetilde{R} \widetilde{g}_{jk} & = - (2a \ddot{a} + \dot{a}^2)\delta_{jk},
		&& (j,k=1,2,3). 
\end{align}
\end{subequations}
We then assume that $\widetilde{\rho} = \widetilde{\rho}(t), \widetilde{p} = \widetilde{p}(t),$ and 
$\widetilde{u}^{\mu} \equiv (1,0,0,0),$ which implies that 
$\partial \widetilde{\Phi} = (\partial_t \widetilde{\Phi}(t),0,0,0).$ We also assume that the equation of state
\eqref{E:EOS} holds. Inserting these ansatzes into the Einstein equations \eqref{E:EinsteinField} $+$ \eqref{E:Tlowerfluid}, we deduce (as in \cite{rW1984}) the following equations, which are known as the \emph{Friedmann equations} in the cosmology
literature:

\begin{subequations}
\begin{align}
	\widetilde{\rho} a^{3(1 + \speed^2)} & \equiv \bar{\rho}  \mathring{a}^{3(1 + \speed^2)} \eqdef \mathring{\kappa},
		\label{E:backgroundrhoafact} \\
	\dot{a} = a \sqrt{\frac{\Lambda}{3} + \frac{\widetilde{\rho}}{3}} & = a \sqrt{\frac{\Lambda}{3} 
		+ \frac{\mathring{\kappa}}{3a^{3(1 + \speed^2)}}}, \label{E:backgroundaequation}
\end{align} 
\end{subequations}
where the \emph{positive} constant $\bar{\rho}$ denotes the initial (uniform) energy density, and
the \emph{positive} constant $\mathring{a}$ is defined by $\mathring{a} \eqdef a(0).$ Observe that the rapid expansion of the background spacetime can be easily deduced from the ODE \eqref{E:backgroundaequation}, which suggests that the asymptotic behavior $a(t) \sim e^{Ht}, H \eqdef \sqrt{\Lambda/3}.$ A more detailed analysis of $a(t)$ is given in Lemma \ref{L:backgroundaoftestimate}.

With $\widetilde{\sigma} \eqdef - \widetilde{g}^{\alpha \beta} (\partial_{\alpha} \widetilde{\Phi})(\partial_{\beta} \widetilde{\Phi}) = (\partial_t \widetilde{\Phi})^2,$ we use \eqref{E:Lequalsp}, \eqref{E:sArelations}, \eqref{E:Lintermsofsigma}, and \eqref{E:backgroundrhoafact} to compute that

\begin{align} \label{E:firstwidetildePhiarelation}
	\speed^2 \mathring{\kappa} a^{-3(1 + \speed^2)} = \speed^2 \widetilde{\rho} = \widetilde{p} = \frac{2\speed^2}{1 + \speed^2} 
		\widetilde{\sigma}^{(1+\speed^2)/(2\speed^2)} = \frac{2\speed^2}{1 + \speed^2} (\partial_t \widetilde{\Phi})^{(1+\speed^2)/(\speed^2)}.
\end{align}
The equalities in \eqref{E:firstwidetildePhiarelation} imply that
\begin{align} \label{E:partialtwidetildePhiisolated}
	\partial_t \widetilde{\Phi} = \bar{\Psi} a^{-3/(2s + 1)}
		= \bar{\Psi} e^{-\decayparameter \Omega},
\end{align}
where

\begin{align}
	\bar{\Psi} &  \eqdef \Big(\bar{\rho} \frac{s+1}{2s+1}\Big)^{1/(2s + 2)} \mathring{a}^\decayparameter, \label{E:barPsidef} \\
	\Omega(t) & \eqdef \ln \Big(a(t)\Big), \label{E:BigOmega} \\
	\decayparameter & \eqdef \frac{3}{2s + 1} = 3 \speed^2.
\end{align}

\begin{remark}
	$\Omega(t)$ has been introduced solely for cosmetic purposes.
\end{remark}

For future use, we note the following consequences of the above discussion:

\begin{subequations}
\begin{align}
	3 \omega^2 - \Lambda & = \widetilde{\rho} = \mathring{\kappa} e^{-2\decayparameter(s+1) \Omega}
		= \Big(\frac{2s+1}{s+1}\Big) \widetilde{\sigma}^{s+1}, \label{E:3omegaSquaredminusLambdaidentity} \\
	- 3 \dot{\omega} & =  3 \omega^2 - \Lambda +  \mathring{\kappa} \Big(\frac{s+2}{2s + 1} \Big) e^{-2\decayparameter(s+1) 
		\Omega} = 3 \mathring{\kappa} \Big(\frac{s+1}{2s + 1} \Big) e^{-2\decayparameter(s+1) \Omega}
		= 3 \widetilde{\sigma}^{s+1}. \label{E:omegadotidentity}
\end{align}
\end{subequations}

\subsection{Analysis of Friedmann's equation}

In the following lemma, we analyze the asymptotic behavior of solutions to the ODE \eqref{E:backgroundaequation}.

\begin{lemma} \label{L:backgroundaoftestimate}
	Let $\mathring{a}, \mathring{\kappa}, P > 0,$ and let $a(t)$ be the solution to the following ODE: 
	\begin{align}
		\frac{d}{dt}a & = a \sqrt{\frac{\Lambda}{3} + \frac{\mathring{\kappa}}{3a^{P}}}, && a(0) = \mathring{a}.
	\end{align}
	Then with $H \eqdef \sqrt{\Lambda/3},$ the solution $a(t)$ is given by
	
	\begin{align}
		a(t) & = \bigg\lbrace\mbox{sinh}\Big(\frac{P H t}{2}\Big) 
			\sqrt{\frac{\mathring{\kappa}}{3H^2} + \mathring{a}^{P}}
		 	+ \mathring{a}^{P/2} \mbox{cosh}\Big(\frac{P H t}{2} \Big) \bigg\rbrace^{2/P}, 
	\end{align}
	and for all integers $N \geq 0,$ there exists a constant $C_N > 0$ such that for all $t \geq 0,$ with \\
	$A \eqdef \bigg\lbrace \frac{1}{2} \Big(\sqrt{\frac{\mathring{\kappa}}{3H^2} + \mathring{a}^{P}} + 
	\mathring{a}^{P/2} \Big) \bigg\rbrace^{2/P},$
	we have that
	
	\begin{subequations}
	\begin{align}
		(1/2)^{2/P} \mathring{a} e^{Ht} \leq a(t) & \leq A e^{Ht}, \\
		\Big| e^{-Ht} \frac{d^N}{dt^N}a(t) - A H^N \Big| & \leq C_Ne^{-PHt}.
	\end{align}
	\end{subequations}
	
	Furthermore, for all integers $N \geq 0,$ there exists a constant $\widetilde{C}_N > 0$ such that
	for all $t \geq 0,$ with 
	\begin{align} 
		\omega \eqdef \frac{\dot{a}}{a}, 
	\end{align}	
	 we have that
	
	\begin{subequations}
	\begin{align}
		H \leq \omega(t) & \leq \sqrt{H^2 + \frac{\mathring{\kappa}}{3\mathring{a}^{P}}}, \\
		\Big| \frac{d^N}{dt^N}\big(\omega(t) - H\big) \Big| & \leq \widetilde{C}_N e^{-P Ht}.
	\end{align}
	\end{subequations}
\end{lemma}

\begin{remark}
	Because of equation \eqref{E:backgroundaequation}, we will assume for the remainder of the article that
	$P = 3(1 + \speed^2).$
\end{remark}

\begin{proof}
	We leave the elementary analysis of this ODE to the reader.
\end{proof}

\section{The Modified Irrotational Euler Einstein System} \label{S:ReducedEquations}
In this section, we introduce our version of wave coordinates. We then use algebraic identities that are valid in wave
coordinates to construct a modified version of the irrotational Einstein equations, which are a system of quasilinear wave equations containing energy-dissipative terms. Next, to facilitate our analysis in later sections, we algebraically decompose the modified system into principal terms and error terms.  Finally, we show that the modified system is equivalent to the un-modified system if the Einstein constraint equations and the wave coordinate condition are both satisfied along the Cauchy hypersurface $\Sigma.$

\subsection{Wave coordinates} \label{SS:harmoniccoodinates}
To hyperbolize the Einstein equations, we use a version of the well-known family of \emph{wave coordinate} systems. More specifically, we use a coordinate system in which the contracted Christoffel symbols 
$\Gamma^{\mu} \eqdef g^{\alpha \beta} \Gamma_{\alpha \ \beta}^{\ \mu}$ of the spacetime metric $g$ are equal to the 
contracted Christoffel symbols $\widetilde{\Gamma}^{\mu} \eqdef \widetilde{g}^{\alpha \beta} \widetilde{\Gamma}_{\alpha \ \beta}^{\ \mu}$ of the background metric $\widetilde{g}.$ This condition is known as a \emph{wave coordinate} condition since $\Gamma^{\mu} \equiv \widetilde{\Gamma}^{\mu}$ if and only if the coordinate functions\footnote{The $x^{\mu}$ are scalar-valued functions, despite the fact that they have indices.} 
$x^{\mu}$ are solutions to the wave equation $g^{\alpha \beta} D_{\alpha} D_{\beta} x^{\mu} + \widetilde{\Gamma}^{\mu} = 0.$ Using \eqref{E:backgroundmetricform} and \eqref{E:BackgroundChristoffel}, we compute that in such a coordinate system

\begin{align} \label{E:HarmonicGauge}
	\Gamma^{\mu} = \widetilde{\Gamma}^{\mu} = 3 \omega \delta_{0}^{\mu}, && \Gamma_{\mu} = 
		g_{\mu \alpha} \Gamma^{\alpha} = 3 \omega g_{0 \mu}, && (\mu=0,1,2,3),
\end{align}
where $\omega(t)$ is defined in \eqref{E:omegadef}.

We now introduce the tensorfields

\begin{subequations}
\begin{align}
	P^{\mu} & \eqdef 3 \omega \delta_0^{\mu}, && P_{\mu} = 3 \omega g_{0 \mu}, && (\mu=0,1,2,3), \\
	Q^{\mu} & \eqdef P^{\mu} - \Gamma^{\mu}, && Q_{\mu} = P_{\mu} - \Gamma_{\mu}, && (\mu=0,1,2,3). \label{E:Qdef}
\end{align}
\end{subequations}

The idea behind wave coordinates is to work in a coordinate system in which $Q^{\mu} \equiv 0,$ 
so that whenever it is convenient, we may replace $\Gamma^{\mu}$ with $3 \omega \delta_{0}^{\mu}$ (and vice-versa) without altering the content of the Einstein equations. The existence of such a coordinate system is nontrivial, and 
it was only in 1952 that Choquet-Bruhat \cite{cB1952} first showed that they exist in general. With this idea in mind, we define the \emph{modified Ricci tensor} $\hat{R}_{\mu \nu}$ by

\begin{align} \label{E:modifiedRicci}
	\hat{R}_{\mu \nu} & \eqdef R_{\mu \nu} + \frac{1}{2} \big(D_{\mu} Q_{\nu} + D_{\nu} Q_{\mu} \big)  \\
	& = -\frac{1}{2} \hat{\Square}_g g_{\mu \nu} + \frac{1}{2} \big(D_{\mu} P_{\nu} + D_{\nu}P_{\mu} \big)
		+ g^{\alpha \beta} g^{\gamma \delta} (\Gamma_{\alpha \gamma \mu} \Gamma_{\beta \delta \nu}
		+ \Gamma_{\alpha \gamma \mu} \Gamma_{\beta \nu \delta}
		+ \Gamma_{\alpha \gamma \nu} \Gamma_{\beta \mu \delta}), \notag
\end{align}
where

\begin{align} \label{E:reducedwaveoperator}
	\hat{\Square}_g \eqdef g^{\alpha \beta} \partial_{\alpha} \partial_{\beta}
\end{align}
is the \emph{reduced wave operator} for the metric $g.$

We now replace the $R_{\mu \nu}$ with $\hat{R}_{\mu \nu}$ in \eqref{E:EulerEinsteinIrrotational}, 
expand the covariant differentiation in \eqref{E:irrotationalfluid}, and
add additional inhomogeneous terms $I_{\mu \nu}$ and $I_{\partial \Phi}$ to the left-hand sides of
\eqref{E:EulerEinsteinIrrotational} and \eqref{E:irrotationalfluid} respectively, arriving at the 
\emph{modified irrotational Euler-Einstein system}:

\begin{subequations}
\begin{align}
	\hat{R}_{\mu \nu} - \Lambda g_{\mu \nu} - T_{\mu \nu}^{(scalar)} + \frac{1}{2} T^{(scalar)} g_{\mu \nu} + I_{\mu \nu} & = 0,
		&& (\mu, \nu = 0,1,2,3), \label{E:FirstmodifiedRicci} \\
	\big[\sigma g^{\alpha \beta} - 2s (\partial^{\alpha} \Phi)(\partial^{\beta} \Phi) \big] 
		\partial_{\alpha} \partial_{\beta} \Phi
		- \sigma \Gamma^{\alpha} \partial_{\alpha} \Phi
		+ 2s \Gamma^{\alpha \lambda \beta} (\partial_{\alpha} \Phi)(\partial_{\lambda} \Phi)  
		(\partial_{\beta} \Phi) + I_{\partial \Phi} & = 0. && \label{E:Firstmodifiedfluid}
\end{align}
\end{subequations}

Here, the additional terms are defined to be

\begin{subequations}
\begin{align} 
	I_{00} & \eqdef -2 \omega Q^0 = 2 \omega(\Gamma^0 - 3 \omega), \label{E:gaugetermI00} \\
	I_{0j} & = I_{j0} \eqdef 2 \omega Q_j = 2 \omega(3 \omega g_{0j} - \Gamma_j), & & (j=1,2,3), \label{E:gaugetermI0j} \\
	I_{jk} & = I_{jk} \eqdef 0, & & (j,k = 1,2,3), \label{E:gaugetermIjk} \\
	I_{\partial \Phi} & \eqdef - \sigma g^{\alpha \beta} Q_{\alpha} \partial_{\beta} \Phi
		= \sigma \Gamma^{\alpha} \partial_{\alpha} \Phi - 3 \omega \sigma \partial_t \Phi. \label{E:gaugetermIpartialPhi}
\end{align}
\end{subequations}

We have several important remarks to make concerning the modified system \eqref{E:FirstmodifiedRicci} - \eqref{E:Firstmodifiedfluid}. First, because the principal term on the left-hand side of 
\eqref{E:FirstmodifiedRicci} is $-\frac{1}{2} \hat{\Square}_g g_{\mu \nu},$ the modified system comprises a 
quasilinear system of wave equations and is of hyperbolic character. Second, the gauge terms $I_{\mu \nu}, I_{\partial \Phi}$
have been added to the system in order to produce an energy dissipation effect that is analogous to the 
effect created by the $3(\partial_t v)^2$ term on the right-hand side of the model equation \eqref{E:modelequation}. These dissipation-inducing terms play a key role in the global existence theorem of Section \ref{S:GlobalExistence}. Finally, in Section \ref{SS:ClassicalLocalExistence}, we will show that if the initial data satisfy the Gauss and Codazzi constraints \eqref{E:Gauss} - \eqref{E:Codazzi}, 
and if the wave coordinate condition $Q_{\mu}|_{t=0} = 0$ is satisfied,
then $Q_{\mu}, I_{\mu \nu}, I_{\partial \Phi} \equiv 0,$ and $\hat{R}_{\mu \nu} \equiv R_{\mu \nu};$ i.e., under these conditions, the solution to \eqref{E:FirstmodifiedRicci} - \eqref{E:Firstmodifiedfluid} is also a solution to the irrotational Euler-Einstein system \eqref{E:EulerEinsteinIrrotational}.

\subsection{Summary of the modified irrotational Euler-Einstein system for the equation of state 
\texorpdfstring{$p = \speed^2 \rho$}{}}

For convenience, we summarize (with the help of \eqref{E:TmunuminushalfTgmunu}) the results of the previous section by listing the modified irrotational Euler-Einstein system: 

\begin{subequations}
\begin{align}
	\hat{R}_{00} + 2 \omega \Gamma^0 - 6 \omega^2 - \Lambda g_{00} - 2 \sigma^{s} (\partial_t \Phi)^2 
		- \Big(\frac{s}{s + 1} \Big) \sigma^{s + 1} g_{00} & = 0, \label{E:Rhat00} \\
	\hat{R}_{0j} + 2 \omega (3 \omega g_{0j} - \Gamma_j) - \Lambda g_{0j} - 2 \sigma^{s} (\partial_t \Phi)
		(\partial_j \Phi) - \Big(\frac{s}{s + 1} \Big) \sigma^{s + 1} g_{0j} & = 0, \qquad(j=1,2,3), \label{E:Rhat0j} \\
	\hat{R}_{jk} - \Lambda g_{jk} - 2 \sigma^{s} (\partial_{j} \Phi) (\partial_{k} \Phi) 
		- \Big(\frac{s}{s + 1} \Big) \sigma^{s + 1} g_{jk} & = 0, \qquad (j,k=1,2,3), \label{E:Rhatjk} \\
	\big[\sigma g^{\alpha \beta} - 2s (\partial^{\alpha} \Phi)(\partial^{\beta} \Phi) \big] 
		\partial_{\alpha} \partial_{\beta} \Phi
		- 3 \omega \sigma \partial_t \Phi
		+ 2s \Gamma^{\alpha \lambda \beta} (\partial_{\alpha} \Phi)(\partial_{\lambda} \Phi)  
			(\partial_{\beta} \Phi) & = 0. \label{E:Irrotationalfluidsummary}
\end{align}
\end{subequations}

\subsection{Construction of initial data for the modified system} \label{SS:IDReduced}

In general, we may consider arbitrary data for the modified equations \eqref{E:Rhat00} - \eqref{E:Irrotationalfluidsummary}.
However, if the solution of the modified system is also to be a solution of the Einstein equations \eqref{E:EinsteinFieldsummary}, then we cannot choose the data arbitrarily. In this section, we assume that we are given initial data $(\Sigma, \bar{g}, \bar{K}, \bard \mathring{\Phi}, \mathring{\Psi})$ for the irrotational Euler-Einstein equations \eqref{E:EulerEinsteinIrrotational} - \eqref{E:irrotationalfluid} as described in Section \ref{SSS:InitialDataOriginalSystem}. In particular, we assume that it satisfies the constraints \eqref{E:Gauss} - \eqref{E:Codazzi}. We will use this data to construct initial data for the modified equations that lead to a solution $(\mathcal{M}, g, \partial \Phi)$ of both the modified system and the Einstein equations; recall that a solution solves both systems $\iff Q_{\mu} \equiv 0,$ where $Q_{\mu}$ is defined in \eqref{E:Qdef}.

To supply data for the modified equations, we must specify along $\Sigma = \lbrace t=0 \rbrace$ the full spacetime metric components $g_{\mu \nu}|_{t=0},$ their (future-directed) normal derivatives $\partial_t g_{\mu \nu}|_{t=0},$ $(\mu,\nu = 0,1,2,3),$ the spatial derivatives $\bard \Phi|_{t=0}$ of the fluid potential, and the normal derivative $\partial_t \Phi|_{t = 0}$ of the fluid potential. To satisfy the requirements
\begin{itemize}
	\item $\Sigma = \lbrace t=0 \rbrace$
	\item $\bar{g}$ is the first fundamental form of $\Sigma$
	\item $\partial_t$ is normal to $\Sigma$
	\item $\bar{K}$ is the second fundamental form of $\Sigma$
	\item $\bard \Phi|_{\Sigma} = \bard \mathring{\Phi}$
	\item $D_{\hat{N}} \Phi|_{\Sigma} = \mathring{\Psi},$ $D_{\hat{N}} \eqdef$ differentiation in the direction of
		the future-directed normal to $\Sigma,$
\end{itemize}
we set (for $j,k=1,2,3$)

\begin{align}
	g_{00}|_{t=0} & = -1, & g_{0j}|_{t=0} & = 0, & g_{jk}|_{t=0} & = \bar{g}_{jk}, \\
	\bard \Phi|_{t=0} & = \bard \mathring{\Phi}, & \partial_t \Phi|_{t=0} & = \mathring{\Psi}, & \partial_t g_{jk}|_{t=0} & = 2 
	\bar{K}_{jk}.
\end{align}

Furthermore, we need to satisfy the wave coordinate condition $Q_{\mu}|_{t=0} = 0,$ $(\mu = 0,1,2,3).$
To meet this need, we first calculate that

\begin{align}
	\Gamma_0|_{t=0} & = - \frac{1}{2} (\partial_{t} g_{00})|_{t=0} - \bar{g}^{ab} \bar{K}_{ab}, &&
		\label{E:Gamma0attequals0} \\
	\Gamma_j|_{t=0} & = -\partial_t g_{0j}|_{t=0} + \frac{1}{2} \bar{g}^{ab}(2 \partial_a \bar{g}_{bj}
		- \partial_j \bar{g}_{ab}), && (j=1,2,3). \label{E:Gammajattequals0}
\end{align}
Using \eqref{E:Gamma0attequals0} and \eqref{E:Gammajattequals0}, the condition $Q_{\mu}|_{t=0} = 0$ is easily seen to be 
equivalent to the following relations, where $\omega$ is defined in \eqref{E:omegadef}:

\begin{align}
	\partial_{t} g_{00}|_{t=0} & = 2 (- 3 \omega|_{t=0} \underbrace{g_{00}|_{t=0}}_{-1} - \bar{g}^{ab} 
		\bar{K}_{ab}) =  2 \big(3 \omega(0)  - \bar{g}^{ab} \bar{K}_{ab}\big), && \\
	\partial_t g_{0j}|_{t=0} & = - 3 \omega|_{t=0} \underbrace{g_{0 j}|_{t=0}}_{0} + \frac{1}{2} \bar{g}^{ab}(2 
		\partial_a \bar{g}_{bj} - \partial_j \bar{g}_{ab}) = 
		\bar{g}^{ab}\Big(\partial_a \bar{g}_{bj} - \frac{1}{2} \partial_j \bar{g}_{ab}\Big), && (j=1,2,3).
\end{align}
This completes our specification of the data for the modified equations.

\subsection{Decomposition of the modified irrotational Euler-Einstein system in wave coordinates} \label{SS:Decomposition}

Naturally, the key to our global existence theorem involves a careful analysis of the nonlinear terms. In order to better see their structure, we dedicate this section to a decomposition of the modified system \eqref{E:Rhat00} - \eqref{E:Irrotationalfluidsummary} into principal terms and error terms, which we denote by variations of the symbol $\triangle.$ The estimates of Section \ref{S:BootstrapConsequences} will justify the claim that the $\triangle$ terms are in fact error terms. We begin by recalling the previously mentioned rescaling $h_{jk}$ of the spatial indices of the metric:

\begin{align} \label{E:hjkdef}
	h_{jk} \eqdef e^{-2 \Omega} g_{jk}, && (j,k = 1,2,3).
\end{align}
The decomposition is captured in the next proposition.

\begin{proposition} \label{P:Decomposition} \textbf{(Decomposition of the modified equations)}
	The equations \eqref{E:Rhat00} - \eqref{E:Rhatjk} can be written as follows:

\begin{subequations}
\begin{align}
	\hat{\Square}_{g} (g_{00} + 1) & = 5 H \partial_t g_{00} + 6 H^2 (g_{00} + 1) + \triangle_{00},
		&& \label{E:finalg00equation} \\
	\hat{\Square}_{g} g_{0j} & = 3 H \partial_t g_{0j} + 2 H^2 g_{0j} - 2Hg^{ab}\Gamma_{a j b} + \triangle_{0j}, 
		&& (j=1,2,3), \label{E:finallg0jequation} \\
	\hat{\Square}_{g} h_{jk} & = 3H \partial_t h_{jk} + \triangle_{jk}, 
		&& (j,k = 1,2,3), \label{E:finalhjkequation}
\end{align}
\end{subequations}
where $H \eqdef \sqrt{\Lambda/3},$ and the error terms $\triangle_{\mu \nu}$ are defined in \eqref{E:triangle00} - \eqref{E:trianglejk} below.

	Furthermore, the fluid equation \eqref{E:Irrotationalfluidsummary} can be written as

\begin{align} \label{E:finalfluidequation} 
	\hat{\Square}_m \Phi	& = \decayparameter \omega \partial_t \Phi + \triangle_{\partial \Phi}, 
\end{align}
where $\omega(t)$ is defined in \eqref{E:omegadef}, $m^{\mu \nu}$ ($\mu, \nu = 0,1,2,3$) is the 
\textbf{reciprocal acoustical metric},
\\ 
$\hat{\Square}_m \eqdef -\partial_t^2 + m^{ab} \partial_a \partial_b + 2m^{0a} \partial_t \partial_a$
is the reduced wave operator corresponding to $m^{\mu \nu},$ the error term $\triangle_{\partial \Phi}$ is 
defined below in \eqref{E:triangledef}, 
		
		\begin{subequations}	
		\begin{align}
		\decayparameter & \eqdef \frac{3}{1 + 2s} = 3\speed^2, \label{E:wdef} \\
		m^{00} & = - 1, \\
		m^{jk} & = \frac{g^{jk} - \triangle_{(m)}^{jk}}
			{(1 + 2s) + \triangle_{(m)}}, \label{E:mjkdef} \\
		m^{0j} & = - \frac{g^{00}g^{0j}(1 + 2s) + 2e^{\Omega}\big[(s + 1)g^{0j}g^{0a} + sg^{00}g^{aj}\big]z_a   
			+ e^{2 \Omega}(g^{0j}g^{ab} + 2sg^{aj}g^{0b})z_a z_b }{(1 + 2s) + \triangle_{(m)}},  \label{E:m0jdef} \\
		z_j & \eqdef \frac{e^{-\Omega} \partial_j \Phi}{\partial_t \Phi},  \label{E:zjdef}
		\end{align}
		
		\begin{align}
		\triangle_{(m)} & = (1 + 2s)(g^{00} + 1)(g^{00} - 1)
			+ 2(1 + 2s) e^{\Omega} g^{00} g^{0a} z_a + e^{2 \Omega}(g^{00}g^{ab} + 2s g^{0a}g^{0b})z_a z_b, \label{E:trianglemdef} \\
		\triangle_{(m)}^{jk} & = (g^{00} + 1)g^{jk} + 2s g^{0j} g^{0k}
			+ 2e^{\Omega}(g^{jk} g^{0a} + 2s g^{0k} g^{aj})z_a
			+ e^{2 \Omega}(g^{jk} g^{ab} + 2s g^{aj}g^{bk})z_a z_b,  \label{E:trianglemjkdef} \\
		\triangle_{\partial \Phi} & = \omega \Big\lbrace\frac{3}{(1 + 2s) + \triangle_{(m)}} - \decayparameter 
		 	\Big\rbrace\partial_t \Phi \label{E:triangledef} \\
		& \ \ - \bigg\lbrace \frac{1}{(1 + 2s) + \triangle_{(m)}} \bigg\rbrace 
			\bigg\lbrace 3\omega (g^{00} + 1) \partial_t \Phi  + 6 \omega e^{\Omega} g^{0a}z_a \partial_t \Phi 
			+ (3 - 2s) \omega e^{2\Omega}g^{ab}z_a z_b \partial_t \Phi \notag \\
		& \ \ \ \ \ \ \ + 2s\Big(\triangle_{(\Gamma)}^{000} \partial_t \Phi
			+ [\triangle_{(\Gamma)}^{a00} + \triangle_{(\Gamma)}^{0a0} + \triangle_{(\Gamma)}^{00a}] \partial_a \Phi 
			+ e^{2 \Omega} [\triangle_{(\Gamma)}^{0ab} + \triangle_{(\Gamma)}^{a0b} + \triangle_{(\Gamma)}^{ab0}] 
				z_a z_b \partial_t \Phi 
			+ e^{2 \Omega} \triangle_{(\Gamma)}^{abc} z_a z_b \partial_c \Phi\Big) \bigg\rbrace, \notag
		\end{align}
		the $\triangle_{(\Gamma)}^{\mu \alpha \nu}$ ($\mu, \alpha, \nu = 0,1,2,3$) 
		are defined below in Lemma \ref{L:raisedChristoffeldecomposition},

\begin{align}
	\frac{1}{2} \triangle_{00} & = \triangle_{A,00} + \triangle_{C,00}
		+ [f(\partial \widetilde{\Phi}) - f(\partial \Phi)] 
		- (g_{00} + 1)  f(\partial \widetilde{\Phi}) - \frac{s}{s+1}(g_{00} + 1) \sigma^{s+1} \label{E:triangle00} \\
	& \ \ \ + \frac{5}{2} (\omega - H) \partial_t g_{00} + 3 (\omega^2 - H^2)(g_{00} + 1), \notag \\
	\frac{1}{2} \triangle_{0j} & = \triangle_{A,0j} + \triangle_{C,0j}
		+ \frac{s-1}{2(s+1)} \widetilde{\sigma}^{s+1} g_{0j} - 2 \sigma^{s} (\partial_t \Phi)(\partial_j \Phi) - \frac{s}{s + 1} 
		\sigma^{s+1} g_{0j} \label{E:triangle0j} \\
	& \ \ \ + \frac{3}{2} (\omega - H) \partial_t g_{0j} + (\omega^2 - H^2) g_{0j} - 
		(\omega - H)g^{ab}\Gamma_{a j b}, \notag \\
	\frac{1}{2} \triangle_{jk} & = e^{-2 \Omega} \triangle_{A,jk} - (g^{00} + 1) \widetilde{\sigma}^{s+1} h_{jk} 
			- 2 \omega g^{0a} \partial_{a} h_{jk} - 2e^{-2 \Omega} \sigma^s (\partial_j \Phi)(\partial_k \Phi) \label{E:trianglejk} \\ 
		& \ \ \  + \frac{s}{s+1} (\widetilde{\sigma}^{s+1} - \sigma^{s+1}) h_{jk} + \frac{3}{2}(\omega - H) \partial_t h_{jk}, 
			\notag \\
	f(\partial \Phi) & \eqdef 2 \sigma^s (\partial_t \Phi)^2 - \frac{s}{s+1} \sigma^{s+1}, \label{E:ffirstdef}
\end{align}
\end{subequations}
the $\triangle_{A,\mu \nu}$ are defined in \eqref{E:triangleA00def} - \eqref{E:triangleAjkdef}, and 
$\triangle_{C,00}, \triangle_{C,0j}$ are defined in \eqref{E:triangleC00def} - \eqref{E:triangleC0jdef}. In the above expressions, quantities associated to the background solution of Section \ref{S:backgroundsolution}
are decorated with the symbol $\ \widetilde \ \ .$ 
\end{proposition}

\begin{remark}
	Note that there are two metrics in this system of PDEs: the inverse spacetime metric $g^{\mu \nu},$ and 
	the reciprocal acoustical metric $m^{\mu \nu}.$
\end{remark}

\begin{remark}
	As discussed in Section \ref{SS:Commentsonanalysis}, the variable $z_j$ has been introduced to facilitate our
	analysis of the ratio of the size of the spatial derivatives of $\Phi$ to its time derivative.
\end{remark}

\begin{proof}
	The proof involves a series of tedious computations, some of which are contained in Lemmas \ref{L:modifiedRiccidecomposition} 
	- \ref{L:raisedChristoffeldecomposition} below. We sketch the proof of \eqref{E:finalhjkequation} and leave the remaining
	equations to the reader. To obtain \eqref{E:finalhjkequation}, we first use equation \eqref{E:Rhatjk}, Lemma 	
	\ref{L:modifiedRiccidecomposition}, and Lemma \ref{L:Amunudecomposition} to obtain the following equation 
	for $h_{jk}=e^{-2\Omega}g_{jk}:$

	\begin{align} \label{E:firsthjkequation}
		\hat{\Square}_g h_{jk} & = 3 \omega \partial_t h_{jk} - 2(g^{00} + 1)(\partial_t \omega) h_{jk} 
			- 4 \omega g^{0a} \partial_{a} h_{jk} + 2 \big[ e^{-2 \Omega} \triangle_{A,jk} - 2 e^{-2 \Omega} \sigma^s (\partial_j 
			\Phi)(\partial_k \Phi) \big] \\
		& \ \ \ + 2 \Big\lbrace (3 \omega^2 - \Lambda + \partial_t \omega)h_{jk} - \frac{s}{s+1} \sigma^{s+1} h_{jk} \Big\rbrace. 
			\notag 
	\end{align}
	We now use \eqref{E:3omegaSquaredminusLambdaidentity} - \eqref{E:omegadotidentity} to substitute
	in equation \eqref{E:firsthjkequation} for $3 \omega^2 - \Lambda$ and $\partial_t \omega$ in terms of $\widetilde{\sigma},$ 
	arriving at the following equation:
	
	\begin{align} \label{E:secondhjkequation}
		\hat{\Square}_g h_{jk} & = 3 \omega \partial_t h_{jk} + 2(g^{00} + 1) \widetilde{\sigma}^{s+1} h_{jk} 
			- 4 \omega g^{0a} \partial_{a} h_{jk} + 2 \big[ e^{-2 \Omega} \triangle_{A,jk} - 2 e^{-2 \Omega} \sigma^s (\partial_j 
			\Phi)(\partial_k \Phi) \big] \\ 
		& \ \ \  + 2 \frac{s}{s+1} (\widetilde{\sigma}^{s+1} - \sigma^{s+1}) h_{jk}. \notag
	\end{align}
	Equation \eqref{E:finalhjkequation} now easily follows from \eqref{E:secondhjkequation}. We remark that the proofs of 
	\eqref{E:finalg00equation} and \eqref{E:finallg0jequation} require the use of Lemma \ref{L:AmunuplusImunudecomposition},
	and that the proof of \eqref{E:finalfluidequation} requires the use of Lemma \ref{L:raisedChristoffeldecomposition}. To
	obtain \eqref{E:finalfluidequation}, it is also useful to note that
	$- (\partial_t \Phi)^2 \big\lbrace (1 + 2s) + \triangle_{(m)} \big\rbrace$ is the coefficient of 
	the differential operator $\partial_t^2$ in equation \eqref{E:Irrotationalfluidsummary}.
\end{proof}

We now state the following four lemmas, which are needed for the proof of Proposition \ref{P:Decomposition}.

\begin{lemma} \label{L:modifiedRiccidecomposition}
	The modified Ricci tensor from \eqref{E:modifiedRicci} can be decomposed as follows:
	\begin{align}
		\hat{R}_{\mu \nu} & = - \frac{1}{2} \hat{\Square}_g g_{\mu \nu}
			+ \frac{3}{2}(g_{0 \mu} \partial_{\nu} \omega + g_{0 \nu} \partial_{\mu} \omega)
			+ \frac{3}{2} \omega \partial_t g_{\mu \nu} + A_{\mu \nu}, &&  (\mu, \nu = 0,1,2,3), \\
			\mbox{where} \notag \\
		A_{\mu \nu} & = g^{\alpha \beta} g^{\kappa \lambda}
			\big[(\partial_{\alpha} g_{\nu \kappa})(\partial_{\beta} g_{\mu \lambda})
			- \Gamma_{\alpha \nu \kappa} \Gamma_{\beta \mu \lambda}\big], && (\mu, \nu = 0,1,2,3).	\label{E:Amunudef}
	\end{align}
	
\end{lemma}

\begin{proof}
	These computations are carried out in Lemma 4 of \cite{hR2008}.
\end{proof}

\begin{lemma} \label{L:Amunudecomposition}
	
	The term $A_{\mu \nu}$ ($\mu, \nu = 0,1,2,3$) defined in \eqref{E:Amunudef} can be decomposed into principal terms and error 
	terms $\triangle_{A,\mu \nu}$ as follows:
	\begin{subequations}
	\begin{align}
		A_{00} & = 3 \omega^2 - \omega g^{ab}\partial_t g_{ab} + 2 \omega g^{ab} \partial_a g_{0b} + \triangle_{A,00}, \\
		A_{0j} & = 2 \omega g^{00} \partial_t g_{0j} - 2 \omega^2 g^{00}g_{0j} - \omega g^{00} \partial_j g_{00} 
			+ \omega g^{ab}\Gamma_{ajb} + \triangle_{A,0j}, && (j=1,2,3), \\
		A_{jk} & = 2 \omega g^{00} \partial_t g_{jk} - 2 \omega^2 g^{00} g_{jk} + \triangle_{A,jk}, && (j,k = 1,2,3),
	\end{align}
	\end{subequations}
	where
	
	\begin{subequations}
	\begin{align}
		\triangle_{A,00} & = (g^{00})^2 \Big\lbrace(\partial_t g_{00})^2 - (\Gamma_{000})^2 \Big\rbrace 
			+ g^{00}g^{0a} \Big\lbrace 2 (\partial_t g_{00})(\partial_t g_{0a} + \partial_a g_{00}) 
			- 4 \Gamma_{000} \Gamma_{00a} \Big\rbrace \label{E:triangleA00def} \\
		& \ \ \ + g^{00}g^{ab} \Big\lbrace (\partial_t g_{0a})(\partial_t g_{0b}) 
				+ (\partial_a g_{00}) (\partial_b g_{00}) 
				- 2 \Gamma_{00a} \Gamma_{00b} \Big\rbrace \notag \\
		& \ \ \ + g^{0a} g^{0b} \Big\lbrace 2(\partial_t g_{00})(\partial_a g_{0b}) + 2(\partial_t g_{0b})(\partial_a g_{00}) 
				- 2 \Gamma_{000} \Gamma_{a0b} - 2 \Gamma_{00b} \Gamma_{00a} \Big\rbrace \notag \\
		& \ \ \ + g^{ab} g^{0l} \Big\lbrace 2(\partial_t g_{0a})(\partial_l g_{0b}) + 2(\partial_b g_{00})(\partial_a g_{0l}) 
				- 4\Gamma_{00a} \Gamma_{l0b}  \Big\rbrace \notag \\
		& \ \ \ + g^{ab}g^{lm}(\partial_a g_{0l})(\partial_b g_{0m}) 
				+ \frac{1}{2} g^{lm}(\underbrace{g^{ab} \partial_t g_{al} - 2\omega \delta_l^b}_{
				e^{2 \Omega} g^{ab} \partial_t h_{al} - 2 \omega g^{0b}g_{0l}})(\partial_b g_{0m} + \partial_m g_{0b}) \notag \\
		& \ \ \ - \frac{1}{4} g^{ab} g^{lm}(\partial_a g_{0l} + \partial_l g_{0a})(\partial_b g_{0m} + \partial_m g_{0b}) 
			- \frac{1}{4}(\underbrace{g^{ab} \partial_t g_{al} - 2\omega \delta_l^b}_{
				e^{2 \Omega} g^{ab} \partial_t h_{al} - 2 \omega g^{0b}g_{0l}}) 
				(\underbrace{g^{lm} \partial_t g_{bm} - 2 \omega \delta_b^l}_{e^{2 \Omega} g^{lm} \partial_t h_{bm} 
				- 2 \omega g^{0l}g_{0b}}), \notag \\
		\triangle_{A,0j} & =  (g^{00})^2 \Big\lbrace (\partial_t g_{00}) (\partial_t g_{0j}) - \Gamma_{000} \Gamma_{0j0} \Big\rbrace
			\label{E:triangleA0jdef} \\
		& \ \ \ + g^{00}g^{0a} \Big\lbrace (\partial_t g_{00})(\partial_t g_{aj} + \partial_a g_{0j}) 
			+(\partial_t g_{0j})(\partial_t g_{0a} + \partial_a g_{00}) 
			- 2 \Gamma_{000} \Gamma_{0ja} - 2 \Gamma_{0j0} \Gamma_{00a} \Big\rbrace \notag \\
		& \ \ \ + g^{00} (\underbrace{g^{ab} \partial_t g_{bj} - 2 \omega \delta_j^a}_{e^{2 \Omega} g^{ab} \partial_t h_{bj}
			- g^{0a}g_{0j}}) \Big(\partial_t g_{0a} - \frac{1}{2} \partial_a g_{00}\Big) 
			+ \frac{1}{2} g^{00} g^{ab}(\partial_a g_{00})(\partial_b g_{0j} + \partial_j g_{0b}) \notag \\
		& \ \ \ + g^{0a} g^{0b} \Big\lbrace (\partial_t g_{00})(\partial_a g_{bj}) + (\partial_t g_{0b})(\partial_a g_{0j})  			
			+ (\partial_a g_{00})(\partial_t g_{bj}) + (\partial_a g_{0b})(\partial_t g_{0j})  \notag  \\
		& \hspace{3.5in} - \Gamma_{000} \Gamma_{ajb} - 2 \Gamma_{00b} \Gamma_{0ja} - \Gamma_{a0b} \Gamma_{0j0} \Big\rbrace \notag \\
		& \ \ \ + g^{ab} g^{0l} \Big\lbrace (\partial_t g_{0a})(\partial_l g_{bj}) + (\partial_l g_{0a})(\partial_t g_{bj}) 
			+ (\partial_b g_{00})(\partial_a g_{lj}) + (\partial_b g_{0l})(\partial_a g_{0j}) - 2 \Gamma_{00a} \Gamma_{ljb} \Big\rbrace \notag \\
		& \ \ \ - g^{ab} g^{0l} \Big\lbrace (\partial_l g_{0a} + \partial_a g_{0l})\Gamma_{0jb} 
			-\frac{1}{2}(\partial_t g_{la})(\partial_b g_{0j} - \partial_j g_{0b}) \Big\rbrace \notag \\
		& \ \ \ + \omega g^{0a}(\underbrace{\partial_t g_{aj} - 2 \omega g_{aj}}_{e^{2 \Omega} \partial_t h_{aj}})
			+ \frac{1}{2}g^{0l}(\underbrace{g^{ab} \partial_t g_{la} - 2\omega \delta_{l}^b}_{e^{2 \Omega} g^{ab} \partial_t h_{la} 
				- 2 \omega g^{0b}g_{0l}})\partial_t g_{bj} \notag \\
		& \ \ \ + g^{ab} g^{lm} \Big\lbrace (\partial_a g_{0l})(\partial_b g_{mj}) 
			- \frac{1}{2} (\partial_a g_{0l} + \partial_l g_{0a})\Gamma_{bjm} \Big\rbrace 
			+ \frac{1}{2} g^{ab} (\underbrace{g^{lm} \partial_t g_{la} - 2\omega \delta_a^m}_{e^{2 \Omega} g^{lm} \partial_t h_{la} - 
			g^{0m}g_{0a}})\Gamma_{bjm}, \notag \\
		\triangle_{A,jk} & = (g^{00})^2 \Big\lbrace (\partial_t g_{0j}) (\partial_t g_{0k}) - \Gamma_{0j0} \Gamma_{0k0} \Big\rbrace
			\label{E:triangleAjkdef} \\
		& \ \ \ + g^{00}g^{0a} \Big\lbrace (\partial_t g_{0j})(\partial_t g_{ak} + \partial_a g_{0k}) 
			+ (\partial_t g_{0k})(\partial_t g_{aj} + \partial_a g_{0j}) 
			- 2 \Gamma_{0j0} \Gamma_{0ka} - 2 \Gamma_{0k0} \Gamma_{0ja} \Big\rbrace \notag \\
		& \ \ \ + g^{00}g^{ab} \Big\lbrace (\partial_a g_{0j})(\partial_b g_{0k})  
			- \frac{1}{2}(\partial_a g_{0j} - \partial_j g_{0a})(\partial_b g_{0k} - \partial_k g_{0b}) \Big\rbrace	\notag \\
		& \ \ \ - \frac{1}{2} g^{00} \Big\lbrace (\underbrace{g^{ab} \partial_t g_{aj} - 2 \omega \delta_j^b}_{
			e^{2 \Omega} g^{ab} \partial_t h_{aj} - 2 \omega g^{0b}g_{0j}})(\partial_b g_{0k} - \partial_k g_{0b})
			+ (\underbrace{g^{ab} \partial_t g_{bk} - 2 \omega \delta_k^a}_{e^{2 \Omega}g^{ab} \partial_t 
			h_{bk} - 2 \omega g^{0a}g_{0k}})(\partial_a g_{0j} - \partial_j g_{0a}) \Big\rbrace \notag \\
		& \ \ \ + \omega g^{00}(\underbrace{g_{bk}g^{ab} - \delta_k^a}_{-g_{0k}g^{0a}})\partial_t g_{aj} 
			+ \frac{1}{2} g^{00}(\underbrace{g^{ab}\partial_t g_{aj} - 2 \omega \delta_j^b}_{e^{2 \Omega} g^{ab} \partial_t 
			h_{aj} - 2 \omega g^{0b}g_{0j}})(\underbrace{\partial_t g_{bk} - 2 \omega g_{bk}}_{e^{2 \Omega}\partial_t h_{bk}}) 
			\notag \\
		& \ \ \ + g^{0a} g^{0b} \Big\lbrace (\partial_t g_{0j})(\partial_a g_{bk}) + (\partial_t g_{bj})(\partial_a g_{0k})  			+ 
			(\partial_a g_{0j})(\partial_t g_{bk}) + (\partial_a g_{bj})(\partial_t g_{0k})  \notag  \\
		& \hspace{3.5in} - \Gamma_{0j0} \Gamma_{akb} - 2 \Gamma_{0jb} \Gamma_{0ka} - \Gamma_{ajb} \Gamma_{0k0} \Big\rbrace \notag \\
		& \ \ \ + g^{ab} g^{0l} \Big\lbrace (\partial_t g_{aj})(\partial_l g_{bk}) + (\partial_l g_{aj})(\partial_t g_{bk}) 
			+ (\partial_b g_{0j})(\partial_a g_{lk}) + (\partial_b g_{lj})(\partial_a g_{0k}) 
			- 2 \Gamma_{0ja} \Gamma_{lkb} - 2 \Gamma_{lja} \Gamma_{0kb} \Big\rbrace \notag \\
		& \ \ \ + g^{ab} g^{ml} \Big\lbrace (\partial_a g_{lj})(\partial_b g_{mk}) - \Gamma_{ajl} \Gamma_{bkm} \Big\rbrace. 
		\notag
	\end{align}
	\end{subequations}
	
\end{lemma}

\begin{proof}
	These computations are carried out in Lemma 5 of \cite{hR2008}.
\end{proof}

\begin{lemma} \label{L:AmunuplusImunudecomposition}
	The sums $A_{00} + I_{00}$ and $A_{0j} + I_{0j},$ $(j=1,2,3),$ can be decomposed into principal terms and error terms as 
	follows, where $I_{00}, I_{0j}$ are defined in \eqref{E:gaugetermI00} - \eqref{E:gaugetermI0j}; $A_{00},A_{0j}$ are 
	defined in \eqref{E:Amunudef}; and $\triangle_{A,00}, \triangle_{A,0j}$ are defined in \eqref{E:triangleA00def} - 
	\eqref{E:triangleA0jdef}:
	
	\begin{subequations}
	\begin{align}
		A_{00} + 2 \omega \Gamma^0 - 6 \omega^2 & = \omega \partial_t g_{00} + 3 \omega^2(g_{00} + 1)
			+ 3 \omega^2 g_{00} + \triangle_{A,00} + \triangle_{C,00}, && \\
		A_{0j} + 2\omega(3 \omega g_{0j} - \Gamma_j) & = 4 \omega^2 g_{0j} - \omega g^{ab} \Gamma_{ajb} + \triangle_{A,0j} 
			+ \triangle_{C,0j}, && (j=1,2,3), 
	\end{align}
	\end{subequations}
	where
	\begin{subequations}
	\begin{align}
		\triangle_{C,00} & = -6 (g_{00})^{-1} \omega^2 \Big\lbrace (g_{00} + 1)^2 - g^{0a}g_{0a} \Big\rbrace 
			- \omega (g^{00} + 1)(\underbrace{g^{ab} \partial_t g_{ab} - 6 \omega}_{e^{2\Omega}g^{ab} \partial_t h_{ab} - 2 \omega 
			g^{0a}g_{0a}}) + 2\omega(g^{00} + 1)g^{ab} \partial_a g_{0b} && 
				\label{E:triangleC00def} \\
		& \ \ \ + \omega(g^{00} + 1)(g^{00} - 1) \partial_t g_{00} + 2\omega g^{00} g^{0a}(\Gamma_{0a0} + 2 \Gamma_{00a})
			+ 4\omega g^{0a} g^{0b} \Gamma_{0ab} + 2 \omega g^{ab} g^{0l} \Gamma_{alb}, \notag \\
		\triangle_{C,0j} & = 2 \omega^2(g^{00} + 1) g_{0j}
		 - 2 \omega g^{0a} \Big\lbrace (\underbrace{\partial_t g_{aj} - 2\omega g_{aj}}_{e^{2 \Omega} \partial_t h_{aj}}) + 
		 \partial_a g_{0j} - \partial_j g_{0a} \Big\rbrace, && (j=1,2,3). \label{E:triangleC0jdef}
	\end{align}
	\end{subequations}
\end{lemma}

\begin{proof}
	These computations are carried out in Lemma 6 of \cite{hR2008}.
\end{proof}

\begin{lemma} \label{L:raisedChristoffeldecomposition}
	The fully raised Christoffel symbols $\Gamma^{\mu \alpha \nu}$ ($\mu, \alpha, \nu = 0,1,2,3$) 
	can be decomposed into principal terms and
	error terms $\triangle_{(\Gamma)}^{\mu \alpha \nu}$ as follows:
	
	\begin{subequations}
	\begin{align}
		\Gamma^{000} & = \triangle_{(\Gamma)}^{000},  \\
		\Gamma^{j00} & = \Gamma^{00j} = \triangle_{(\Gamma)}^{j00}, && (j=1,2,3),   \\
		\Gamma^{0j0} & = \triangle_{(\Gamma)}^{0j0}, && (j=1,2,3), \\
		\Gamma^{0jk} & = \Gamma^{kj0} = - \omega g^{jk} + \triangle_{(\Gamma)}^{0jk}, && (j,k=1,2,3),   \\
		\Gamma^{j0k} & = \Gamma^{k0j} = \omega g^{jk}  + \triangle_{(\Gamma)}^{j0k}, && (j,k=1,2,3),   \\
		\Gamma^{ijk} & = \triangle_{(\Gamma)}^{ijk}, && (i,j,k = 1,2,3), 
\end{align}
\end{subequations}
where
	
	\begin{subequations}
	\begin{align}
		\triangle_{(\Gamma)}^{000} & = \frac{1}{2} (g^{00})^3 \partial_t g_{00}
			+ \frac{1}{2} g^{00} g^{0a} g^{0b} (\partial_t g_{ab} + 2 \partial_a g_{0b}) \label{E:triangleGamma000} \\
		& \ \ \ + \frac{1}{2} (g^{00})^2 g^{0a} (2\partial_t g_{0a} + \partial_{a}g_{00})
			+ \frac{1}{2} g^{0a}g^{0b}g^{0l} \partial_a g_{bl}, \notag \\
		\triangle_{(\Gamma)}^{j00} & = \frac{1}{2} (g^{00})^2 (g^{0j} \partial_t g_{00} + g^{aj}\partial_a g_{00})
			+ g^{00} g^{0j} g^{a0} \partial_t g_{0a} + g^{00} g^{aj} g^{0b} \partial_a g_{0b}
			+ \frac{1}{2} g^{aj} g^{0b} g^{0l} \partial_a g_{bl}, \label{E:triangleGammaj00}  \\
		\triangle_{(\Gamma)}^{0j0} & = \frac{1}{2}(g^{00})^2(g^{0j} \partial_t g_{00} + 2 g^{aj} \partial_t g_{0a} 
			- g^{aj} \partial_a g_{00}) \label{E:triangleGamma0j0} \\ 
		& \ \ \ + g^{00}(g^{aj}g^{0b} \partial_t g_{ab}
			+ g^{a0} g^{0j} \partial_a g_{00} + g^{a0} g^{bj} \partial_a g_{0b} 
			- g^{aj} g^{0b} \partial_a g_{0b}) \notag \\
		& \ \ \ + g^{0a} g^{bj} g^{0l} \partial_a g_{bl} - \frac{1}{2} g^{0a} g^{bj} g^{0l} \partial_b g_{al}, \notag \\
		\triangle_{(\Gamma)}^{0jk} & =  \frac{1}{2} g^{00}\Big\lbrace g^{0j} g^{0k} \partial_t g_{00}
			+ g^{0j} g^{ak} \partial_a g_{00} + g^{aj} g^{0k} (2\partial_t g_{0a} - \partial_a g_{00})
			+ g^{aj} g^{kb}(\partial_b g_{0a} - \partial_a g_{0b}) \Big\rbrace \label{E:triangleGamma0jk}  \\
		& \ \ \ + \frac{1}{2} g^{0a}\Big\lbrace g^{0j} g^{0k} \partial_a g_{00} + g^{0j} g^{bk}(\partial_a g_{0b} + 
			\partial_b g_{0a} - \partial_t g_{ab})  \notag \\
		& \hspace{1in} + g^{bj} g^{0k} (\partial_t g_{ab} + \partial_a g_{0b} - \partial_b g_{0a})
			+ g^{bj} g^{lk}(\partial_a g_{bl} + \partial_l g_{ab} - \partial_b g_{al}) \Big\rbrace \notag \\ 
		& \ \ \ + \frac{1}{2} g^{00} \big(\underbrace{g^{aj} g^{bk} \partial_t g_{ab} - 2\omega g^{jk}}_{
			g^{bk}(e^{2 \Omega} g^{aj}\partial_t h_{ab} - 2 \omega g^{0j}g_{0b})} \big)
			+ \omega (g^{00} + 1) g^{jk}, \notag \\
		\triangle_{(\Gamma)}^{j0k} & = \frac{1}{2}g^{00} \Big\lbrace g^{0j}g^{0k} \partial_t g_{00}
			+ g^{0j} g^{ak} \partial_a g_{00} + g^{aj} g^{0k} \partial_a g_{00} 
			+ g^{aj} g^{bk} (\partial_a g_{0b} + \partial_b g_{0a}) \Big\rbrace \label{E:triangleGammaj0k}  \\ 
		& \ \ \ + \frac{1}{2} g^{0a} \Big\lbrace g^{0j}g^{0k}(2\partial_t g_{0a} - \partial_a g_{00}) 
			+ g^{0j}g^{bk}(\partial_t g_{ab} + \partial_b g_{0a} - \partial_a g_{0b})
			 \notag \\
		&	\hspace{1in} + g^{bj} g^{0k} (\partial_t g_{ab} + \partial_b g_{0a} - \partial_a g_{0b})
			+ g^{bj} g^{lk}(\partial_b g_{al} + \partial_l g_{ab} - \partial_a g_{bl}) \Big\rbrace \notag \\
		& \ \ \ - \frac{1}{2} g^{00} \big(\underbrace{g^{aj} g^{bk} \partial_t g_{ab} - 2\omega g^{jk}}_{
			g^{bk}(e^{2 \Omega} g^{aj}\partial_t h_{ab} - 2 \omega g^{0j} g_{0b})} \big)
		  - \omega (g^{00} + 1) g^{jk}, \notag \\
		\triangle_{(\Gamma)}^{ijk} & = \frac{1}{2}g^{0i} g^{0j} g^{0k} \partial_t g_{00}
			+ \frac{1}{2}\big(g^{ai} g^{0j} g^{0k} + g^{0i} g^{0j} g^{ak} - g^{0i} g^{aj} g^{0k}\big) \partial_a g_{00}
			+ g^{0i} g^{aj} g^{0k} \partial_t g_{0a} \label{E:triangleGammaijk} \\
		& \ \ \ + \frac{1}{2}\big(g^{ai} g^{bj} g^{0k} + g^{ai} g^{0j} g^{bk} + g^{bi} g^{0j} g^{ak} + g^{0i} g^{bj} g^{ak}
			- g^{bi} g^{aj} g^{0k} - g^{0i} g^{aj} g^{bk}\big)\partial_a g_{0b} \notag \\
		& \ \ \ + \frac{1}{2}\big(g^{0i} g^{aj} g^{bk} + g^{ai} g^{bj} g^{0k} - g^{ai} g^{0j} g^{bk}\big)\partial_t g_{ab} 
		  + \frac{1}{2}\big(g^{ai} g^{bj} g^{lk} + g^{bi} g^{lj} g^{ak} - g^{bi} g^{aj} g^{lk}\big) \partial_a g_{bl}. \notag 
	\end{align}
	\end{subequations}
\end{lemma}

\begin{proof}
	The proof is again a series of tedious computations that follow from the formula \\
	$\Gamma^{\mu \alpha \nu} = \frac{1}{2} g^{\mu \kappa} g^{\alpha \lambda} g^{\nu \sigma}
	(\partial_{\kappa} g_{\lambda \sigma} + \partial_{\sigma} g_{\kappa \lambda} - \partial_{\lambda} g_{\kappa \sigma}).$ 
\end{proof}

\subsection{Classical local existence} \label{SS:ClassicalLocalExistence}

In this section, we discuss classical local existence results for the modified system of PDEs \eqref{E:finalg00equation} - \eqref{E:finalfluidequation}. The theorems in this section are stated without proof; they can be proved using standard methods described in e.g. \cite[Ch. VI]{lH1997}, \cite[Ch. 16]{mT1997III}. Also see \cite[Proposition 1]{hR2008}.

\begin{theorem} \label{T:LocalExistence} \textbf{(Local existence)}
	Let $N \geq 3$ be an integer. Let $\mathring{g}_{\mu \nu} = g_{\mu \nu}|_{t=0},$ $2 \mathring{K}_{\mu \nu} = (\partial_t 
	g_{\mu \nu})|_{t=0},$ $(\mu, \nu = 0,1,2,3),$ $\bard \mathring{\Phi} = \bard \Phi|_{t=0},$ and $\mathring{\Psi} = (\partial_t 
	\Phi)|_{t=0}$ be initial data (not necessarily satisfying the Einstein constraints) on the manifold 
	$\Sigma = \mathbb{T}^3$ for the modified system \eqref{E:finalg00equation} - \eqref{E:finalfluidequation} 
	satisfying (for $j,k = 1,2,3$)

	\begin{subequations}
	\begin{align}
		\mathring{g}_{00} + 1 \in H^{N+1}, && \mathring{g}_{0j} \in H^{N+1}, && \bard \mathring{g}_{jk} \in H^{N}, \\
		\mathring{K}_{00} \in H^{N}, && \mathring{K}_{0j} \in H^{N}, && \mathring{K}_{jk} - \omega(0)e^{2 
			\Omega(0)}\mathring{g}_{jk} \in H^{N}, \\ 
		\ \bard \mathring{\Phi} \in H^{N}, && \mathring{\Psi} - \bar{\Psi} \in H^N, && 
	\end{align}
	\end{subequations}
	where $\bar{\Psi} > 0$ is a constant. Assume further that there are constants $C_1, C_2 > 1$ 
	and $C_3 > 0$ such that
	
	\begin{align} \label{E:localexistencemathringgjklowerequivalenttostandardmetric}
		C_1^{-1} \delta_{ab} X^a X^b & \leq \mathring{g}_{ab}X^a X^b \leq C_1 \delta_{ab} X^a X^b,
			&& \forall(X^1,X^2,X^3) \in \mathbb{R}^3, \\
		C_2^{-1} \delta^{ab} X_a X_b & \leq \mathring{m}^{ab}X_a X_b \leq C_2 \delta^{ab} X_a X_b,
			&& \forall(X_1,X_2,X_3) \in \mathbb{R}^3, \\
		\mathring{g}_{00} & \leq - C_3, &&
	\end{align}
	where the $m^{jk},$ $(j,k=1,2,3)$ are the spatial components of the reciprocal acoustical metric (see \eqref{E:mjkdef}).
	Then these data launch a unique classical solution $(g_{\mu \nu}, \partial_{\mu} \Phi),$ $(\mu, \nu = 0,1,2,3),$ to the 
	modified system existing on an interval $(T_-, T_+),$ with $T_- < 0 < T_+,$ such that
	
	\begin{align}
		g_{\mu \nu} \in C_b^{N-1}((T_-, T_+) \times \mathbb{T}^3), && \partial_{\mu} 
		\Phi \in C_b^{N-2}((T_-, T_+) \times \mathbb{T}^3),
	\end{align}
	such that $g_{00} < 0,$ and such that the eigenvalues of the $3 \times 3$ matrices $g_{jk}$ and $m^{jk}$ are uniformly bounded
	below from $0$ and from above. 
	
	The solution has the following regularity properties:
	
	\begin{subequations}
	\begin{align}
		\ \ g_{00} + 1 \in C^0((T_-, T_+),H^{N+1}), &&  \ \ g_{0j} \in C^0((T_-, T_+),H^{N+1}),
			&& \bard g_{jk} \in C^0((T_-, T_+),H^{N}), \\
		\partial_t g_{00} \in C^0((T_-, T_+),H^{N}), && \partial_t g_{0j} \in C^0((T_-, T_+),H^{N}),	
			&& \partial_t g_{jk} - 2\omega(t) g_{jk} \in C^0((T_-, T_+),H^{N}), \\ 
		\bard \Phi \in C^0((T_-, T_+),H^{N}), && e^{\decayparameter \Omega}\partial_t \Phi - \bar{\Psi} \in C^0((T_-, T_+),H^{N}).
	\end{align}
	\end{subequations}
	Furthermore, $g_{\mu \nu}$ is a smooth Lorentz metric on $(T_-,T_+) \times \mathbb{T}^3,$
	and the sets $\lbrace t \rbrace \times \mathbb{T}^3$ are Cauchy hypersurfaces in the Lorentzian manifold 
	$(\mathcal{M} \eqdef (T_-, T_+) \times \mathbb{T}^3,g)$ for $t \in (T_-, T_+).$
	
	In addition, there exists an open neighborhood $\mathcal{O}$ of
	$(\mathring{g}_{\mu \nu}, \mathring{K}_{\mu \nu}, \bard \mathring{\Phi}, \mathring{\Psi})$
	such that all data belonging to $\mathcal{O}$ launch solutions that also exist on the interval $(T_-, T_+)$ and that have the 
	same regularity properties as $(g,\partial \Phi).$ Furthermore, on $\mathcal{O},$ the map from the initial data to the 
	solution is continuous.\footnote{By continuous, we mean continuous relative to the norms on the data and the norms on the 
	solution that are stated in the hypotheses and above conclusions of the theorem.}
	
	Finally, if, as described in Section \ref{SS:IDReduced}, the data for the modified system are constructed from data for the 
	irrotational Einstein-Euler system satisfying the constraints \eqref{E:Gauss} - \eqref{E:Codazzi} on an open subset 
	$S \subset \mathbb{T}^3,$ and if the wave coordinate condition $Q_{\mu}|_{S} = 0$ holds, then $(g_{\mu \nu}, 
	\partial_{\mu} \Phi)$ is also a solution to the un-modified equations \eqref{E:EinsteinFieldsummary} - 
	\eqref{E:FluidEquationSummary} on 
	$\mathcal{D}(S),$ the Cauchy development of $S.$

\end{theorem}

\begin{remark}
	The hypotheses in Theorem \ref{T:LocalExistence} have been stated in a manner that allows us to apply to it initial data 
	near that of the background solution of Section \ref{S:backgroundsolution}. Furthermore, we remark that
	the assumptions and conclusions concerning the metric components $g_{jk}$ would appear more natural if expressed in terms of
	the variables $h_{jk} \eqdef e^{-2 \Omega} g_{jk}$; these rescaled quantities are the ones that we use in our global 
	existence proof.
\end{remark}

\begin{remark}
	The fact that $(g_{\mu \nu},\partial_{\mu} \Phi)$ is also a solution to the modified equations if the constraints and the 
	wave coordinate condition $Q_{\mu}|_{\Sigma} = 0$ are satisfied is discussed 
	more fully in Section \ref{SS:PreservationofHarmonicGauge}.
\end{remark}

In our proof of Theorem \ref{T:GlobalExistence}, we will use the following continuation principle, which provides
criteria that are sufficient to ensure that a solution to the modified equations exists globally in time.
See \cite[Theorem 6.4.11]{lH1997} for the ideas behind a proof.

\begin{theorem} \label{T:ContinuationCriterion} \textbf{(Continuation principle)}
	Let $T_{max}$ be the supremum over all times $T_+$ such that the solution $(g_{\mu \nu}, \partial_{\mu} \Phi),$ 
	$(\mu,\nu=0,1,2,3),$ exists on the 
	interval $[0,T_+)$ and has the properties stated in the conclusions of Theorem \ref{T:LocalExistence}. Let $m^{\mu \nu}$ be 
	the reciprocal acoustical metric, which is defined above in Proposition \ref{P:Decomposition}. 
	Then if $T_{max} < \infty,$ one of the following four possibilities must occur:
	\begin{enumerate}
		\item There is a sequence $(t_n,x_n) \in [0,T_{max}) \times \mathbb{T}^3$ such that
			$\lim_{n \to \infty} g_{00}(t_n,x_n) = 0.$
		\item There is a sequence $(t_n,x_n) \in [0,T_{max}) \times \mathbb{T}^3$ such that the smallest eigenvalue of
			the $3 \times 3$ matrix $g_{jk}(t_n,x_n)$ converges to $0$ as $n \to \infty.$ 
		\item There is a sequence $(t_n,x_n) \in [0,T_{max}) \times \mathbb{T}^3$ such that the smallest eigenvalue of
			the $3 \times 3$ matrix $m^{jk}(t_n,x_n)$ converges to $0$ as $n \to \infty.$ 
		\item $\lim_{t \to T_{max}} \sup_{0 \leq \tau \leq t} \Bigg\lbrace 
			\sum_{\mu,\nu =0}^3 \| g_{\mu \nu} \|_{C_b^1}
			+ \sum_{\mu,\nu =0}^3 \| m^{\mu \nu} \|_{C_b^1}
			+ \sum_{\mu = 0}^3 \| \partial_{\mu} \Phi \|_{L^{\infty}} \Bigg\rbrace = \infty.$
		\end{enumerate}
		Similar results hold for an interval of the form $(T_{min},0].$
\end{theorem}

\begin{remark}   \label{R:mLorentziancondition}
	If one the first two possibilities occurs, then the hyperbolicity of equations 
	\eqref{E:finalg00equation} - \eqref{E:finalhjkequation} breaks down. Similarly, if the third possibility occurs, then the 
	hyperbolicity of equation \eqref{E:finalfluidequation} breaks down.
\end{remark}

\subsection{Preservation of the wave coordinate condition} \label{SS:PreservationofHarmonicGauge}

In Section \ref{SS:IDReduced}, from given initial data for the Einstein equations, we constructed initial data for the modified equations that in particular satisfy the wave coordinate condition along the Cauchy hypersurface $\Sigma;$ i.e., $Q_{\mu}|_{t=0} = 0.$ As mentioned in the statement of Theorem \ref{T:LocalExistence}, these data launch a solution of both the modified equations and the Einstein equations. As mentioned in Section \ref{S:ReducedEquations}, this fact follows from the fact that $Q_{\mu} = 0$ in $\mathcal{D}(\Sigma).$ In the next proposition, we prove this fact.

\begin{proposition} \label{P:Preservationofgauge} \textbf{(Preservation of wave coordinates)}
	Let $\big(\Sigma = \lbrace x \in \mathcal{M} \ | \ t=0 \rbrace, \bar{g}_{jk}, \bar{K}_{jk}, \bard \mathring{\Phi}, 
	\mathring{\Psi}\big),$ \\
	$(j,k=1,2,3),$ be initial data for the Einstein equations \eqref{E:EinsteinFieldsummary} 
	- \eqref{E:FluidEquationSummary} that satisfy 
	the constraints \eqref{E:Gauss} - \eqref{E:Codazzi}. Let $\big( g_{\mu \nu}|_{t=0}, \partial_t g_{\mu \nu}|_{t=0}, 
	\bard \mathring{\Phi}, \mathring{\Psi}\big),$ $(\mu,\nu=0,1,2,3),$ be initial data for the modified equations 
	\eqref{E:FirstmodifiedRicci} - \eqref{E:Firstmodifiedfluid} that are constructed from the data for 
	the Einstein equations as described in Section \ref{SS:IDReduced}. In particular, we recall that the construction of Section 
	\ref{SS:IDReduced} leads the fact that $Q_{\mu}|_{t=0} = 0$ where $Q_{\mu}$
	is defined in \eqref{E:Qdef}. Let $(\mathcal{M},g_{\mu \nu}, \partial_{\mu} \Phi)$ be the maximal 
	globally hyperbolic development of the data for the modified equations. Then $Q_{\mu} = 0$ in $\mathcal{D}(\Sigma).$
\end{proposition}

\begin{remark}
	Although we assume that the equation of state satisfies the assumptions of Section \ref{S:IrrotationalEE}, 
	we do not assume here that it is of the form $p= \speed^2 \rho.$ 
\end{remark}

\begin{proof}

To prove Proposition \ref{P:Preservationofgauge}, we first compute that for a solution of the modified equations, it follows that $Q_{\mu}$ is a solution to the \emph{hyperbolic} system
\begin{align} \label{E:Qwaveequation}
	g^{\alpha \beta} D_{\alpha} D_{\beta} Q_{\mu} + R_{\mu}^{\ \alpha} Q_{\alpha} + 2 g^{\alpha \beta} D_{\alpha} I_{\mu \beta} - 
	g^{\alpha \beta} D_{\mu} I_{\alpha \beta} = - 4 I_{\partial \Phi} D_{\mu} \Phi, && (\mu=0,1,2,3).
\end{align}
where $I_{\mu \nu}$ and $I_{\partial \Phi},$ which depend linearly on the $Q_{\mu},$ are defined in \eqref{E:gaugetermI00} - \eqref{E:gaugetermIpartialPhi}. More specifically, to obtain equation \eqref{E:Qwaveequation}, apply $D^{\nu}$ to each side of equation \eqref{E:ModifiedEE} below, and use the Bianchi identity $D^{\nu}(R_{\mu \nu} - \frac{1}{2}R g_{\mu \nu}) = 0,$ equation \eqref{E:Inhomogeneousenergymomentumdivergence}, and the curvature relation $D_{\mu} D^{\alpha} Q_{\alpha} =
D^{\alpha} D_{\mu} Q_{\alpha} - R_{\mu}^{\ \alpha} Q_{\alpha}.$

Since \eqref{E:Qwaveequation} is a system of wave equations and is of hyperbolic character, the fact that $Q_{\mu} = 0$ in $\mathcal{D}(\Sigma)$ would follow from a standard uniqueness theorem for such systems (see e.g. \cite[Ch. VI]{lH1997}, \cite[Ch.16]{mT1997III}), together with the knowledge that \emph{both} $Q_{\mu}|_{\Sigma} = 0$ \emph{and} $\partial_t Q_{\mu}|_{\Sigma} = 0$ hold. However, in constructing the data for the modified equations, we have already exhausted our gauge freedom. Although the construction of Section \ref{SS:IDReduced} has led to the condition $Q_{\mu}|_{\Sigma} = 0,$ it seems that we have no way to enforce the condition $\partial_t Q_{\mu}|_{\Sigma} = 0.$ The remarkable fact, first exploited by Choquet-Bruhat in \cite{cB1952}, is that the assumption that the original data for the Einstein equations satisfies the constraints \eqref{E:Gauss} - \eqref{E:Codazzi}, together with the assumption
$Q_{\mu}|_{\Sigma} = 0,$ \emph{automatically imply} that $\partial_t Q_{\mu}|_{\Sigma} = 0.$ The remainder of the proof is dedicated to proving this fact.

First, using definition \eqref{E:modifiedRicci}, we compute that for a solution of the modified equation 
\eqref{E:FirstmodifiedRicci}, the following identity holds:

\begin{align} \label{E:ModifiedEE}
	R_{\mu \nu} - \frac{1}{2}R g_{\mu \nu} + \Lambda g_{\mu \nu} 
		- T_{\mu \nu}^{(scalar)} & = - \frac{1}{2} \big(D_{\mu} Q_{\nu} + D_{\nu} Q_{\mu} \big)
		+ \frac{1}{2} (D^{\alpha} Q_{\alpha}) g_{\mu \nu} - I_{\mu \nu} + \frac{1}{2} g^{\alpha \beta} I_{\alpha \beta} 
		g_{\mu \nu}, && (\mu,\nu=0,1,2,3).
\end{align}

The left-hand side of \eqref{E:ModifiedEE} is simply the difference of the left and right sides of the Einstein equations 
\eqref{E:EinsteinFieldsummary}. Since the initial data $(\Sigma, \bar{g}, \bar{K}, \bard \mathring{\Phi}, \mathring{\Psi})$ for the Einstein equations are assumed to satisfy the constraints \eqref{E:Gauss} - \eqref{E:Codazzi}, as described in Section \ref{SSS:InitialDataOriginalSystem}, it follows that the Einstein equations are satisfied along $\Sigma;$ i.e., the left-hand side of \eqref{E:ModifiedEE} vanishes at $t=0.$ Furthermore, since $Q_{\mu}|_{t=0} = 0,$ it follows from 
definitions \eqref{E:gaugetermI00} - \eqref{E:gaugetermIpartialPhi} that $I_{\mu \nu}|_{t=0}=0,$ $(\mu, \nu = 0,1,2,3),$ 
and $I_{\partial \Phi}|_{t=0}=0.$ Using these facts and \eqref{E:ModifiedEE}, we conclude that the following equation holds:

\begin{align} \label{E:partialtQmuequals0}
	- \frac{1}{2} \big(D_{\mu} Q_{\nu} + D_{\nu} Q_{\mu} \big)|_{t=0}
		+ \frac{1}{2} (D^{\alpha} Q_{\alpha})|_{t=0}  g_{\mu \nu}|_{t=0} = 0, && (\mu, \nu = 0,1,2,3). 
\end{align}

Now let $\hat{N}^{\mu}=\delta_0^{\mu}$ denote the future-directed unit normal to $\Sigma,$ and let 
$X_{(j)}^{\nu} = \delta_j^{\nu},$ a vectorfield tangent to $\Sigma.$ Contracting \eqref{E:partialtQmuequals0} 
with $\hat{N}^{\mu} X^{\nu},$ setting $t=0,$ using the facts that $Q_{\mu}|_{t=0} = 0$ and 
$(X_{(j)}^{\nu} \partial_{\nu} Q_{\mu})|_{t=0} = 0,$ we have that

\begin{align} \label{E:partialtQjvanishesinitially}
	0 & = -\frac{1}{2}\hat{N}^{\mu}X_{(j)}^{\nu}(\partial_{\mu} Q_{\nu} + \partial_{\nu} Q_{\mu})|_{t=0} = 
		\partial_t Q_j|_{t=0}, && (j = 1,2,3). 
\end{align}
We also contract \eqref{E:partialtQmuequals0} with $\hat{N}^{\mu} \hat{N}^{\nu}$ set $t=0,$ and use \eqref{E:partialtQjvanishesinitially}, the fact that $Q_{\mu}|_{t=0} = 0,$ and the fact that $g^{00}|_{t=0} = -1,$ obtaining

\begin{align} \label{E:partialtQ0vanishesinitially}
	\partial_t Q_0|_{t=0} = 0.
\end{align}
From \eqref{E:partialtQjvanishesinitially} and \eqref{E:partialtQ0vanishesinitially},
we conclude that the data for the system \eqref{E:Qwaveequation} are trivial.

\end{proof}

\section{Norms and Energies} \label{S:NormsandEnergies}

In this section, we define the Sobolev norms\footnote{Technically, $\supfluidnorm{N}$ is a norm of the difference of $\partial \Phi$ and the derivatives $\partial \widetilde \Phi$ of the background potential.} and energies that will play a central role in our global existence theorem of Section \ref{S:GlobalExistence}. Let us make a few comments on them. First, we remark that in Section \ref{S:EnergyNormEquivalence}, we will show that if the norms are sufficiently small, then they are equivalent to the energies; i.e., the energies can be used to control Sobolev norms of solutions. Next, we recall that the background solution variable $\partial \widetilde{\Phi}$ satisfies $\partial_t \widetilde{\Phi} = \bar{\Psi} e^{-\decayparameter \Omega},$ $\bard \widetilde{\Phi} = 0,$ where $\bar{\Psi} > 0$ is the constant defined in \eqref{E:barPsidef}. The quantity $\supfluidnorm{N},$ 
which is introduced below in \eqref{E:Fluidsupnorm}, measures the difference of the perturbed variable $(\partial_t \Phi, \bard \Phi)$ from the background $(\partial_t \widetilde{\Phi}, 0).$ We also follow Ringstr\"{o}m by introducing scalings by $e^{\alpha \Omega},$ where $\alpha$ is a number, in the definitions of the norms and energies. The effect of these scalings is that in our proof of global existence, a convenient and viable bootstrap assumption to make for these quantities is that they are of size $\epsilon,$ where $\epsilon$ is sufficiently small. Finally, we remark that the positive number $q$ that appears in this section and throughout this article is defined in \eqref{E:qdef} below, and we remind the reader that $h_{jk} \eqdef e^{-2\Omega} g_{jk},$ ($j,k = 1,2,3$).

\subsection{Norms for \texorpdfstring{$g$ and $\partial \Phi$}{}} \label{SS:NormsforgandpartialPhi}

In this section, we introduce the weighted Sobolev norms that will be used in Section \ref{S:BootstrapConsequences} 
to estimate the terms appearing in the modified equations. The weights are designed in order to make the bootstrap argument of Section \ref{S:GlobalExistence} easy to close.

\begin{definition} \label{D:Norms}
We define the norms $\gzerozeronorm{N}(t),$ $\gzerostarnorm{N}(t),$ $\hstarstarnorm{N}(t),$
$\gnorm{N}(t),$ $\fluidnorm{N}(t),$ $\totalnorm{N}(t),$ $\supgzerozeronorm{N}(t),$ $\supgzerostarnorm{N}(t),$ $\suphstarstarnorm{N}(t),$ $\supgnorm{N}(t),$ and $\supfluidnorm{N}(t)$ as follows:

\begin{subequations}
\begin{align}
	\gzerozeronorm{N} & \eqdef e^{q \Omega} 
			\| \partial_t g_{00} \|_{H^N} +  e^{q \Omega} \| g_{00} + 1 \|_{H^N}
			+ e^{(q - 1)\Omega} \| \bard g_{00} \|_{H^N}, \label{E:mathfrakSMg00} \\
	\gzerostarnorm{N} & \eqdef 
		\sum_{j=1}^3 \Big( e^{(q - 1) \Omega} \| \partial_t g_{0j} \|_{H^N} 
			+  e^{(q - 1) \Omega} \| g_{0j} \|_{H^N} 
			+ e^{(q - 2)\Omega} \| \bard g_{0j} \|_{H^N} \Big),  	
			\label{E:mathfrakSMg0*} \\
	\hstarstarnorm{N} & \eqdef \sum_{j,k=1}^3 \Big( e^{q \Omega} 
		\| \partial_t h_{jk} \|_{H^N} + \underbrace{\| \bard h_{jk} \|_{H^{N-1}}}_{\mbox{absent if} \ N=0} 
		+ e^{(q - 1)\Omega} \| \bard h_{jk} \|_{H^N} \Big), \label{E:mathfrakSMh**} \\
	\fluidnorm{N} & \eqdef \| e^{\decayparameter \Omega} \partial_t \Phi - \bar{\Psi} \|_{H^N} + 
		e^{(\decayparameter-1)\Omega} \| \bard \Phi \|_{H^N}, \label{E:Fluidnorm} \\
\end{align}
\end{subequations}	

\begin{subequations}
\begin{align}
	\supgzerozeronorm{N}(t) & \eqdef \sup_{0 \leq \tau \leq t} 
			\gzerozeronorm{N}(\tau),   \label{E:mathfrakSMsupg00} \\
	\supgzerostarnorm{N}(t) & \eqdef \sup_{0 \leq \tau \leq t} 
		\gzerostarnorm{N}(\tau), \label{E:mathfrakSMsupg0*} \\
	\suphstarstarnorm{N}(t) & \eqdef \sup_{0 \leq \tau \leq t} 
		\hstarstarnorm{N}(\tau), \label{E:mathfrakSMsuph**} \\
	\supgnorm{N} & \eqdef \supgzerozeronorm{N} + \supgzerostarnorm{N} + \suphstarstarnorm{N}		
		\label{E:totalsupgnorm}, \\
	\supfluidnorm{N}(t) & \eqdef \sup_{0 \leq \tau \leq t} \fluidnorm{N}(\tau), \label{E:Fluidsupnorm} \\
	\suptotalnorm{N} & \eqdef \supgnorm{N} + \supfluidnorm{N}. 	\label{E:totalsupnorm}
\end{align}
\end{subequations}

\end{definition}

\subsection{Energies for the metric \texorpdfstring{$g$}{}}

\subsubsection{The building block energy for $g$}

The energies for the metric components will be built from the quantities defined in the following lemma. They are designed 
with equations \eqref{E:finalg00equation} - \eqref{E:finalhjkequation} in mind.

\begin{lemma} \label{L:buildingblockmetricenergy}
 
Let $v$ be a solution to the scalar equation
\begin{align} \label{E:vscalar}
	\hat{\Square}_g v = \alpha H \partial_t v + \beta H^2 v + F,
\end{align}

\noindent where $	\hat{\Square}_g = g^{\lambda \kappa} \partial_{\lambda} \partial_{\kappa},$ 
$\alpha > 0$ and $\beta \geq 0.$ For any constants $\gamma \geq 0, \delta \geq 0,$ we define

\begin{align} \label{E:mathcalEdef}
		\mathcal{E}_{(\gamma,\delta)}^2[v,\partial v] \eqdef \frac{1}{2} \int_{\mathbb{T}^3} 
		\lbrace -g^{00} (\partial_t v)^2 + g^{ab}(\partial_a v)(\partial_b v) - 2 \gamma H g^{00} v \partial_t v
		+ \delta H^2 v^2 \rbrace \, d^3 x.
\end{align}

Then there exist constants 
$\eta > 0, C >0, \delta \geq 0,$ and $\gamma \geq 0,$ with $\eta$ and $C$ depending 
on $\alpha, \beta, \gamma$ and $\delta,$ such that 

\begin{align}
	|g^{00} + 1| \leq \eta
\end{align}
implies that

\begin{align} \label{E:mathcalEfirstlowerbound}
	\mathcal{E}_{(\gamma,\delta)}^2[v,\partial v] \geq C \int_{\mathbb{T}^3} (\partial_t v)^2 
		+ g^{ab}(\partial_a v)(\partial_b v) + C_{(\gamma)} v^2 \, d^3x,
\end{align} where $C_{(\gamma)} = 0$ if $\gamma = 0$ and $C_{(\gamma)} = 1$ if $\gamma > 0.$ Furthermore, if
$\beta = 0,$ then $\gamma = \delta = 0.$ Finally, we have that

\begin{align} \label{E:mathcalEtimederivativebound}
	\frac{d}{dt} (\mathcal{E}_{(\gamma,\delta)}^2[v,\partial v]) & \leq - \eta H \mathcal{E}_{(\gamma,\delta)}^2[v,\partial v] 
		+ \int_{\mathbb{T}^3} \Big\lbrace - (\partial_t v + \gamma H v)F + \triangle_{\mathcal{E};(\gamma, \delta)}[v,\partial v] \Big\rbrace \, d^3 x,
\end{align}
where

\begin{align} \label{E:trianglemathscrEdef}
	\triangle_{\mathcal{E};(\gamma, \delta)}[v,\partial v] & = - \gamma H (\partial_a g^{ab}) v \partial_b v
		- 2 \gamma H (\partial_a g^{0a}) v \partial_t v - 2 \gamma H g^{0a}(\partial_a v)(\partial_t v) \\
	& \ \ - (\partial_a g^{0a})(\partial_t v)^2 - (\partial_a g^{ab})(\partial_b v)(\partial_t v)
		- \frac{1}{2}(\partial_t g^{00})(\partial_t v)^2 \notag \\
	& \ \ + \bigg(\frac{1}{2} \partial_t g^{ab} + \omega g^{ab} \bigg) (\partial_a v) (\partial_b v)
		+ \big(H - \omega \big) g^{ab} (\partial_a v) (\partial_b v) \notag \\
	& \ \ - \gamma H (\partial_t g^{00}) v \partial_t v - \gamma H (g^{00} + 1)(\partial_t v)^2. \notag
\end{align}

\end{lemma}

\begin{proof}
	A proof based on a standard integration by parts argument \big(multiply both sides of equation \eqref{E:vscalar} by
	$- (\partial_t v + \gamma H v)$ before integrating over $\mathbb{T}^3$)\big)
	is given in Lemma 15 of \cite{hR2008}. In particular, we quote the following identity:
	
	\begin{align} \label{E:mathcalEtimederivativeformula}
		\frac{d}{dt} (\mathcal{E}_{(\gamma,\delta)}^2[v,\partial v]) & = \int_{\mathbb{T}^3} 
		\Big\lbrace -(\alpha - \gamma) H (\partial_t v)^2 + (\delta - \beta - \gamma \alpha) H^2 v \partial_t v 
		- \beta \gamma H^3 v^2 \\ 
	& \hspace{1in} - (1 + \gamma) H g^{ab}(\partial_a v)(\partial_b v) - (\partial_t v + \gamma H v)F 
		+ \triangle_{\mathcal{E};(\gamma, \delta)}[v,\partial v] \Big\rbrace \, d^3 x. \notag
	\end{align}
\end{proof}

\subsubsection{Energies for the components of $g$}
In this section, we will use re-scaled versions of energies of the form \eqref{E:mathcalEdef} to construct energies for the components of $g.$

\begin{definition}  \label{D:energiesforg}
We define the \emph{non-negative} energies $\gzerozeroenergy{N}(t),$ 
$\supgzerozeroenergy{N}(t),$ $\gzerostarenergy{N}(t),$ $\supgzerostarenergy{N}(t),$
$\hstarstarenergy{N}(t),$ $\suphstarstarenergy{N}(t),$ $\genergy{N}(t),$ and
$\supgenergy{N}(t)$ as follows:

 \begin{subequations}
\begin{align}
	\gzerozeroenergy{N}^2 & \eqdef \sum_{|\vec{\alpha}| \leq N}
		e^{2q \Omega} \mathcal{E}_{(\gamma_{00},\delta_{00})}^2[\partial_{\vec{\alpha}} (g_{00} + 1),
		\partial (\partial_{\vec{\alpha}} g_{00})], \label{E:g00energydef} \\
	\gzerostarenergy{N}^2 & \eqdef \sum_{|\vec{\alpha}| \leq N} \sum_{j=1}^3 
		e^{2(q-1) \Omega}\mathcal{E}_{(\gamma_{0*},\delta_{0*})}^2[\partial_{\vec{\alpha}} g_{0j},
			\partial (\partial_{\vec{\alpha}} g_{0j})], \label{E:g0*energydef} \\
	\hstarstarenergy{N}^2 & \eqdef \sum_{|\vec{\alpha}| \leq N} 
		\Big\lbrace \sum_{j,k=1}^3 e^{2q \Omega} 
		\mathcal{E}_{(0,0)}^2[0,\partial (\partial_{\vec{\alpha}} h_{jk})] 
		+ \frac{1}{2} \int_{\mathbb{T}^3} c_{\vec{\alpha}} H^2 \big(\partial_{\vec{\alpha}} h_{jk}(\tau)\big)^2 \, d^3 x 
		\Big\rbrace,	 \label{E:h**energydef} \\
	\genergy{N}^2 & \eqdef \gzerozeroenergy{N}^2 + \gzerostarenergy{N}^2 + 	
		\hstarstarenergy{N}^2, \label{E:gtotalenergydef}
\end{align}
\end{subequations}

\begin{subequations}		
\begin{align}
	\supgzerozeroenergy{N}(t) & \eqdef \sup_{0 \leq \tau \leq t} \gzerozeroenergy{N}(\tau),
		\label{E:g00supenergydef} \\
	\supgzerostarenergy{N}(t) & \eqdef \sup_{0 \leq \tau \leq t} 
		\gzerostarenergy{N}(\tau), \label{E:g0*supenergydef} \\
	\suphstarstarenergy{N}(t) & \eqdef \hstarstarenergy{N}(\tau),	
		\label{E:h**supenergydef} \\
	\supgenergy{N}^2 & \eqdef \supgzerozeroenergy{N}^2 + \supgzerostarenergy{N}^2 + \suphstarstarenergy{N}^2
	\label{E:gtotalsupenergydef},
\end{align}
\end{subequations}
where

\begin{subequations}
\begin{align}
	h_{jk} & \eqdef e^{-2 \Omega} g_{jk}, && (j,k = 1,2,3), \\
	c_{\vec{\alpha}} & \eqdef 0, && \mbox{if} \ |\vec{\alpha}| = 0, \\
	c_{\vec{\alpha}} & \eqdef 1, && \mbox{if} \ |\vec{\alpha}| > 0, 
\end{align}
\end{subequations}
and $(\gamma_{00}, \delta_{00}), (\gamma_{0*}, \delta_{0*}),$ and $(\gamma_{**}, \delta_{**})=(0,0)$ are the constants generated by applying Lemma \ref{L:buildingblockmetricenergy} to equations \eqref{E:finalg00equation} - \eqref{E:finalhjkequation} respectively.

\end{definition}

In the next lemma, we provide a preliminary estimate of the time derivative of these energies.

\begin{lemma} \label{L:metricfirstdiferentialenergyinequality}
	
	Assume that $(g_{\mu \nu}, \partial_{\mu} \Phi),$ $(\mu,\nu=0,1,2,3),$ is a solution to the modified equations 
	\eqref{E:finalg00equation} - \eqref{E:finalhjkequation}, and let $\gzerozeroenergy{N}, \gzerostarenergy{N},$ and
	$\hstarstarenergy{N}$ be as in Definition \ref{D:energiesforg}. Let 
	$[\hat{\Square}_g ,\partial_{\vec{\alpha}}]$ denote the commutator of the operators $\hat{\Square}_g$ and
	$\partial_{\vec{\alpha}}.$ Then the following differential inequalities are satisfied, where 
	$\triangle_{\mathcal{E};(\gamma, \delta)}[\cdot, \partial(\cdot )]$ is defined in \eqref{E:trianglemathscrEdef}, and the
	constants $(\gamma_{00}, \delta_{00}), (\gamma_{0*}, \delta_{0*}),$ and $(\gamma_{**}, \delta_{**})=(0,0)$ 
	are defined in Definition \ref{D:energiesforg}:
	
	\begin{subequations}
	\begin{align}
		\frac{d}{dt}(\gzerozeroenergy{N}^2) & \leq (2q - \eta_{00})H\gzerozeroenergy{N}^2 
			+ 2q(\omega - H) \gzerozeroenergy{N}^2 \\
		& \ \ \ - \sum_{|\vec{\alpha}| \leq N } \int_{\mathbb{T}^3} e^{2q \Omega} \big\lbrace \partial_t \partial_{\vec{\alpha}}(g_{00}+1)
		  + \gamma_{00} H \partial_{\vec{\alpha}}(g_{00}+1) \big\rbrace \big\lbrace \partial_{\vec{\alpha}} \triangle_{00} 
		  + [\hat{\Square}_g ,\partial_{\vec{\alpha}}](g_{00}+1) \big\rbrace \, d^3 x \notag \\
		& \ \ \ + \sum_{|\vec{\alpha}| \leq N } \int_{\mathbb{T}^3} 
			e^{2q \Omega} \triangle_{\mathcal{E};(\gamma_{00}, \delta_{00})}[\partial_{\vec{\alpha}}(g_{00}+1)
				,\partial (\partial_{\vec{\alpha}} g_{00})]   \,d^3 x, \notag \\
		\frac{d}{dt}(\gzerostarenergy{N}^2) & \leq [2(q-1) - \eta_{0*}]H\gzerostarenergy{N}^2 
			+ 2(q-1)(\omega - H) \gzerostarenergy{N}^2 \\
		& \ \ \ - \sum_{|\vec{\alpha}| \leq N } \sum_{j=1}^3 \int_{\mathbb{T}^3} e^{2(q-1) \Omega} \big\lbrace \partial_t 
			\partial_{\vec{\alpha}} g_{0j}
		  + \gamma_{0*} H \partial_{\vec{\alpha}} g_{0j} \big\rbrace \big\lbrace -2H \partial_{\vec{\alpha}} (g^{ab}\Gamma_{a j b})
		  + \partial_{\vec{\alpha}} \triangle_{0j} 
		  + [\hat{\Square}_g ,\partial_{\vec{\alpha}}] g_{0j} \big\rbrace \, d^3 x \notag \\
		& \ \ \ + \sum_{|\vec{\alpha}| \leq N } \sum_{j=1}^3 \int_{\mathbb{T}^3} 
			e^{2(q-1)\Omega} \triangle_{\mathcal{E};(\gamma_{0*}, \delta_{0*})}[\partial_{\vec{\alpha}} g_{0j},
			\partial( \partial_{\vec{\alpha}} g_{0j})]   \,d^3 x, \notag \\
		\frac{d}{dt}(\hstarstarenergy{N}^2) & \leq (2q - \eta_{**})H\hstarstarenergy{N}^2 
			+ 2q(\omega - H) \hstarstarenergy{N}^2 \\
		& \ \ \ - \sum_{|\vec{\alpha}| \leq N } \sum_{j,k=1}^3 \int_{\mathbb{T}^3} e^{2q \Omega} \big\lbrace \partial_t 
			\partial_{\vec{\alpha}} h_{jk}
		  + \underbrace{\gamma_{**}}_0 H \partial_{\vec{\alpha}} h_{jk} \big\rbrace \big\lbrace \partial_{\vec{\alpha}} 
		  	\triangle_{jk} 
		  + [\hat{\Square}_g ,\partial_{\vec{\alpha}}] h_{jk} \big\rbrace \, d^3 x \notag \\
		& \ \ \ + \sum_{|\vec{\alpha}| \leq N } \sum_{j,k=1}^3 \int_{\mathbb{T}^3} 
			e^{2q \Omega} \triangle_{\mathcal{E};(0, 0)}[0,\partial (\partial_{\vec{\alpha}} h_{jk})]  \,d^3 x   
		+ \sum_{1 \leq |\vec{\alpha}| \leq N } \int_{\mathbb{T}^3}  
			H^2 (\partial_{\vec{\alpha}} \partial_t h_{jk})(\partial_{\vec{\alpha}} h_{jk}) \, d^3 x. \notag
	\end{align}	
	\end{subequations}
\end{lemma}

\begin{proof}
	Lemma \ref{L:metricfirstdiferentialenergyinequality} follows easily from definitions \eqref{E:g00energydef} - \eqref{E:gtotalsupenergydef}, and from \eqref{E:mathcalEtimederivativebound}.
\end{proof}

The following corollary follows easily from Lemma \ref{L:metricfirstdiferentialenergyinequality}, 
definitions \eqref{E:g00energydef} - \eqref{E:h**energydef}, and the Cauchy-Schwarz inequality for integrals.

\begin{corollary} \label{C:metricfirstdiferentialenergyinequality}
	Under the assumptions of Lemma \ref{L:metricfirstdiferentialenergyinequality}, we have that
	
	\begin{subequations}
	\begin{align}
		\frac{d}{dt}(\gzerozeroenergy{N}^2) & \leq (2q - \eta_{00})H\gzerozeroenergy{N}^2 
			+ 2q(\omega - H) \gzerozeroenergy{N}^2 
			+ \gzerozeronorm{N} e^{q \Omega}  \| \triangle_{00} \|_{H^N}\\
		& \ \ \ + \gzerozeronorm{N} \sum_{|\vec{\alpha}| \leq N} 
		 	e^{q \Omega} \| [\hat{\Square}_g ,\partial_{\vec{\alpha}}](g_{00}+1) \|_{L^2} 	
			+ \sum_{|\vec{\alpha}| \leq N } e^{2q \Omega} \| \triangle_{\mathcal{E};(\gamma_{00}, 
			\delta_{00})}[\partial_{\vec{\alpha}}(g_{00}+1),\partial (\partial_{\vec{\alpha}} g_{00})] \|_{L^1}, \notag \\
	 	\frac{d}{dt}(\gzerostarenergy{N}^2) & \leq [2(q-1) - \eta_{0*}]H\gzerostarenergy{N}^2 
			+ 2(q-1)(\omega - H) \gzerostarenergy{N}^2 	\label{E:underlinemathfrakEg0*firstdifferential} \\
		& \ \ \ + 2H \gzerostarnorm{N} \sum_{j=1}^3 e^{(q-1) \Omega} \| g^{ab} \Gamma_{ajb} \|_{H^N}
		  + \gzerostarnorm{N} \sum_{j=1}^3 e^{(q-1) \Omega} \| \triangle_{0j} \|_{H^N}
			\notag \\
		&	\ \ \ + \gzerostarnorm{N} 
			\sum_{|\vec{\alpha}| \leq N } \sum_{j=1}^3 e^{(q-1) \Omega} \| [\hat{\Square}_g ,\partial_{\vec{\alpha}}] g_{0j} \|_{L^2}
			+ \sum_{|\vec{\alpha}| \leq N } \sum_{j=1}^3 e^{2(q-1) \Omega} 
			\| \triangle_{\mathcal{E};(\gamma_{0*}, \delta_{0*})}[\partial_{\vec{\alpha}} g_{0j}
			,\partial(\partial_{\vec{\alpha}} g_{0j})] \|_{L^1}, \notag \\
		\frac{d}{dt}(\hstarstarenergy{N}^2) & \leq (2q - \eta_{**})H\hstarstarenergy{N}^2 
			+ 2q(\omega - H) \hstarstarenergy{N}^2 + \hstarstarnorm{N} \sum_{j,k=1}^3 \| \triangle_{jk} \|_{H^N} \\
		& \ \ \ + \hstarstarnorm{N} \sum_{|\vec{\alpha}| \leq N } \sum_{j,k=1}^3 
				\| [\hat{\Square}_g ,\partial_{\vec{\alpha}}] h_{jk} \|_{L^2}
			+ \sum_{|\vec{\alpha}| \leq N } \sum_{j,k=1}^3 
				\| \triangle_{\mathcal{E};(0, 0)}[0,\partial(\partial_{\vec{\alpha}} h_{jk}) \|_{L^1}
			+ H^2 e^{-q \Omega}\hstarstarnorm{N}^2, \notag
	\end{align}	
	\end{subequations}
	where the norms $\gzerozeronorm{N}, \gzerostarnorm{N}, 
	\hstarstarnorm{N}$ are defined in Definition \ref{D:Norms}.
\end{corollary}

\subsection{Energies for the fluid variable \texorpdfstring{$\partial \Phi$}{}} \label{SS:fluidvariableEnergies}
In this section, we define the energies that we will use to study the irrotational fluid equation \eqref{E:finalfluidequation}.
We begin by stating their definitions.

\begin{definition}
	Let $\bar{\Psi}$ be the positive constant defined in \eqref{E:barPsidef}. Then we define the following
	\emph{non-negative} energies $\fluidenergy{N}(t),\supfluidenergy{N}(t)$ for $\partial \Phi$
	as follows:
	
	\begin{subequations}
	\begin{align} \label{E:fluidenergydef}
		\underline{E}_{\partial \Phi;0}^2 & \eqdef \frac{1}{2} \int_{\mathbb{T}^3} (e^{\decayparameter \Omega}\partial_t \Phi - 
			\bar{\Psi})^2 + e^{2\decayparameter \Omega} m^{ab} (\partial_a \Phi)(\partial_b \Phi) \,d^3 x, && \\
		\fluidenergy{N}^2 & \eqdef \underline{E}_0^2 + \sum_{1 \leq |\vec{\alpha}| \leq N} \frac{1}{2} 
			\int_{\mathbb{T}^3} e^{2 \decayparameter \Omega}(\partial_t \partial_{\vec{\alpha}}\Phi)^2 
			+ e^{2\decayparameter \Omega} m^{ab} (\partial_a \partial_{\vec{\alpha}} \Phi)(\partial_b \partial_{\vec{\alpha}} \Phi) 
			\,d^3x, && (N \geq 1), \\
		\supfluidenergy{N}(t) & \eqdef \sup_{0 \leq \tau \leq t} \fluidenergy{N}(\tau). && 
			\label{E:fluidsupenergydef} 
	\end{align}
	\end{subequations}
	
\end{definition}

In the next lemma, we provide a preliminary estimate of the time derivative of the fluid energies.
 
\begin{lemma} \label{L:fluidenergytimederivative}
	Let $\partial \Phi$ be a solution to
	\begin{align} \label{E:linearizedPhi}
		\hat{\Square}_{m} \Phi  & = \decayparameter \omega(t) \partial_t \Phi
			 + \triangle_{\partial \Phi}, 
	\end{align}
	
	\noindent where
	
	\begin{align}
		\hat{\Square}_{m} & \eqdef - \partial_t^2 + m^{ab} \partial_a \partial_b 
			+ 2m^{0a} \partial_t \partial_a 
	\end{align}
	is the reduced wave operator corresponding to the reciprocal acoustical metric $m^{\mu \nu},$ ($\mu, \nu = 0,1,2,3$), and
	\begin{align}
		\decayparameter & = \frac{3}{1 + 2s} = 3 \speed^2.
	\end{align}
	Let $[\hat{\Square}_m ,\partial_{\vec{\alpha}}]$ denote the commutator of the operators $\hat{\Square}_m$ and
	$\partial_{\vec{\alpha}}.$ Then  	
	\begin{align} \label{E:fluidenergytimederivative}
		\frac{d}{dt}(\fluidenergy{N}^2) & = - \int_{\mathbb{T}^3}  
			(\partial_a m^{0a})( e^{\decayparameter \Omega} \partial_t \Phi - \bar{\Psi})^2 \, d^3x
			\ - \sum_{1 \leq |\vec{\alpha}| \leq N} \int_{\mathbb{T}^3} e^{2\decayparameter \Omega}  
			(\partial_a m^{0a}) (\partial_t \partial_{\vec{\alpha}} \Phi)^2 \, d^3x \\
		& \ \ - \int_{\mathbb{T}^3} e^{\decayparameter \Omega} 
			(\partial_a m^{ab})(e^{\decayparameter \Omega} \partial_t \Phi - \bar{\Psi})(\partial_b \Phi) \, d^3x
			\ - \sum_{1 \leq |\vec{\alpha}| \leq N} \int_{\mathbb{T}^3} e^{2\decayparameter \Omega}  
			(\partial_b m^{ab}) (\partial_t \partial_{\vec{\alpha}} \Phi) (\partial_b \partial_{\vec{\alpha}} \Phi) \, d^3x \notag \\
		& \ \ - \int_{\mathbb{T}^3} e^{\decayparameter \Omega} (\triangle_{\partial \Phi}) (e^{\decayparameter \Omega} \partial_t 
			\Phi - \bar{\Psi}) \,d^3x 
		- \sum_{1 \leq |\vec{\alpha}| \leq N} \int_{\mathbb{T}^3} e^{2\decayparameter \Omega}  (\partial_{\vec{\alpha}} 
			\triangle_{\partial \Phi}) \partial_t \partial_{\vec{\alpha}} \Phi \,d^3x \notag \\
		& \ \ - \sum_{1 \leq |\vec{\alpha}| \leq N} \int_{\mathbb{T}^3} e^{2\decayparameter \Omega} 
		([\hat{\Square}_m, \partial_{\vec{\alpha}}] \Phi)
			\partial_t \partial_{\vec{\alpha}} \Phi \,d^3x 
		+ \frac{1}{2} \sum_{|\vec{\alpha}| \leq N} \int_{\mathbb{T}^3} 
		e^{2\decayparameter \Omega}(\partial_t m^{ab} + 2 \decayparameter \omega m^{ab})(\partial_a \partial_{\vec{\alpha}} \Phi)
			(\partial_b \partial_{\vec{\alpha}} \Phi) \,d^3 x. \notag 
\end{align}
	
\end{lemma}

\begin{proof}
	This is a standard integration by parts lemma that can be proved using the ideas of Lemma \ref{L:buildingblockmetricenergy}.
	We provide a sketch of the proof. We begin by differentiating under the integral in the definition of $\fluidenergy{N}^2$ to 
	conclude that
	
	\begin{align} \label{E:EMpreliminarytimederivative}
		\frac{d}{dt}(\fluidenergy{N}^2)	& = \sum_{|\vec{\alpha}| \leq N} \int_{\mathbb{T}^3} 
			\big\lbrace \partial_{\vec{\alpha}} (e^{\decayparameter \Omega}\partial_t \Phi - \bar{\Psi}) \big\rbrace
			\partial_t \partial_{\vec{\alpha}} (e^{\decayparameter \Omega}\partial_t \Phi) 
			+ e^{2\decayparameter \Omega} m^{ab} (\partial_t \partial_a \partial_{\vec{\alpha}} \Phi)
			(\partial_b \partial_{\vec{\alpha}} \Phi) \, d^3 x \\
			& \ \ \ + \sum_{|\vec{\alpha}| \leq N} \int_{\mathbb{T}^3}
				e^{2\decayparameter \Omega}\Big(\frac{1}{2} \partial_t m^{ab} + \decayparameter \omega m^{ab}\Big)
				(\partial_a \partial_{\vec{\alpha}} \Phi)(\partial_b \partial_{\vec{\alpha}}\Phi)\,d^3 x. \notag
	\end{align}
	For each fixed $\vec{\alpha},$ we will now eliminate the highest derivatives of $\Phi$ in 
	\eqref{E:EMpreliminarytimederivative} (i.e., the derivatives of order $|\vec{\alpha}| + 2$). 
	To this end, we first differentiate equation \eqref{E:linearizedPhi} with 
	$\partial_{\vec{\alpha}}$ and multiply both sides of the equation by $e^{\decayparameter \Omega},$
	which allows us to express the resulting equality as
	
	\begin{align} \label{E:linearizedPhirewritten}
		- \partial_t \partial_{\vec{\alpha}} (e^{\decayparameter \Omega}\partial_t \Phi- \bar{\Psi}) 
			+ e^{\decayparameter \Omega} m^{ab} \partial_a \partial_b \partial_{\vec{\alpha}} \Phi
			= e^{\decayparameter \Omega} \partial_{\vec{\alpha}} \triangle_{\partial \Phi} 
				+ e^{\decayparameter \Omega} [\hat{\Square}_m, \partial_{\vec{\alpha}}] \Phi
			- 2 e^{\decayparameter \Omega} m^{0a}\partial_t \partial_a \partial_{\vec{\alpha}} \Phi.
	\end{align}
	We then multiply both sides of \eqref{E:linearizedPhirewritten} by 
	$- \partial_{\vec{\alpha}} (e^{\decayparameter \Omega}\partial_t \Phi- \bar{\Psi}),$ integrate over $\mathbb{T}^3,$
	and integrate by parts. Inserting the identity that arises into \eqref{E:EMpreliminarytimederivative}, we
	arrive at \eqref{E:fluidenergytimederivative}.
\end{proof}

We now state the following corollary, which follows easily from definition \eqref{E:Fluidnorm}, Lemma \ref{L:fluidenergytimederivative}, and the Cauchy-Schwarz inequality for integrals.

\begin{corollary} \label{C:fluidenergytimederivative}
 Under the hypotheses of Lemma \ref{L:fluidenergytimederivative}, we have that
	
	\begin{align} \label{E:fluidenergytimederivativeCauchySchwarz}
		\frac{d}{dt}(\fluidenergy{N}^2) & \leq 
			\fluidnorm{N}^2 \| \partial_a m^{0a} \|_{L^{\infty}} 
			+ \fluidnorm{N}^2 \sum_{b=1}^3 e^{\Omega} \| \partial_a m^{ab} \|_{L^{\infty}} 
			+ \fluidnorm{N}  e^{\decayparameter \Omega} \|\triangle_{\partial \Phi} \|_{H^N}   \\
		& \ \ \ + \fluidnorm{N}  \sum_{1 \leq |\vec{\alpha}| \leq N} e^{\decayparameter \Omega} \| 
			[\hat{\Square}_m,\partial_{\vec{\alpha}}] \Phi \|_{L^2} 
			+ \frac{1}{2} \fluidnorm{N}^2 \sum_{a,b=1}^3 e^{2 \decayparameter \Omega} 
				\| \partial_t m^{ab} + 2 \omega m^{ab} \|_{L^{\infty}} \notag \\
		& \ \ \ + (\decayparameter - 1) \omega \sum_{|\vec{\alpha}| \leq N} \int_{\mathbb{T}^3} 
			e^{2 \decayparameter \Omega} m^{ab} (\partial_a \partial_{\vec{\alpha}} \Phi) (\partial_b \partial_{\vec{\alpha}} \Phi) 
			\,d^3 x, \notag
	\end{align}
	where $\fluidnorm{N}$ is defined in \eqref{E:Fluidnorm}.
\end{corollary}

\subsection{The total energy \texorpdfstring{$\suptotalenergy{N}$}{}} \label{SS:TotalEnergy}

\begin{definition}
Let $\supgenergy{N}$ and $\supfluidenergy{N}$ be the metric component and fluid energies defined in \eqref{E:gtotalsupenergydef}
and \eqref{E:fluidsupenergydef} respectively. We define $\suptotalenergy{N},$ the total energy associated to $(g_{\mu \nu}, \partial_{\mu} \Phi),$ $(\mu,\nu = 0,1,2,3),$ as follows:
	
	\begin{align} \label{E:totalenergy}
		\suptotalenergy{N} & \eqdef \supgenergy{N} + \supfluidenergy{N}.
	\end{align}
\end{definition}

\section{Linear-Algebraic Estimates of \texorpdfstring{$g_{\mu \nu}$ and $g^{\mu \nu},$ $(\mu,\nu=0,1,2,3)$}{}} \label{S:LinearAlgebra}
\setcounter{equation}{0}

In this section, we provide some linear-algebraic estimates of $g_{\mu \nu}$ and $g^{\mu \nu}.$
In addition to providing some rough $L^{\infty}$ estimates that we will use in Sections \ref{S:BootstrapConsequences} and \ref{S:EnergyNormEquivalence}, the lemmas that we prove will guarantee that $g_{\mu \nu}$ is a Lorentzian metric. We remark that we already made use of this fact in our statement of the conclusions of Theorem \ref{T:LocalExistence}. The estimates are based on the following rough assumptions, which we will upgrade during our global existence argument.

\begin{center}
	\large{Rough Bootstrap Assumptions for $g_{\mu \nu}:$}
\end{center}

In this section, we assume that there are constants $\eta > 0$ and $c_1 \geq 1$ such that

\begin{subequations}
\begin{align}
	|g_{00} + 1| & \leq \eta, && \label{E:metricBAeta} \\
	c_1^{-1} \delta_{ab} X^a X^b 
		& \leq e^{-2 \Omega} g_{ab}X^{a}X^{b} 
		\leq c_1 \delta_{ab} X^a X^b, && \forall (X^1,X^2,X^3) \in \mathbb{R}^3, \label{E:gjkBAvsstandardmetric}   \\
	\sum_{a=1}^3 |g_{0a}|^2 & \leq \eta c_1^{-1} e^{2(1 - q) \Omega}. && \label{E:g0jBALinfinity}
\end{align}
\end{subequations}
For our global existence argument, we will assume that $\eta = \eta_{min},$ where $\eta_{min}$ 
is defined in Section \ref{SS:etamin}.

\begin{lemma} \label{L:ginverseformluas}
	Let $g_{\mu \nu}$ be a symmetric $4 \times 4$ matrix of real numbers. 
	Let $(g_{\flat})_{jk}$ be the $3 \times 3$ matrix defined by $(g_{\flat})_{jk} = g_{jk},$ let $(g_{\flat}^{-1})^{jk}$
	be the $3 \times 3$ inverse of $(g_{\flat})_{jk}.$ Assume that $g_{00} < 0$ and that $(g_{\flat})_{jk}$ is positive definite. 
	Then $g_{\mu \nu}$ is a Lorentzian metric with inverse $g^{\mu \nu},$ $g^{00} < 0,$ and the $3 \times 3$ matrix 
	$(g^{\#})^{jk}$ defined by $(g^{\#})^{jk} \eqdef g^{jk}$ is positive definite. Furthermore, the following formulas hold:
	
	\begin{subequations}
	\begin{align}
		g^{00} & = \frac{1}{g_{00} - d^2}, \\
		\frac{g_{00}}{g_{00} - d^2} (g_{\flat}^{-1})^{ab} X_{a}X_{b} 
			& \leq (g^{\#})^{ab} X_{a}X_{b} \leq (g_{\flat}^{-1})^{ab} X_{a}X_{b}, && \forall (X_1,X_2,X_3) \in \mathbb{R}^3, \\
		g^{0j} & = \frac{1}{d^2 - g_{00}} (g_{\flat}^{-1})^{aj} g_{0a}, && (j=1,2,3),
	\end{align}
	\end{subequations}
	where 
	
	\begin{align}
		d^2 = (g_{\flat}^{-1})^{ab} g_{0a} g_{0b}.
	\end{align}
	
\end{lemma}

\begin{proof}
	Lemma \ref{L:ginverseformluas} is a combining of Lemmas 1 and 2 of \cite{hR2008}.
\end{proof}

\begin{lemma} \label{L:ginverseestimates}
	Let $g_{\mu \nu}$ be a symmetric $4 \times 4$ matrix of real numbers satisfying \eqref{E:metricBAeta} - 
	\eqref{E:g0jBALinfinity}, where $\Omega \geq 0$ and $0 \leq q < 1.$ 
	Then $g$ is a Lorentzian metric, and there exists a constant $\eta_0 > 0$ 
	such that $0 \leq \eta \leq \eta_0$ implies that the following estimates hold for the 
	its inverse $g^{\mu \nu}:$ 
	
	\begin{subequations}
	\begin{align}
		|g^{00} + 1| & \leq 4 \eta, && \label{E:g00upperplusoneroughestimate} \\
		\sqrt{\sum_{a=1}^3 |g^{0a}|^2} & \leq \eta c_1^{-1} e^{-2 \Omega} \sqrt{\sum_{a=1}^3 |g_{0a}|^2}, \\
		|g^{0a} g_{0a}| & \leq 2 c_1 e^{-2 \Omega} \sum_{a=1}^3 |g_{0a}|^2, && \\
		\frac{2}{3c_1} \delta^{ab}X_{a} X_{b}
		& \leq e^{2 \Omega} g^{ab}X_{a}X_{b} 
			\leq \frac{3c_1}{2}\delta^{ab}X_{a}X_{b}, && \forall (X_1,X_2,X_3) \in \mathbb{R}^3. 
				\label{E:gjkuppercomparetostandard}  
	\end{align}
	\end{subequations}
\end{lemma}

\begin{proof}
	Lemma \ref{L:ginverseestimates} is proved as Lemma 7 in \cite{hR2008}.
\end{proof}

\section{The Bootstrap Assumption for \texorpdfstring{$\suptotalnorm{N}$ and the Definition of $N, \eta_{min}$ and $q$}{}} \label{S:BootstrapAssumptions}

In this short section, we define the quantities $N, \eta_{min},$ and $q.$ We then introduce some assumptions
that will be used in our derivation of the estimates of Sections \ref{S:BootstrapConsequences} and \ref{S:EnergyNormEquivalence}.

\subsection{The definition of \texorpdfstring{$N$ and the assumption $\suptotalnorm{N} \leq \epsilon$}{}}

For the remainder of the article, we will assume that $N$ is an integer subject to one of the following requirements:

\begin{align}
	N &\geq 3, && (\mbox{this is large enough for all of our results except some of the conclusions of 
		Theorem \ref{T:Asymptotics}}), \\ \label{E:Ndef}
	N &\geq 5, &&	(\mbox{for the full results of Theorem \ref{T:Asymptotics} to be valid}).
\end{align}
We require $N$ to be of this size to ensure that various Sobolev embedding results are valid; see also
Remark \ref{R:Nlarger}.

In our global existence argument, we will make the following bootstrap assumption:

\begin{align} \label{E:fundamentalbootstrap}
	\suptotalnorm{N}[g, \partial \Phi] \leq \epsilon,
\end{align}
where $\suptotalnorm{N}$ is defined in \eqref{E:totalsupnorm}, and $\epsilon$ is a sufficiently 
small positive number. Above we have written $\suptotalnorm{N}[g, \partial \Phi]$ to emphasize the dependence of
$\suptotalnorm{N}$ on $(g, \partial \Phi).$ Observe that $\suptotalnorm{N}[\widetilde{g}, \partial \widetilde{\Phi}] = 0,$ where $\widetilde{g}, \partial \widetilde{\Phi}$ is the background solution from Section \ref{S:backgroundsolution}; i.e., 
$\suptotalnorm{N}[g, \partial \Phi]$ measures how much $(g, \Phi)$ differs from the background solution.

\subsection{The definitions of \texorpdfstring{$\eta_{min}$ and $q$}{}} \label{SS:etamin}

\begin{definition}
	Let us apply Lemma \ref{L:buildingblockmetricenergy} to each of the equations \eqref{E:finalg00equation} - 
	\eqref{E:finalhjkequation}, denoting the constant $\eta$ produced by the lemma in each case by
	 $\eta_{00}, \eta_{0*},$ and $\eta_{**}$ respectively. Furthermore, let $\eta_0$ be the constant from
	 Lemma \ref{L:ginverseestimates}. We now define the positive quantities (recalling that $0 < \decayparameter < 1$ when
	 $0 < \speed < \sqrt{1/3}$) $\eta_{min}$ and $q$ by
	 \begin{align}
	 	\eta_{min} & \eqdef \frac{1}{8} \mbox{min} \big\lbrace 1, \eta_0, \eta_{00}, \eta_{0*}, \eta_{**} \big\rbrace, 
	 		\label{E:etamindef} \\
	 	q & \eqdef \frac{2}{3} \mbox{min}\big\lbrace \eta_{min}, \decayparameter, 1 - \decayparameter \big\rbrace. 
	 		\label{E:qdef}
	 \end{align}
\end{definition}
We remark that $\eta_{min}$ and $q$ have been chosen to be small enough so that the bootstrap argument for global existence given in Section \ref{SS:globalexistencetheorem} will close. In particular, inequality \eqref{E:g00upperplusoneroughestimate}, with $\eta \leq \eta_{min},$ guarantees that the energies $\mathcal{E}_{(\gamma,\delta)}[\cdot,\partial (\cdot)]$ 
for solutions to \eqref{E:finalg00equation} - \eqref{E:finalhjkequation} have the coercive property \eqref{E:mathcalEfirstlowerbound}.

\begin{remark} \label{R:RoughBootstrapAutomatic}

By Sobolev embedding and Lemma \ref{L:ginverseestimates}, if $\epsilon$ is small enough, then the assumption $\suptotalnorm{N} \leq \epsilon$ implies that there exists a constant $C > 0$ such that

\begin{subequations}
\begin{align} \label{E:g00plusonesmall}
	|g_{00} + 1| & \leq C \epsilon, \\
	\sum_{j=1}^3 |g_{0j}| & \leq C \epsilon e^{(1 - q) \Omega}. \label{E:g0jremarkbound}
\end{align}
\end{subequations}

Therefore, if $\epsilon$ is sufficiently small, the inequalities \eqref{E:metricBAeta} and \eqref{E:g0jBALinfinity} 
are an automatic consequence of \eqref{E:gjkBAvsstandardmetric} and the assumption $\suptotalnorm{N} \leq \epsilon.$

\end{remark}

\section{Sobolev Estimates} \label{S:BootstrapConsequences}

In this section, we use the bootstrap assumptions of Sections \ref{S:LinearAlgebra} and \ref{S:BootstrapAssumptions} to derive estimates of all of the terms appearing in the modified equations \eqref{E:finalg00equation} -  \eqref{E:finalfluidequation}
in terms of the norms defined in Section \ref{SS:NormsforgandpartialPhi}. The main goal is to show that the error terms are small compared to the principal terms, which is the main step in closing the bootstrap argument in the proof of Theorem \ref{T:GlobalExistence}. More specifically, in Section \ref{SS:IntegralInequalities}, the estimates of this section will be coupled with the energy inequalities of Corollary  \ref{C:metricfirstdiferentialenergyinequality} and Corollary \ref{C:fluidenergytimederivative}
in order to derive a system of integral inequalities for the solution. We divide the analysis into two propositions: Proposition \ref{P:BoostrapConsequences} provides estimates for $g$ and $\partial_t \Phi,$ while Proposition \ref{P:Nonlinearities} provides estimates for the error terms, the reciprocal acoustical metric, and for $\partial_{\mu} \Phi,$
$(\mu = 0,1,2,3).$ In particular, the latter proposition provides estimates for the ratio $z_j \eqdef \frac{e^{-\Omega} \partial_j \Phi}{\partial_t \Phi},$ $(j=1,2,3),$ which are crucial for closing the bootstrap argument of Theorem \ref{T:GlobalExistence}. The main tools for proving the propositions are standard Sobolev-Moser product-type estimates, which we have collected together in the Appendix for convenience.

\subsection{Estimates of \texorpdfstring{$g_{\mu \nu}, g^{\mu \nu},$ $(\mu, \nu =0,1,2,3),$ $\| e^{\decayparameter \Omega} \partial_t \Phi - \bar{\Psi} \|_{L^{\infty}},$ and $\Big\| \frac{1}{\partial_t \Phi} \Big\|_{L^{\infty}}$}{}} \label{SS:Bootstrapconsequencesg}

In this section, we state and prove the first proposition that will be used to deduce the energy inequalities of Section
\ref{SS:IntegralInequalities}.

\begin{proposition} \label{P:BoostrapConsequences}
Let $N \geq 3$ be an integer, and assume that the bootstrap assumptions \eqref{E:metricBAeta} - \eqref{E:g0jBALinfinity} hold 
on the spacetime slab $[0,T) \times \mathbb{T}^3$ for some constant $c_1 \geq 1$ and for $\eta = \eta_{min}.$
Then there exists a constant $\epsilon' > 0$ and a constant 
$C > 0,$ where $C$ depends on $N, c_1,$ and $\eta_{min},$ such that if and $\suptotalnorm{N}(t) < \epsilon'$ on $[0,T),$
then the following estimates also hold on $[0,T),$ where $h_{jk} = e^{-2 \Omega} g_{jk}:$

\begin{subequations}
\begin{align}
	\| g_{00} \|_{L^{\infty}} & \leq 2, \label{E:g00lowerLinfinity} \\
	\| g_{0j} \|_{L^{\infty}} & \leq C e^{(1 - q) \Omega} \supgnorm{N}, \label{E:g0jlowerLinfinity} \\
	\| g_{jk} \|_{L^{\infty}} & \leq C e^{2 \Omega},  \label{E:gjklowerLinfinity}
\end{align}
\end{subequations}

\begin{subequations}
\begin{align}
	\| \bard g^{00} \|_{H^{N-1}} & \leq C e^{-q \Omega}  \supgnorm{N}, \label{E:partialg00upperHNminusone} \\
	\|g^{00} \|_{L^{\infty}} & \leq 5, \label{E:g00upperLinfinity} \\
	\| g^{00} + 1 \|_{H^N} & \leq C e^{-q \Omega} \supgnorm{N}, \label{E:g00upperplusoneHN} \\
	\| g^{00} + 1 \|_{L^{\infty}} & \leq C e^{-q \Omega} \supgnorm{N}, \label{E:g00upperplusoneLinfinity} \\
	\| g^{jk} \|_{L^{\infty}} & \leq C e^{-2 \Omega},  \label{E:gjkupperLinfinity} \\
	\| \bard g^{jk} \|_{H^{N-1}} & \leq C e^{-2 \Omega}  \supgnorm{N}, \label{E:partialgjkupperHNminusone} \\
	\| \bard g^{jk} \|_{L^{\infty}} & \leq C e^{-2 \Omega}  \supgnorm{N}, \label{E:partialgjkupperLinfinity} \\
	\| g^{0j} \|_{H^N} & \leq C e^{-(1 + q)\Omega}  \supgnorm{N}, \label{E:g0jupperHN} \\
	\| g^{0j} \|_{C_b^1} & \leq C e^{-(1 + q)\Omega} \supgnorm{N}, \label{E:g0jupperCb1} 
\end{align}
\end{subequations}

\begin{subequations}
\begin{align}
	\| \partial_t g_{jk} - 2 \omega g_{jk} \|_{H^N} & \leq C e^{(2 - q) \Omega} \suphstarstarnorm{N}, 
		\label{E:partialtgjkminusomegagjklowerHN} \\
	\| \partial_t g_{jk} - 2 \omega g_{jk} \|_{C_b^1} & \leq C e^{(2 - q) \Omega} \suphstarstarnorm{N}, 
		\label{E:partialtgjkminusomegagjklowerLinfinity} \\
	\| \partial_t g_{jk} \|_{C_b^1} & \leq C e^{2 \Omega}, \label{E:partialtgjklowerC1} 
\end{align}
\end{subequations}

\begin{subequations}
\begin{align}
	\| g^{aj} \partial_t g_{ak} - 2 \omega \delta_k^j \|_{H^N} & \leq C e^{-q \Omega}  \supgnorm{N}, 
		\label{E:gjaupperpartialtgaklowerminus2omegadeltakjHN} \\
	\| g^{aj} \partial_t g_{ak} - 2 \omega \delta_k^j \|_{L^{\infty}} & \leq C e^{-q \Omega}  \supgnorm{N}, 
		\label{E:gjaupperpartialtgaklowerminus2omegadeltakjLinfinity} 
\end{align}
\end{subequations}
	
\begin{subequations}
\begin{align}
	\| \partial_t g^{jk} + 2 \omega g^{jk} \|_{H^N} & \leq C e^{-(2+q)\Omega} \supgnorm{N},  
		\label{E:partialtgjkupperplusomegagjkHN} \\
	\| \partial_t g^{jk} + 2 \omega g^{jk} \|_{L^{\infty}} & \leq C e^{-(2+q)\Omega} \supgnorm{N},  
		\label{E:partialtgjkupperplusomegagjkLinfinity} \\
	\| \partial_t g^{00} \|_{L^{\infty}} & \leq C e^{-q \Omega} \supgnorm{N},  
		\label{E:partialtg00upperLinfinity} \\
	\| \partial_t g^{0j} \|_{L^{\infty}} & \leq C e^{-(1 + q) \Omega} \supgnorm{N}, 
		\label{E:partialtg0jupperLinfinity} \\
	\| \partial_t g^{jk} \|_{L^{\infty}} & \leq C e^{-2 \Omega},  
		\label{E:partialtgjkupperLinfinity} 
\end{align}
\end{subequations}

\begin{subequations}
\begin{align} 
	\| g^{ab}\Gamma_{a j b} \|_{H^N} & \leq C e^{(1-q)\Omega}\suphstarstarnorm{N}, 
		\label{E:gabupperGammaajblowerHN} \\
	\| g^{ab}\Gamma_{a j b} \|_{H^{N-1}} & \leq C \suphstarstarnorm{N}.
		\label{E:gabupperGammaajblowerHNminusone} 
\end{align}
\end{subequations}

Furthermore, the following estimates for $\partial_t \Phi$ hold:

\begin{subequations}
\begin{align}
	\| e^{\decayparameter \Omega} \partial_t \Phi - \bar{\Psi} \|_{L^{\infty}} & \leq C \supfluidnorm{N},
		\label{E:partialtPhiLinfinity} \\
	\Big\| \frac{1}{\partial_t \Phi} \Big\|_{L^{\infty}} & \leq \Big( \frac{1}{\bar{\Psi}} 
		+ C \supfluidnorm{N} \Big) e^{\decayparameter \Omega} \label{E:1overpartialtPhiLinfinity}.
\end{align}
\end{subequations}

In the above estimates, the norms $\suphstarstarnorm{N}, \supgnorm{N},$ and $\supfluidnorm{N}$
are defined in Definition \ref{D:Norms}.
\end{proposition}

\begin{proof}

Most of these estimates can be found in the statements and proofs of Lemmas 9, 11, 18, and 20 of \cite{hR2008}. The exceptions
are \eqref{E:partialtg0jupperLinfinity}, \eqref{E:partialtgjkupperLinfinity}, \eqref{E:partialtPhiLinfinity}, and \eqref{E:1overpartialtPhiLinfinity}. For brevity, we do not repeat all of the details of the estimates that are proved in \cite{hR2008}.

\begin{remark} \label{R:ProofsRemark}

Throughout all of the remaining proofs in this article, we freely use the results of Lemma \ref{L:backgroundaoftestimate}, the definitions of the norms from Section \ref{SS:NormsforgandpartialPhi}, the definitions \eqref{E:etamindef}, \eqref{E:qdef} of $\eta_{min}$ and $q,$ and the Sobolev embedding result $H^{M+2}(\mathbb{T}^3) \hookrightarrow C_b^M(\mathbb{T}^3),$ $(M \geq 0)$. \textbf{We also freely use the assumption that $\suptotalnorm{N},$ which is defined in \eqref{E:totalsupnorm}, is sufficiently small without explicitly mentioning it every time. Furthermore, the smallness is adjusted as necessary at each step in the proof.} To avoid overburdening the paper with details, we don't give explicit estimates for how small $\suptotalnorm{N}$ must be. We also remark that as discussed in Section \ref{SS:runningconstants}, the constants $c,C,C_*$ that appear throughout the article can be chosen uniformly (however, they may depend on $N$) as long as $\suptotalnorm{N}$ is sufficiently small. Finally, we prove statements in logical order, rather than the order in which they are stated in the proposition.

\end{remark}
 
Before beginning the proof, we observe the following Counting Principle, which provides a very useful heuristic guideline for the estimates concerning the spacetime metric, its inverse, and its derivatives. This tool is only intended to help guide the reader through the estimates; we will provide full proofs of many of the estimates.

\begin{center}
	\textbf{Counting Principle}
\end{center}	
	
	\emph{Consider a product which contains as factors metric components $g_{\mu \nu},$ the inverse metric components
	$g^{\mu \nu},$ and first derivatives of these quantities. If $u$ denotes the total number of \textbf{upstairs spatial metric} 
	among these factors, and $d$ denotes the total number of \textbf{downstairs spatial metric} indices, then the expected 
	contribution to the rate of growth/decay of the $H^N$ norm of the product coming from these terms is no larger than 
	$e^{\Omega(d-u)}.$ For purposes of counting, a spatial derivative $\partial_j$ is considered to be a downstairs spatial 
	index, while time derivatives are neutral. We remark that by these criteria, $h_{jk},$ $(j,k = 1,2,3),$
	is an order $1$ term that doesn't contribute to the decay rate. Furthermore, each factor in a product, excluding $h_{jk}$ but 
	including $\partial_a h_{jk},$ that is equal to one of quantities
	under an $H^N$ norm in definitions \eqref{E:mathfrakSMg00} - \eqref{E:mathfrakSMh**} contributes an additional decay factor 
	of $e^{-q \Omega}.$} 
	\\
	
	As an example, we remark that $\| g^{ab} \Gamma_{ajb} \|_{H^N}$ grows at a rate no larger than $e^{(1-q) \Omega}$ 
	(compare with inequality \eqref{E:gabupperGammaajblowerHN}). As a second example, we note the identity $g^{00} + 1 = 
	\frac{1}{g_{00}}[(g_{00} + 1) - g^{0a}g_{0a}],$ the right-hand side of which allows us to conclude that 
	$\|g^{00} + 1 \|_{H^N}$ grows at a rate no larger than $e^{-q \Omega}$ 
	(because $g_{00} + 1$ is one of the factors under an $H^N$ norm on the right-hand side of \eqref{E:mathfrakSMg00}).
	The net effect is that with \emph{only a single exception}, the products 
	appearing on the right-hand sides of the energy inequalities of Corollary \ref{C:metricfirstdiferentialenergyinequality} and 
	Corollary \ref{C:fluidenergytimederivative} can
	be bounded by $e^{-q \Omega}$ times one or more of the norms (which, according to our assumption
	$\suptotalnorm{N} \leq \epsilon',$ should be thought of as size $\leq \epsilon'$). The one exception is a term
	in inequality \eqref{E:underlinemathfrakEg0*firstdifferential}, namely the dangerous term 
	$\gzerostarnorm{N} e^{(q-1) \Omega} \| g^{ab} \Gamma_{ajb} \|_{H^N}.$ Consequently, we have to pay special attention to it 
	during the global existence argument of Section \ref{S:GlobalExistence}; see Remark \ref{R:Dangerousterm}.

\begin{remark}
	The Counting Principle applies only to the $H^N$ norms of products of the metric components, the inverse metric components, 
	and first derivatives of these quantities. For example, one should not estimate the growth rate of
	$\| \partial_a \partial_b g_{00} \|_{H^{N-1}}$ by simply observing that there are two spatial indices downstairs.
	The correct rate in this case, which can be discerned from definition \eqref{E:mathfrakSMg00},
	is $e^{(1 - q) \Omega}.$ Furthermore, the Counting Principle is not intended to be used on the fluid quantities. 
\end{remark}

\noindent \emph{Proof of \eqref{E:g00lowerLinfinity} - \eqref{E:gjklowerLinfinity}}:
The estimates \eqref{E:g00lowerLinfinity} and \eqref{E:g0jlowerLinfinity} follow from the definition 
\eqref{E:totalsupnorm} of $\suptotalnorm{N}$ and Sobolev embedding. The estimate \eqref{E:gjklowerLinfinity} follows from the
bootstrap assumption \eqref{E:gjkBAvsstandardmetric}. \\
 
\noindent \emph{Proof of \eqref{E:g00upperLinfinity}, \eqref{E:gjkupperLinfinity}, and 
the preliminary estimate $\| g^{0j} \|_{L^{\infty}} \leq C e^{-(1 + q)\Omega} \supgnorm{N}$}: The $L^{\infty}$ estimates \eqref{E:g00upperLinfinity}, \eqref{E:gjkupperLinfinity}, as  well as the estimate $\| g^{0j} \|_{L^{\infty}} \leq C e^{-(1 + q)\Omega}\supgnorm{N},$ which we will need shortly, follow from the definition \eqref{E:totalsupgnorm} of $\supgnorm{N},$ 
Lemma \ref{L:ginverseformluas}, and Lemma \ref{L:ginverseestimates}. \\
 
\noindent \emph{Proof of \eqref{E:partialg00upperHNminusone}, \eqref{E:partialgjkupperHNminusone}, \eqref{E:partialgjkupperLinfinity}, \eqref{E:g0jupperHN}, and \eqref{E:g0jupperCb1}}:
The proofs of \eqref{E:partialg00upperHNminusone}, \eqref{E:partialgjkupperHNminusone}, and \eqref{E:g0jupperHN}
begin with the fact that $\partial_{\vec{\alpha}} g^{\mu \nu}$ is a linear combination of terms of the form
 
 \begin{align} \label{E:gmunuupperdifferentiated}
 		g^{\lambda_1 \mu} g^{\lambda_2 \kappa_1} \cdots g^{\lambda_n \kappa_{n-1}} g^{\nu \kappa_n} 
 			(\partial_{\vec{\alpha}_1} g_{\lambda_1 \kappa_1}) \cdots (\partial_{\vec{\alpha}_n} g_{\lambda_n \kappa_n}),
 \end{align}
	where $\vec{\alpha}_1 + \cdots \vec{\alpha}_n = \vec{\alpha},$ where each $|\vec{\alpha}_i| > 0.$
We remark that \eqref{E:gmunuupperdifferentiated} can be shown inductively, using the identity
$\partial_{\lambda} g^{\mu \nu} = - g^{\alpha \mu} g^{\beta \nu} \partial_{\lambda} g_{\alpha \beta}.$

 To prove \eqref{E:partialg00upperHNminusone}, we must estimate \eqref{E:gmunuupperdifferentiated} in $L^2$ for $1 \leq 
 |\vec{\alpha}| \leq N,$ with $\mu, \nu$ in \eqref{E:gmunuupperdifferentiated} equal to $0,0.$ To this end, 	
 we first bound the terms $g^{\lambda_1 0} g^{\lambda_2 \kappa_1} \cdots g^{\lambda_n \kappa_{n-1}} g^{0 \kappa_n}$ in 
 $L^{\infty},$ and then estimate the remaining product $(\partial_{\vec{\alpha}_1} g_{\lambda_1 \kappa_1}) \cdots 
 (\partial_{\vec{\alpha}_n} g_{\lambda_n \kappa_n})$ in $L^2$ using Proposition \ref{P:derivativesofF1FkL2}. The $L^{\infty}$ 
 terms are bounded using \eqref{E:g00upperLinfinity}, \eqref{E:gjkupperLinfinity}, and 	
 the estimate $\| g^{0j} \|_{L^{\infty}} \leq C e^{-(1 + q)\Omega}\supgnorm{N}$ shown above. 
 The $L^2$ terms are controlled by $\supgnorm{N}.$ The 
 only difficulty is keeping track of the powers of $e^{\Omega},$ which Ringstr\"{o}m accomplishes inductively
 through a counting argument that is analogous to the Counting Principle stated above;
 the details can be found in the proof of Lemma 9 of \cite{hR2008}. The proofs of \eqref{E:partialgjkupperHNminusone} and 
 \eqref{E:g0jupperHN} are similar. Inequalities \eqref{E:partialgjkupperLinfinity} and \eqref{E:g0jupperCb1}
 now follow from \eqref{E:partialgjkupperHNminusone}, \eqref{E:g0jupperHN}, and Sobolev embedding. 
\\

\noindent \emph{Proof of \eqref{E:g00upperplusoneHN} - \eqref{E:g00upperplusoneLinfinity}}:
The estimate \eqref{E:g00upperplusoneHN} follows from the identity $g^{00} + 1 = \frac{1}{g_{00}}[(g_{00} + 1) - g^{0a}g_{0a}],$
Corollary \ref{C:DifferentiatedSobolevComposition} with $v = g_{00}$ and $F(g_{00}) = \frac{1}{g_{00}}$ in the corollary,  Proposition \ref{P:F1FkLinfinityHN}, the definition \eqref{E:totalsupgnorm} of $\supgnorm{N},$ \eqref{E:g0jupperHN}, and \eqref{E:g0jupperCb1}. Inequality \eqref{E:g00upperplusoneLinfinity} then follows from \eqref{E:g00upperplusoneHN} and Sobolev embedding. 
\\

\noindent \emph{Proof of \eqref{E:partialtgjkminusomegagjklowerHN} - \eqref{E:partialtgjklowerC1}}:
The estimate \eqref{E:partialtgjkminusomegagjklowerHN} follows directly from the definition \eqref{E:mathfrakSMsuph**}
of $\suphstarstarnorm{N}$ and the observation that $\partial_t h_{jk} = e^{-2 \Omega}(\partial_t g_{jk} - 2 \omega g_{jk}).$
Inequality \eqref{E:partialtgjkminusomegagjklowerLinfinity} then follows from \eqref{E:partialtgjkminusomegagjklowerHN} and Sobolev embedding. Inequality \eqref{E:partialtgjklowerC1} now follows from \eqref{E:gjklowerLinfinity},
\eqref{E:partialtgjkminusomegagjklowerLinfinity}, and Sobolev embedding.
\\

\noindent \emph{Proof of \eqref{E:gjaupperpartialtgaklowerminus2omegadeltakjHN} - \eqref{E:gjaupperpartialtgaklowerminus2omegadeltakjLinfinity}}:
Recall that $\delta_k^j = g^{0j} g_{0k} + g^{aj}g_{ak}.$ Using this fact,
the estimate \eqref{E:gjaupperpartialtgaklowerminus2omegadeltakjHN} follows from Proposition \ref{P:F1FkLinfinityHN}, the definition \eqref{E:totalsupgnorm} of $\supgnorm{N},$ and \eqref{E:gjkupperLinfinity} - \eqref{E:partialtgjkminusomegagjklowerLinfinity}. Inequality \eqref{E:gjaupperpartialtgaklowerminus2omegadeltakjLinfinity} then follows from \eqref{E:gjaupperpartialtgaklowerminus2omegadeltakjHN} and Sobolev embedding. \\

\noindent \emph{Proof of \eqref{E:partialtgjkupperplusomegagjkHN} - \eqref{E:partialtgjkupperLinfinity}}:
To prove \eqref{E:partialtgjkupperplusomegagjkHN}, first note the identity $\partial_t g^{jk} = - g^{\alpha j} g^{\beta k} \partial_t g_{\alpha \beta},$ which was mentioned below the statement of \eqref{E:gmunuupperdifferentiated}. We then
use Proposition \ref{P:F1FkLinfinityHN}, the definition \eqref{E:totalsupgnorm} of
$\supgnorm{N},$ Sobolev embedding, \eqref{E:partialg00upperHNminusone}, \eqref{E:g00upperLinfinity},
\eqref{E:gjkupperLinfinity}, \eqref{E:partialgjkupperHNminusone}, \eqref{E:g0jupperHN}, and \eqref{E:gjaupperpartialtgaklowerminus2omegadeltakjHN} to conclude that

\begin{align} \label{E:partialtgjkplus2omegagjkproofinequality}
	\| \partial_t g^{jk} + 2 \omega g^{jk} \|_{H^N} 
		& \leq \|  g^{bk} (g^{aj} \partial_t g_{ab} - 2 \omega \delta_{b}^j)\|_{H^N} 
			+ \| g^{0j} g^{bk} \partial_t g_{0b} \|_{H^N}
			+ \| g^{aj} g^{0k} \partial_t g_{0a} \|_{H^N}
			+ \| g^{0j}g^{0k} \partial_t g_{00} \|_{H^N} \\
		& \leq C e^{-(2 + q) \Omega} \supgnorm{N}, \notag
\end{align}
which gives \eqref{E:partialtgjkupperplusomegagjkHN}. Inequality \eqref{E:partialtgjkupperplusomegagjkLinfinity} now follows from \eqref{E:partialtgjkupperplusomegagjkHN} and Sobolev embedding. Inequality \eqref{E:partialtgjkupperLinfinity} 
then follows from  \eqref{E:gjkupperLinfinity} and \eqref{E:partialtgjkupperplusomegagjkLinfinity}. The proofs of \eqref{E:partialtg00upperLinfinity} and \eqref{E:partialtg0jupperLinfinity} are similar, and we omit the details. 
\\

\noindent \emph{Proof of \eqref{E:gabupperGammaajblowerHN}}:
To prove \eqref{E:gabupperGammaajblowerHN}, we first use Proposition \ref{P:F1FkLinfinityHN} to conclude that

\begin{align} \label{E:gabGammaajbfirstestimate}
	\| g^{ab}\Gamma_{a j b} \|_{H^N} \leq C \big\lbrace \| g^{ab} \|_{L^{\infty}} \| \Gamma_{ajb} \|_{H^N} 
		+ \| \bard g^{ab} \|_{H^{N-1}} \| \Gamma_{ajb} \|_{L^{\infty}} \big\rbrace.
\end{align} 

Recalling that $\Gamma_{ajb} = \frac{1}{2} \big( \partial_a g_{bj} + \partial_b g_{aj} - \partial_j g_{ab} \big)$ and
that $g_{jk} = e^{2 \Omega} h_{jk},$ and using \eqref{E:gjkupperLinfinity}, \eqref{E:partialgjkupperHNminusone}, the definitions \eqref{E:mathfrakSMsuph**} and \eqref{E:totalsupnorm} of $\suphstarstarnorm{N}$ and
$\suptotalnorm{N},$ and Sobolev embedding, it follows that the right-hand side of \eqref{E:gabGammaajbfirstestimate} is bounded from above by $C e^{(1-q) \Omega} \suphstarstarnorm{N}.$ This proves \eqref{E:gabupperGammaajblowerHN}. The proof of
\eqref{E:gabupperGammaajblowerHNminusone} follows similarly. 
\\

\noindent \emph{Proof of \eqref{E:partialtPhiLinfinity} - \eqref{E:1overpartialtPhiLinfinity}}:
Inequality \eqref{E:partialtPhiLinfinity} follows immediately from the definition \eqref{E:Fluidsupnorm}
of $\supfluidnorm{N}$ and Sobolev embedding. Assuming that $\supfluidnorm{N}$ is sufficiently small, inequality \eqref{E:1overpartialtPhiLinfinity} follows trivially from
\eqref{E:partialtPhiLinfinity}. This concludes the proof of Proposition \ref{P:BoostrapConsequences}.

\end{proof}

\subsection{Estimates of the error terms, \texorpdfstring{$m^{\mu \nu},$ and $\partial_{\mu} \Phi,$ $(\mu,\nu=0,1,2,3)$}{}} \label{SS:Nonlinearities}

In this section, we state and prove the second proposition that will be used to deduce the energy inequalities of Section
\ref{SS:IntegralInequalities}.

\begin{proposition} \label{P:Nonlinearities}
	Let $N \geq 3$ be an integer, and assume that the bootstrap assumptions \eqref{E:metricBAeta} - \eqref{E:g0jBALinfinity} hold 
	on the spacetime slab $[0,T) \times \mathbb{T}^3$ for some constant $c_1 \geq 1$ and for $\eta = \eta_{min}.$
	Then there exists a constant $\epsilon'' > 0$ and a constant 
	$C > 0,$ where $C$ depends on $N, c_1,$ and $\eta_{min},$ such that if and $\suptotalnorm{N}(t) < \epsilon''$ on $[0,T),$
	then the following estimates also hold on $[0,T)$ for the 
	quantities $\triangle_{A,\mu \nu}, \triangle_{C,00},$ and $\triangle_{C,0j},$ defined in \eqref{E:triangleA00def} - 
	\eqref{E:triangleAjkdef} and \eqref{E:triangleC00def} - \eqref{E:triangleC0jdef}:
	
	\begin{subequations}
	\begin{align}
		\| \triangle_{A,00} \|_{H^N} & \leq C e^{-2q \Omega} \supgnorm{N}^2, \label{E:triangleA00HN} \\
		\| \triangle_{A,0j} \|_{H^N} & \leq C e^{(1-2q) \Omega} \supgnorm{N}^2, \label{E:triangleA0jHN} \\
		\| \triangle_{A,jk} \|_{H^N} & \leq C e^{(2-2q) \Omega} \supgnorm{N}^2, \label{E:triangleAjkHN} \\
		\| \triangle_{C,00} \|_{H^N} & \leq C e^{-2q \Omega} \supgnorm{N}^2, \label{E:triangleC00HN} \\
		\| \triangle_{C,0j} \|_{H^N} & \leq C e^{(1-2q) \Omega} \supgnorm{N}^2. \label{E:triangleC0jHN}
	\end{align}
	\end{subequations}
	
	Additionally, for the quantities $\triangle_{\mu \nu}$ defined in \eqref{E:triangle00} - \eqref{E:trianglejk}, we have the
	following estimates on $[0,T):$
	
	\begin{subequations}
	\begin{align}
		\| \triangle_{00} \|_{H^N} & \leq C e^{-2q \Omega} \suptotalnorm{N}, \label{E:triangle00HN} \\
		\| \triangle_{0j} \|_{H^N} & \leq C  e^{(1 - 2q)\Omega} \suptotalnorm{N},  \label{E:triangle0jHN} \\
		\| \triangle_{jk} |_{H^N} & \leq C e^{-2q \Omega} \suptotalnorm{N}. \label{E:trianglejkHN}
	\end{align}
	\end{subequations}
	
	For the commutator terms from Corollary \ref{C:metricfirstdiferentialenergyinequality}, 
	we have the following estimates on $[0,T):$
	
	\begin{subequations}	
	\begin{align}
		\| [\hat{\Square}_g, \partial_{\vec{\alpha}} ] (g_{00} + 1) \|_{L^2} & \leq C e^{-2q \Omega} \suptotalnorm{N},  
			&& (|\vec{\alpha}| \leq N), \label{E:g00commutatorL2} \\
		\| [\hat{\Square}_g, \partial_{\vec{\alpha}} ] g_{0j} \|_{L^2} & \leq 
			 C e^{(1-2q) \Omega} \suptotalnorm{N}, && (|\vec{\alpha}| \leq N), \label{E:g0jcommutatorL2} \\
		\| [\hat{\Square}_g, \partial_{\vec{\alpha}} ] h_{jk} \|_{L^2} & \leq  
			 C e^{-2q \Omega} \suptotalnorm{N}, && (|\vec{\alpha}| \leq N).  \label{E:hjkcommutatorL2} 
	\end{align}
	\end{subequations}
	
	For the terms from Corollary \ref{C:metricfirstdiferentialenergyinequality},
	where $\triangle_{\mathcal{E};(\gamma, \delta)}[v,\partial v]$ is defined in \eqref{E:trianglemathscrEdef},
	we have the following estimates on $[0,T):$
	
	\begin{subequations}
	\begin{align}
		e^{2q \Omega} \| \triangle_{\mathcal{E};(\gamma_{00}, \delta_{00})}[\partial_{\vec{\alpha}} (g_{00} + 1)
			,\partial (\partial_{\vec{\alpha}} g_{00})] \|_{L^1}
			& \leq C e^{-q \Omega} \gzerozeronorm{N} \suptotalnorm{N}, && (|\vec{\alpha}| \leq N),   
			\label{E:triangleEgamma00delta00L1} \\
		e^{2(q-1) \Omega} \| \triangle_{\mathcal{E};(\gamma_{0*}, \delta_{0*})}[\partial_{\vec{\alpha}} g_{0j},
			\partial (\partial_{\vec{\alpha}} g_{0j})] \|_{L^1}
			& \leq C e^{-q \Omega} \gzerostarnorm{N} \suptotalnorm{N}, && (|\vec{\alpha}| \leq N),
			 \label{E:triangleEgamma0jdelta0*L1} \\
		e^{2q \Omega} \| \triangle_{\mathcal{E};(0,0)}[0,\partial (\partial_{\vec{\alpha}} h_{jk})] \|_{L^1}
			& \leq C e^{-q \Omega} \hstarstarnorm{N} \suptotalnorm{N}, && (|\vec{\alpha}| \leq N).
			\label{E:triangleEgamma**delta**L1}
	\end{align}
	\end{subequations}

	For the error terms corresponding to the fully raised Christoffel symbols
	of Lemma \ref{L:raisedChristoffeldecomposition} we have the following estimates on $[0,T):$
	
	\begin{subequations}
	\begin{align}
		\| \triangle_{(\Gamma)}^{000} \|_{H^N} & \leq C e^{-q \Omega} \supgnorm{N}, \label{E:triangleGamma000HN} \\
		\| \triangle_{(\Gamma)}^{j00} \|_{H^N} & \leq C e^{-(1 + q)\Omega} \supgnorm{N}, \label{E:triangleGammaj00HN} \\
		\| \triangle_{(\Gamma)}^{0j0} \|_{H^N} & \leq C e^{-(1 + q)\Omega} \supgnorm{N}, \\
		\| \triangle_{(\Gamma)}^{0jk} \|_{H^N} & \leq C e^{-(2 + q) \Omega} \supgnorm{N}, \\
		\| \triangle_{(\Gamma)}^{j0k}  \|_{H^N} & \leq C e^{-(2 + q) \Omega} \supgnorm{N}, \\ 
		\| \triangle_{(\Gamma)}^{ijk}  \|_{H^N} & \leq C e^{-(3 + q) \Omega} \supgnorm{N}. \label{E:triangleGammaijkHN}
 \end{align}
 \end{subequations}
	
	For the reciprocal acoustical metric error terms $\triangle_{(m)},$ $\triangle_{(m)}^{jk},$
	and $\triangle_{\partial \Phi}$ defined in \eqref{E:trianglemdef}, \eqref{E:trianglemjkdef},
	\eqref{E:triangledef} respectively, we have the following estimates on $[0,T):$
	
	\begin{subequations}
	\begin{align}
		\| \triangle_{(m)} \|_{H^N} & \leq C e^{-q \Omega} \suptotalnorm{N}, \label{E:trianglemHN} \\
		\| \triangle_{(m)} \|_{L^{\infty}} & \leq C e^{-q \Omega} \suptotalnorm{N},  \label{E:trianglemLinfinity} \\
			\| \triangle_{(m)}^{jk} \|_{H^N} & \leq C e^{-(2 + q) \Omega} \suptotalnorm{N},
			\label{E:trianglemjkHN} \\
		\| \triangle_{(m)}^{jk} \|_{L^{\infty}} & \leq C e^{-(2 + q) \Omega} \suptotalnorm{N},  
			\label{E:trianglemjkLinfinity} \\
		e^{\decayparameter \Omega} \| \triangle_{\partial \Phi} \|_{H^N} & \leq C e^{-q \Omega} \suptotalnorm{N}. \label{E:triangleHN}
	\end{align}
	\end{subequations}

	For the reciprocal acoustical metric components $m^{jk}$ and $m^{0j}$ defined in \eqref{E:mjkdef} and \eqref{E:m0jdef},
	we have the following estimates on $[0,T):$
	
	\begin{subequations}
	\begin{align}
		\| m^{0j} \|_{H^N} & \leq C e^{- \Omega} \suptotalnorm{N}, \label{E:m0jHN} \\
		\| m^{0j} \|_{C_b^1} & \leq C e^{- \Omega} \suptotalnorm{N},  \label{E:m0jC1} \\
		\Big\| m^{jk} - \frac{1}{2s+1} g^{jk} \Big\|_{L^{\infty}} & \leq C e^{-(2 + q) \Omega} \suptotalnorm{N}, 
		\label{E:mjkuppercomparedtogjk} 
		\\
		\| m^{jk} \|_{L^{\infty}} &	\leq C e^{-2 \Omega}, \label{E:mjkLinfinity} \\
		C^{-1} \delta^{ab}X_{a} X_{b}
			& \leq e^{2 \Omega} m^{ab}X_{a}X_{b} 
			\leq C \delta^{ab}X_{a}X_{b}, && \forall (X_1,X_2,X_3) \in \mathbb{R}^3, \label{E:mjkpositivedefinite} \\
		\| \bard m^{jk} \|_{H^{N-1}} & \leq C e^{-2 \Omega} \suptotalnorm{N}, \label{E:partialmjkHNminusone} \\
		\| \bard m^{jk} \|_{L^{\infty}} & \leq C e^{- 2 \Omega} \suptotalnorm{N}.  \label{E:partialmjkHNLinfinity} 
		\end{align}
	\end{subequations}

	For the commutator terms from Corollary \ref{C:fluidenergytimederivative}, we have the following estimates
	on $[0,T):$
	
	\begin{align} \label{E:SquaremcommutatorL2}
		\| [\hat{\Square}_m, \partial_{\vec{\alpha}} ] \Phi \|_{L^2} & \leq C e^{-(1 + \decayparameter)\Omega} \suptotalnorm{N}^2,
			&& (|\vec{\alpha}| \leq N).
	\end{align}

	For the Square of the enthalpy per particle $\sigma \eqdef - g^{\alpha \beta}(\partial_{\alpha} \Phi)(\partial_{\beta} 
	\Phi),$ we have the following estimates on $[0,T):$
	
	\begin{subequations}
	\begin{align}
		\sigma & > 0, \label{E:sigmaispositive} \\
		\| \sigma \|_{L^{\infty}} & \leq C e^{-2\decayparameter \Omega}, \label{E:sigmasLinfinity} \\
		\Big \| \frac{1}{\sigma} \Big \|_{L^{\infty}} & \leq C e^{2\decayparameter \Omega}. \label{E:1oversigmasLinfinity} 
	\end{align}
	\end{subequations}
	
	For $\bard \Phi$ and the ratio $z_j = \frac{e^{-\Omega} \partial_j \Phi}{\partial_t \Phi},$ $(j=1,2,3),$ 
	defined in \eqref{E:zjdef}, we have the following estimates on $[0,T):$
	
	\begin{subequations}
	\begin{align}
		\| \partial_{\vec{\alpha}} \Phi \|_{L^{\infty}} & \leq C \supfluidnorm{N}, && (1 \leq |\vec{\alpha}| \leq N - 2), 
			\label{E:partialvecalphaPhiLinfinity} \\
		\| \partial_{\vec{\alpha}} z_j \|_{L^{\infty}} & \leq C e^{(\decayparameter - 1) \Omega} \supfluidnorm{N}, 
			&& (|\vec{\alpha}| \leq N - 3), \label{E:partialvecalphazjLinfinity} \\
		\| z_j \|_{H^N} & \leq C \supfluidnorm{N}. &&  \label{E:zjHN} 
	\end{align}
	\end{subequations}

	For the time derivatives of the fluid-related quantities, we have the following estimates on $[0,T):$
	
	\begin{subequations}
	\begin{align}
		\| \partial_t^2 \Phi \|_{L^{\infty}} & \leq C e^{-\decayparameter \Omega}, \label{E:partialtSquaredPhiLinfinity} \\
		\| \partial_t z_j \|_{L^{\infty}} & \leq  C e^{(\decayparameter - 1)\Omega} \supfluidnorm{N},  
			\label{E:partialtzjLinfinity} \\
		\| \partial_t \triangle_{(m)} \|_{L^{\infty}} & \leq C e^{-q \Omega} \suptotalnorm{N} ,
			\label{E:partialttrianglemLinfinity} \\
		\| \partial_t \triangle_{(m)}^{jk} \|_{L^{\infty}} & \leq C e^{-(2+q) \Omega} \suptotalnorm{N},  
			\label{E:partialttrianglemjkLinfinity} \\
	 	\| \partial_t m^{0j} \|_{L^{\infty}} & \leq C e^{-(1+q) \Omega} \suptotalnorm{N}, \label{E:partialtm0jLinfinity} \\
	 	\| \partial_t m^{jk} + 2 \omega m^{jk} \|_{L^{\infty}} & \leq C e^{-(2+q) \Omega} \suptotalnorm{N}
			\label{E:partialtmjkplusomegamjkLinfinity}, \\
		\| \partial_t m^{jk} \|_{L^{\infty}} & \leq C e^{-2 \Omega} \suptotalnorm{N}.
			\label{E:partialtmjkLinfinity}
	\end{align}
	\end{subequations}
	
	In the above estimates, the norms $\gzerozeronorm{N}, \gzerostarnorm{N}, \hstarstarnorm{N},
	\supgzerozeronorm{N}, \supgzerostarnorm{N}, \suphstarstarnorm{N}, 
	\supgnorm{N}, \supfluidnorm{N},$ and $\suptotalnorm{N}$ are defined in Definition \ref{D:Norms}.
\end{proposition}

The proof of Proposition \ref{P:Nonlinearities} is located in Section \ref{SSS:ProofofPropositionNonlinearities}.
Before we prove the proposition, we first prove some Lemmas that will be useful for estimating
quantities involving $\partial \Phi.$ We also prove the estimates \eqref{E:sigmaispositive} - \eqref{E:1oversigmasLinfinity},
since we need them for the proofs of some of the Lemmas. See Remark \ref{R:ProofsRemark} for some conventions that we use throughout our proofs.

\subsection{Some preliminary lemmas for estimating \texorpdfstring{$\partial_{\mu} \Phi,$ $(\mu=0,1,2,3),$
and proofs of \eqref{E:sigmaispositive} - \eqref{E:1oversigmasLinfinity}}{}}

\begin{lemma} \label{L:sigmawidetildesigma}
	Under the assumptions of Proposition \ref{P:Nonlinearities} and Remark \ref{R:ProofsRemark}, there exists a constant $C>0$
	such that
	\begin{align} \label{E:sigmawidetildesigma}
		\| \sigma - \widetilde{\sigma} \|_{H^N} & \leq C e^{-2\decayparameter \Omega}\suptotalnorm{N}.
	\end{align}
\end{lemma}

\begin{proof}
Recall that $\sigma = - g^{\alpha \beta}(\partial_{\alpha} \Phi)(\partial_{\beta} \Phi)$ and 
$\widetilde{\sigma} = - \widetilde{g}^{\alpha \beta}(\partial_{\alpha} \widetilde{\Phi})(\partial_{\beta} \widetilde{\Phi})
= (\bar{\Psi} e^{-\decayparameter \Omega})^2.$ We therefore can decompose $\sigma - \widetilde{\sigma}$ as follows:

\begin{align}
	\sigma - \widetilde{\sigma} = (\partial_t \Phi - \bar{\Psi} e^{-\decayparameter \Omega})(\partial_t \Phi + \bar{\Psi} e^{-\decayparameter \Omega})
	- (g^{00} + 1) (\partial_t \Phi)^2 - 2g^{0a} (\partial_t \Phi)(\partial_a \Phi) - g^{ab}(\partial_a \Phi)(\partial_b \Phi).
\end{align}
Inequality \eqref{E:sigmawidetildesigma} now follows from Proposition \ref{P:F1FkLinfinityHN}, 
the definition \eqref{E:totalsupnorm} of $\suptotalnorm{N},$ Sobolev embedding, \eqref{E:g00upperplusoneHN}, 
\eqref{E:gjkupperLinfinity}, \eqref{E:partialgjkupperHNminusone}, \eqref{E:g0jupperHN}, and \eqref{E:partialtPhiLinfinity}.

\end{proof}

\begin{lemma} \label{L:sigmas}
Under the assumptions of Proposition \ref{P:Nonlinearities} and Remark \ref{R:ProofsRemark}, there exists a constant $C>0$
such that
	
	\begin{subequations}
	\begin{align}
		\sigma & > 0, \label{E:sigmaispositiveagain} \\
		\| \sigma \|_{L^{\infty}} & \leq C e^{-2\decayparameter \Omega}, \\
		\Big \| \frac{1}{\sigma} \Big \|_{L^{\infty}} & \leq C e^{2 \decayparameter \Omega},  \label{E:1oversigmasLinfinityagain} \\
		\| \bard(\sigma^s) \|_{H^{N-1}} & \leq C e^{-2\decayparameter s \Omega} \suptotalnorm{N}.  \label{E:partialsigmasHNminusone}
	\end{align}
	\end{subequations}
	
\end{lemma}

\begin{proof}
	Using the fact that $\widetilde{\sigma} = (\bar{\Psi} e^{-\decayparameter \Omega})^2,$ 
	\eqref{E:sigmawidetildesigma}, and Sobolev embedding, it follows that if $\suptotalnorm{N}$ is sufficiently small, then
	
	\begin{align}
		\sigma & > 0,  \\
		\| \sigma \|_{L^{\infty}} & \leq C e^{-2\decayparameter \Omega}, \\
		\Big \| \frac{1}{\sigma} \Big \|_{L^{\infty}} & \leq C e^{2 \decayparameter \Omega},
	\end{align}
	which implies \eqref{E:sigmaispositiveagain} - \eqref{E:1oversigmasLinfinityagain}
	(and which are the same as \eqref{E:sigmaispositive} - \eqref{E:1oversigmasLinfinity}).
	
	To prove \eqref{E:partialsigmasHNminusone},
	we first note that because $\widetilde{\sigma}$ is a function of $t$ alone, it follows from 
	\eqref{E:sigmawidetildesigma} that
	
	\begin{align} \label{E:partialsigmaHNminus1}
		\| \bard \sigma \|_{H^{N-1}} = \| \bard (\sigma - \widetilde{\sigma}) \|_{H^{N-1}}
			\leq \| \sigma - \widetilde{\sigma} \|_{H^{N}} \leq C e^{-2\decayparameter \Omega}\suptotalnorm{N}.
	\end{align}
	By Corollary \ref{C:DifferentiatedSobolevComposition}, with $v= \sigma$ and
	$F(\sigma) = \sigma^s$ in the corollary, we have that
	
	\begin{align} \label{E:partialsigmasPropositioninequality}
		\| \bard(\sigma^s) \|_{H^{N-1}} & 
			\leq C \| \bard \sigma \|_{H^{N-1}}
	 		\sum_{l=1}^{N} \| \sigma^{s-l} \|_{L^{\infty}} \| \sigma \|_{L^{\infty}}^{l - 1}.
	\end{align}
	Inequality \eqref{E:partialsigmasHNminusone} now follows from \eqref{E:sigmasLinfinity}, \eqref{E:1oversigmasLinfinity},
	\eqref{E:partialsigmaHNminus1}, and \eqref{E:partialsigmasPropositioninequality}.

\end{proof}

\begin{corollary}	
Under the assumptions of Proposition \ref{P:Nonlinearities} and Remark \ref{R:ProofsRemark}, there exists a constant $C>0$
such that
	
	\begin{align} \label{E:sigmawidetildesigmaHN}
		\| \sigma^s - \widetilde{\sigma}^s \|_{H^N} & \leq C e^{-2\decayparameter s\Omega}\suptotalnorm{N}.
	\end{align}

\end{corollary}

\begin{proof}
	By Taylor's theorem, we have that
	$\sigma^s - \widetilde{\sigma}^s = s (\sigma^*)^{s-1}(\sigma - \widetilde{\sigma})$ for some
	$\sigma^*$ lying between $\widetilde{\sigma}$ and $\sigma.$ By Lemma \ref{L:sigmas}, it follows that
	$(\sigma^*)^{s-1} \leq C e^{-2\decayparameter (s-1) \Omega}.$ By Lemma \ref{L:sigmawidetildesigma}, it therefore
	follows that
	
	\begin{align} \label{E:sigmaswidetildesigmasL2}
		\| \sigma^s - \widetilde{\sigma}^s  \|_{L^2} \leq C e^{-2\decayparameter (s-1) \Omega} 
			\| \sigma - \widetilde{\sigma} \|_{L^2} 
		\leq C e^{-2\decayparameter s \Omega} \suptotalnorm{N}. 
	\end{align}
	
	To bound the spatial derivatives of $\sigma^s - \widetilde{\sigma}^s$ in $L^2,$ we note that
	$\widetilde{\sigma}$ is a function of $t$ alone. Therefore, by \eqref{E:partialsigmasHNminusone},
	it follows that
	\begin{align} \label{E:partialsigmaswidetildesigmasHNminusone}
		\| \bard(\sigma^s - \widetilde{\sigma}^s) \|_{H^{N-1}} = \| \bard(\sigma^s) \|_{H^{N-1}} 
		\leq C e^{-2\decayparameter s \Omega} \suptotalnorm{N}.
	\end{align}
	Combining \eqref{E:sigmaswidetildesigmasL2} and \eqref{E:partialsigmaswidetildesigmasHNminusone}, we deduce
	 \eqref{E:sigmawidetildesigmaHN}.
	
\end{proof}

\begin{corollary} \label{C:fofPhiminusfofwidetildePhiHN}
Let $f(\partial \Phi)$ be defined as in \eqref{E:ffirstdef}:

\begin{align}
	f(\partial \Phi) \eqdef 2 \sigma^s (\partial_t \Phi)^2 - \frac{s}{s+1} \sigma^{s+1}.
\end{align}
Then under the assumptions of Proposition \ref{P:Nonlinearities} and Remark \ref{R:ProofsRemark}, there exists a constant $C>0$
such that

\begin{align}  \label{E:fofPhiminusfofwidetildePhiHN}
	\| f(\partial \Phi) - f(\partial \widetilde{\Phi}) \|_{H^N} & \leq 
		C e^{-2\decayparameter(s + 1) \Omega}\suptotalnorm{N}.
\end{align}

\end{corollary}

\begin{proof}

Using Proposition \ref{P:F1FkLinfinityHN}, the definition \eqref{E:totalsupnorm} of $\suptotalnorm{N},$ Sobolev embedding,
\eqref{E:partialtPhiLinfinity}, \eqref{E:sigmasLinfinity}, \eqref{E:partialsigmasHNminusone}, and \eqref{E:sigmawidetildesigmaHN}, we estimate

\begin{align}
	\| f(\partial \Phi) - f(\partial \widetilde{\Phi}) \|_{H^N} 
	& \leq C \Big\lbrace
		\| (\partial_t \Phi - \bar{\Psi} e^{-\decayparameter \Omega})(\partial_t \Phi + \bar{\Psi} 
		e^{-\decayparameter\Omega})\sigma^s \|_{H^N} \\
	& \hspace{1in} + (\bar{\Psi} e^{-\decayparameter \Omega})^2 \| \sigma^{s} - \widetilde{\sigma}^s \|_{H^N}
			+ \| \sigma^{s+1} - \widetilde{\sigma}^{s+1} \|_{H^N} \Big\rbrace \notag \\
	& \leq C e^{-2\decayparameter(s + 1) \Omega} \suptotalnorm{N}. \notag
\end{align}
This demonstrates \eqref{E:fofPhiminusfofwidetildePhiHN}.

\end{proof}

\subsection{Proof of Proposition \ref{P:Nonlinearities}} \label{SSS:ProofofPropositionNonlinearities}

\begin{proof} We now provide a proof of Proposition \ref{P:Nonlinearities}. We prove statements in logical order, rather than in the order in which they are listed in the conclusions of the proposition. See Remark \ref{R:ProofsRemark} for some conventions that we use throughout the proof. \ \\

\noindent \emph{Proofs of \eqref{E:partialvecalphaPhiLinfinity} - \eqref{E:zjHN}}:
	By Sobolev embedding and the definition \eqref{E:Fluidnorm} of $\supfluidnorm{|\vec{\alpha}| + 2}$, the following 
	inequality holds for $1 \leq |\vec{\alpha}| \leq N - 2:$
	\begin{align} \label{E:partialtpartialvecalphaPhiLinfinity}
		\|\partial_t \partial_{\vec{\alpha}} \Phi \|_{L^{\infty}} & \leq C e^{-\decayparameter \Omega} 
		\supfluidnorm{|\vec{\alpha}| + 2}.
	\end{align}
	Integrating $\partial_t \partial_{\vec{\alpha}} \Phi$ in time, using \eqref{E:partialtpartialvecalphaPhiLinfinity}, and using 
	that $\supfluidnorm{|\vec{\alpha}| + 2}(t)$ is increasing, we conclude that
	\begin{align} \label{E:partialvecalphaPhiLinfinitysmall}
		\| \partial_{\vec{\alpha}} \Phi (t) \|_{L^{\infty}} \leq C \Big( \supfluidnorm{|\vec{\alpha}| + 2}(0)
			+ \int_{0}^{t} e^{-\decayparameter \Omega(\tau)} \supfluidnorm{|\vec{\alpha}| + 2}(\tau) \, d \tau \Big) & 
			\leq C \supfluidnorm{|\vec{\alpha}| + 2}(t),  
	\end{align}
	which implies inequality \eqref{E:partialvecalphaPhiLinfinity}.

To prove \eqref{E:partialvecalphazjLinfinity} and \eqref{E:zjHN}, we use the definition \eqref{E:Fluidsupnorm} of $\supfluidnorm{N},$
Corollary \ref{C:DifferentiatedSobolevComposition}, with $v = \partial_t \Phi$ and
$F(\partial_t \Phi) \eqdef \frac{1}{\partial_t \Phi}$ in the corollary, together with
\eqref{E:partialtPhiLinfinity} and \eqref{E:1overpartialtPhiLinfinity} to conclude that

\begin{align} \label{E:partial1overpartialtPhiHNminusone}
	\Big \| \bard \Big( \frac{1}{\partial_t \Phi} \Big) \Big\|_{H^{N-1}}
  	& \leq C e^{-\decayparameter \Omega}\supfluidnorm{N} \sum_{l=1}^N e^{\decayparameter(l+1)\Omega} 
    e^{-\decayparameter(l - 1)}  = C e^{\decayparameter \Omega} \supfluidnorm{N}.
\end{align}
By \eqref{E:1overpartialtPhiLinfinity}, Sobolev embedding, and \eqref{E:partial1overpartialtPhiHNminusone}, it follows that
\begin{align} 
	\Big \| \partial_{\vec{\alpha}} \Big( \frac{1}{\partial_t \Phi} \Big) \Big\|_{L^{\infty}} 
	& \leq C e^{\decayparameter \Omega}, && (|\vec{\alpha}| \leq N-2). \label{E:partialvecalpha1overpartialtPhiLinfinity}
\end{align}

Now by the definition $z_j \eqdef \frac{e^{-\Omega} \partial_j \Phi}{\partial_t \Phi},$ the product rule, Proposition \ref{P:F1FkLinfinityHN}, \eqref{E:1overpartialtPhiLinfinity}, \eqref{E:partialvecalphaPhiLinfinity}, \eqref{E:partial1overpartialtPhiHNminusone}, and \eqref{E:partialvecalpha1overpartialtPhiLinfinity}, we have that
\begin{align}
	\| \partial_{\vec{\alpha}} z_j \|_{L^{\infty}} & \leq C e^{(\decayparameter -1) \Omega} 
		\supfluidnorm{N}, && (|\vec{\alpha}| \leq N - 3), \\
	\| z_j \|_{H^N} & \leq C e^{- \Omega} \bigg\lbrace \Big \| \frac{1}{\partial_t \Phi} \Big\|_{L^{\infty}} 
		\| \bard \Phi \|_{H^N} + \| \bard \Phi \|_{L^{\infty}} \Big \| \bard \Big( \frac{1}{\partial_t \Phi} \Big) 
		\Big\|_{H^{N-1}} \bigg\rbrace \leq C \supfluidnorm{N}, && (j=1,2,3),
\end{align}
which proves \eqref{E:partialvecalphazjLinfinity} and \eqref{E:zjHN}.
\\

\noindent \emph{Proofs of \eqref{E:triangleA00HN} - \eqref{E:triangleC0jHN}}: To prove \eqref{E:triangleA00HN}, we first recall equation \eqref{E:triangleA00def}:

\begin{align}
	\triangle_{A,00} & = (g^{00})^2 \Big\lbrace(\partial_t g_{00})^2 - (\Gamma_{000})^2 \Big\rbrace 
			+ g^{00}g^{0a} \Big\lbrace 2 (\partial_t g_{00})(\partial_t g_{0a} + \partial_a g_{00}) 
			- 4 \Gamma_{000} \Gamma_{00a} \Big\rbrace \label{E:triangleA00defagain} \\
		& \ \ \ + g^{00}g^{ab} \Big\lbrace (\partial_t g_{0a})(\partial_t g_{0b}) 
				+ (\partial_a g_{00}) (\partial_b g_{00}) 
				- 2 \Gamma_{00a} \Gamma_{00b} \Big\rbrace \notag \\
		& \ \ \ + g^{0a} g^{0b} \Big\lbrace 2(\partial_t g_{00})(\partial_a g_{0b}) + 2(\partial_t g_{0b})(\partial_a g_{00}) 
				- 2 \Gamma_{000} \Gamma_{a0b} - 2 \Gamma_{00b} \Gamma_{00a} \Big\rbrace \notag \\
		& \ \ \ + g^{ab} g^{0l} \Big\lbrace 2(\partial_t g_{0a})(\partial_l g_{0b}) + 2(\partial_b g_{00})(\partial_a g_{0l}) 
				- 4\Gamma_{00a} \Gamma_{l0b}  \Big\rbrace \notag \\
		& \ \ \ + g^{ab}g^{lm}(\partial_a g_{0l})(\partial_b g_{0m}) 
				+ \frac{1}{2} g^{lm}(\underbrace{g^{ab} \partial_t g_{al} - 2\omega \delta_l^b}_{
				e^{2 \Omega} g^{ab} \partial_t h_{al} - 2 \omega g^{0b}g_{0l}})(\partial_b g_{0m} + \partial_m g_{0b}) \notag \\
		& \ \ \ - \frac{1}{4} g^{ab} g^{lm}(\partial_a g_{0l} + \partial_l g_{0a})(\partial_b g_{0m} + \partial_m g_{0b}) 
			- \frac{1}{4}(\underbrace{g^{ab} \partial_t g_{al} - 2\omega \delta_l^b}_{
				e^{2 \Omega} g^{ab} \partial_t h_{al} - 2 \omega g^{0b}g_{0l}}) 
				(\underbrace{g^{lm} \partial_t g_{bm} - 2 \omega \delta_b^l}_{e^{2 \Omega} g^{lm} \partial_t h_{bm} 
				- 2 \omega g^{0l}g_{0b}}). \notag 
\end{align}

We now use Proposition \ref{P:F1FkLinfinityHN}, the definition \eqref{E:totalsupnorm} of $\suptotalnorm{N},$ Sobolev embedding, 
\eqref{E:g00upperplusoneHN}, \eqref{E:g00upperplusoneLinfinity}, \eqref{E:gjkupperLinfinity}, \eqref{E:partialgjkupperHNminusone}, 
\eqref{E:g0jupperHN}, \eqref{E:g0jupperCb1}, \eqref{E:partialtgjkminusomegagjklowerHN}, \eqref{E:partialtgjklowerC1}, \eqref{E:gjaupperpartialtgaklowerminus2omegadeltakjHN}, \eqref{E:gjaupperpartialtgaklowerminus2omegadeltakjLinfinity}, and the  relation $\Gamma_{\mu \alpha \nu} = \frac{1}{2}(\partial_{\mu} g_{\alpha \nu} + \partial_{\nu} g_{\alpha \mu} - \partial_{\alpha} g_{\mu \nu})$ to conclude that

\begin{align}
	\| \triangle_{A,00} \|_{H^N} & \leq C e^{-2q \Omega} \supgnorm{N}^2.
\end{align}
This proves \eqref{E:triangleA00HN}. The proofs of \eqref{E:triangleA0jHN} - \eqref{E:triangleC0jHN} are similar, and we omit the details. We also remark
that slight variations of these inequalities were proved in Lemma 12 of \cite{hR2008}. 
\\

\noindent \emph{Proofs of \eqref{E:triangle00HN} - \eqref{E:trianglejkHN}}:
To prove \eqref{E:triangle00HN}, we first use equation \eqref{E:triangle00}
and Proposition \ref{P:F1FkLinfinityHN} to arrive at the following estimate:

\begin{align}
	\| \triangle_{00} \|_{H^N} & \leq C\Big\lbrace \| \triangle_{A,00} \|_{H^N} + \| \triangle_{C,00} \|_{H^N}
		+ \| f(\partial \widetilde{\Phi}) - f(\partial \Phi) \|_{H^N}  
		+ |f(\partial \widetilde{\Phi})| \| g_{00} + 1 \|_{H^N} \\
	& \ \ \ + \|g_{00} + 1\|_{H^N} \| \sigma^{s+1} \|_{L^{\infty}}
		+ \| g_{00} + 1 \|_{L^{\infty}} \| \bard(\sigma^{s+1}) \|_{H^{N-1}} 
		+ |\omega - H| \| \partial_t g_{00}\|_{H^N} + |\omega + H||\omega - H| \| g_{00} + 1 \|_{H^N} \Big\rbrace. \notag
\end{align}

We now use \eqref{E:g00upperplusoneHN}, \eqref{E:g00upperplusoneLinfinity}, \eqref{E:triangleA00HN}, \eqref{E:triangleC00HN},  
\eqref{E:partialsigmasHNminusone}, \eqref{E:sigmasLinfinity}, \eqref{E:fofPhiminusfofwidetildePhiHN}, 
the definition \eqref{E:totalsupnorm} of $\suptotalnorm{N},$ the fact that
$f(\partial \widetilde{\Phi}) = \big(\frac{s+2}{s+1}\big)\bar{\Psi}^{2(s+1)} e^{-2(s+1)\decayparameter \Omega},$
and the trivial estimate $e^{-2(s+1)\decayparameter \Omega} \leq e^{-2q\Omega}$ to conclude that

\begin{align}
	\| \triangle_{00} \|_{H^N} < C e^{-2q \Omega}\suptotalnorm{N},
\end{align}
which proves \eqref{E:triangle00HN}.

Inequalities \eqref{E:triangle0jHN} and \eqref{E:trianglejkHN} can be proved using similar reasoning;
we omit the details.
\\

\noindent \emph{Proofs of \eqref{E:triangleEgamma00delta00L1} - \eqref{E:triangleEgamma**delta**L1}}:
To begin, we first recall equation \eqref{E:trianglemathscrEdef}:

\begin{align} \label{E:trianglemathscrEdefagain}
	\triangle_{\mathcal{E};(\gamma, \delta)}[v,\partial v] & = - \gamma H (\partial_a g^{ab}) v \partial_b v
		- 2 \gamma H (\partial_a g^{0a}) v \partial_t v - 2 \gamma H g^{0a}(\partial_a v)(\partial_t v) \\
	& \ \ \ - (\partial_a g^{0a})(\partial_t v)^2 - (\partial_a g^{ab})(\partial_b v)(\partial_t v)
		- \frac{1}{2}(\partial_t g^{00})(\partial_t v)^2 \notag \\
	& \ \ \ + \Big(\frac{1}{2} \partial_t g^{ab} + \omega g^{ab} \Big) (\partial_a v) (\partial_b v)
		+ (H - \omega) g^{ab} (\partial_a v) (\partial_b v) \notag \\
	& \ \ \ - \gamma H (\partial_t g^{00}) v \partial_t v - \gamma H (g^{00} + 1)(\partial_t v)^2. \notag
\end{align}

We now claim that the following inequality holds for any function $v$ for which the right-hand side is finite:
\begin{align} \label{E:trianglegammadeltaL1}
	\| \triangle_{\mathcal{E};(\gamma, \delta)}[v,\partial v] \|_{L^1} & \leq C \Big\lbrace e^{-q \Omega}  
		\| \partial_t v \|_{L^2}^2 + e^{-(2 + q) \Omega} \| \bard v 
		\|_{L^2}^2 + C_{(\gamma)} e^{-q \Omega} \| v \|_{L^2}^2 \Big\rbrace,
\end{align}
where $C_{(\gamma)}$ is defined in \eqref{E:mathcalEfirstlowerbound} To obtain \eqref{E:trianglegammadeltaL1}, we use the Cauchy-Schwarz inequality for integrals, \eqref{E:g00upperplusoneLinfinity}, \eqref{E:gjkupperLinfinity}, \eqref{E:partialgjkupperLinfinity}, \eqref{E:g0jupperCb1}, \eqref{E:partialtgjkupperplusomegagjkLinfinity}, and \eqref{E:partialtg00upperLinfinity}. Inequalities \eqref{E:triangleEgamma00delta00L1} - \eqref{E:triangleEgamma**delta**L1} now easily follow from definitions \eqref{E:mathfrakSMg00} - \eqref{E:totalsupnorm} and \eqref{E:trianglegammadeltaL1}.
\\

\noindent \emph{Proofs of \eqref{E:triangleGamma000HN} - \eqref{E:triangleGammaijkHN}}:
To estimate $\triangle_{(\Gamma)}^{000},$ we first recall equation \eqref{E:triangleGamma000}: 

\begin{align} 
	\triangle_{(\Gamma)}^{000} & = \frac{1}{2} (g^{00})^3 \partial_t g_{00}
		+ \frac{1}{2} g^{00} g^{0a} g^{0b} (\partial_t g_{ab} + 2 \partial_a g_{0b}) \label{E:triangleGamma000again} \\
	& \ \ \ + \frac{1}{2} (g^{00})^2 g^{0a} (2\partial_t g_{0a} + \partial_{a}g_{00})
		+ \frac{1}{2} g^{0a}g^{0b}g^{0l} \partial_a g_{bl}. \notag  
\end{align}
By Proposition \ref{P:F1FkLinfinityHN}, the definition \eqref{E:totalsupgnorm} of $\supgnorm{N},$ Sobolev embedding, 
 \eqref{E:g00upperplusoneHN}, \eqref{E:g0jupperHN}, \eqref{E:g0jupperCb1}, and \eqref{E:partialtgjklowerC1}, it follows that
 
\begin{align}
	\| \triangle_{(\Gamma)}^{000} \|_{H^N} \leq C e^{- q \Omega} \supgnorm{N},
\end{align}
which proves \eqref{E:triangleGamma000HN}. The estimates \eqref{E:triangleGammaj00HN} - \eqref{E:triangleGammaijkHN} can be proved similarly using Propositions \ref{P:BoostrapConsequences} and \ref{P:F1FkLinfinityHN}. 
\\

\noindent \emph{Proofs of \eqref{E:trianglemHN} - \eqref{E:trianglemjkLinfinity}}:
To prove \eqref{E:trianglemHN}, we first recall equation \eqref{E:trianglemdef}:

\begin{align}
	\triangle_{(m)} & = (1 + 2s)(g^{00} + 1)(g^{00} - 1)
		+ 2(1 + 2s) e^{\Omega} g^{00} g^{0a} z_a + e^{2 \Omega}(g^{00}g^{ab} + 2s g^{0a}g^{0b})z_a z_b. 
		\label{E:trianglemdefagain}
\end{align}
By Proposition \ref{P:F1FkLinfinityHN}, the definition \eqref{E:totalsupnorm} of $\suptotalnorm{N},$ Sobolev embedding, \eqref{E:g00upperplusoneHN}, \eqref{E:g00upperplusoneLinfinity},
\eqref{E:gjkupperLinfinity}, \eqref{E:partialgjkupperHNminusone}, \eqref{E:g0jupperHN}, \eqref{E:g0jupperCb1}, 
\eqref{E:partialvecalphazjLinfinity}, and \eqref{E:zjHN}, it follows that

\begin{align}
	\| \triangle_{(m)} \|_{H^N} & \leq C e^{-q \Omega} \suptotalnorm{N},
\end{align}
which proves \eqref{E:trianglemHN}. Inequality \eqref{E:trianglemLinfinity} now follows from \eqref{E:trianglemHN} and Sobolev embedding.

Inequality \eqref{E:trianglemjkHN} follows from equation \eqref{E:trianglemjkdef} using similar reasoning.
Inequality \eqref{E:trianglemjkLinfinity} then follows from \eqref{E:trianglemjkHN} and Sobolev embedding.
\\

\noindent \emph{Proof of \eqref{E:triangleHN}}: To prove \eqref{E:triangleHN}, we first 
multiply each side of equation \eqref{E:triangledef} by $e^{\decayparameter \Omega}$ to deduce that

\begin{multline}
	e^{\decayparameter \Omega}\triangle_{\partial \Phi} = \omega e^{\decayparameter \Omega}
		\big\lbrace 3F(\triangle_{(m)}) - \decayparameter \big\rbrace\partial_t \Phi 
	- e^{\decayparameter \Omega} F(\triangle_{(m)}) 
		\bigg\lbrace 3\omega (g^{00} + 1) \partial_t \Phi  + 6 \omega e^{\Omega} g^{0a}z_a \partial_t \Phi 
		+ (3 - 2s) \omega e^{2\Omega}g^{ab}z_a z_b \partial_t \Phi \label{E:trianglePhidefagain} \\
		+ 2s\Big(\triangle_{(\Gamma)}^{000} \partial_t \Phi
		+ [\triangle_{(\Gamma)}^{a00} + \triangle_{(\Gamma)}^{0a0} + \triangle_{(\Gamma)}^{00a}] \partial_a \Phi 
		+ e^{2 \Omega} [\triangle_{(\Gamma)}^{0ab} + \triangle_{(\Gamma)}^{a0b} + \triangle_{(\Gamma)}^{ab0}] 
			z_a z_b \partial_t \Phi 
		+ e^{2 \Omega} \triangle_{(\Gamma)}^{abc} z_a z_b \partial_c \Phi\Big) \bigg\rbrace,
\end{multline}
where

\begin{align}
	F(\triangle_{(m)}) & \eqdef \big[ (1 + 2s) + \triangle_{(m)} \big]^{-1}, \label{E:Fproofdef}
\end{align}
and $\triangle_{(m)}$ is defined in \eqref{E:trianglemdef}.

We begin by estimating $3F(\triangle_{(m)}) - \decayparameter$ where $\decayparameter = \frac{3}{1 + 2s} = 3F(0)$ (as in \eqref{E:wdef}). Using Corollary \ref{C:SobolevTaylor}, with $v = \triangle_{(m)}$ in the corollary,
and \eqref{E:trianglemHN}, it follows that

\begin{align} \label{E:FoftrianglemminusFof0HN}
	\Big\| 3F(\triangle_{(m)}) - 3F(0) \Big\|_{H^N} \leq C \| \triangle_{(m)} \|_{H^N} 
		\leq C e^{-q \Omega} \suptotalnorm{N}.
\end{align}
As a simple consequence of \eqref{E:FoftrianglemminusFof0HN}, we note that

\begin{align}
	\Big\| F(\triangle_{(m)}) - \overbrace{\frac{1}{1 + 2s}}^{F(0)}\Big\|_{L^{\infty}} & 
		\leq C e^{-q \Omega} \suptotalnorm{N}, \label{E:FLinfinity} \\
	\| \bard \big(F(\triangle_{(m)})\big) \|_{H^{N-1}} & \leq C e^{-q \Omega} \suptotalnorm{N}. \label{E:partialFHNminusone}
\end{align}

We now estimate the terms on the right-hand side of \eqref{E:trianglePhidefagain}. The bound \\ 
$e^{\decayparameter \Omega} \omega \Big\| \Big\lbrace 3F(\triangle_{(m)}) - 3F(0) \Big\rbrace \partial_t \Phi \Big \|_{H^N}  \leq C e^{-q \Omega} \suptotalnorm{N}$ follows from Proposition \ref{P:F1FkLinfinityHN}, the definition \eqref{E:totalsupnorm} of $\suptotalnorm{N},$ Sobolev embedding, \eqref{E:partialtPhiLinfinity}, and \eqref{E:FoftrianglemminusFof0HN}. 
 
The remaining term to be estimated is $e^{\decayparameter \Omega}\Big\| F(\triangle_{(m)}) \Big\lbrace G \Big\rbrace \Big\|_{H^N},$ where $G $ is an abbreviation for the terms in the large braces on the right-hand side of \eqref{E:trianglePhidefagain}. By Proposition \ref{P:F1FkLinfinityHN}, the definition \eqref{E:totalsupnorm} of $\suptotalnorm{N},$ Sobolev embedding, \eqref{E:g00upperplusoneHN}, \eqref{E:gjkupperLinfinity}, \eqref{E:partialgjkupperHNminusone}, \eqref{E:g0jupperHN},
\eqref{E:partialtPhiLinfinity}, \eqref{E:triangleGamma000HN} - \eqref{E:triangleGammaijkHN},
\eqref{E:partialvecalphaPhiLinfinity}, \eqref{E:partialvecalphazjLinfinity}, and \eqref{E:zjHN},
we have that $e^{\decayparameter \Omega} \| G \|_{H^N} \leq C e^{-q\Omega} \suptotalnorm{N}.$ We may therefore again apply Proposition \ref{P:F1FkLinfinityHN}, using \eqref{E:FLinfinity} and \eqref{E:partialFHNminusone} to deduce that $e^{\decayparameter \Omega}\Big\| F(\triangle_{(m)}) \Big\lbrace G \Big\rbrace \Big\|_{H^N} \leq C e^{-q\Omega} \suptotalnorm{N}.$ This concludes our proof of \eqref{E:triangleHN}.
\\

\noindent \emph{Proofs of \eqref{E:m0jHN} - \eqref{E:partialmjkHNLinfinity}}: To prove \eqref{E:mjkuppercomparedtogjk} - \eqref{E:partialmjkHNLinfinity}, we use \eqref{E:mjkdef} to decompose 

\begin{align} \label{E:mjkmultiplicativedecomposition}
 m^{jk} = F(\triangle_{(m)}) \ M^{jk}, 
\end{align} 
where $F(\triangle_{(m)})$ is defined in \eqref{E:Fproofdef}, and

\begin{align} \label{E:Mjkdecomposition}
	M^{jk} & \eqdef g^{jk} - \triangle_{(m)}^{jk}. 
\end{align}
By Proposition \ref{P:F1FkLinfinityHN}, \eqref{E:gjkupperLinfinity}, \eqref{E:partialgjkupperHNminusone},  
\eqref{E:trianglemHN}, and \eqref{E:trianglemLinfinity}, we have that

\begin{align}
	\| M^{jk} \|_{L^{\infty}} & \leq C e^{-2 \Omega}, \\
	\| \bard M^{jk} \|_{H^{N-1}} & \leq C e^{-2 \Omega} \suptotalnorm{N}. \label{E:partialbigmHNminusone}
\end{align}
It therefore follows from Proposition \ref{P:F1FkLinfinityHN}, \eqref{E:gjkupperLinfinity}, and \eqref{E:FLinfinity} - \eqref{E:partialbigmHNminusone} that
\begin{align}
	\Big\|m^{jk} - \frac{1}{2s + 1} g^{jk} \Big\|_{L^{\infty}} & \leq C e^{-(2 + q) \Omega} \suptotalnorm{N}, \\
	\| m^{jk} \|_{L^{\infty}} & \leq C e^{-2 \Omega}, \\
	\| \bard m^{jk} \|_{H^{N-1}} & \leq C e^{-2 \Omega} \suptotalnorm{N},
\end{align}
which proves \eqref{E:mjkuppercomparedtogjk}, \eqref{E:mjkLinfinity},
and \eqref{E:partialmjkHNminusone}. Inequality \eqref{E:partialmjkHNLinfinity} now follows from Sobolev embedding and \eqref{E:partialmjkHNminusone}. Inequality \eqref{E:mjkpositivedefinite} follows from \eqref{E:gjkuppercomparetostandard} and \eqref{E:mjkuppercomparedtogjk}. 

Inequalities \eqref{E:m0jHN} and \eqref{E:m0jC1} can be proved using the ideas similar to the ones we used to prove \eqref{E:mjkuppercomparedtogjk} - \eqref{E:partialmjkHNLinfinity}; we omit the details. 
\\

\noindent \emph{Proofs of \eqref{E:g00commutatorL2} - \eqref{E:hjkcommutatorL2}}:
To prove \eqref{E:g00commutatorL2} - \eqref{E:hjkcommutatorL2}, we first use equations \eqref{E:finalg00equation} - \eqref{E:finalhjkequation} to conclude that

\begin{subequations}
\begin{align} 
	\partial_t^2 g_{00} & = (g^{00})^{-1}\Big(-g^{ab} \partial_a \partial_b g_{00}  
		-2 g^{0a} \partial_a \partial_t g_{00} + 5 H \partial_t g_{00} + 6 H^2 (g_{00} + 1) + \triangle_{00} \Big),
		\label{E:partialtSquaredg00isolated} \\
	\partial_t^2 g_{0j} & = (g^{00})^{-1}\Big(-g^{ab} \partial_a \partial_b g_{0j}  
		-2 g^{0a} \partial_a \partial_t g_{0j} + 3 H \partial_t g_{0j} + 2 H^2 g_{0j} - 2Hg^{ab}\Gamma_{a j b} 
		+ \triangle_{0j} \Big), \label{E:partialtSquaredg0jisolated} \\
	\partial_t^2 h_{jk} & = (g^{00})^{-1}\Big(-g^{ab} \partial_a \partial_b h_{jk}  
		-2 g^{0a} \partial_a \partial_t h_{jk} + 3 H \partial_t h_{jk} + \triangle_{jk} \Big).
		\label{E:partialtSquaredhjkisolated}
\end{align}
\end{subequations}
We will now estimate $\| \partial_t^2 g_{00} \|_{H^{N-1}}, \| \partial_t^2 g_{0j} \|_{H^{N-1}},$ and 
$\| \partial_t^2 h_{jk} \|_{H^{N-1}},$ $(j,k=1,2,3).$ Using \eqref{E:partialtSquaredg00isolated}, Corollary \ref{C:SobolevTaylor}, with $v = g^{00} + 1$ and $F(v) = (v - 1)^{-1} = (g^{00})^{-1}$ in the corollary,
Proposition \ref{P:F1FkLinfinityHN}, the definition \eqref{E:totalsupnorm} of $\suptotalnorm{N},$ Sobolev embedding, 
\eqref{E:gjkupperLinfinity}, \eqref{E:partialgjkupperHNminusone}, \eqref{E:g0jupperHN}, 
and \eqref{E:triangle00HN}, it follows that

\begin{subequations}
\begin{align} \label{E:partialtSquaredg00HNminusone}
	\| \partial_t^2 g_{00} \|_{H^{N-1}} \leq C e^{- q \Omega} \suptotalnorm{N}.
\end{align}
Using similar reasoning, we also deduce that

\begin{align}
	\| \partial_t^2 g_{0j} \|_{H^{N-1}} & \leq C e^{(1 - q) \Omega} \suptotalnorm{N}, && (j=1,2,3), 
		\label{E:partialtSquaredg0jHNminusone} \\
	\| \partial_t^2 h_{jk} \|_{H^{N-1}} & \leq C e^{- q \Omega} \suptotalnorm{N}, && (j,k = 1,2,3). \label{E:partialtSquaredhjkHNminusone}
\end{align}
\end{subequations}
Note that because of inequalities \eqref{E:gabupperGammaajblowerHN} and \eqref{E:triangle0jHN},
the inequality \eqref{E:partialtSquaredg0jHNminusone} has an additional power of $e^{\Omega}$ on its
right-hand side compared to \eqref{E:partialtSquaredg00HNminusone} and \eqref{E:partialtSquaredhjkHNminusone}.

We now use Proposition \ref{P:SobolevMissingDerivativeProposition}, the definition \eqref{E:totalsupnorm} of $\suptotalnorm{N},$ Sobolev embedding, \eqref{E:partialg00upperHNminusone}, \eqref{E:partialgjkupperHNminusone},
\eqref{E:g0jupperHN}, \eqref{E:partialtSquaredg00HNminusone} to obtain the following inequalities, valid for
$|\vec{\alpha}| \leq N:$

\begin{align}
	\| [\hat{\Square}_g, \partial_{\vec{\alpha}}] (g_{00} + 1) \|_{L^2} & \leq 
		\| g^{00}\partial_{\vec{\alpha}} (\partial_t^2g_{00}) 
		- \partial_{\vec{\alpha}}(g^{00} \partial_t^2 g_{00}) \|_{L^2}
		+ \| g^{ab}\partial_{\vec{\alpha}} (\partial_{a} \partial_b g_{00}) 
		- \partial_{\vec{\alpha}}(g^{ab} \partial_{a} \partial_b g_{00}) \|_{L^2} \\
	& \ \ \ + 2\| g^{0a}\partial_{\vec{\alpha}} (\partial_t \partial_a g_{00})
		- \partial_{\vec{\alpha}}(g^{0a} \partial_t \partial_a g_{00}) \|_{L^2} \notag \\
	& \leq C \Big(\|\bard g^{00} \|_{H^{N-1}} \| \partial_t^2 g_{00} \|_{H^{N-1}}
		+ \|\bard g^{ab} \|_{H^{N-1}} \| \partial_a \partial_b g_{00} \|_{H^{N-1}}
		+ \| \bard g^{0a} \|_{H^{N-1}} \| \partial_a \partial_t g_{00} \|_{H^{N-1}} \Big) \notag \\
	& \leq C e^{-2q \Omega} \suptotalnorm{N}.  \notag
\end{align}
This completes the proof of \eqref{E:g00commutatorL2}. Inequalities \eqref{E:g0jcommutatorL2} and 
\eqref{E:hjkcommutatorL2} can be proved similarly; we omit the remaining details.
\\

\noindent \emph{Proof of \eqref{E:SquaremcommutatorL2}}:
To prove \eqref{E:SquaremcommutatorL2}, we use Proposition \ref{P:SobolevMissingDerivativeProposition}, 
the definition \eqref{E:totalsupnorm} of $\suptotalnorm{N},$ Sobolev embedding, \eqref{E:partialmjkHNminusone}, 
and \eqref{E:m0jHN} to obtain

\begin{align}
	\| [\hat{\Square}_m, \partial_{\vec{\alpha}} ] \Phi \|_{L^2} & \leq 
		\| m^{ab}\partial_{\vec{\alpha}} (\partial_{a} \partial_b \Phi) 
		- \partial_{\vec{\alpha}}(m^{ab} \partial_{a} \partial_b \Phi) \|_{L^2}
		+ 2\| m^{0a}\partial_{\vec{\alpha}} (\partial_t \partial_a \Phi)
		- \partial_{\vec{\alpha}}(m^{0a} \partial_t \partial_a \Phi) \|_{L^2} \\
	& \leq C \Big(\|\bard m^{ab} \|_{H^{N-1}} \| \partial_a \partial_b \Phi \|_{H^{N-1}}
		+ \| \bard m^{0a} \|_{H^{N-1}} \| \partial_a \partial_t \Phi \|_{H^{N-1}} \Big) \notag \\
	& \leq C e^{-(1 + \decayparameter)\Omega} \suptotalnorm{N}^2. \notag
\end{align}
This demonstrates \eqref{E:SquaremcommutatorL2}.
\\

\noindent \emph{Proof of \eqref{E:partialtSquaredPhiLinfinity}}:
To prove \eqref{E:partialtSquaredPhiLinfinity}, we first solve for $\partial_t^2 \Phi$ 
using \eqref{E:finalfluidequation}:
\begin{align} \
	\partial_t^2 \Phi = m^{jk} \partial_j \partial_k \Phi + 2m^{0j} \partial_t \partial_j \Phi 
		- \decayparameter \omega \partial_t \Phi - \triangle_{\partial \Phi}.	
\end{align}
By the definition \eqref{E:totalsupnorm} of $\suptotalnorm{N},$ Sobolev embedding, \eqref{E:partialtPhiLinfinity}, 
\eqref{E:triangleHN}, \eqref{E:mjkLinfinity}, \eqref{E:m0jC1}, we have that

\begin{align}
	\| \partial_t^2 \Phi \|_{L^{\infty}} & \leq C e^{-\decayparameter \Omega}.
\end{align}
This proves \eqref{E:partialtSquaredPhiLinfinity}.
\\

\noindent \emph{Proof of \eqref{E:partialtzjLinfinity}}:
By the definition \eqref{E:zjdef} of $z_j,$ we have that

\begin{align} \label{E:partialtzj}
	\partial_t z_j = -\omega z_j + e^{-\Omega} \frac{\partial_j \partial_t \Phi}{\partial_t \Phi}
		- e^{- \Omega} \frac{(\partial_j \Phi) (\partial_t^2 \Phi)}{(\partial_t \Phi)^2}.
\end{align}
Using the definition \eqref{E:totalsupnorm} of $\suptotalnorm{N},$ Sobolev embedding, \eqref{E:1overpartialtPhiLinfinity},
\eqref{E:partialvecalphaPhiLinfinity}, \eqref{E:partialvecalphazjLinfinity}, \eqref{E:partialtSquaredPhiLinfinity}, \eqref{E:partialtzj}, and Sobolev embedding,
we have that

\begin{align}
	\| \partial_t z_j \|_{L^{\infty}} & \leq C e^{(\decayparameter - 1)\Omega} \supfluidnorm{N}.
\end{align}
This demonstrates \eqref{E:partialtzjLinfinity}.
\\

\noindent \emph{Proofs of \eqref{E:partialttrianglemLinfinity} - \eqref{E:partialttrianglemjkLinfinity}}:
We first recall equation \eqref{E:trianglemdef}:

\begin{align} \label{E:trianglemdefagainII}
	\triangle_{(m)} & = (1 + 2s)(g^{00} + 1)(g^{00} - 1)
		+ 2(1 + 2s) e^{\Omega} g^{00} g^{0a} z_a + e^{2 \Omega}(g^{00}g^{ab} + 2s g^{0a}g^{0b})z_a z_b.
\end{align}
Differentiating \eqref{E:trianglemdefagainII} with respect to $t$ and using  
\eqref{E:g00upperplusoneLinfinity}, \eqref{E:gjkupperLinfinity}, \eqref{E:g0jupperCb1},
\eqref{E:partialtg00upperLinfinity}, \eqref{E:partialtg0jupperLinfinity}, \eqref{E:partialtgjkupperLinfinity}, \eqref{E:partialvecalphazjLinfinity}, \eqref{E:partialtzjLinfinity}, it follows that

\begin{align}
	\| \partial_t \triangle_{(m)} \|_{L^{\infty}} & \leq C e^{-q \Omega} \suptotalnorm{N}.
\end{align}
This concludes the proof of \eqref{E:partialttrianglemLinfinity}.

The proof of \eqref{E:partialttrianglemjkLinfinity} follows similarly, and we omit the details.
\\

\noindent \emph{Proof of \eqref{E:partialtm0jLinfinity}}: 
To prove \eqref{E:partialtm0jLinfinity}, we first recall equation \eqref{E:m0jdef}:

\begin{align}
	m^{0j} & = - \frac{g^{00}g^{0j}(1 + 2s) + 2e^{\Omega}\big[(s + 1)g^{0j}g^{0a} + sg^{00}g^{aj}\big]z_a   
	+ e^{2 \Omega}(g^{0j}g^{ab} + 2sg^{aj}g^{0b})z_a z_b}{(1 + 2s) + \triangle_{(m)}}.  \label{E:m0jdefagain}
\end{align}
Differentiating \eqref{E:m0jdefagain} with $\partial_t$ and using
\eqref{E:g00upperplusoneLinfinity}, \eqref{E:gjkupperLinfinity}
\eqref{E:g0jupperCb1},  \eqref{E:partialtg00upperLinfinity}, \eqref{E:partialtg0jupperLinfinity}, \eqref{E:partialtgjkupperLinfinity}, 
\eqref{E:trianglemLinfinity}, \eqref{E:partialvecalphazjLinfinity},
\eqref{E:partialtzjLinfinity}, and \eqref{E:partialttrianglemLinfinity},
it follows that $\| \partial_t m^{0j} \|_{L^{\infty}} \leq C e^{-(1+q)\Omega} \suptotalnorm{N}.$ This proves \eqref{E:partialtm0jLinfinity}.
\\

\noindent \emph{Proofs of \eqref{E:partialtmjkplusomegamjkLinfinity} - \eqref{E:partialtmjkLinfinity}}:
First, using \eqref{E:trianglemLinfinity} and \eqref{E:partialttrianglemLinfinity}, it follows that

\begin{align} \label{E:partialtFLinfinity}
	\| \partial_t \big(F(\triangle_{(m)})\big) \|_{L^{\infty}} & \leq C e^{-q \Omega} \suptotalnorm{N},
\end{align}
where $F(\triangle_{(m)})$ is defined in \eqref{E:Fproofdef}. Next, using the decompositions \eqref{E:mjkmultiplicativedecomposition} and \eqref{E:Mjkdecomposition}, together with 
simple triangle inequality estimates, we have that

\begin{align} \label{E:partialtmjkplus2omegamjkexpression}
	\| \partial_t m^{jk} + 2\omega m^{jk} \|_{L^{\infty}} & \leq 
		\| \partial_t \big(F(\triangle_{(m)})\big) \|_{L^{\infty}}  \| g^{jk} \|_{L^{\infty}}
		+ \| \partial_t \big(F(\triangle_{(m)})\big) \|_{L^{\infty}}  \| \triangle_{(m)}^{jk} \|_{L^{\infty}} \\
	& \ \ + \| F(\triangle_{(m)}) \|_{L^{\infty}}  \| \partial_t \triangle_{(m)}^{jk} \|_{L^{\infty}}
		+ \| F(\triangle_{(m)}) \|_{L^{\infty}} \| \partial_t g^{jk} + 2 \omega g^{jk} \|_{L^{\infty}} \notag \\
	& \ \ + 2 \omega \| F(\triangle_{(m)}) \|_{L^{\infty}} \| \triangle_{(m)}^{jk} \|_{L^{\infty}}. \notag
\end{align}
Finally, using \eqref{E:gjkupperLinfinity}, \eqref{E:partialtgjkupperplusomegagjkLinfinity}, \eqref{E:trianglemjkLinfinity}, \eqref{E:partialttrianglemjkLinfinity}, \eqref{E:FLinfinity}, and \eqref{E:partialtFLinfinity} to bound
the right-hand side of \eqref{E:partialtmjkplus2omegamjkexpression}, we conclude that

\begin{align}
	\| \partial_t m^{jk} + 2\omega m^{jk} \|_{L^{\infty}} & \leq C e^{-(2 + q)\Omega} \suptotalnorm{N}.
\end{align}
This proves \eqref{E:partialtmjkplusomegamjkLinfinity}. 

Inequality \eqref{E:partialtmjkLinfinity} follows
from \eqref{E:mjkLinfinity} and \eqref{E:partialtmjkplusomegamjkLinfinity}. This concludes our proof of Proposition \ref{P:Nonlinearities}.

\end{proof}

\section{The Equivalence of Sobolev and Energy Norms} \label{S:EnergyNormEquivalence}
\setcounter{equation}{0}

As is typical in the theory of hyperbolic PDEs, our global existence proof is based on showing that the energies of 
Section \ref{S:NormsandEnergies} remain finite (they happen to be uniformly bounded for $t \geq 0$ in the problem studied here). However, the boundedness of the energies does not in itself preclude the possibility of blow-up; to show that the blow-up scenario (4) from the conclusions of Theorem \ref{T:ContinuationCriterion} does \emph{not} occur, we
will control appropriate Sobolev norms of $\partial g$ and $\partial \Phi.$ In this short section, we supply the bridge between the energies and the norms. More specifically, in the following proposition, we prove that under appropriate bootstrap assumptions, the Sobolev-type norms and energies defined in Section \ref{S:NormsandEnergies} are equivalent.

\begin{proposition} \label{P:energynormomparison} \textbf{(Equivalence of Sobolev norms and energy norms)}
Let $N \geq 3$ be an integer, and assume that the bootstrap assumptions \eqref{E:metricBAeta} - \eqref{E:g0jBALinfinity} hold 
on the spacetime slab $[0,T) \times \mathbb{T}^3$ for some constant $c_1 \geq 1$ and for $\eta = \eta_{min}.$ Let $(\delta, \gamma)$ be any of the pairs of constants given in Definition \ref{D:energiesforg}, and let $C_{(\gamma)}$ be the corresponding
constant from Lemma \ref{L:buildingblockmetricenergy}. There exist constants $\epsilon''' > 0$ and $C > 0$ 
depending on $N,$ $c_1,$ $\eta_{min},$ $\gamma,$ and $\delta,$ such that if $\suptotalnorm{N} \leq \epsilon''',$ then
the following inequalities hold on the interval $[0,T)$ for the norms and energies defined in \eqref{E:mathfrakSMsupg00} - \eqref{E:totalsupnorm}, \eqref{E:mathcalEdef}, \eqref{E:g00supenergydef} - \eqref{E:gtotalsupenergydef}, \eqref{E:fluidsupenergydef}, and \eqref{E:totalenergy}:
	
	\begin{subequations}
	\begin{align}
		C^{-1}  \big\lbrace \| \partial_t v \|_{L^2} + C_{(\gamma)} \| v \|_{L^2} + e^{- \Omega} \| \bard v \|_{L^2} \big\rbrace 
			\leq \mathcal{E}_{(\gamma,\delta)}[v,\partial v] & \leq 	C  \big\lbrace \| \partial_t v \|_{L^2} 
			+  C_{(\gamma)} \| v \|_{L^2} + e^{- \Omega} \| \bard v 
			\|_{L^2} \big\rbrace, \label{E:mathcalEcomparison} \\
		C^{-1} \supgzerozeroenergy{N} \leq \supgzerozeronorm{N} & \leq C \supgzerozeroenergy{N}, 
			\label{E:mathfrakENg00mathfrakSMcomparison} \\
		C^{-1} \supgzerostarenergy{N} \leq \supgzerostarnorm{N} & \leq C \supgzerostarenergy{N}, \\
		C^{-1} \suphstarstarenergy{N} \leq \suphstarstarnorm{N} & \leq C \suphstarstarenergy{N}, 
			\label{E:mathfrakENh**mathfrakSMcomparison} \\
		C^{-1} \supgenergy{N} \leq \supgnorm{N} & \leq C \supgenergy{N}, \label{E:mathfrakENmathfrakSMcomparison} \\
		C^{-1} \supfluidenergy{N} \leq \supfluidnorm{N} & \leq C \supfluidenergy{N}, \label{E:ENSNcomparison} \\
		C^{-1} \suptotalenergy{N}  \leq \suptotalnorm{N} & \leq C \suptotalenergy{N}.  \label{E:mathcalQNQNcomparison}
	\end{align}
	\end{subequations}
	
	Analogous inequalities hold if we make the replacements $(\supgzerozeroenergy{N},
	\supgzerostarenergy{N}, \suphstarstarenergy{N}, \supgenergy{N}, \supfluidenergy{N})$ \\
	$\rightarrow (\gzerozeroenergy{N},
	\gzerostarenergy{N}, \hstarstarenergy{N}, \genergy{N}, \fluidenergy{N})$ 
	and $(\supgzerozeronorm{N},
	\supgzerostarnorm{N}, \suphstarstarnorm{N}, \supgnorm{N}, \supfluidnorm{N})$ \\
	$\rightarrow (\gzerozeronorm{N},
	\gzerostarnorm{N}, \hstarstarnorm{N}, \gnorm{N}, \fluidnorm{N}).$
\end{proposition}

\begin{proof}
	The proof of the proposition will follow quite easily from the estimates already derived. The inequalities
	in \eqref{E:mathcalEcomparison} follow from the definition \eqref{E:mathcalEdef} of $\mathcal{E}_{(\gamma,\delta)}[v, 
	\partial v],$ the definition \eqref{E:totalsupnorm} of $\suptotalnorm{N},$ \eqref{E:mathcalEfirstlowerbound},  
	and \eqref{E:gjkuppercomparetostandard}. The inequalities in \eqref{E:mathfrakENg00mathfrakSMcomparison} - 
	\eqref{E:mathfrakENh**mathfrakSMcomparison} then follow from definitions \eqref{E:mathfrakSMsupg00} - 
	\eqref{E:mathfrakSMsuph**}, definitions \eqref{E:g00supenergydef} - \eqref{E:h**supenergydef}, 
	and \eqref{E:mathcalEcomparison}. The inequalities in \eqref{E:ENSNcomparison} follow
	from definitions \eqref{E:Fluidsupnorm} and \eqref{E:fluidsupenergydef}, and from \eqref{E:gjkuppercomparetostandard} plus
	\eqref{E:mjkuppercomparedtogjk}. Finally, \eqref{E:mathfrakENmathfrakSMcomparison} and 
	\eqref{E:mathcalQNQNcomparison} follow trivially from definitions \eqref{E:totalsupgnorm}, \eqref{E:totalsupnorm}, 
	\eqref{E:gtotalsupenergydef}, and \eqref{E:totalenergy}, and from the previous inequalities.
	
\end{proof}

\section{Global Existence} \label{S:GlobalExistence}

In this section, we use the estimates derived in Sections \ref{S:BootstrapConsequences} and \ref{S:EnergyNormEquivalence}
to prove two main theorems. In the first theorem, we show that the modified system \eqref{E:finalg00equation} - \eqref{E:finalfluidequation} has global solutions for initial data near that of the background solution $(\widetilde{g},\partial \widetilde{\Phi})$ on $[0,\infty) \times \mathbb{T}^3,$ which is described in Section \ref{S:backgroundsolution}. As described in Section \ref{SS:PreservationofHarmonicGauge}, if the Einstein constraint equations and the wave coordinate condition $Q_{\mu} = 0,$ $(\mu = 0,1,2,3),$ are both satisfied along the Cauchy hypersurface $\Sigma = \lbrace x \in \mathcal{M} \ | \ t = 0 \rbrace,$ then the solution to the modified equations is also a solution to the irrotational Euler-Einstein system. The main idea of the proof is to show that the energies satisfy a system of integral inequalities that forces them (via Gronwall-type estimates) to remain uniformly small on the time interval of existence. Since Proposition \ref{P:energynormomparison} shows that the norms of the solution must also remain small, Theorem \ref{T:ContinuationCriterion} can be applied to conclude that the solution exists globally in time. In the second theorem, we provide for convenience a proof of Propositions $3$ and $4$ of \cite{hR2008}, which provide criteria for the initial data that are sufficient to ensure that the spacetime they launch is a future geodesically complete solution to irrotational Euler-Einstein system.

\subsection{Integral inequalities for the energies} \label{SS:IntegralInequalities}

In this section, we derive the system of integral inequalities that was mentioned in the previous paragraph.

\begin{proposition} \label{P:IntegralEnergyInequalities} \textbf{(Integral inequalities)}
	Let $N \geq 3$ be an integer. Assume that on the spacetime slab $[0,T) \times \mathbb{T}^3,$
	$(g_{\mu \nu}, \partial_{\mu} \Phi),$ $(\mu, \nu = 0,1,2,3),$ is a classical solution to 
	the modified system \eqref{E:finalg00equation} - \eqref{E:finalfluidequation}, and that the 
	bootstrap assumptions \eqref{E:metricBAeta} - \eqref{E:g0jBALinfinity} hold for 
	some constant $c_1 \geq 1$ and for $\eta = \eta_{min}.$ Then there exist constants $\epsilon'''' > 0$ and  
	$C > 0,$ where $C$ depends on $N, c_1,$ and $\eta_{min},$ such that if $\suptotalnorm{N}(t) \leq \epsilon''''$ on $[0,T)$ and 
	$t_1 \in [0,T),$ then the following system of integral inequalities is also satisfied for $t \in [t_1,T):$ 	
	
	\begin{subequations}
	\begin{align}
		\fluidenergy{N}^2(t) & \leq \fluidenergy{N}^2(t_1) + C \int_{t_1}^t e^{-q H \tau} \suptotalenergy{N}^2(\tau) \, d \tau, 
			\label{E:ENintegral}  \\
		\gzerozeroenergy{N}^2(t) 
			& \leq \gzerozeroenergy{N}^2(t_1) 
			+ C \int_{t_1}^t e^{-q H \tau} \suptotalenergy{N}^2(\tau) \, d \tau, \label{E:mathfrakENg00integral} \\
		\gzerostarenergy{N}^2(t)  
			& \leq \gzerostarenergy{N}^2(t_1) + \int_{t_1}^t -4qH \gzerostarenergy{N}^2(\tau) 
				+ C \suphstarstarenergy{N}(\tau)\gzerostarenergy{N}(\tau)
				+ C e^{-q H \tau} \suptotalenergy{N}(\tau) \gzerostarenergy{N}(\tau) \, d \tau, 
				\label{E:mathfrakENg0*integral} \\
		\hstarstarenergy{N}^2(t) & \leq \hstarstarenergy{N}^2(t_1)	
			+ C \int_{t_1}^t e^{-q H \tau} \suptotalenergy{N}^2(\tau) \, d \tau. \label{E:mathfrakENh**integral}
	\end{align}
	\end{subequations}

\end{proposition}

\begin{proof}
	
	We apply Corollary \ref{C:fluidenergytimederivative}, using \eqref{E:triangleHN}, \eqref{E:partialmjkHNLinfinity},  
	\eqref{E:m0jC1}, \eqref{E:SquaremcommutatorL2}, and \eqref{E:partialtmjkplusomegamjkLinfinity}
	to estimate the terms on the right-hand side of \eqref{E:fluidenergytimederivativeCauchySchwarz}, using 
	Proposition \ref{P:energynormomparison} to replace the norms with corresponding energies, and dropping the term
	$(\decayparameter - 1) \omega \sum_{|\vec{\alpha}| \leq N} \int_{\mathbb{T}^3} e^{2\decayparameter \Omega} m^{ab} (\partial_a \partial_{\vec{\alpha}} 
	\Phi) (\partial_b \partial_{\vec{\alpha}} \Phi) \,d^3 x$ on the right-hand side of \eqref{E:fluidenergytimederivativeCauchySchwarz}, which 
	by \eqref{E:mjkpositivedefinite} is non-positive for $\decayparameter < 1,$ thereby arriving at the following inequality:
	
	\begin{align} 
		\frac{d}{dt} \Big(\fluidenergy{N}^2(t) \Big) \leq C e^{-qHt}\suptotalenergy{N}^2.
	\end{align}	
Integrating from $t_1$ to $t$ gives \eqref{E:ENintegral}.

To prove \eqref{E:mathfrakENg0*integral}, we apply Corollary \ref{C:metricfirstdiferentialenergyinequality},
using \eqref{E:gabupperGammaajblowerHN}, \eqref{E:triangle0jHN}, \eqref{E:g0jcommutatorL2}, and \eqref{E:triangleEgamma0jdelta0*L1} to estimate the terms on the right-hand side of \eqref{E:underlinemathfrakEg0*firstdifferential}, using 
Proposition \ref{P:energynormomparison} to replace the norms with corresponding energies, and using 
definition \eqref{E:qdef} to deduce that $2(q-1) - \eta_{0*} \leq - 4q,$ thereby arriving at the following inequality:

\begin{align} \label{E:g0*energyfinaldifferentialinequality}
	\frac{d}{dt} \Big(\gzerostarenergy{N}^2(t) \Big)
	& \leq -4qH	\gzerostarenergy{N}^2(t) + 
	C \suphstarstarenergy{N} \gzerostarenergy{N}(t) +  C e^{-q H t} \suptotalenergy{N}(t)	
	\gzerostarenergy{N}(t).
\end{align}
Inequality \eqref{E:mathfrakENg0*integral} now follows by integrating from $t_1$ to $t.$ Inequalities \eqref{E:mathfrakENg00integral} and \eqref{E:mathfrakENh**integral} can be proved similarly; we omit the details.
\end{proof}

\begin{remark} \label{R:Dangerousterm}
	The term $C \suphstarstarenergy{N} \gzerostarenergy{N}$
	in inequality \eqref{E:mathfrakENg0*integral} arises from the $\gzerostarnorm{N} \sum_{j=1}^3 e^{(q-1) 
	\Omega} \| g^{ab}\Gamma_{a j b} \|_{H^N}$ term on the right-hand side of \eqref{E:underlinemathfrakEg0*firstdifferential}. 
	This term is dangerous in the sense that it does not contain an exponentially decaying factor, and looks like it could lead 
	to the growth of $\gzerostarenergy{N}.$ However, as we shall see in 
	the proof of Theorem \ref{T:GlobalExistence}, there is a partial decoupling in the integral inequalities in the sense that
	the $C \suphstarstarenergy{N}$ factor in the dangerous term can be controlled from inequality \eqref{E:mathfrakENh**integral}
	alone. We will then insert this information into inequality \eqref{E:mathfrakENg0*integral}, and also make use of 
	the negative term $-4qH \gzerostarenergy{N}^2$ to obtain a bound for $\gzerostarenergy{N}.$
\end{remark}

For completeness, we state the following version of Gronwall's inequality; we omit the simple proof. We will use it in Section \ref{SS:globalexistencetheorem}.

\begin{lemma} \label{L:Gronwall}
	Let $b(t) \geq 0$ be a continuous function on the interval $[t_1,T],$ and let $B(t)$ be an anti-derivative of $b(t).$
	Suppose that $A \geq 0$ and that $y(t) \geq 0$ is a continuous function satisfying the inequality
	\begin{align}
		y(t) \leq A + \int_{t_1}^t b(\tau)y(\tau) \, d \tau
	\end{align}
	for $t \in [t_1,T].$ Then for $t \in [t_1,T],$
	we have
	
	\begin{align}
		y(t) \leq A \exp \Big[B(t) - B(t_1)\Big].
	\end{align}
	
\end{lemma}

In addition, in Section \ref{SS:globalexistencetheorem}, we will apply the following integral inequality to \eqref{E:mathfrakENg0*integral} in order to estimate the energy $\supgzerostarenergy{N}(t).$ 

\begin{lemma} \label{L:integralinequality}
	Let $b(t) > 0$ be a continuous \textbf{non-decreasing} function on the interval $[0,T],$ and let $\epsilon > 0.$
	Suppose that for each $t_1 \in [0,T],$ $y(t) \geq 0$ is a continuous function satisfying the inequality
	\begin{align} \label{E:integralinequality}
		y^2(t) \leq y^2(t_1) + \int_{\tau = t_1}^t - b(\tau)y^2(\tau) + \epsilon y(\tau) \, d \tau 
	\end{align}
	for $t \in [t_1,T].$ Then for any $t_1, t \in [0,T]$ with $t_1 \leq t,$
	we have that
	
	\begin{align} \label{E:integralinequalityconclusion}
		y(t) \leq y(t_1) + \frac{\epsilon}{b(t_1)}.
	\end{align}
	
\end{lemma}

\begin{proof}
	Let $\mathcal{C}$ be the ``highest'' curve in the $(t,y)$ plane on which the integrand in \eqref{E:integralinequality} 
	vanishes; i.e. $\mathcal{C} = \lbrace (t,y) | y = \frac{\epsilon}{b(t)} \rbrace.$  Then by 
	\eqref{E:integralinequality},
	above $\mathcal{C}$ (i.e. for larger $y$ values), $y(t)$ is \emph{strictly} decreasing. Let $y(t)$ have achieve
	its maximum at $t_{max} \in [t_1,T].$ We separate the proof of \eqref{E:integralinequalityconclusion} into two cases. 
	Case i) assume that $t_{max} = t_1.$ Then $y(t) \leq y(t_{max}) = y(t_1)$ for $t \in 
	[t_1, T],$ which implies \eqref{E:integralinequalityconclusion}. Case ii) assume that $t_{max} \in (t_1,T].$ We claim that 
	$y(t_{max}) \leq \frac{\epsilon}{b(t_{max})}.$ For otherwise, the point 
	$\big(t_{max}, y(t_{max})\big)$ lies above $\mathcal{C}.$ Since $y(t)$ is then strictly decreasing in a neighborhood of 
	$t_{max},$ it follows that there are times $t_* < t_{max},$ with $t_* \in (t_1,T),$ at which $y(t_*) < y(t_{max}).$ This 
	contradicts the 
	definition of $t_{max}.$ Using also the fact that $\frac{1}{b(t)}$ is non-increasing, it follows that $y(t) 
	\leq y(t_{max}) \leq \frac{\epsilon}{b(t_{max})} \leq \frac{\epsilon}{b(t_1)};$ this concludes the proof of 
	\eqref{E:integralinequalityconclusion}.
\end{proof}

\subsection{The global existence theorem} \label{SS:globalexistencetheorem}

In this section, we state and prove our main theorem, which provides global existence criteria for the 
modified system \eqref{E:finalg00equation} - \eqref{E:finalfluidequation}.

\begin{theorem} \label{T:GlobalExistence} \textbf{(Global existence)}
Let $N \geq 3$ be an integer, and assume that $0 < \speed < \sqrt{1/3},$ where $\speed$ denotes the speed of sound. 
Let $(\mathring{g}_{\mu \nu}, \mathring{K}_{\mu \nu}, \bard \mathring{\Phi}, \mathring{\Psi}),$ $(\mu, \nu = 0,1,2,3),$ be initial data (not necessarily satisfying the wave coordinate condition or the Einstein constraints) on the manifold $\mathbb{T}^3$ for the modified Euler-Einstein system \eqref{E:finalg00equation} - \eqref{E:finalfluidequation}, and let $\suptotalnorm{N} \eqdef \supgzerozeronorm{N} + \supgzerostarnorm{N} + \suphstarstarnorm{N} + \supfluidnorm{N}$ be the norm defined in \eqref{E:totalsupnorm}. Assume that there is a positive constant $c_1 \geq 2$ such that 

\begin{align} \label{E:mathringgjklowerequivalenttostandardmetric}
	\frac{2}{c_1} \delta_{ab} X^a X^b \leq \mathring{g}_{ab}X^a X^b \leq \frac{c_1}{2} \delta_{ab} X^a X^b,
		&& \forall(X^1,X^2,X^3) \in \mathbb{R}^3.
\end{align}
Then there exists a small constant $\epsilon_0,$ with $0 < \epsilon_0 < 1,$ 
and a large constant $C_*,$ both depending on $c_1$ and $N,$ such that if $\epsilon \leq \epsilon_0$ and $\suptotalnorm{N}(0) \leq C_*^{-1} \epsilon,$ then the classical solution $(g_{\mu \nu}, \partial_{\mu} \Phi)$ provided by Theorem \ref{T:LocalExistence} exists on $[0,\infty) \times \mathbb{T}^3,$ and

\begin{align} \label{E:QNgloballessthanepsilon}
	\suptotalnorm{N}(t) & \leq \epsilon 
\end{align}
holds for all $t \geq 0.$  Furthermore, the time $T_{max}$ from the hypotheses of Theorem \ref{T:ContinuationCriterion} is infinite.  
\end{theorem}

\begin{proof}
	See Remark \ref{R:ProofsRemark} for some conventions that we use throughout this proof.
	Our proof relies upon a standard bootstrap-style argument that ultimately relies on Theorem \ref{T:ContinuationCriterion};
	i.e., we will make assumptions concerning the size of the energies and concerning $g_{\mu \nu}$
	on a spacetime slab of the form $[0,T) \times \mathbb{T}^3,$ and we will use these assumptions, together with assumptions on 
	the data, to deduce an improvement on the same slab. In effect, we will avoid the four breakdown possibilities of Theorem 
	\ref{T:ContinuationCriterion}, which will allow us to conclude global existence.
	
	Let $\mathring{\epsilon} \eqdef \suptotalnorm{N}(0),$ and let $\eta_{min}$ be the constant defined in \eqref{E:etamindef}.
	We are assuming that $\mathring{\epsilon} \leq C_*^{-1} \epsilon,$ where $C_*$ and $\epsilon$ will be adjusted throughout the 
	proof. To begin the detailed analysis, we first note that if $\epsilon$ is 
	sufficiently small and $C_*$ is sufficiently large, then by \eqref{E:g00plusonesmall} and
	\eqref{E:mjkpositivedefinite}, it follows that all of the hypotheses of
	Theorem \ref{T:LocalExistence} are satisfied. Therefore, there is a local solution $(g, \partial \Phi)$ existing on a 
	(non-trivial) \emph{maximal interval} $[0,T)$ on which the following bootstrap assumptions hold:
	
	\begin{subequations}
	\begin{align} 
		\suptotalnorm{N}(t) & \leq \epsilon, && \label{E:proofbootstrapQN} \\
		c_1^{-1} \delta_{ab}X^{a}X^{b} & \leq e^{-2 \Omega} g_{ab} X^{a}X^{b} 
			\leq c_1 \delta_{ab}X^{a}X^{b}, && \forall (X^1,X^2,X^3) \in \mathbb{R}^3. 
			\label{E:gjproofBAkvsstandardmetric} 
	\end{align}
	\end{subequations}
	We remark that if $\epsilon$ is sufficiently small, then \eqref{E:proofbootstrapQN} -  
	\eqref{E:gjproofBAkvsstandardmetric} imply that the rough bootstrap assumptions 
	\eqref{E:metricBAeta} - \eqref{E:g0jBALinfinity}, with $\eta = \eta_{min}$ in \eqref{E:metricBAeta},
	hold on $[0,T) \times \mathbb{T}^3;$ these assumptions were used the proofs of the two propositions
	in Section \ref{S:BootstrapConsequences}.  
	
	By maximal interval, we mean that 
	\begin{align}
		T \eqdef \sup \big\lbrace t \geq 0 \ | \ \mbox{The solution exists classically on} \ [0,t) \times \mathbb{T}^3, \ 
		\mbox{and} \ \eqref{E:proofbootstrapQN} - \eqref{E:gjproofBAkvsstandardmetric} \ \mbox{hold} \big\rbrace.
	\end{align}
	The remainder of this proof is dedicated to showing that 
	$T = \infty$ if $\epsilon$ is small enough and $C_*$ is large enough. Throughout the proof, we 
	assume that $\epsilon$ is small enough so that Propositions \ref{P:BoostrapConsequences}, 
	\ref{P:Nonlinearities}, \ref{P:energynormomparison}, and \ref{P:IntegralEnergyInequalities} are valid on $[0,T).$ We will 
	make repeated use of Proposition \ref{P:energynormomparison} throughout this proof without explicitly mentioning it each time.

	We now address the bootstrap assumption \eqref{E:gjproofBAkvsstandardmetric},
	with the intent of showing an improvement. First, we note that the assumption $\suptotalnorm{N} 
	\leq \epsilon$ implies that
	\begin{align} \label{E:partialthnksmall}
		\| \partial_t (e^{-2 \Omega}g_{jk}) \|_{L^{\infty}}= \| \partial_t h_{jk} \|_{L^{\infty}} & \leq C \epsilon e^{-qHt}.
	\end{align}
	We then use \eqref{E:partialthnksmall} to integrate in time from $t=0,$ concluding that
	
	\begin{align} \label{E:partialthnksmallintegrated}
		\| e^{-2 \Omega} g_{jk}(t,\cdot) - \mathring{g}_{jk}(\cdot) \|_{L^{\infty}}
		\leq C \epsilon. 
	\end{align}
	By \eqref{E:mathringgjklowerequivalenttostandardmetric} and \eqref{E:partialthnksmallintegrated}, 
	it follows that if $\epsilon$ is small enough, then on $[0,T) \times \mathbb{T}^3,$ we have that
	
	\begin{align} \label{E:gabequivalentbootstrapimprovement}
		\frac{3}{2c_1} \delta_{ab} X^a X^b \leq e^{-2 \Omega} g_{ab}X^a X^b \leq \frac{2c_1}{3} \delta_{ab} X^a X^b,
			&& \forall (X^1,X^2,X^3) \in \mathbb{R}^3.
	\end{align}
	Inequality \eqref{E:gabequivalentbootstrapimprovement}
	shows that if $\epsilon$ is small enough, then the bootstrap assumption \eqref{E:gjproofBAkvsstandardmetric} can be 
	\emph{strictly improved} on the interval $[0,T).$

	To complete our proof of the theorem, we will show that if $\epsilon$ is small enough and $C_*$ is large enough, then the 
	bootstrap assumption \eqref{E:proofbootstrapQN} can be improved by replacing $\epsilon$ with $\epsilon/2;$ the primary tool 
	for deducing an improvement is of course Proposition \ref{P:IntegralEnergyInequalities}. To begin our proof of an improvement 
	of \eqref{E:proofbootstrapQN}, we use a very non-optimal application of Proposition 
	\ref{P:IntegralEnergyInequalities} with $t_1 = 0,$ deducing that on $[0,T),$ we have that
	
	\begin{align} \label{E:mathcalQNGronwallready}
		\suptotalenergy{N}(t) & \leq \suptotalenergy{N}(0) + \int_{\tau=0}^{t} c \suptotalenergy{N}(\tau) \, d\tau.
	\end{align}
	Applying Lemma \ref{L:Gronwall} (Gronwall's inequality) to \eqref{E:mathcalQNGronwallready}, using
	$\suptotalnorm{N}(0) = \mathring{\epsilon},$ and using Proposition \ref{P:energynormomparison}, we 
	conclude that the following preliminary inequalities hold on $[0,T):$
	
	\begin{align} 
		\suptotalnorm{N}(t) & \leq C \mathring{\epsilon} e^{ct}, \label{E:QNweakenergyinequality} \\
		\suptotalenergy{N}(t) & \leq C \mathring{\epsilon} e^{ct}. \label{E:mathcalQNweakenergyinequality}
	\end{align}
	
	\begin{remark} \label{R:weaklifespan}
		By modifying the argument in the last paragraph of this proof, Theorem \ref{T:ContinuationCriterion} and inequality 
		\eqref{E:QNweakenergyinequality} can be used to deduce that the time 
		of existence is at least of order $c^{-1} |\mbox{ln}(C \mathring{\epsilon})|,$ if $\mathring{\epsilon}$ is sufficiently 
		small.
	\end{remark}

	We now fix a time $t_1 \in [0,T);$ $t_1$ will be adjusted at the end of the proof. Roughly speaking, it will play the role of
	a time that is large enough so that the exponentially damped terms on the right-hand sides of the inequalities of
	Proposition \ref{P:IntegralEnergyInequalities} are of size $\ll \epsilon.$ To estimate $\supfluidenergy{N}(t)$ on $[t_1,T),$ 
	we simply use \eqref{E:proofbootstrapQN} and \eqref{E:mathcalQNweakenergyinequality} to estimate the two terms on the 
	right-hand side of \eqref{E:ENintegral}:
	
	\begin{align} \label{E:ENdifferentialtwoterms}
		\supfluidenergy{N}^2(t) \leq \supfluidenergy{N}^2(t_1) + C \epsilon^2 \int_{\tau=t_1}^{t} e^{-q H \tau}  \, d\tau
			\leq C \lbrace \mathring{\epsilon} e^{ct_1} + \epsilon e^{-qHt_1/2} \rbrace^2. 
	\end{align}
	Using \eqref{E:mathcalQNweakenergyinequality} again to estimate $\supfluidenergy{N}(t)$ on $[0,t_1],$ 
	we thus conclude that the following inequality is valid on $[0,T):$
	
	\begin{align} \label{E:ENGronwall}
		\supfluidenergy{N}(t) & \leq C \lbrace \mathring{\epsilon} e^{ct_1} + \epsilon e^{-qHt_1/2} \rbrace.
	\end{align}
	
	Applying similar reasoning to inequalities \eqref{E:mathfrakENg00integral} and
	\eqref{E:mathfrakENh**integral}, we also have the following inequalities on $[0,T):$

	\begin{align} 
		\supgzerozeroenergy{N}(t) & \leq C \lbrace \mathring{\epsilon} e^{ct_1} + \epsilon e^{-qHt_1 /2} \rbrace, 
			\label{E:mathfrakEg00Gronwall} \\
		\suphstarstarenergy{N}(t) & \leq C \lbrace \mathring{\epsilon} e^{ct_1} + \epsilon e^{-qHt_1/2} \rbrace.
			\label{E:mathfrakEh**Gronwall}
	\end{align}
	
	To estimate $\supgzerostarenergy{N}(t),$ we use \eqref{E:mathfrakENg0*integral}, the bootstrap assumption 
	\eqref{E:proofbootstrapQN}, and \eqref{E:mathfrakEh**Gronwall} to arrive at the 
	following inequality valid for $t \in [t_1,T)$:
	
	\begin{align} \label{E:mathfrakE0*Gronwallready}
		\gzerostarenergy{N}^2(t) & \leq 
			\gzerostarenergy{N}^2(t_1) 
			+ \int_{\tau = t_1}^t -4qH \gzerostarenergy{N}^2(\tau) + C \lbrace \mathring{\epsilon} e^{ct_1} 
			+ \epsilon e^{-qHt_1 /2} \rbrace \gzerostarenergy{N}(\tau)  \, d \tau. 
	\end{align}
	Applying Lemma \ref{L:integralinequality} to \eqref{E:mathfrakE0*Gronwallready}, with 	
	$y(\tau) = \gzerostarenergy{N}(\tau)$ and 
	$b(\tau) = 4 q H$ in the lemma, and also using
	\eqref{E:mathcalQNweakenergyinequality}, we conclude that the following inequality holds on $[0,T):$
	
	\begin{align} \label{E:mathfrakE0*Gronwall}
		\supgzerostarenergy{N}(t) & \leq C \lbrace \mathring{\epsilon} e^{ct_1} + \epsilon e^{-qHt_1/2} \rbrace. 
	\end{align}

	Adding  \eqref{E:ENGronwall},
	\eqref{E:mathfrakEg00Gronwall}, \eqref{E:mathfrakEh**Gronwall}, 
	and \eqref{E:mathfrakE0*Gronwall}, referring to definition 
	\eqref{E:totalsupnorm}, and using Proposition \ref{P:energynormomparison}, it follows that on $[0,T),$
	we have
	
	\begin{align} \label{E:QNGronwall}
		\suptotalnorm{N}(t) & \leq C \lbrace \mathring{\epsilon} e^{ct_1} + \epsilon e^{-qHt_1/2} \rbrace.
	\end{align}
	We now choose $t_1$ such that $C e^{-qHt_1/2} \leq \frac{1}{4},$ and $\mathring{\epsilon}$ such that
	$C \mathring{\epsilon} e^{c t_1} \leq \frac{1}{4} \epsilon,$ which implies that the following inequality holds
	on $[0,T):$ 
	
	\begin{align}
		\suptotalnorm{N}(t) & \leq \frac{1}{2} \epsilon. \label{E:QNbootstrapimprovement}
	\end{align}
	We remark that in order to guarantee that the solution exists long enough (i.e. that $T$ is large enough) 
	so that $t_1 \in [0,T),$ we may have to further shrink $\epsilon_0;$ see Remark \ref{R:weaklifespan}. We also
	remark that the constant $C_*$ from the conclusions of the theorem can be chosen to be $4 C e^{c t_1},$
	where $C$ is from the right-hand side of \eqref{E:QNGronwall}.
	
	We now claim that $T = \infty.$ We argue by contradiction, assuming that $T < \infty.$ Then by
	combining \eqref{E:gabequivalentbootstrapimprovement} and \eqref{E:QNbootstrapimprovement},
	it follows that none of the four existence-breakdown scenarios stated in the conclusions of Theorem 
	\ref{T:ContinuationCriterion} occur: $(1)$ is ruled out by the Sobolev embedding result 
	$\|g_{00} + 1 \|_{L^{\infty}} \leq C e^{-q \Omega} \suptotalnorm{N};$ $(2)$ is ruled out by 
	\eqref{E:gabequivalentbootstrapimprovement};
	$(3)$ is ruled out by \eqref{E:mjkpositivedefinite}; and $(4)$ is ruled out by the Sobolev embedding results 
	$\|g_{00} + 1 \|_{C_b^1} + \|\partial_t g_{00} \|_{L^{\infty}} \leq C e^{-q \Omega} \suptotalnorm{N},$ 
	$ \sum_{j=1}^3 \big( \|g_{0j}\|_{C_b^1} + \|\partial_t g_{0j} \|_{L^{\infty}}\big) \leq 
	C e^{(1-q)\Omega}\suptotalnorm{N},$ $\sum_{j,k=1}^3 \| \bard g_{jk} \|_{L^{\infty}} \leq C e^{2 \Omega}\suptotalnorm{N},$
	$\|\bard \Phi \|_{L^{\infty}} \leq C e^{(1 - \decayparameter)\Omega} \suptotalnorm{N},$ together with 
	inequalities \eqref{E:gjklowerLinfinity}, \eqref{E:partialtgjklowerC1}, \eqref{E:partialtPhiLinfinity},
	\eqref{E:m0jC1}, \eqref{E:mjkLinfinity}, \eqref{E:partialmjkHNLinfinity},
	\eqref{E:partialtm0jLinfinity}, and \eqref{E:partialtmjkLinfinity}. Also using the continuity of $\suptotalnorm{N}(t),$ it 
	thus follows from Theorem \ref{T:ContinuationCriterion} that there exists a $\delta > 0$ such that the solution can be 
	extended to the interval $[0, T + \delta)$ on which the bootstrap assumptions \eqref{E:proofbootstrapQN} - 
	\eqref{E:gjproofBAkvsstandardmetric} hold. This contradicts the definition of $T,$ which shows that $T$ is not finite.

\end{proof}

\subsection{On the breakdown of the proof for \texorpdfstring{$\speed \geq 1/3$ (i.e. $s \leq 1, \decayparameter \geq 1$)}{}} \label{SS:Breakdown}
In this short section, we give a brief example of how our proof breaks down when $\speed \geq 1/3.$ If $\decayparameter \geq 1,$ we cannot use our previous reasoning to bound the following term, which is the last term on the right-hand side of 
\eqref{E:fluidenergytimederivative}:
\begin{align} \label{E:abadterm}
	\frac{1}{2} \sum_{|\vec{\alpha}| \leq N} \int_{\mathbb{T}^3} 
	e^{2\decayparameter \Omega}(\partial_t m^{ab} + 2 \decayparameter \omega m^{ab})(\partial_a \partial_{\vec{\alpha}} \Phi)
	(\partial_b \partial_{\vec{\alpha}} \Phi) \,d^3 x.
\end{align}

Previously, we had split this term into two pieces, one of which is 
$(\decayparameter - 1)\omega \sum_{|\vec{\alpha}| \leq N} \int_{\mathbb{T}^3} m^{ab}(\partial_a \partial_{\vec{\alpha}} \Phi) (\partial_b \partial_{\vec{\alpha}} \Phi) \,d^3 x,$ which, as is explained in our proof of Proposition \ref{P:IntegralEnergyInequalities}, could be 
discarded from the energy inequality \eqref{E:ENintegral} because it is non-positive when $\decayparameter < 1.$ Obviously, we can no longer discard this
term when $\decayparameter > 1.$ Furthermore, even in the case $\decayparameter = 1,$ inequality \eqref{E:partialtmjkplusomegamjkLinfinity} must be replaced with $\| \partial_t m^{jk} + 2\omega m^{jk} \|_{L^{\infty}} \leq C \lbrace e^{-2\Omega} \suptotalnorm{N}^2 + \mbox{Positive terms} \rbrace.$ Ultimately, this fact can be traced to the fact that the $L^{\infty}$ estimate for $z_j$ from \eqref{E:partialvecalphazjLinfinity} must be replaced with $\| z_j \|_{L^{\infty}} \leq C \suptotalnorm{N}.$ 

Using Proposition \ref{P:energynormomparison}, it can be shown that the net result in both the case $\decayparameter = 1$ and the case $\decayparameter > 1$ is that the term \eqref{E:abadterm} leads to (as in our proof of Proposition \ref{P:IntegralEnergyInequalities}) an integral inequality of the form 
\begin{align} \label{E:breakdown}
	\supfluidenergy{N}^2(t) \leq \supfluidenergy{N}^2(t_1) + \mbox{positive terms} +  C \epsilon \int_{\tau = t_1}^t 
	\supfluidenergy{N}^2(\tau)\, d \tau.
\end{align}
Inequality \eqref{E:breakdown} allows for the possibility of growth of $\supfluidenergy{N}(t);$
i.e., unlike in the cases $0 < \speed^2 < 1/3,$ there is no $e^{-qH\tau}$ factor with $q>0$ in the integrand. Therefore, 
this inequality provides no means to improve the bootstrap assumption $\suptotalnorm{N}(t) \leq \epsilon.$

\subsection{Future geodesic completeness}

In this section, we prove our second main theorem, which provides criteria for the initial data under which the global solutions
provided by Theorem \ref{T:GlobalExistence} are future geodesically complete. The theorem and its proof are based on Propositions 3 and 4 of \cite{hR2008}.  

\begin{theorem} \label{T:FutureComplete} \textbf{(Future geodesic completeness)}
Let $N \geq 3$ be an integer, and assume that $0 < \speed < \sqrt{1/3},$ where $\speed$ denotes the speed of sound.
Let $([0,\infty) \times \mathbb{T}^3,\widetilde{g}_{\mu \nu}, \partial_{\mu} \widetilde{\Phi}),$ $(\mu, \nu = 0,1,2,3),$ be one of the FLRW background solutions described in Section \ref{S:backgroundsolution}. Note that the background is a solution to both the un-modified system \eqref{E:EulerEinsteinIrrotational} - \eqref{E:irrotationalfluid} and the modified Euler-Einstein system \eqref{E:finalg00equation} - \eqref{E:finalfluidequation}. Let $(\mathring{g}_{\mu \nu}, \mathring{K}_{\mu \nu}, \bard \mathring{\Phi}, \mathring{\Psi}),$ $(\mu, \nu = 0,1,2,3),$ be initial data for the modified Euler-Einstein system on the manifold $\mathbb{T}^3$ that are constructed from initial data $(\mathbb{T}^3, \bar{g}_{jk}, \bar{K}_{jk}, \bard \mathring{\Phi}, \mathring{\Psi}),$ $(j,k=1,2,3),$ for the un-modified system \eqref{E:EulerEinsteinIrrotational} - \eqref{E:irrotationalfluid} (and which are viewed as a perturbation of the data for the FLRW solution) as described in Section \ref{SSS:InitialDataOriginalSystem}. Let $\suptotalnorm{N}$ be the norm defined in \eqref{E:totalsupnorm}, and let $\supgnorm{N},$ a norm for the metric components, be the quantity defined in \eqref{E:totalsupgnorm}. Assume that the data 
for the modified system are near that of the FLRW solution
in the sense that $\suptotalnorm{N}(0) \leq C_*^{-1} \epsilon_0,$
where $\epsilon_0$ and $C_*$ are the constants from the conclusion of Theorem \ref{T:GlobalExistence}
(note that $\suptotalnorm{N} \equiv 0$ for the FLRW solution). Also assume that the perturbed data satisfy the inequality \eqref{E:mathringgjklowerequivalenttostandardmetric}, so that all of the hypotheses of Theorem \ref{T:GlobalExistence} are satisfied. Let $(\mathcal{M} \eqdef [0,\infty) \times \mathbb{T}^3,g_{\mu \nu}, \partial_{\mu} \Phi)$ be the global (to the future) solution (which by Proposition \ref{P:Preservationofgauge} is a solution to both the un-modified system and the modified system) guaranteed by Theorem \ref{T:GlobalExistence}, and let $\gamma$ be a future-directed causal curve in $\mathcal{M}$ with domain $[s_0, s_{max})$ such that $\gamma(s_0) = 0.$ Let $\gamma^{\mu}$ denote the coordinates of this curve in the universal covering space of the spacetime (i.e. $[0,\infty) \times \mathbb{R}^3).$ Then there exist constants $C > 0$ and $\epsilon_1,$ where $0 < \epsilon_1 \leq \epsilon_0,$ such that if $\supgnorm{N}(t) \leq \epsilon_1$ for all $t \geq 0,$ then $\dot{\gamma}^0 > 0,$ and furthermore, the length of the spatial part of the curve as measured by the metric $\mathring{g}_{jk} = \bar{g}_{jk}$ satisfies 

\begin{align} \label{E:lengthestimate}
	\int_{s_0}^{s_{max}} \sqrt{\bar{g}_{ab}(\pi \circ \gamma) \dot{\gamma}^a \dot{\gamma}^b} \, ds \leq C,
\end{align}
where $\pi$ denotes projection onto spatial indices; i.e., $\pi^j \circ{\gamma} \eqdef \gamma^j.$ We remark that the constants
$C$ and $\epsilon_1$ \textbf{can be chosen to be independent of} $\gamma.$ Additionally, if $\gamma$ is future-inextendible, then $\gamma^{0}(s) \uparrow \infty$ as $s \uparrow s_{max}.$ Finally, the spacetime $(\mathcal{M},g_{\mu \nu})$ is future causally geodesically complete. 

\end{theorem}

\begin{remark}
	Our smallness assumption concerns $\supgnorm{N}$ because this is the quantity 
	that controls $g_{\mu \nu};$ i.e., the influence of the fluid plays no direct role in the conclusions of this theorem.
\end{remark}

\begin{proof}
	See Remark \ref{R:ProofsRemark} for some conventions that we use throughout this proof,
	but as noted above in the statement of the theorem, the smallness assumption on $\suptotalnorm{N}$ mentioned 
	in the remark can be replaced with a smallness assumption on $\supgnorm{N}.$ We also note that $\gamma^0(s)$ can be 
	identified with the spacetime coordinate $t.$ 
	
	We begin by showing that $\dot{\gamma}^0 > 0,$ where $\cdot \eqdef \frac{d}{ds}.$ Since $\gamma$ is causal and 
	future-directed, and since $\partial_t$ is future-directed and timelike, we have that
	\begin{align}
		g_{\alpha \beta} \dot{\gamma}^{\alpha} \dot{\gamma}^{\beta} & \leq 0, \label{E:causal} \\
		g_{0 \alpha} \dot{\gamma}^{\alpha} & < 0. \label{E:futuredirected}
	\end{align}
	
	Our first goal is to prove that if $\epsilon_1$ is small enough, then the following estimates hold:
	\begin{align} 
		g_{ab} \dot{\gamma}^{a} \dot{\gamma}^{b} & \leq C_{\eta} (\dot{\gamma}^0)^2, \label{E:gabdotgammaadotgammabvsdotgamma0} \\
		\delta_{ab} \dot{\gamma}^{a} \dot{\gamma}^{b} & \leq C e^{-2 \Omega} (\dot{\gamma}^0)^2, 
			\label{E:deltaabdotgammaadotgammabvsdotgamma0}
	\end{align}
	where $\eta \eqdef C \epsilon_1$ is from the right-hand side of 
	\eqref{E:metricBAeta}. To this end, we use the Cauchy-Schwarz inequality, \eqref{E:g0jBALinfinity}, 
	the Cauchy-Schwarz inequality again, and \eqref{E:gjkBAvsstandardmetric} to conclude that

	\begin{align} \label{E:g0adotgamma0firstestimate}
		|2g_{0a} \dot{\gamma}^0 \dot{\gamma}^{a}| & \leq \eta^{1/2} (\dot{\gamma}^0)^2 + 
			\eta^{-1/2} |g_{0a} \dot{\gamma}^a|^2 \leq \eta^{1/2} (\dot{\gamma}^0)^2 
			+ c_1^{-1} \eta^{1/2} e^{2(1-q)\Omega} \delta_{ab} \dot{\gamma}^a \dot{\gamma}^b \\
		& \leq \eta^{1/2} (\dot{\gamma}^0)^2 
			+ \eta^{1/2} e^{-2q \Omega} g_{ab} \dot{\gamma}^a \dot{\gamma}^b. \notag
	\end{align}
	Combining \eqref{E:causal}, \eqref{E:g0adotgamma0firstestimate}, and the estimate
	$- g_{00} < 1 + \eta,$ which follows from \eqref{E:metricBAeta}, we have that
	\begin{align} \label{E:gabdotgammaadotgammabvsdotgamma0firstestimate}
		g_{ab} \dot{\gamma}^{a} \dot{\gamma}^{b} \leq - g_{00} (\dot{\gamma}^0)^2 + |2g_{0a} \dot{\gamma}^0 \dot{\gamma}^{a}|
			\leq  (1 + \eta) (\dot{\gamma}^0)^2 + \eta^{1/2} (\dot{\gamma}^0)^2 
			+ \eta^{1/2} e^{-2q \Omega} g_{ab} \dot{\gamma}^a \dot{\gamma}^b.
	\end{align}
	It now easily follows from \eqref{E:gabdotgammaadotgammabvsdotgamma0firstestimate} that 
	if $\epsilon_1$ is small enough, then there exists a constant $C_{\eta} > 0$ such that 
	\eqref{E:gabdotgammaadotgammabvsdotgamma0} holds. Inequality \eqref{E:deltaabdotgammaadotgammabvsdotgamma0}
	then follows from \eqref{E:gjproofBAkvsstandardmetric} and \eqref{E:gabdotgammaadotgammabvsdotgamma0}.
	
	We now claim that $|\dot{\gamma}^{0}| > 0.$ For if $|\dot{\gamma}^{0}| = 0,$ then inequality 
	\eqref{E:deltaabdotgammaadotgammabvsdotgamma0} shows that $\sum_{\alpha = 0}^3 |\dot{\gamma}^{\alpha}| = 0,$
	which contradicts \eqref{E:futuredirected}. We may therefore divide each side of \eqref{E:g0adotgamma0firstestimate} by 
	$|\dot{\gamma}^0| > 0$ and use \eqref{E:gabdotgammaadotgammabvsdotgamma0} to arrive at the following inequality:
	
	\begin{align} \label{E:g0agammaainequality}
		|g_{0a} \dot{\gamma}^a| \leq \eta^{1/2}(1 + C_{\eta} e^{-2q \Omega}) |\dot{\gamma}^0|.
	\end{align}
	Using \eqref{E:g0agammaainequality} and the estimate $|g_{00}| > 1 - \eta,$ we conclude that if $\epsilon_1$ is
	sufficiently small, then
	
	\begin{align} \label{E:sgnequality}
		\mbox{sgn}(g_{0 \alpha} \dot{\gamma}^{\alpha}) = \mbox{sgn} (g_{00} \dot{\gamma}^0),
	\end{align}
	where $\mbox{sgn}(y) = 1$ if $y > 0,$ and $\mbox{sgn}(y) = -1$ if $y < 0.$ Since \eqref{E:futuredirected} implies that the 
	left-hand side of \eqref{E:sgnequality} is negative, and since $g_{00} < 0,$ we conclude that $\dot{\gamma}^0 > 0.$ 
	
	We now show the estimate \eqref{E:lengthestimate}. First, by the hypothesis 
	\eqref{E:mathringgjklowerequivalenttostandardmetric}, the fact that $\bar{g}_{jk} = \mathring{g}_{jk},$
	and the estimate \eqref{E:gjkBAvsstandardmetric}, it follows that
	
	\begin{align} \label{E:bargjkestimate}
		\bar{g}_{ab} \dot{\gamma}^a \dot{\gamma}^b \leq \frac{c_1}{2} \delta_{ab} \dot{\gamma}^a \dot{\gamma}^b 
			\leq \frac{c_1^2}{2} e^{-2 \Omega} g_{ab} \dot{\gamma}^a \dot{\gamma}^b.
	\end{align}
	Integrating the Square root of each side of \eqref{E:bargjkestimate} from $s_0$ to $s_{max},$ 
	recalling that $e^{-\Omega(\gamma^{0}(s))} \leq C e^{-H \gamma^{0}(s)},$
	and using \eqref{E:gabdotgammaadotgammabvsdotgamma0}, we have that
	
	\begin{align}
		\int_{s_0}^{s_{max}} \sqrt{\bar{g}_{ab}({\pi \circ \gamma}) \dot{\gamma}^a \dot{\gamma}^b} \, ds
			\leq \int_{s_0}^{s_{max}} C e^{- H \gamma^0} \dot{\gamma}^0 \, ds 
			= - \frac{C}{H} \int_{s_0}^{s_{max}} \frac{d}{ds} e^{- H \gamma^0} \, ds \leq C,
	\end{align}
	which proves \eqref{E:lengthestimate}.
	
	We now show that $\gamma^0$ converges to infinity as $s \uparrow s_{max}$ if $\gamma$ is a causal future-inextendible
	curve. Since $\dot{\gamma}^0 > 0,$ it follows that either $\gamma^0$ converges to infinity as $s \uparrow s_{max},$ which 
	is the desired result, or that $\gamma^0(s)$ converges to a finite number. In the latter case, by 
	\eqref{E:gabdotgammaadotgammabvsdotgamma0}, we also conclude 
	 the $\gamma^{j}(s)$ converge to finite numbers as 
	$s \uparrow s_{max},$ which contradicts the definition of future-inextendibility. 
	
	To show that the spacetime is future causally geodesically complete, we consider a future-directed causal
	geodesic $\gamma.$ We recall that the geodesic equations are $\ddot{\gamma}^{\mu} + \Gamma_{\alpha \ \beta}^{\ \mu} 
	\dot{\gamma}^{\alpha} \dot{\gamma}^{\beta} = 0,$ which in the case of $\mu=0$ reads
	
	\begin{align} \label{E:gamma0geodesicequation}
		\ddot{\gamma}^{0} + \Gamma_{\alpha \ \beta}^{\ 0} \dot{\gamma}^{\alpha} \dot{\gamma}^{\beta} = 0.
	\end{align}
	To analyze equation \eqref{E:gamma0geodesicequation}, we will use the following estimates for the Christoffel symbols:
	
	\begin{align}
		|\Gamma_{0 \ 0}^{\ 0}| & \leq C \epsilon_1 e^{-qHt}, \label{E:Gamma0down0up0downLinfinity} \\
		|\Gamma_{0 \ j}^{\ 0}| &  \leq C \epsilon_1 e^{(1-q)Ht}, \label{E:Gamma0down0upjdownLinfinity} \\
		|\Gamma_{j \ k}^{\ 0} - \omega g_{jk}| & \leq C \epsilon_1 e^{(2-q)Ht}. \label{E:Gammajdown0upkdownLinfinity}
	\end{align}
	We will prove the estimate \eqref{E:Gammajdown0upkdownLinfinity}; the estimates \eqref{E:Gamma0down0up0downLinfinity} - 
	\eqref{E:Gamma0down0upjdownLinfinity} can be shown similarly. To begin, we use the definition \eqref{E:EMBIChristoffeldef}
	of $\Gamma_{j \ k}^{\ 0}$ and the triangle inequality to obtain the following estimate:
		
		\begin{align} \label{E:Gammajdown0upkdowntriangleinequality}
			|\Gamma_{j \ k}^{\ 0} - \omega g_{jk}| & \leq 
				\frac{1}{2} |g^{00}||\partial_{j} g_{0k}
			+ \partial_{k} g_{0j}| + \frac{1}{2} |g^{00} + 1| |\partial_t g_{jk}|
			+ \frac{1}{2} |\partial_t g_{jk} - 2\omega g_{jk}|
			+ \frac{1}{2} |g^{0 a}||\partial_{j} g_{a k} 
			+ \partial_{k} g_{j a} - \partial_{a} g_{j k}|.
		\end{align}
		Inequality \eqref{E:Gammajdown0upkdownLinfinity} now follows easily from \eqref{E:Gammajdown0upkdowntriangleinequality},
		the assumption $\supgnorm{N} < \epsilon_1,$ Sobolev embedding, \eqref{E:g00upperplusoneLinfinity}, \eqref{E:g0jupperCb1}, 
		\eqref{E:partialtgjkminusomegagjklowerLinfinity}, and \eqref{E:partialtgjklowerC1}.

	We now use \eqref{E:gabdotgammaadotgammabvsdotgamma0}, \eqref{E:deltaabdotgammaadotgammabvsdotgamma0}, and 
	\eqref{E:Gamma0down0up0downLinfinity} - \eqref{E:Gammajdown0upkdownLinfinity}
	to arrive at the following inequality:
	
	\begin{align} \label{E:Christoffelgeodesicinequality}
		|\Gamma_{0 \ 0}^{\ 0}|(\dot{\gamma}^0)^2 + 2 |\Gamma_{0 \ a}^{\ 0}\dot{\gamma}^0 \dot{\gamma}^a|
			+ |\Gamma_{a \ b}^{\ 0} - \omega g_{ab}| \dot{\gamma}^a \dot{\gamma}^b 
			& \leq C \epsilon_1 e^{-qH \gamma^0(s)} (\dot{\gamma}^0)^2,
	\end{align}
	where it is understood that both sides of \eqref{E:Christoffelgeodesicinequality} 
	are evaluated along the curve $\gamma(s).$ Using \eqref{E:gamma0geodesicequation},
	\eqref{E:Christoffelgeodesicinequality}, and the positive-definiteness of the $3 \times 3$ matrix $g_{jk},$ it 
	follows that
	
	\begin{align} \label{E:gamma0ddotestimate}
		\ddot{\gamma}^{0} & = - \Gamma_{0 \ 0}^{\ 0}(\dot{\gamma}^0)^2
			- 2 \Gamma_{0 \ a}^{\ 0}\dot{\gamma}^0 \dot{\gamma}^a
			- \Gamma_{a \ b}^{\ 0} \dot{\gamma}^a \dot{\gamma}^b
			\leq |\Gamma_{0 \ 0}^{\ 0}|(\dot{\gamma}^0)^2 + 2 |\Gamma_{0 \ a}^{\ 0}\dot{\gamma}^0 \dot{\gamma}^a|
			+ |\Gamma_{a \ b}^{\ 0} - \omega g_{ab}| \dot{\gamma}^a \dot{\gamma}^b \\
		& \leq C \epsilon_1 e^{-qH \gamma^0(s)} (\dot{\gamma}^0)^2. \notag 
	\end{align}
	Since we have already shown that $\dot{\gamma}^0$ is positive if $\epsilon_1$ is small enough, we may divide inequality 
	\eqref{E:gamma0ddotestimate} by it and integrate:
	\begin{align}
		\mbox{ln} \bigg( \frac{\dot{\gamma}^0(s)}{\dot{\gamma}^0(s_0)} \bigg) 
			& = \int_{s_0}^s \frac{\ddot{\gamma}^0(s')}{\dot{\gamma}^0(s')}  \, ds' 
			\leq C \epsilon_1 \int_{s_0}^s e^{-qH \gamma^0(s')} \dot{\gamma}^0(s')  \, ds'
			= -\frac{C}{qH} \epsilon_1 \int_{s_0}^s \frac{d}{ds'}e^{-qH\gamma^0(s')} \, ds' \\
			& \leq \frac{C}{qH} e^{-qH \gamma^0(s)}. \notag
	\end{align}
	We therefore conclude that $\dot{\gamma}^0$ is bounded from above:
	\begin{align} \label{E:dotgamma0bounded}
		\dot{\gamma}^0 \leq C.
	\end{align}
	Integrating \eqref{E:dotgamma0bounded} from $s_0$ to $s,$ we have that
	\begin{align} \label{E:gamma0sminusgamma0s0boundfrombelow}
		\gamma^0(s) - \gamma^0(s_0) = \int_{s_0}^s \dot{\gamma}^0(s') \, ds' \leq C |s - s_0|.
	\end{align}
	Finally, since we have already shown that 
	$\gamma^0 \uparrow \infty$ as $s \uparrow s_{max},$ it follows from \eqref{E:gamma0sminusgamma0s0boundfrombelow} that $s_{max} = \infty.$
	
\end{proof}

\section{Asymptotics} \label{S:Asymptotics}
\setcounter{equation}{0} 	

In this section, we strengthen the conclusions of Theorem \ref{T:GlobalExistence} by showing that
$g_{\mu \nu}, g^{\mu \nu}, \partial_{\mu} \Phi,$ $(\mu, \nu = 0,1,2,3),$ and various coordinate derivatives of these quantities converge as $t \rightarrow \infty.$ Because our strategy is to integrate bounds for time derivatives, we will
lose at least one order of differentiability in our convergence estimates. Furthermore, we note that although our bootstrap assumptions were sufficient to close the global existence argument, they are far from optimal from the point of view of decay rates. Thus, at the cost of a few more derivatives, we will also revisit the modified equations and derive improved rates of decay compared to what can be directly concluded from the estimate $\suptotalnorm{N} \leq \epsilon.$ In particular, the Counting Principle of Section \ref{S:BootstrapConsequences} is not precise enough to detect these refinements. These results should be viewed as an initial investigation of the asymptotics; it is clear that more information could be extracted at the expense of more work. This theorem is analogous to \cite[Proposition 2]{hR2008}.

\begin{theorem} \label{T:Asymptotics} \textbf{(Asymptotics)}
 Assume that initial data $(\mathring{g}_{\mu \nu}, \mathring{K}_{\mu \nu}, \bard \mathring{\Phi}, \mathring{\Psi}),$
 $(\mu, \nu = 0,1,2,3),$ for the modified system \eqref{E:finalg00equation} - \eqref{E:finalfluidequation} satisfy the 
 assumptions of Theorem \ref{T:GlobalExistence}, including the smallness assumption 
 ${\suptotalnorm{N} \leq C_*^{-1} \epsilon},$ where 
 $0 \leq \epsilon \leq \epsilon_0.$ 
 Let $\mathring{g}^{\mu \nu}$ denote the inverse of $\mathring{g}_{\mu \nu}.$
 Assume in addition that $N \geq 5,$ and let $(g_{\mu \nu}, \partial_{\mu} \Phi)$ 
 be the global solution launched by the data. Then there exists a constant $\epsilon_2$ satisfying
 $0 < \epsilon_2 \leq \epsilon_0$ such that if $\epsilon < \epsilon_2,$ then there exists a Riemann metric $g_{jk}^{(\infty)},$ 
 ($j,k = 1,2,3$), with corresponding Christoffel symbols $\Gamma_{ijk}^{(\infty)},$ $(i,j,k = 1,2,3),$
 and inverse $g_{(\infty)}^{jk}$ on $\mathbb{T}^3,$ and (time independent) functions $\bard \Phi_{(\infty)}, \Psi_{(\infty)}$ 
 on $\mathbb{T}^3$ such that $g_{jk}^{(\infty)} - \mathring{g}_{jk} \in H^{N},$ 
 $g_{(\infty)}^{jk} - \mathring{g}^{jk} \in H^N,$ $\bard \Phi_{(\infty)} \in H^{N-1},$ and
 $\Psi_{(\infty)} - \bar{\Psi} \in H^{N-1}$ where $\bar{\Psi}$ is defined 
 in \eqref{E:barPsidef}, and such that the following estimates hold for all $t \geq 0:$

	\begin{subequations}
		\begin{align}
			\| g_{jk}^{(\infty)} - \mathring{g}_{jk} \|_{H^N} & \leq C \epsilon, \label{E:gjklowerinfinityHN} \\
			\| g_{(\infty)}^{jk} - \mathring{g}^{jk} \|_{H^N} & \leq C \epsilon, \label{E:gjkupperinfinityHN} 
		\end{align}
	\end{subequations}
	
	\begin{subequations}
	\begin{align}
		\| e^{-2 \Omega} g_{jk} - g_{jk}^{(\infty)} \|_{H^N} & \leq C \epsilon e^{-qHt},  \label{E:gjklowerconvergence} \\
		\| e^{-2 \Omega} g_{jk} - g_{jk}^{(\infty)} \|_{H^{N-2}} & \leq C \epsilon e^{-2Ht},  
			\label{E:improvedgjklowerconvergence} \\
		\| e^{2 \Omega} g^{jk} - g_{(\infty)}^{jk} \|_{H^N} & \leq C \epsilon e^{-qHt},  \label{E:gjkupperconvergence} \\
		\| e^{2 \Omega} g^{jk} - g_{(\infty)}^{jk} \|_{H^{N-2}} & \leq C \epsilon e^{-2Ht},  
			\label{E:improvedgjkupperconvergence} \\
		\| e^{-2 \Omega} \partial_t g_{jk} - 2 \omega g_{jk}^{(\infty)} \|_{H^N} & \leq C \epsilon e^{-qHt}, 
			\label{E:e2Omegapartialgjkminusrhojklowerconvergence} \\
		\| e^{-2 \Omega} \partial_t g_{jk} - 2 \omega g_{jk}^{(\infty)} \|_{H^{N-2}} & \leq C \epsilon e^{-2Ht},
			\label{E:improvede2Omegapartialgjkminusrhojklowerconvergence} \\
			\| e^{2 \Omega} \partial_t g^{jk} + 2 \omega g_{(\infty)}^{jk}\|_{H^N} & \leq C \epsilon e^{-qHt}, 
			\label{E:e2Omegapartialgjkminusrhojkupperconvergence} \\
		\| e^{2 \Omega} \partial_t g^{jk} + 2 \omega g_{(\infty)}^{jk} \|_{H^{N-2}} & \leq C \epsilon e^{-2Ht},
			\label{E:improvede2Omegapartialgjkminusrhojkupperconvergence}
	\end{align}
	\end{subequations}

	\begin{subequations}
	\begin{align}
		\| g_{0j} - H^{-1} g_{(\infty)}^{ab} \Gamma_{ajb}^{(\infty)} \|_{H^{N-3}} & \leq C \epsilon e^{-qHt},
			\label{E:g0jlowerconvergence} \\
		\| \partial_t g_{0j} \|_{H^{N-3}} & \leq C \epsilon e^{-qHt},
			\label{E:partialtg0jconvergence} 
	\end{align}
	\end{subequations}
	
	\begin{subequations}
	\begin{align}
		\| g_{00} + 1\|_{H^N} & \leq C \epsilon e^{-qHt}, \label{E:g00plusoneconvergence} \\
		\| g_{00} + 1\|_{H^{N-2}} & \leq C \epsilon (1 + t) e^{-2Ht}, \label{E:improvedg00plusoneconvergence} \\
		\| \partial_t g_{00} \|_{H^N} & \leq C \epsilon e^{-qHt}, \label{E:partialtg00convergence} \\
		\| \partial_t g_{00} + 2 \omega(g_{00} + 1) \|_{H^{N-2}} & \leq C \epsilon e^{-2Ht}, 
			\label{E:improvedpartialtg00convergence}
	\end{align}
	\end{subequations}
	
	\begin{subequations}
	\begin{align} 
		\| e^{-2 \Omega} K_{jk} - \omega g_{jk}^{(\infty)} \|_{H^{N-1}} & \leq C \epsilon e^{-qHt}, \label{E:Kjkconvergence} \\
		\| e^{-2 \Omega} K_{jk} - \omega g_{jk}^{(\infty)} \|_{H^{N-2}} & \leq C (1+t) \epsilon e^{-2Ht}. 				
			\label{E:improvedKjkconvergence}
	\end{align}
	\end{subequations}
	In the above inequalities, $K_{jk}$ is the second fundamental form of the hypersurface $t= const.$
	
	Furthermore, we have that
	\begin{subequations}
	\begin{align}
		\| e^{\decayparameter \Omega} \partial_t \Phi - \Psi_{(\infty)} \|_{H^{N-1}} & \leq C \epsilon e^{-qHt}, 
			\label{E:partialtPhiconvergence} \\
		\| e^{\decayparameter \Omega} \partial_t \Phi - \Psi_{(\infty)} \|_{H^{N-2}} 
			& \leq C \epsilon e^{-2(1 - \kappa)Ht}, \label{E:improvedpartialtPhiconvergence} \\
		\| \bard \Phi - \bard \Phi_{(\infty)} \|_{H^{N-1}} & \leq C \epsilon e^{- \decayparameter Ht}, 
			\label{E:partialPhiconvergence} \\
		\| \bard \Phi_{(\infty)} \|_{H^{N-1}} & \leq C \epsilon.
			\label{E:partialPhiinfinityHNminusone} 
	\end{align}
	\end{subequations}
	
\end{theorem}

\begin{remark} \label{R:Nlarger}
	We assume that $N \geq 5$ so that we can use standard Sobolev-Moser estimates during our proofs the improved
	rates of decay.
\end{remark}

\begin{proof}
	See Remark \ref{R:ProofsRemark} for some conventions that we use throughout this proof. Furthermore, in our proofs below,
	we will introduce new energies, and the differential inequalities that we will derive are valid only under the assumption 
	that the energies are sufficiently small; we don't explicitly mention the smallness assumption each time we make it. 
	In the interest of brevity, we will only  sketch the proofs of the estimates involving the improved decay rates. We also 
	remind the reader of the conclusion \eqref{E:QNgloballessthanepsilon} of Theorem 
	\ref{T:GlobalExistence}, which is that $\suptotalnorm{N} \eqdef \supgzerozeronorm{N} + \supgzerostarnorm{N} + 
	\suphstarstarnorm{N} + \supfluidnorm{N}$ satisfies $\suptotalnorm{N}(t) \leq \epsilon$ for $t \geq 0.$ 
	\\

\noindent \emph{Proofs of \eqref{E:gjklowerinfinityHN}, \eqref{E:gjkupperinfinityHN},
\eqref{E:gjklowerconvergence}, \eqref{E:gjkupperconvergence}, \eqref{E:e2Omegapartialgjkminusrhojklowerconvergence}, and \eqref{E:e2Omegapartialgjkminusrhojkupperconvergence}}:
	It follows from the definition \eqref{E:totalsupnorm} of $\suptotalnorm{N}$ that
	\begin{align} \label{E:partialthjkdecays}
		\| \partial_t h_{jk} \|_{H^N} \leq C \epsilon e^{-q H t}.		
	\end{align}
	Integrating $\partial_t h_{jk}$ and using \eqref{E:partialthjkdecays}, it follows that for $t_1 \leq t_2,$ we have that
	
	\begin{align} \label{E:Cauchysequence}
		\| h_{jk}(t_2) - h_{jk}(t_1) \|_{H^N} & \leq C \epsilon e^{-q H t_1}.
	\end{align}
	Using \eqref{E:Cauchysequence} and the fact that $h_{jk} = e^{-2 \Omega} g_{jk}$, 
	it easily follows that there exist functions $g_{jk}^{(\infty)}(x^1,x^2,x^3)$ such that
	\begin{align} \label{E:gjklowerconvergenceproof}
		\| e^{-2 \Omega} g_{jk} - g_{jk}^{(\infty)} \|_{H^N} \leq C \epsilon e^{-qHt},
	\end{align}
	and such that
	
	\begin{align} \label{E:rhominusmathringg}
		\| g_{jk}^{(\infty)} - \mathring{g}_{jk} \|_{H^N} \leq C \epsilon.
	\end{align}
	We have thus shown \eqref{E:gjklowerinfinityHN} and \eqref{E:gjklowerconvergence} . Inequality 
	\eqref{E:e2Omegapartialgjkminusrhojklowerconvergence} follows from  
	\eqref{E:gjklowerconvergence} and \eqref{E:partialthjkdecays}.

	To obtain the asymptotics for $g^{jk},$ we use \eqref{E:partialtgjkupperplusomegagjkHN}, which implies that
	\begin{align} \label{E:partialte2Omegagjkupperdecays}
		\| \partial_t \big(e^{2 \Omega} g^{jk} \big) \|_{H^N} \leq C \epsilon e^{-q H t}.
	\end{align}
	From \eqref{E:partialte2Omegagjkupperdecays} and the estimates \eqref{E:gjkupperLinfinity}, 
	\eqref{E:partialgjkupperHNminusone} at $t=0,$ it follows as in the previous argument that there exist functions 
	$g_{(\infty)}^{jk}(x^1,x^2,x^3)$ such that
	
	\begin{align} 
		\| e^{2 \Omega} g^{jk} - g_{(\infty)}^{jk} \|_{H^N} \leq C \epsilon e^{-q Ht},
	\end{align}
	and such that
	
	\begin{align} 
		\| g_{(\infty)}^{jk} - \mathring{g}^{jk} \|_{H^N} & \leq C \epsilon, \label{E:rhominusmathringgupper} \\
		\| \bard g_{(\infty)}^{jk} \|_{H^{N-1}} & \leq C \epsilon, \label{E:partialrhojkupperHNminusone} \\
		\| g_{(\infty)}^{jk} \|_{L^{\infty}} & \leq C,  \label{E:partialrhojkupperLinfinity}
	\end{align}
	where $\mathring{g}^{jk} \eqdef g^{jk}|_{t=0}.$ This proves 
	\eqref{E:gjkupperinfinityHN} and \eqref{E:gjkupperconvergence}. 
	\eqref{E:e2Omegapartialgjkminusrhojkupperconvergence} then follows from \eqref{E:gjkupperconvergence} and 
	\eqref{E:partialte2Omegagjkupperdecays}. Furthermore, since
	
	\begin{align}
		g^{aj}g_{ak} + g^{0j}g_{0k} = \delta_k^j,
	\end{align}
	and since Proposition \ref{P:F1FkLinfinityHN}, $\suptotalnorm{N} \leq \epsilon,$ \eqref{E:g0jupperHN}, 
	and \eqref{E:g0jupperCb1} imply that
	
	\begin{align}
		\| g^{0j}g_{0k} \|_{H^N} \leq C \epsilon^2 e^{-2 q Ht},
	\end{align}
	it follows that $g_{jk}^{(\infty)}$ are the components of a Riemannian metric $g^{(\infty)}$ and that
	$g_{(\infty)}^{jk}$ are the components of its inverse $g_{(\infty)}.$ 
	\\

\noindent \emph{Proofs of \eqref{E:g00plusoneconvergence} and \eqref{E:partialtg00convergence}}:
	The estimates \eqref{E:g00plusoneconvergence} and \eqref{E:partialtg00convergence} follow trivially from
	definition \eqref{E:mathfrakSMsupg00}.
	\\

\noindent \emph{Proofs of \eqref{E:partialtPhiconvergence}, \eqref{E:partialPhiconvergence}, and \eqref{E:partialPhiinfinityHNminusone}}:
To prove \eqref{E:partialtPhiconvergence}, we first recall equation \eqref{E:finalfluidequation},
which can be re-expressed as follows:
	
	\begin{align} \label{E:partialtewOmegapartialtPhidecay}
		\partial_t \Big( e^{\decayparameter \Omega} \partial_t \Phi - \bar{\Psi} \Big) = e^{\decayparameter \Omega} 
		\triangle'_{\partial \Phi},		
	\end{align}
	where $\bar{\Psi}$ is defined in \eqref{E:barPsidef} and 
	
	\begin{align} 
		\triangle'_{\partial \Phi} = -m^{ab} \partial_a \partial_b \Phi - 2 m^{0a} \partial_a \partial_t \Phi + 
		\triangle_{\partial \Phi}.
	\end{align}
	Using Proposition \ref{P:F1FkLinfinityHN}, the definition \eqref{E:totalsupnorm} of $\suptotalnorm{N},$ Sobolev embedding,
	\eqref{E:triangleHN}, \eqref{E:mjkLinfinity}, \eqref{E:partialmjkHNminusone}, 	
	and \eqref{E:m0jHN}, it follows that
	
	\begin{align} \label{E:ewOmegatrianglePhiHNminusone}
		\| e^{\decayparameter \Omega} \triangle'_{\partial \Phi} \|_{H^{N-1}} \leq C \epsilon e^{-qHt}.
	\end{align}
	As in our proof of \eqref{E:gjklowerconvergenceproof}, it easily follows from \eqref{E:partialtewOmegapartialtPhidecay}, 
	\eqref{E:ewOmegatrianglePhiHNminusone}, and the initial condition 
	$\| e^{\decayparameter \Omega(0)} \partial_t \Phi(0,\cdot) - \bar{\Psi} \|_{H^N} \leq C \epsilon$ that there exists a 
	function $\Psi_{(\infty)}(x^1,x^2,x^3)$ with $\Psi_{(\infty)} - 
	\bar{\Psi} \in H^{N-1}$ such that
	
	\begin{align}
		\| e^{\decayparameter \Omega} \partial_t \Phi(t,\cdot) - \Psi_{(\infty)} \|_{H^{N-1}} & \leq C \epsilon e^{-qHt}, \\
		\| \Psi_{(\infty)} - \bar{\Psi} \|_{H^{N-1}} & \leq C \epsilon, \label{E:PsiinfinityHNminusone}
	\end{align}
	which proves \eqref{E:partialtPhiconvergence}.
	
	Furthermore, the bound $\suptotalnorm{N}(t) \leq \epsilon$ implies that
	
	\begin{align} \label{E:partialtpartialPhidecay}
		\| \partial_t \bard \Phi \|_{H^{N-1}} \leq C \epsilon e^{-\decayparameter Ht}.
	\end{align}
	It follows easily from \eqref{E:partialtpartialPhidecay} and the initial condition 
	$\| \bard \Phi(0,\cdot) \|_{H^N} \leq C \epsilon$ that there exists a function $\bard \Phi_{(\infty)}(x^1,x^2,x^3)$ 
	(see Remark \ref{R:Phiremark}) satisfying $\bard \Phi_{(\infty)} \in H^{N-1}$ such that 
	
	\begin{align}
		\| \bard \Phi(t,\cdot) - \bard \Phi_{(\infty)} \|_{H^{N-1}} & \leq C \epsilon e^{- \decayparameter H t}, \\
		\| \bard \Phi_{(\infty)} \|_{H^{N-1}} & \leq C \epsilon,
	\end{align}
	which proves \eqref{E:partialPhiconvergence} and \eqref{E:partialPhiinfinityHNminusone}.
	\\

\noindent \emph{Proof of \eqref{E:Kjkconvergence}}:
	We first observe that $\hat{N},$ the future-directed normal
	to the surface $\lbrace t=const \rbrace,$ can be expressed in components as
	\begin{align}
		\hat{N}^{\mu} = - (-g^{00})^{-1/2} g^{0\mu}.
	\end{align}
	By the definition of $K,$ it thus follows that
	
	\begin{align} \label{E:kijexpression}
		K_{jk} = g_{\alpha k} D_j \hat{N}^{\alpha} = 
			- \big(\partial_j \big[ (-g^{00})^{-1/2} g^{0\alpha} \big] \big) g_{\alpha k} - (-g^{00})^{-1/2} g^{0\alpha} 
			\Gamma_{jk\alpha}.
	\end{align}
	
	The dominant term on the right-hand side of \eqref{E:kijexpression} is the one involving $\Gamma_{jk0}:$ 
	it follows from Corollary \ref{C:SobolevTaylor}, with $v= g^{00} + 1$ and $F(v) = (1-v)^{-1/2} = (-g^{00})^{-1/2}$ in the 
	Corollary, Proposition \ref{P:F1FkLinfinityHN},
	 the definition \eqref{E:totalsupnorm} of $\suptotalnorm{N},$ Sobolev embedding, 
	\eqref{E:gjklowerLinfinity}, \eqref{E:g00upperplusoneHN}, \eqref{E:g0jupperHN}, \eqref{E:partialtgjkminusomegagjklowerHN},
	and definition \eqref{E:Christoffelglowered} that 
	
	\begin{align}
		e^{-2 \Omega} \big\|K_{jk} + (-g^{00})^{-1/2} g^{00} \Gamma_{jk0} \big\|_{H^{N-1}} 
			& = e^{-2 \Omega} \big\| - \big(\partial_j \big[ (-g^{00})^{-1/2} g^{0\alpha} \big]\big) g_{\alpha k} 
			- (-g^{00})^{-1/2} g^{0a} \Gamma_{jka} \big\|_{H^{N-1}}
			\leq C \epsilon e^{- q Ht}, \label{E:Kjktriangle1} \\
		e^{-2 \Omega} \Big\| (-g^{00})^{-1/2} g^{00} \Gamma_{jk0} + \frac{1}{2}\partial_t g_{jk} \Big \|_{H^{N-1}}
		 	& \leq \frac{1}{2} e^{-2 \Omega} \big\| (-g^{00})^{-1/2} g^{00}(\partial_j g_{0k} - \partial_k g_{0j}) \big\|_{H^{N-1}}
		 	\label{E:Kjktriangle2} \\
		 	& \ \ \ + \frac{1}{2} e^{-2 \Omega} \big\| [(-g^{00})^{-1/2} g^{00} + 1 ] \partial_t g_{jk}  \big\|_{H^{N-1}}
		 		\notag \\
		 	& \leq C \epsilon e^{-qHt}. \notag  
	\end{align}
	Inequality \eqref{E:Kjkconvergence} now follows from combining \eqref{E:e2Omegapartialgjkminusrhojklowerconvergence}, 
	\eqref{E:Kjktriangle1}, and \eqref{E:Kjktriangle2}.

\end{proof}
\ \\

\noindent \emph{Proofs of \eqref{E:g0jlowerconvergence} - \eqref{E:partialtg0jconvergence}}:
	Our proofs of \eqref{E:g0jlowerconvergence} - \eqref{E:partialtg0jconvergence} are based on two refined versions of the 
	energy inequality \eqref{E:mathfrakENg0*integral}. The main point is that the even though the
	energy $\gzerostarenergy{N}$ defined in \eqref{E:g0*energydef} allows us to efficiently
	close the bootstrap argument of Theorem \ref{T:GlobalExistence}, there is room for improvement. In particular, Theorem
	\ref{T:GlobalExistence} only allows us to conclude that  $\gzerostarenergy{N} \leq \suptotalnorm{N} \leq C \epsilon,$
	which implies that $\|\partial_t g_{0j}\|_{H^N} \leq \epsilon e^{(1 - q) Ht}$ and 
	$\|g_{0j}\|_{H^N} \leq \epsilon e^{(1 - q) Ht}.$ As we will see, it is possible to improve these estimates
	by a factor of $e^{(q-1) Ht}.$ This is a preliminary step that we will need in our remaining proofs. 
	This improvement is based on the following simple identity:
	
	\begin{align} \label{E:g0jupperidentity}
		g^{0j} = - \frac{1}{g_{00}} g^{aj} g_{0a}.
	\end{align}
	
	We begin our proof of the improvement by defining a new energy $\underline{\mathscr{E}}_{g_{0*};N-1}$ 
	for the $g_{0j},$ $(j=1,2,3),$ by
	
	\begin{align} \label{E:rescaledg0*energy}
		\underline{\mathscr{E}}_{g_{0*};N-1}^2 & \eqdef \sum_{|\vec{\alpha}| \leq N-1} \sum_{j=1}^3
			\mathcal{E}_{(\gamma_{0*},\delta_{0*})}^2[\partial_{\vec{\alpha}} g_{0j}, \partial (\partial_{\vec{\alpha}} g_{0j})],
	\end{align}
	where $\mathcal{E}_{(\gamma_{0*},\delta_{0*})}^2[\partial_{\vec{\alpha}} g_{0j}, \partial (\partial_{\vec{\alpha}} g_{0j})]$
	is defined in \eqref{E:mathcalEdef}, and 
	the constants $\gamma_{0*}, \delta_{0*}$ are defined in Definition \ref{D:energiesforg}. Note that 
	the scaling in \eqref{E:rescaledg0*energy} differs from the scaling used
	in definition \eqref{E:g0*energydef} by a factor of $e^{(1-q)\Omega}.$ Furthermore, we are using an ${N-1}^{st}$ 
	order energy, rather than an $N^{th}$ order energy, because we will make use of the improved rates of decay 
	for lower derivatives that are already discernible from the fact that $\suptotalnorm{N} \leq \epsilon$
	(consider e.g. inequality \eqref{E:gabupperGammaajblowerHNminusone}). Now using 
	\eqref{E:mathcalEcomparison}, we have the following comparison estimate:
	
	\begin{align} \label{E:g0*newenergynormequivalence}
		C^{-1} \underline{\mathscr{E}}_{g_{0*};N-1} \leq \sum_{j=1}^3 \| \partial_t g_{0j} \|_{H^{N-1}} +
			e^{- \Omega} \| \bard g_{0j} \|_{H^{N-1}} + C_{(\gamma_{0*})} \| g_{0j} \|_{H^{N-1}} 
			\leq C \underline{\mathscr{E}}_{g_{0*};N-1}.
	\end{align}
	
	Using the definition of $\underline{\mathscr{E}}_{g_{0*};N-1}$ and the comparison estimate, 
	it follows that the energy inequality \eqref{E:underlinemathfrakEg0*firstdifferential} can be replaced
	with	
	\begin{align} 
		\frac{d}{dt}(\underline{\mathscr{E}}_{g_{0*};N-1}^2) 
			& \leq - \eta_{0*}H\underline{\mathscr{E}}_{g_{0*};N-1}^2 
			+ C \underline{\mathscr{E}}_{g_{0*};N-1} \sum_{j=1}^3 \| g^{ab} \Gamma_{ajb} \|_{H^{N-1}}
		  + C \underline{\mathscr{E}}_{g_{0*};N-1} \sum_{j=1}^3 \| \triangle_{0j} \|_{H^{N-1}}
			\label{E:newunderlinemathfrakEg0*firstdifferential} \\ 
		&	\ \ \ + C \underline{\mathscr{E}}_{g_{0*};N-1} 
			\sum_{|\vec{\alpha}| \leq N-1} \sum_{j=1}^3 \| [\hat{\Square}_g ,\partial_{\vec{\alpha}}] g_{0j} \|_{L^2}
			+ \sum_{|\vec{\alpha}| \leq N-1} \sum_{j=1}^3 \| \triangle_{\mathcal{E};(\gamma_{0*}, 
			\delta_{0*})}[\partial_{\vec{\alpha}} g_{0j} ,\partial(\partial_{\vec{\alpha}} g_{0j})] \|_{L^1}. \notag 
	\end{align}

	We now claim that the following improvements of \eqref{E:triangleA0jHN}, \eqref{E:triangleC0jHN}, \eqref{E:triangle0jHN}, 
	\eqref{E:g0jcommutatorL2}, and \eqref{E:triangleEgamma0jdelta0*L1} hold:
	
	\begin{align}
		\| \triangle_{A,0j} \|_{H^{N-1}} & \leq C \epsilon e^{-qHt} \underline{\mathscr{E}}_{g_{0*};N-1}
			+ C \epsilon e^{-qHt} , && \label{E:improvedtriangleA0jHN} \\
		\| \triangle_{C,0j} \|_{H^{N-1}} & \leq C \epsilon e^{-qHt} \underline{\mathscr{E}}_{g_{0*};N-1}, &&
			\label{E:improvedtriangleC0jHN} \\
		\| \triangle_{0j} \|_{H^{N-1}} & \leq C e^{-qHt} \underline{\mathscr{E}}_{g_{0*};N-1}
			+ C \epsilon e^{-qHt}, && \label{E:triangle0jnewbound} \\	
		\| [\hat{\Square}_g, \partial_{\vec{\alpha}} ] g_{0j} \|_{L^2} & \leq 
			 C \epsilon e^{-qHt}\underline{\mathscr{E}}_{g_{0*};N-1}
			 + C \epsilon e^{-qHt}, && (|\vec{\alpha}| \leq N-1), \label{E:newg0jcommutatorL2} \\					\| 
			 \triangle_{\mathcal{E};(\gamma_{0*}, \delta_{0*})}[\partial_{\vec{\alpha}}g_{0j},
			\partial (\partial_{\vec{\alpha}} g_{0j})] \|_{L^1}
			& \leq C \epsilon e^{-q Ht} \underline{\mathscr{E}}_{g_{0*};N-1}
				+ C e^{-qHt} \underline{\mathscr{E}}_{g_{0*};N-1}^2, && (|\vec{\alpha}| \leq N-1). \label{E:newtriangleEgamma0jdelta0*L1}	
	\end{align}
	These improved estimates can be derived using the 
	same methods as in our proofs of the original estimates, together with the identity \eqref{E:g0jupperidentity}. More 
	specifically, in our proofs of \eqref{E:triangleA0jHN}, \eqref{E:triangleC0jHN}, \eqref{E:triangle0jHN}, 
	\eqref{E:g0jcommutatorL2}, and \eqref{E:triangleEgamma0jdelta0*L1}, we used the estimates 
	$\|\partial_t g_{0j}\|_{H^N} \leq e^{(1-q)\Omega} \gzerostarnorm{N},$
	$\|g_{0j}\|_{H^N} \leq e^{(1-q)\Omega} \gzerostarnorm{N},$ and 
	$\|\bard g_{0j}\|_{H^N} \leq e^{(2-q)\Omega} \gzerostarnorm{N},$ which follow directly from the definition
	of $\gzerostarnorm{N},$ together with \eqref{E:g0jupperHN}, which reads $\|g^{0j}\|_{H^N} \leq C e^{-(1 + 
	q)\Omega} \supgnorm{N}.$ However, whenever it is convenient, these estimates can be replaced with 
	
	\begin{align}
			\| \partial_t g_{0j} \|_{H^{N-1}} & \leq C \underline{\mathscr{E}}_{g_{0*};N-1}, 
			\label{E:partialtg0jlowerHNimprovement} \\
		\|g_{0j}\|_{H^{N-1}} & \leq C \underline{\mathscr{E}}_{g_{0*};N-1}, 
			\label{E:g0jlowerHNimprovement} \\
		\|\bard g_{0j}\|_{H^{N-1}} & \leq C e^{\Omega} \underline{\mathscr{E}}_{g_{0*};N-1}, 
			\label{E:barpartialg0jlowerHNimprovement} \\ 
		\|g^{0j}\|_{H^{N-1}} & \leq C e^{-2\Omega} \underline{\mathscr{E}}_{g_{0*};N-1} 
			\label{E:g0jupperHNimprovement}
	\end{align}
	respectively, where \eqref{E:partialtg0jlowerHNimprovement} - \eqref{E:barpartialg0jlowerHNimprovement}
	follow from \eqref{E:g0*newenergynormequivalence}, while \eqref{E:g0jupperHNimprovement} follows from applying  
	\eqref{E:g0*newenergynormequivalence}, Proposition \ref{P:F1FkLinfinityHN}, Sobolev embedding,
	\eqref{E:gjkupperLinfinity}, and \eqref{E:partialgjkupperHNminusone} to the identity \eqref{E:g0jupperidentity}.
	We remark that the $C e^{-qHt} \underline{\mathscr{E}}_{g_{0*};N-1}$ term on the right-hand
	side of \eqref{E:triangle0jnewbound} arises from from e.g.
	the $(\omega - H) \partial_t g_{0j}$ term on the right-hand side of \eqref{E:triangle0j}. Similarly, the
	$C e^{-qHt} \underline{\mathscr{E}}_{g_{0*};N-1}^2$ term on the right-hand side of \eqref{E:newtriangleEgamma0jdelta0*L1}
	arises from the e.g. the $(H - \omega) g^{ab} (\partial_a v) (\partial_b v)$ term
	on the right-hand side of \eqref{E:trianglemathscrEdef}.
	See \cite[Section 14]{hR2008} for additional details on these improved estimates.

	Now using \eqref{E:gabupperGammaajblowerHNminusone}, \eqref{E:newunderlinemathfrakEg0*firstdifferential} and 
	\eqref{E:triangle0jnewbound} - \eqref{E:newtriangleEgamma0jdelta0*L1}, we argue as in our proof of 
	\eqref{E:g0*energyfinaldifferentialinequality} to deduce the following inequality:
	
	\begin{align}  \label{E:newg0*energyfinaldifferentialinequality}
		\frac{d}{dt} \Big(\underline{\mathscr{E}}_{g_{0*};N-1}^2 \Big)
			& \leq - \eta_{0*}H \underline{\mathscr{E}}_{g_{0*};N-1}^2  
			+ C e^{-qHt} \underline{\mathscr{E}}_{g_{0*};N-1}^2
			+ C \epsilon \underline{\mathscr{E}}_{g_{0*};N-1}, 
	\end{align}
	where the $C \epsilon \underline{\mathscr{E}}_{g_{0*};N-1}$ term on the right-hand side of
	\eqref{E:newg0*energyfinaldifferentialinequality} arises from applying
	inequality \eqref{E:gabupperGammaajblowerHNminusone} to the second term on the right-hand side of 
	\eqref{E:newunderlinemathfrakEg0*firstdifferential}. Integrating \eqref{E:newg0*energyfinaldifferentialinequality} 
	from $0$ to $t,$ using the smallness condition $\underline{\mathscr{E}}_{g_{0*};N-1}(0) \leq C \epsilon$
	for small $t,$ and applying Lemma \ref{L:integralinequality} for large $t$ (the 
	$C e^{-qHt} \underline{\mathscr{E}}_{g_{0*};N-1}^2$ term is dominated by the  
	$- \eta_{0*}H \underline{\mathscr{E}}_{g_{0*};N-1}^2 $ term for large $t$),
	we conclude that the following bound holds for all $t \geq 0:$
	
	\begin{align} \label{E:rescaledg0*energybounded}
		\underline{\mathscr{E}}_{g_{0*};N-1} & \leq C \epsilon.
	\end{align}
	This completes the proof of our preliminary improved estimate.
	
	We are now ready for the proofs of 
	\eqref{E:g0jlowerconvergence} and \eqref{E:partialtg0jconvergence}. Defining 
	
	\begin{align}
		v_j \eqdef g_{0j} - H^{-1} g_{(\infty)}^{ab} \Gamma_{ajb}^{(\infty)},
	\end{align}
	and recalling that $g_{0j}$ is a solution to 
	
	\begin{align}
		\hat{\Square}_{g} g_{0j} & = 3 H \partial_t g_{0j} + 2 H^2 g_{0j} - 2Hg^{ab}\Gamma_{a j b} + \triangle_{0j},	
	\end{align}
	it follows that $v_j$ is a solution to
	
	\begin{align} \label{E:finallg0jequationagain}
		\hat{\Square}_{g} v_j = 3H \partial_t v_j + 2H^2 v_j + \triangle_j,
	\end{align}
	where
	
	\begin{align} \label{E:newtrianglejdef}
		\triangle_j = \triangle_{0j} + 2H \big(g_{(\infty)}^{ab} \Gamma_{ajb}^{(\infty)} - g^{ab}\Gamma_{ajb}\big) 
			- H^{-1} g^{lm} \partial_l \partial_m \big(g_{(\infty)}^{ab} \Gamma_{ajb}^{(\infty)} \big).
	\end{align}
	
	To estimate $v_j, \ (j = 1,2,3),$ we will use the energy
	
	\begin{align} \label{E:rescaledv*energy}
		\underline{\mathscr{E}}_{v_{*};N-3}^2 & \eqdef e^{2qHt} \sum_{|\vec{\alpha}| \leq N-3} \sum_{j=1}^3
			\mathcal{E}_{(\gamma_{0*},\delta_{0*})}^2[\partial_{\vec{\alpha}} v_j, \partial (\partial_{\vec{\alpha}} v_j)],
	\end{align}
	where $\mathcal{E}_{(\gamma_{0*},\delta_{0*})}^2[\partial_{\vec{\alpha}} v_j, \partial (\partial_{\vec{\alpha}} v_j)]$ is 
	defined in \eqref{E:mathcalEdef}, and the constants $\gamma_{0*}, \delta_{0*}$ are defined in Definition 
	\ref{D:energiesforg}. Note that in response to the last term on the right-hand side of 
	\eqref{E:newtrianglejdef}, we have further reduced the number
	of derivatives in the definition of our energy by two. Using \eqref{E:mathcalEcomparison}, it follows that 
	
	\begin{align} \label{E:v*newenergynormequivalence}
		C^{-1} \underline{\mathscr{E}}_{v_*;N-3} \leq \sum_{j=1}^3 e^{qHt} \| \partial_t v_j \|_{H^{N-3}}
			 + C_{(\gamma_{0*})} e^{qHt} \| v_j \|_{H^{N-3}} + e^{(q-1)Ht} \| \bard v_j \|_{H^{N-3}}
			\leq C \underline{\mathscr{E}}_{v_*;N-3}.
	\end{align}	
	Arguing as in our proof of \eqref{E:newunderlinemathfrakEg0*firstdifferential}, we have that
	
\begin{align}  \label{E:newunderlinemathfrakEv*firstdifferential}
		\frac{d}{dt}(\underline{\mathscr{E}}_{v_{*};N-3}^2) 
			& \leq (2q - \eta_{0*})H \underline{\mathscr{E}}_{v_{*};N-3}^2 
			+ C e^{qHt} \underline{\mathscr{E}}_{v_{*};N-3} \sum_{j=1}^3 \| \triangle_{j} \|_{H^{N-3}}
			 \\ 
		&	\ \ \ + C e^{qHt} \underline{\mathscr{E}}_{v_{*};N-3} 
			\sum_{|\vec{\alpha}| \leq N-3} \sum_{j=1}^3 \| [\hat{\Square}_g ,\partial_{\vec{\alpha}}] v_j \|_{L^2}
			+ e^{2qHt} \sum_{|\vec{\alpha}| \leq N-3} \sum_{j=1}^3 \| \triangle_{\mathcal{E};(\gamma_{0*}, 
			\delta_{0*})}[\partial_{\vec{\alpha}} v_j ,\partial(\partial_{\vec{\alpha}} v_j)] \|_{L^1}. \notag 
	\end{align}	
	Note that the inequality \eqref{E:newunderlinemathfrakEv*firstdifferential} does \emph{not} have a term corresponding to the 
	$\| g^{ab} \Gamma_{ajb} \|_{H^{N-1}}$ term in \eqref{E:newunderlinemathfrakEg0*firstdifferential}; the remnants of
	this term are present in $\| \triangle_{j} \|_{H^{N-3}}.$ 	
	
	We will now estimate $\| \triangle_j \|_{H^{N-3}},$ our goal being to show that it decays like
	$\epsilon e^{-qHt}.$ We begin by estimating the term $2Hg^{ab}\Gamma_{a j b}$ from the right-hand side of \eqref{E:newtrianglejdef}.
	Let us first recall the definitions of the lowered Christoffel symbols 
	$\Gamma_{ijk}^{(\infty)}$ and $\Gamma_{\mu \alpha \nu}$ corresponding to $g^{(\infty)}$ and $g$ respectively:
	
	\begin{align} 
		\Gamma_{ijk}^{(\infty)} & \eqdef \frac{1}{2} (\partial_i g_{jk}^{(\infty)} + \partial_k g_{ij}^{(\infty)} 
			- \partial_j g_{ik}^{(\infty)}), \label{E:Christoffelrholowered} \\
		\Gamma_{\mu \alpha \nu} & \eqdef \frac{1}{2} (\partial_{\mu} g_{\alpha \nu} + \partial_{\nu} g_{\mu \alpha} 
			- \partial_{\alpha} g_{\mu \nu}). \label{E:Christoffelglowered}
	\end{align}
	Using Proposition \ref{P:F1FkLinfinityHN}, the definition \eqref{E:totalsupnorm} of $\suptotalnorm{N},$ Sobolev embedding,
	\eqref{E:gjklowerinfinityHN},
	\eqref{E:gjkupperinfinityHN}, \eqref{E:gjklowerconvergence}, \eqref{E:gjkupperconvergence}, 
	\eqref{E:partialrhojkupperLinfinity}, \eqref{E:Christoffelrholowered}, and \eqref{E:Christoffelglowered}, we conclude that 
	
	\begin{align} 
		\| g^{ab}(t,\cdot) \Gamma_{ajb}(t,\cdot) - g_{(\infty)}^{ab} \Gamma_{ajb}^{(\infty)} \|_{H^{N-1}} 
			& \leq || e^{2 \Omega} g^{ab}(t,\cdot) - g_{(\infty)}^{ab}||_{H^{N-1}} 
				\| e^{-2 \Omega} \Gamma_{ajb}(t,\cdot) \|_{H^{N-1}} 
			\label{E:gabGammaajbminusrhoabGammaajbHN}  \\
			& \ \ + \| g_{(\infty)}^{ab} \|_{L^{\infty}} 
				\|  e^{-2 \Omega} \Gamma_{ajb}(t,\cdot) - \Gamma_{ajb}^{(\infty)} 
				\|_{H^{N-1}} \notag \\
			& \ \ + \| \bard g_{(\infty)}^{ab} \|_{H^{N-2}} 
				\| e^{-2 \Omega} \Gamma_{ajb}(t,\cdot) - \Gamma_{ajb}^{(\infty)} \|_{L^{\infty}}	\notag \\
		& \leq C \epsilon e^{-q Ht}, \notag \\
		\| g_{(\infty)}^{ab} \Gamma_{ajb}^{(\infty)} \|_{H^{N-1}} & \leq C \epsilon. \label{E:rhoabGammaajbHN}
	\end{align}

	We now estimate the term $H^{-1} g^{lm} \partial_l \partial_m \big(g_{(\infty)}^{ab} \Gamma_{ajb}^{(\infty)} \big)$ 
	from the right-hand side of \eqref{E:newtrianglejdef}.
	By Proposition \ref{P:F1FkLinfinityHN}, \eqref{E:gjkupperLinfinity},
	\eqref{E:partialgjkupperHNminusone}, and \eqref{E:rhoabGammaajbHN}, it follows that
	\begin{align} \label{E:hatSquarerhoabGammaajbHNminus2}
		\| g^{lm} \partial_l \partial_m (g_{(\infty)}^{ab} \Gamma_{ajb}^{(\infty)}) \|_{H^{N-3}} \leq C \epsilon e^{-2 Ht}.
	\end{align}
	Applying the estimates \eqref{E:triangle0jnewbound},
	\eqref{E:rescaledg0*energybounded},
	\eqref{E:gabGammaajbminusrhoabGammaajbHN}, and \eqref{E:hatSquarerhoabGammaajbHNminus2}
	to the terms in \eqref{E:newtrianglejdef}, we conclude the
	desired estimate for $\| \triangle_j \|_{H^{N-3}}:$
	
	\begin{align} \label{E:trianglejHNminus3}
		\| \triangle_j \|_{H^{N-3}} \leq C \epsilon e^{-q Ht}.
	\end{align}
	
	In addition, arguing as in our proof of \eqref{E:g0jcommutatorL2} and \eqref{E:newtriangleEgamma0jdelta0*L1}, 
	and in particular making use of the improved estimates \eqref{E:partialtg0jlowerHNimprovement} - \eqref{E:g0jupperHNimprovement} and \eqref{E:rescaledg0*energybounded}, it can be shown that
	
	\begin{align}
		\| [\hat{\Square}_g, \partial_{\vec{\alpha}} ] v_j \|_{L^2} 
			& \leq C Squareon e^{-q Ht} \underline{\mathscr{E}}_{v_*;N-3}
				+ C \epsilon e^{-q Ht}, && (|\vec{\alpha}| \leq N - 3), \label{E:newvjcommutatorL2} 
			\\	
		\| \triangle_{\mathcal{E};(\gamma_{0*}, \delta_{0*})}[\partial_{\vec{\alpha}} v_j,
		\partial (\partial_{\vec{\alpha}} v_j)] \|_{L^1}
			& \leq C \epsilon e^{-2qHt} \underline{\mathscr{E}}_{v_*;N-3}
			+ C e^{-3qHt} \underline{\mathscr{E}}_{v_*;N-3}^2, && (|\vec{\alpha}| \leq N - 3).
			\label{E:newtriangleEgammajdelta0*L1} 
	\end{align}
	
	Using \eqref{E:newunderlinemathfrakEv*firstdifferential}, \eqref{E:trianglejHNminus3}, \eqref{E:newvjcommutatorL2}, 
	and \eqref{E:newtriangleEgammajdelta0*L1}, it follows that
	
	\begin{align}  \label{E:v*energyfinaldifferentialinequality}
		\frac{d}{dt} \Big(\underline{\mathscr{E}}_{v_*;N-3}^2 \Big) 
		& \leq (2q - \eta_{0*})H \underline{\mathscr{E}}_{v_*;N-3}^2 
			+ C e^{-qHt} \underline{\mathscr{E}}_{v_*;N-3}^2
			+ C \epsilon \underline{\mathscr{E}}_{v_*;N-3}.
	\end{align}
	Using the fact that $2q - \eta_{0*} < 0,$
	we argue as in our proof of \eqref{E:rescaledg0*energybounded} to conclude that the following inequality
	holds for all $t \geq 0:$
	
	\begin{align} \label{E:v*newenergynormboundedbyepsilon}
		\underline{\mathscr{E}}_{v_{*};N-3} \leq C \epsilon.
	\end{align}
	Inequalities \eqref{E:g0jlowerconvergence} and \eqref{E:partialtg0jconvergence} now follow from
	the comparison estimate \eqref{E:v*newenergynormequivalence} and from \eqref{E:v*newenergynormboundedbyepsilon}.
	\\

\noindent \emph{Proofs of \eqref{E:improvedgjklowerconvergence},
\eqref{E:improvedgjkupperconvergence},
\eqref{E:improvede2Omegapartialgjkminusrhojklowerconvergence}, 
\eqref{E:improvede2Omegapartialgjkminusrhojkupperconvergence}, 
\eqref{E:improvedg00plusoneconvergence}, and \eqref{E:improvedpartialtg00convergence}}:

We begin by using equations \eqref{E:finalg00equation} and \eqref{E:finalhjkequation}, 
together with the identity $g^{00} + 1 = g^{00}(g_{00} + 1) + g^{0a}g_{0a}$
to derive the following equations:

\begin{align}
		\partial_t^2\big[e^{2\Omega}(g_{00} + 1) \big]& = - \omega \partial_t \big[e^{2\Omega}(g_{00} + 1)\big] 
			+ 2\dot{\omega} e^{2\Omega} (g_{00} + 1) + e^{2\Omega} \triangle'_{00}, \label{E:improvedfinalg00equation} \\
		\partial_t [e^{2\Omega}h_{jk}] & = - \omega e^{2\Omega} \partial_t h_{jk} + e^{2\Omega} \triangle'_{jk},
			\label{E:improvedfinalhjkequation}
\end{align}
where

\begin{align}
		\triangle'_{00} & = (g^{00})^{-1} \Big\lbrace -2 g^{0a} \partial_a \partial_t g_{00} -  g^{ab} \partial_a \partial_b g_{00} 
			+ 5 \omega \big[ g^{00}(g_{00} + 1) + g^{0a}g_{0a} \big] \partial_t g_{00} 
			+ 6 \omega^2 \big[ g^{00}(g_{00} + 1)^2 + g^{0a}g_{0a}(g_{00} + 1) \big]
				\Big\rbrace \label{E:triangle00prime} \\
		& \ \ \ + 2(g^{00})^{-1} \Big\lbrace \triangle_{A,00} + \triangle_{C,00}+ 
			[f(\partial \widetilde{\Phi}) - f(\partial \Phi)] 
			- (g_{00} + 1)  f(\partial \widetilde{\Phi}) - \frac{s}{s+1}(g_{00} + 1) \sigma^{s+1} \Big\rbrace, \notag \\
		 \triangle'_{jk} & = (g^{00})^{-1} \Big\lbrace - 2 g^{0a} \partial_a \partial_t h_{jk} -  g^{ab} \partial_a \partial_b 	
				h_{jk} + 3\omega \big[ g^{00}(g_{00} + 1) + g^{0a}g_{0a} \big] \partial_t h_{jk}  \Big\rbrace 
				\label{E:trianglejkprime} \\
		& \ \ \ + 2(g^{00})^{-1} \Big\lbrace e^{-2 \Omega} \triangle_{A,jk} 
			- \big[ g^{00}(g_{00} + 1) + g^{0a}g_{0a} \big] \widetilde{\sigma}^{s+1} h_{jk} 
			- 2 \omega g^{0a} \partial_{a} h_{jk} - 2e^{-2 \Omega} \sigma^s (\partial_j \Phi)(\partial_k \Phi) \notag \\ 
		& \hspace{4in}  + \frac{s}{s+1} (\widetilde{\sigma}^{s+1} - \sigma^{s+1}) h_{jk} \Big\rbrace. \notag 
\end{align}

Using the fact that $\suptotalnorm{N} \leq \epsilon,$ the improved estimates (whenever they are convenient)
\eqref{E:partialtg0jlowerHNimprovement} - \eqref{E:g0jupperHNimprovement}, \eqref{E:rescaledg0*energybounded},
and the Sobolev-Moser type inequalities in the Appendix, we derive the following inequalities: 

\begin{align}
	\| \triangle'_{00} \|_{H^{N-2}} & \leq C \epsilon e^{-2\Omega}
		+ C \epsilon e^{-qHt}\big\lbrace \| g_{00} + 1 \|_{H^{N-2}}
		+ \| \partial_t g_{00} \|_{H^{N-2}}  
		+ \| \partial_t h_{jk} \|_{H^{N-2}}\big\rbrace, 
		 \label{E:triangleprime00firstestimate} \\
	\| \triangle'_{jk} \|_{H^{N-2}} & \leq C \epsilon e^{-2\Omega} + C \epsilon e^{-qHt} \| \partial_t h_{jk} \|_{H^{N-2}}. 
		\label{E:triangleprimejkfirstestimate}
\end{align}
Since the derivation of the above inequalities is similar to many other estimates
proved in this article, we have left the tedious details up to the reader. However, we remark that
the $\epsilon e^{-qHt}  \| g_{00} + 1 \|_{H^{N-2}}$  term on the right-hand side of \eqref{E:triangleprime00firstestimate}
arises from e.g. the first term on the right-hand side of \eqref{E:triangleC00def}, that
the $\epsilon e^{-qHt}  \| \partial_t g_{00} \|_{H^{N-2}}$  term on the right-hand side of \eqref{E:triangleprime00firstestimate} arises from e.g. the first term on the right-hand side of \eqref{E:triangleA00def}, that
the $\epsilon e^{-qHt} \| \partial_t h_{jk} \|_{H^{N-2}}$ term on the right-hand side of \eqref{E:triangleprime00firstestimate}
arises from e.g. the last term on the right-hand side of \eqref{E:triangleA00def}, and that the 
$\epsilon e^{-qHt} \| \partial_t h_{jk} \|_{H^{N-2}}$ term on the right-hand side of \eqref{E:triangleprimejkfirstestimate}
arises from e.g. terms on the fourth and fifth lines of \eqref{E:triangleAjkdef}. Furthermore, we remark that we are using
$H^{N-2}$ norms because of the presence of the terms on the right-hand sides of \eqref{E:triangle00prime} - \eqref{E:trianglejkprime} that contain second derivatives, e.g. the term $\partial_a \partial_b g_{00};$ by examining e.g. the definition \eqref{E:mathfrakSMg00}, we see that the $H^{N-2}$ norm of such a term has a more favorable rate of decay than its
$H^{N-1}$ norm. We need this additional decay to deduce \eqref{E:triangleprime00firstestimate} - \eqref{E:triangleprimejkfirstestimate}.

Now in order to derive our desired inequalities, it is convenient to introduce the following 
non-negative energies $\underline{\mathscr{E}}_{e^{2\Omega}(g_{00}+1);N-2},$ $\underline{\mathscr{E}}_{\partial_t [e^{2\Omega}(g_{00}+1)];N-2},$ and $\underline{\mathscr{E}}_{e^{2 \Omega}\partial_t h_{**};N-2},$ which are defined by

\begin{subequations}
\begin{align}
	\underline{\mathscr{E}}_{e^{2\Omega}(g_{00}+1);N-2}^2 & \eqdef \sum_{|\vec{\alpha}|\leq N-2} 
		\int_{\mathbb{T}^3} \big(\partial_{\vec{\alpha}} \big(e^{2 \Omega}(g_{00} + 1)]\big)^2 \, d^3 x 
		= \| e^{2 \Omega}(g_{00} + 1) \|_{H^{N-2}}^2, \label{E:newg00energy}  \\
	\underline{\mathscr{E}}_{\partial_t [e^{2\Omega}(g_{00}+1)];N-2}^2 & \eqdef \sum_{|\vec{\alpha}|\leq N-2}  
		\int_{\mathbb{T}^3} \big(\partial_{\vec{\alpha}} \partial_t [e^{2 \Omega}(g_{00} + 1)]\big)^2 \, d^3 x
		= \| \partial_t [e^{2\Omega}(g_{00} + 1)] \|_{H^{N-2}}^2, 
		\label{E:newpartialtg00energy} \\
	\underline{\mathscr{E}}_{e^{2 \Omega}\partial_t h_{**};N-2}^2 & \eqdef  \sum_{|\vec{\alpha}|\leq N-2} \sum_{j,k=1}^3 
		\int_{\mathbb{T}^3} e^{4 \Omega}(\partial_{\vec{\alpha}} \partial_t h_{jk})^2 \, d^3 x
			= \sum_{j,k=1}^3 \| e^{2 \Omega} \partial_t h_{jk} \|_{H^{N-2}}^2.
		\label{E:newpartialthjkenergy}
\end{align}
\end{subequations}
Observe the following simple consequence of the definitions \eqref{E:newg00energy} and \eqref{E:newpartialtg00energy}:

\begin{align} \label{E:e2Omegapartialtg00newenergiesboundHNminus3}
	e^{2 \Omega} \| \partial_t g_{00} \|_{H^{N-2}} \leq C \big\lbrace \underline{\mathscr{E}}_{e^{2\Omega}(g_{00}+1);N-2} + 
		\underline{\mathscr{E}}_{\partial_t [e^{2\Omega}(g_{00}+1)];N-2} \big\rbrace.	
\end{align}

Now from \eqref{E:triangleprime00firstestimate} and \eqref{E:triangleprimejkfirstestimate}, and \eqref{E:e2Omegapartialtg00newenergiesboundHNminus3}, it follows that

\begin{align}
	e^{2\Omega} \| \triangle'_{00} \|_{H^{N-2}} & \leq C \epsilon 
		+ C \epsilon e^{-qHt} \big\lbrace \underline{\mathscr{E}}_{e^{2\Omega}(g_{00}+1);N-2} + \underline{\mathscr{E}}_{\partial_t [e^{2\Omega}(g_{00}+1)];N-2} + \underline{\mathscr{E}}_{e^{2 \Omega}\partial_t h_{**};N-2} \big\rbrace, 
		\label{E:triangleprime00HNminus3} \\
	e^{2\Omega} \| \triangle'_{jk} \|_{H^{N-2}} & \leq C \epsilon + C \epsilon e^{-qHt}\underline{\mathscr{E}}_{e^{2 \Omega}\partial_t h_{**};N-2}.
	\label{E:triangleprimejkHNminus3}
\end{align}
We may therefore use equations \eqref{E:improvedfinalg00equation} - \eqref{E:improvedfinalhjkequation},  together with \eqref{E:triangleprime00HNminus3} - \eqref{E:triangleprimejkHNminus3} to derive the following system of differential inequalities, which is valid if 
$\underline{\mathscr{E}}_{e^{2\Omega}(g_{00}+1);N-2}, \underline{\mathscr{E}}_{\partial_t [e^{2\Omega}(g_{00}+1)];N-2},
$ and $\underline{\mathscr{E}}_{e^{2 \Omega}\partial_t h_{**};N-2}$ are sufficiently small:

\begin{subequations}
\begin{align}
	\frac{d}{dt} \big(\underline{\mathscr{E}}_{e^{2\Omega}(g_{00}+1);N-2}^2 \big) & \leq 
		2 \underline{\mathscr{E}}_{e^{2\Omega}(g_{00}+1);N-2} \underline{\mathscr{E}}_{\partial_t [e^{2\Omega}(g_{00}+1)];N-2}, 
		\label{E:newg00energyfinaldifferentialinequality} \\
	\frac{d}{dt} \big(\underline{\mathscr{E}}_{\partial_t [e^{2\Omega}(g_{00}+1)];N-2}^2 \big) & \leq - 2 \omega \underline{\mathscr{E}}_{\partial_t [e^{2\Omega}(g_{00}+1)];N-2}^2
	 + 2 e^{2 \Omega} \underline{\mathscr{E}}_{\partial_t [e^{2\Omega}(g_{00}+1)];N-2} \| \triangle'_{00} \|_{H^{N-2}} 
	  \label{E:newpartialtg00energyfinaldifferentialinequality} \\
	 & \ \ + 4 |\dot{\omega}| \underline{\mathscr{E}}_{e^{2\Omega}(g_{00}+1);N-2} 
	  \underline{\mathscr{E}}_{\partial_t [e^{2\Omega}(g_{00}+1)];N-2} \notag \\
	  & \leq - 2 \omega \underline{\mathscr{E}}_{\partial_t [e^{2\Omega}(g_{00}+1)];N-2}^2 
	 	+ C e^{-qHt} \underline{\mathscr{E}}_{e^{2\Omega}(g_{00}+1);N-2}\underline{\mathscr{E}}_{\partial_t
	 	[e^{2\Omega}(g_{00}+1)];N-2} 
	 	+ C \epsilon \underline{\mathscr{E}}_{\partial_t [e^{2\Omega}(g_{00}+1)];N-2} \notag \\
	 & \ \ \ + C \epsilon e^{-qHt} \underline{\mathscr{E}}_{\partial_t [e^{2\Omega}(g_{00}+1)];N-2}
	 	\big\lbrace \underline{\mathscr{E}}_{e^{2\Omega}(g_{00}+1);N-2} + \underline{\mathscr{E}}_{\partial_t 
	 	[e^{2\Omega}(g_{00}+1)];N-2} + \underline{\mathscr{E}}_{e^{2 \Omega}\partial_t h_{**};N-2} \big\rbrace, \notag \\
	 \frac{d}{dt} \big(\underline{\mathscr{E}}_{e^{2 \Omega}\partial_t h_{**};N-2}^2 \big) & \leq - 2 \omega 
		\underline{\mathscr{E}}_{e^{2 \Omega}\partial_t h_{**};N-2}^2
	 	+ 2 e^{2 \Omega} \underline{\mathscr{E}}_{e^{2 \Omega}\partial_t h_{**};N-2} \| \triangle'_{jk} \|_{H^{N-2}} 
	 	\label{E:newpartialthjkenergyfinaldifferentialinequality} \\
	 	& \leq - 2 \omega \underline{\mathscr{E}}_{e^{2 \Omega}\partial_t h_{**};N-2}^2 + C \epsilon \underline{\mathscr{E}}_{e^{2 
	 	\Omega}\partial_t h_{**};N-2} +  C \epsilon e^{-qHt}\underline{\mathscr{E}}_{e^{2 \Omega}\partial_t h_{**};N-2}^2. \notag
\end{align}	 
\end{subequations}
We remark that \eqref{E:newg00energyfinaldifferentialinequality} follows from a simple application of the Cauchy-Schwarz inequality, after bringing the time derivative under the integral over $\mathbb{T}^3.$ Using the initial conditions $\underline{\mathscr{E}}_{e^{2\Omega}(g_{00}+1);N-2}(0) \leq C \epsilon, 
\underline{\mathscr{E}}_{\partial_t [e^{2\Omega}(g_{00}+1)];N-2}(0) \leq C \epsilon,$ 
and $\underline{\mathscr{E}}_{e^{2 \Omega}\partial_t h_{**};N-2}(0) \leq C \epsilon,$
and assuming that $\epsilon$ is sufficiently small, we apply a Gronwall-type inequality to the system \eqref{E:newg00energyfinaldifferentialinequality} - \eqref{E:newpartialthjkenergyfinaldifferentialinequality}, concluding that the following inequalities hold for all $t \geq 0:$

\begin{subequations}
\begin{align}
	\underline{\mathscr{E}}_{e^{2\Omega}(g_{00}+1);N-2} & \leq C \epsilon (1 + t), 
		\label{E:newg00energybound} \\
	\underline{\mathscr{E}}_{\partial_t [e^{2\Omega}(g_{00}+1)];N-2} & \leq C \epsilon,  
	\label{E:newpartialtg00energybound} \\
\underline{\mathscr{E}}_{e^{2 \Omega}\partial_t h_{**};N-2} & \leq C \epsilon. \label{E:newpartialthjkenergybound}
\end{align}
\end{subequations}
Inequalities \eqref{E:improvedg00plusoneconvergence} and \eqref{E:improvedpartialtg00convergence} now follow from
\eqref{E:newg00energybound} and \eqref{E:newpartialtg00energybound}. Similarly, \eqref{E:newpartialthjkenergybound} 
implies that

\begin{align} \label{E:partialthjkimprovedboundproof}
	\|\partial_t h_{jk} \|_{H^{N-2}} & \leq C \epsilon e^{-2 \Omega},
\end{align}
from which \eqref{E:improvedgjklowerconvergence} and
\eqref{E:improvede2Omegapartialgjkminusrhojklowerconvergence} easily follow.

To obtain \eqref{E:improvedgjkupperconvergence}, we first use the improved estimates
\eqref{E:partialtg0jlowerHNimprovement} - \eqref{E:g0jupperHNimprovement}, \eqref{E:rescaledg0*energybounded}, and \eqref{E:partialthjkimprovedboundproof} to upgrade the inequality \eqref{E:partialtgjkplus2omegagjkproofinequality} as
follows:

\begin{align} \label{E:improvedpartialtgupperjkplus2omegagupperjk}
	\| \partial_t g^{jk} + 2 \omega g^{jk} \|_{H^{N-2}} & \leq C \epsilon e^{-4 \Omega}.
\end{align}
\eqref{E:improvedgjkupperconvergence} then follows from \eqref{E:improvedpartialtgupperjkplus2omegagupperjk} as in our proof of \eqref{E:gjkupperconvergence}. \eqref{E:improvede2Omegapartialgjkminusrhojkupperconvergence} then follows from
\eqref{E:improvedgjkupperconvergence} and \eqref{E:improvedpartialtgupperjkplus2omegagupperjk}. 
\\

\noindent \emph{Proof of \eqref{E:improvedKjkconvergence}}:
To prove \eqref{E:improvedKjkconvergence}, we use the identity $g^{00} + 1 = \frac{1}{g_{00}}[(g_{00} + 1) - g^{0a}g_{0a}]$
and the improved estimates \eqref{E:improvedg00plusoneconvergence}, \eqref{E:improvedpartialtg00convergence},
\eqref{E:partialtg0jlowerHNimprovement} - \eqref{E:g0jupperHNimprovement}, \eqref{E:rescaledg0*energybounded}, and \eqref{E:partialthjkimprovedboundproof} to upgrade inequalities \eqref{E:Kjktriangle1} - \eqref{E:Kjktriangle2} as follows:

\begin{align}
		e^{-2 \Omega} \big\|K_{jk} + (-g^{00})^{-1/2} g^{00} \Gamma_{jk0} \big\|_{H^{N-2}} 
			& \leq C \epsilon e^{-2Ht}, \label{E:improvedKjktriangle1} \\
		e^{-2 \Omega} \Big\| (-g^{00})^{-1/2} g^{00} \Gamma_{jk0} + \frac{1}{2}\partial_t g_{jk} \Big \|_{H^{N-2}}
		 	& \leq C \epsilon (1 + t) e^{-2Ht}. \label{E:improvedKjktriangle2} 
	\end{align}
	Combining \eqref{E:improvede2Omegapartialgjkminusrhojklowerconvergence}, \eqref{E:improvedKjktriangle1}, and
	\eqref{E:improvedKjktriangle2}, we deduce \eqref{E:improvedKjkconvergence} (as in our proof of \eqref{E:Kjkconvergence}).
	\\
	
\noindent \emph{Proof of \eqref{E:improvedpartialtPhiconvergence}}:
	To prove \eqref{E:improvedpartialtPhiconvergence}, we use the improved estimates \eqref{E:improvedg00plusoneconvergence},
	\eqref{E:improvedpartialtg00convergence},
	\eqref{E:partialtg0jlowerHNimprovement} - 
	\eqref{E:g0jupperHNimprovement}, \eqref{E:rescaledg0*energybounded}, and \eqref{E:partialthjkimprovedboundproof}
	to upgrade inequality \eqref{E:ewOmegatrianglePhiHNminusone} to 
	
	\begin{align} \label{E:improvedewOmegatrianglePhiHNminusone}
		\| e^{\decayparameter \Omega} \triangle'_{\partial \Phi} \|_{H^{N-2}} \leq C \epsilon e^{-2(1 - \decayparameter)Ht}.
	\end{align}
	\eqref{E:improvedpartialtPhiconvergence} then follows from \eqref{E:improvedewOmegatrianglePhiHNminusone},
	as in our proof of \eqref{E:partialtPhiconvergence}.

\section*{Acknowledgments}
JS thanks Princeton University, where the bulk of this project was completed, for its support. The authors would like to thank Alan Rendall for offering some helpful comments, and for pointing out the reference \cite{uBaRoR1994}.

\begin{center}
	\textbf{\huge{Appendices}}
\end{center}
\setcounter{section}{0}
   \setcounter{subsection}{0}
   \setcounter{subsubsection}{0}
   \setcounter{paragraph}{0}
   \setcounter{subparagraph}{0}
   \setcounter{figure}{0}
   \setcounter{table}{0}
   \setcounter{equation}{0}
   \setcounter{theorem}{0}
   \setcounter{definition}{0}
   \setcounter{remark}{0}
   \setcounter{proposition}{0}
   \renewcommand{\thesection}{\Alph{section}}
   \renewcommand{\theequation}{\Alph{section}.\arabic{equation}}
   \renewcommand{\theproposition}{\Alph{section}-\arabic{proposition}}
   \renewcommand{\thecorollary}{\Alph{section}.\arabic{corollary}}
   \renewcommand{\thedefinition}{\Alph{section}.\arabic{definition}}
   \renewcommand{\thetheorem}{\Alph{section}.\arabic{theorem}}
   \renewcommand{\theremark}{\Alph{section}.\arabic{remark}}
   \renewcommand{\thelemma}{\Alph{section}-\arabic{lemma}}

\section{Sobolev-Moser Inequalities} \label{A:SobolevMoser}
		In this Appendix, we provide some standard Sobolev-Moser type estimates that play a fundamental role in
		our analysis of the nonlinear terms in our equations. The propositions and corollaries stated below can be proved
		using methods similar to those used in \cite[Chapter 6]{lH1997} and in \cite{sKaM1981}. The proofs given in the 
		literature are commonly based on a version of the Gagliardo-Nirenberg inequality \cite{lN1959},
		which we state as Lemma \ref{L:GN}, together with repeated use of H\"{o}lder's inequality and/or Sobolev embedding. 
		Throughout this appendix, we abbreviate $L^p=L^p(\mathbb{T}^3),$ and $H^M=H^M(\mathbb{T}^3).$ 

\begin{lemma}                                                                                           \label{L:GN}
    If $M,N$ are integers such that $0 \leq M \leq N,$ and $v$ is a function on $\mathbb{T}^3$ such that $v \in
    L^{\infty}, \|\bard^{(N)} v \|_{L^2} < \infty,$ then
    \begin{align}
        \| \bard^{(M)} v \|_{L^{2N/M}} \leq C(M,N) \| v \|_{L^{\infty}}^{1 -
        \frac{M}{N}}\|\bard^{(N)} v \|_{L^2}^{\frac{M}{N}}.
    \end{align}
\end{lemma}

\begin{proposition} \label{P:derivativesofF1FkL2}
	Let $M \geq 0$ be an integer. If $\lbrace v_a \rbrace_{1 \leq a \leq l}$ are functions such that $v_a \in
    L^{\infty}, \|\bard^{(M)} v_a \|_{L^2} < \infty$ for $1 \leq a \leq l,$ and
	$\vec{\alpha}_1, \cdots \vec{\alpha}_l$ are spatial derivative multi-indices with 
	$|\vec{\alpha}_1| + \cdots |\vec{\alpha}_l| = M,$ then
	\begin{align}
		\| (\partial_{\vec{\alpha}_1}v_1) (\partial_{\vec{\alpha}_2}v_2) \cdots (\partial_{\vec{\alpha}_l}v_l)\|_{L^2}
		& \leq C(l,M) \sum_{a=1}^l \Big( \| \bard^{(M)} v_a  \|_{L^2} \prod_{b \neq a} \|v_{b} \|_{L^{\infty}} \Big).
	\end{align}
\end{proposition}

\begin{corollary}                                                  \label{C:DifferentiatedSobolevComposition}
    Let $M \geq 1$ be an integer, let $\mathfrak{K}$ be a compact set, and let $F \in C_b^M(\mathfrak{K})$ be a 
    function. Assume that $v$ is a function such that $v(\mathbb{T}^3) \subset \mathfrak{K}$ and $ \bard v \in H^{M-1}.$
    Then $\bard (F \circ v) \in H^{M-1},$ and
    \begin{align} 														\label{E:DifferentiatedModifiedSobolevEstimate}
    	\| \bard (F \circ v) \|_{H^{M-1}} 
    		& \leq C(M) \| \bard v \|_{H^{M-1}} \sum_{l=1}^M |F^{(l)}|_{\mathfrak{K}} 
    		\| v \|_{L^{\infty}}^{l - 1}.
    \end{align}
\end{corollary}

\begin{corollary}                                                                                             \label{C:SobolevTaylor}
     Let $M \geq 1$ be an integer, let $\mathfrak{K}$ be a compact, convex set, and let $F \in C_b^M(\mathfrak{K})$ be a 
     function. Assume that $v$ is a function such that $v(\mathbb{T}^3) \subset \mathfrak{K}$ and $v - \bar{v} \in H^M,$
    where $\bar{v} \in \mathfrak{K}$ is a constant. Then $F \circ v - F \circ \bar{v} \in H^M,$ and
    \begin{align} 														\label{E:ModifiedSobolevEstimateConstantArray}
    	\|F \circ v - F \circ \bar{v} \|_{H^M} 
    		\leq C(M) \Big\lbrace |F^{(1)}|_{\mathfrak{K}}\| v - \bar{v} \|_{L^2} 
    		+ \| \bard v \|_{H^{M-1}} \sum_{l=1}^M  |F^{(l)}|_{\mathfrak{K}} 
    		\| v \|_{L^{\infty}}^{l - 1} \Big\rbrace.
    \end{align}
\end{corollary}

\begin{proposition} \label{P:F1FkLinfinityHN}
	Let $M \geq 1, l \geq 2$ be integers. Suppose that $\lbrace v_a \rbrace_{1 \leq a \leq l}$ are functions such that $v_a \in
    L^{\infty}$ for $1 \leq a \leq l,$ that $v_l \in H^M,$ and that
	$\bard v_a \in H^{M-1}$ for $1 \leq a \leq l - 1.$
	Then
	\begin{align}
		\| v_1 v_2 \cdots v_l \|_{H^M} \leq C(l,M) \Big\lbrace \| v_l \|_{H^M} \prod_{a=1}^{l-1} \| v_a \|_{L^{\infty}}  
		+ \sum_{a=1}^{l-1} \| \bard v_a \|_{H^{M-1}} \prod_{b \neq a} \| v_b \|_{L^{\infty}} \Big\rbrace.
	\end{align}	
\end{proposition}

\begin{remark}
	The significance of this proposition is that only one of the functions, namely $v_l,$ is estimated in $L^2.$
\end{remark}

\begin{proposition}                                                                             \label{P:SobolevMissingDerivativeProposition}
    Let $M \geq 1$ be an integer, let $\mathfrak{K}$ be a compact, convex set, and
    let $F \in C_b^M(\mathfrak{K})$ be a function.
    Assume that $v_1$ is a function such that $v_1(\mathbb{T}^3) \subset \mathfrak{K},$ that $\bard v_1 \in 
    L^{\infty},$ and that $\bard^{(M)} v_1 \in L^2.$ Assume that $v_2 \in L^{\infty},$ that $\bard^{(M-1)} v_2 \in L^2,$ 
    and let $\vec{\alpha}$ be a spatial derivative multi-index with with $|\vec{\alpha}| = M.$ Then 
    $\partial_{\vec{\alpha}} \left((F \circ v_1 )v_2\right) - (F \circ v_1)\partial_{\vec{\alpha}} v_2 \in L^2,$ 
    and
        
 			\begin{align}       \label{E:SobolevMissingDerivativeProposition}
      	\|\partial_{\vec{\alpha}} & \left((F \circ v_1 )v_2\right) - (F \circ v_1)\partial_{\vec{\alpha}} v_2\|_{L^2} \notag \\
        	& \leq C(M) \Big\lbrace |F^{(1)}|_{\mathfrak{K}} \|\bard v_1 \|_{L^{\infty}} \| \bard^{(M-1)} v_2 \|_{L^2} 
					+ \| v_2 \|_{L^{\infty}} \| \bard v_1 \|_{H^{M-1}} \sum_{l=1}^M |F^{(l)}|_{\mathfrak{K}} 
    			\| v_1 \|_{L^{\infty}}^{l - 1} \Big\rbrace.
       \end{align}
\end{proposition}

\begin{remark}
	The significance of this proposition is that the $M^{th}$ order derivatives of $v_2$ do not play a role in the conclusions.
\end{remark}

\bibliographystyle{siam}
\bibliography{JBib}

\bigskip

\noindent {Igor Rodnianski \\
Department of Mathematics, Princeton University \ \\
Fine Hall, Washington Road \ \\
Princeton, NJ 08544-1000, USA} \ \\
\texttt{irod@math.princeton.edu}

\medskip

\noindent {Jared Speck \\
Department of Mathematics, Princeton University\footnote{This article was
finalized while J. Speck was a postdoctoral researcher at the University of Cambridge.} \ \\
Fine Hall, Washington Road \ \\
Princeton, NJ 08544-1000, USA} \ \\
\texttt{jspeck@math.princeton.edu}

\end{document}